\pdfoutput=1 
\documentclass[a4paper,twoside,10pt,fleqn]{article}
\usepackage{helvet}

\usepackage[latin1]{inputenc}
\usepackage{amsfonts}
\usepackage{amsmath}
\usepackage{amssymb,marvosym,units}
\usepackage{amsthm}
\usepackage[cmtip,arrow, matrix, curve]{xy}
\usepackage{pb-diagram,pb-xy}

\usepackage{fontenc}
\usepackage{graphicx}
\usepackage[svgnames,rgb]{xcolor}
\usepackage{epsfig}
\usepackage{tensor}

\usepackage{hyperref}
\usepackage{zref, xkeyval, ifpdf, ifthen, calc, marginnote, pdfcomment}

\usepackage{bbm}

\hypersetup{
    pagebackref=true,
    bookmarks=true,         
    unicode=false,          
    pdftoolbar=true,        
    pdfmenubar=true,        
    pdffitwindow=true,     
    pdfstartview={FitH},    
    pdftitle={A new formulation of the Weyl $C^*$-algebra of  Loop Quantum Gravity},  
    pdfauthor={Diana Kaminski},     
    pdfsubject={},  
    pdfkeywords={Loop Quantum Gravity, Weyl algebra, dynamical sytem}, 
    pdfnewwindow=true,      
    colorlinks=false,       
    linkcolor=black,          
    pdfborder= 0 0 0,
}

\title{Algebras of Quantum Variables for Loop Quantum Gravity\\[5pt]
\textbf{II. A new formulation of the Weyl $C^*$-algebra}}
\author{Diana Kaminski\\[3pt]
kaminski@math.uni-paderborn.de\\ 
\small{Europe - Germany}}
\date{August 19, 2011}

\newcommand{\A}{\begin{large}\mathcal{A}\end{large}}
\newcommand{\Ab}{\begin{large}\bar{\mathcal{A}}\end{large}}

\newcommand{\Alg}{\begin{large}\mathfrak{A}\end{large}}

\newcommand{\Aut}{\begin{large}\mathfrak{Aut}\end{large}}

\newcommand{\CD}{\mathcal{C}}

\newcommand{\Goid}{\mathcal{G}}

\newcommand{\Gop}{\mathbbm{G}}
\newcommand{\GG}{G}

\newcommand{\HS}{\mathcal{H}}

\newcommand{\KD}{\mathcal{K}}

\newcommand{\la}{\langle}
\newcommand{\LD}{\mathcal{L}}

\newcommand{\MD}{\mathcal{M}}

\newcommand{\N}{\mathbb{N}}

\newcommand{\op}{\mathfrak{o}}

\newcommand{\PD}{\mathcal{P}}

\newcommand{\Ss}{\mathcal{S}}

\newcommand{\ra}{\rangle}

\newcommand{\QD}{\mathcal{Q}}
\newcommand{\R}{\mathbb{R}}

\newcommand{\SimGroup}{\mathfrak{G}}
\newcommand{\surf}{\mathbb{S}}

\newcommand{\WF}{\mathfrak{W}}
\newcommand{\WD}{\mathcal{W}}

\newcommand{\ZD}{\mathcal{Z}}

\newcommand{\Zop}{\mathcal{Z}(\Gop_{\breve S,\gamma})}

\newcommand{\ho}{\mathfrak{h}}

\newcommand{\go}{\mathfrak{g}}

\newcommand{\Go}{\mathfrak{G}}

\DeclareMathOperator{\Act}{Act}

\DeclareMathOperator{\adm}{A}

\DeclareMathOperator{\dif}{d}
\DeclareMathOperator{\diff}{surf}

\DeclareMathOperator{\Diff}{Diff}

\DeclareMathOperator{\Hol}{Hol}
\DeclareMathOperator{\Hom}{Hom}

\DeclareMathOperator{\id}{id}

\DeclareMathOperator{\loc}{loc}

\DeclareMathOperator{\Map}{Map}
\DeclareMathOperator{\Mor}{Mor}

\DeclareMathOperator{\ori}{or}

\DeclareMathOperator{\Rep}{Rep}

\newcommand{\gp}{{\gamma^\prime}}
\newcommand{\gpi}{{\gamma^\prime_i}}
\newcommand{\gpj}{{\gamma^\prime_j}}
\newcommand{\gpk}{{\gamma^\prime_K}}

\newcommand{\gppi}{{\gamma^{\prime\prime}_i}}
\newcommand{\gppj}{{\gamma^{\prime\prime}_j}}
\newcommand{\gppje}{{\gamma^{\prime\prime}_{j+1}}}
\newcommand{\gppl}{{\gamma^{\prime\prime}_l}}
\newcommand{\gpe}{{\gamma^\prime_1}}
\newcommand{\gpz}{{\gamma^\prime_2}}
\newcommand{\gpm}{{\gamma^\prime_M}}
\newcommand{\gppe}{{\gamma^{\prime\prime}_1}}

\newcommand{\gppm}{{\gamma^{\prime\prime}_M}}
\newcommand{\gpppe}{{\gamma^{\prime\prime\prime}_1}}

\newcommand{\gpppm}{{\gamma^{\prime\prime\prime}_M}}
\newcommand{\gpppn}{{\gamma^{\prime\prime\prime}_N}}
\newcommand{\tg}{{\tilde\gamma}}

\newcommand{\gpp}{\gamma^{\prime\prime}}

\newcommand{\Gp}{{\Gamma^\prime}}

\newcommand{\Gpp}{\Gamma^{\prime\prime}}
\newcommand{\Gppp}{\Gamma^{\prime\prime\prime}}

\newcommand{\idf}{\mathbbm{1}}

\newcommand{\bra}{[}
\newcommand{\ket}{]}

\newcommand{\beq}{\begin{equation}\begin{aligned}}
\newcommand{\beqs}{\begin{equation*}\begin{aligned}}
\newcommand{\be}{\begin{flalign}}
\newcommand{\bes}{\begin{equation*}}
\newcommand{\eq}{\end{aligned}\end{equation}}
\newcommand{\eqs}{\end{aligned}\end{equation*}}
\newcommand{\ee}{\end{flalign}}
\newcommand{\ees}{\end{equation}}

\newcommand{\limi}{\underset{i\rightarrow\infty}{\underrightarrow{\lim}}}

\newtheorem{theo}{Theorem }
\newtheorem{lem}[theo]{Lemma}
\newtheorem{rem}[theo]{Remark}
\newtheorem{prop}[theo]{Proposition}
\newtheorem{cor}[theo]{Corollary}

\newtheorem{defi}[theo]{Definition}

\newenvironment{proofs}[1][Proof ]{\noindent\textbf{#1}: }{\ \begin{flushright}
                                                                         \rule{0.5em}{0.5em}
                                                                        \end{flushright}}
\newenvironment{proofo}[1][Proof]{\noindent\textbf{#1} }{\ \begin{flushright}
                                                                         \rule{0.5em}{0.5em}
                                                                        \end{flushright}}
\newcounter{exa}[section]
 \newenvironment{exa}{\refstepcounter{exa}
  \textbf{Example} \thesection.\arabic{exa}: }{ {\begin{flushright}
                                                                         \rule{0.2em}{0.2em}
                                                                        \end{flushright}}}

\newcounter{problem}[subsection]
 \newenvironment{problem}{\refstepcounter{problem}
  \textbf{Problem} \thesection.\arabic{problem}: }{{\begin{flushright}
                                                                         \rule{0.2em}{0.2em}
                                                                        \end{flushright}}}
\newcommand{\GGi}{\xymatrix{
  \Goid_1  \ar@<-2pt>[r] \ar@<2pt>[r] &  \Goid^0_1    \\
}}
\newcommand{\GGii}{\xymatrix{
  \Goid_2  \ar@<-2pt>[r] \ar@<2pt>[r] &  \Goid^0_2    \\
}}
\newcommand{\GGm}{\xymatrix{
  \Goid  \ar@<-1pt>[r]^{s} \ar@<1pt>[r]_{t} &  \Goid^0    \\
}}
\newcommand{\GGim}{\xymatrix{
  \Goid_1  \ar@<-1pt>[r]^{s_1} \ar@<1pt>[r]_{t_1} &  \Goid^0_1    \\
}}
\newcommand{\GGiim}{\xymatrix{
  \Goid_2  \ar@<-1pt>[r]^{s_2} \ar@<1pt>[r]_{t_2} &  \Goid^0_2    \\
}}

\newcommand{\PGm}{\xymatrix{
  \PD  \ar@<-1pt>[r]^{s} \ar@<1pt>[r]_{t} &  \Sigma    \\
}}
\newcommand{\PG}{\PD\rightrightarrows\Sigma}
\newcommand{\PGs}{\PD\Sigma\rightrightarrows\Sigma}
\newcommand{\PGoS}{\PD\rightrightarrows\Sigma}
\newcommand{\fPGm}{\xymatrix{
  \PD_\Gamma  \ar@<-1pt>[r]^{s} \ar@<1pt>[r]_{t} &  V_\Gamma    \\
}}
\newcommand{\PGsm}{\xymatrix{
  \PD\Sigma \ar@<-1pt>[r]^{s_{\PD\Sigma}} \ar@<1pt>[r]_{t_{\PD\Sigma}} &  \Sigma   \\
}}
\newcommand{\fPGms}{\xymatrix{
  \PD^s_\Gamma  \ar@<-1pt>[r]^{s} \ar@<1pt>[r]_{t} &  V_\Gamma    \\
}}

\newcommand{\fPG}{\PD_\Gamma\Sigma \rightrightarrows V_\Gamma
}

\newcommand{\fPSGm}{\xymatrix{
  \PD_\Gamma\Sigma  \ar@<-1pt>[r]^{s} \ar@<1pt>[r]_{t} &  V_\Gamma    \\
}}
\newcommand{\fPSG}{\xymatrix{
  \PD_\Gamma\Sigma  \ar@<-2pt>[r] \ar@<2pt>[r] &  V_\Gamma    \\
}}
\newcommand{\fgHGm}{\xymatrix{
  H(\Gamma)  \ar@<-1pt>[r]^/0.3em/{\hat s_H} \ar@<1pt>[r]_/0.3em/{\hat t_H} &  V_\Gamma    \\
}}

\newcommand{\fHGm}{\xymatrix{
  H_\Gamma  \ar@<-1pt>[r]^/0.3em/{\hat s_H} \ar@<1pt>[r]_/0.3em/{\hat t_H} &  V_\Gamma    \\
}}

\newcommand{\fGGm}{\xymatrix{
  \G^G_\Gamma  \ar@<-1pt>[r]^/0.3em/{s_P} \ar@<1pt>[r]_/0.3em/{t_P} &  V_\Gamma    \\
}}
\newcommand{\fGHm}{\xymatrix{
  \G^H_\Gamma  \ar@<-1pt>[r]^/0.3em/{s_P} \ar@<1pt>[r]_/0.3em/{t_P} &  V_\Gamma    \\
}}

\newcounter{count}
\newcommand{\citetableB}{\cite[table 11.2]{KaminskiPHD}}

\newcommand{\refapprepgroupalg}{Appendix}

\begin{document}
\maketitle
\begin{abstract}\noindent In this article a new formulation of the Weyl $C^*$-algebra, which has been invented by Fleischhack \cite{Fleischhack06}, in terms of $C^*$-dynamical systems is presented. The quantum configuration variables are given by the holonomies along paths in a graph. Functions depending on these quantum variables form the analytic holonomy $C^*$-algebra. Each classical flux variable is quantised as an element of a flux group associated to a certain surface set and a graph. The quantised spatial diffeomorphisms are elements of the group of bisections of a finite graph system. Then different actions of the flux group associated to surfaces and the group of bisections on the analytic holonomy $C^*$-algebra are studied. The Weyl $C^*$-algebra for surfaces is generated by unitary operators, which implements the group-valued quantum flux operators, and certain functions depending on holonomies along paths that satisfy canonical commutation relations. Furthermore there is a unique pure state on the commutative Weyl $C^*$-algebra for surfaces, which is a path- or graph-diffeomorphism invariant.
\end{abstract}

\thispagestyle{plain}
\pdfbookmark[0]{\contentsname}{toc}
\tableofcontents

\section{Introduction}
\subsection*{The LQG-viewpoint }

In LQG the algebra of holonomy variables has been introduced by Ashtekar and Lewandowski \cite{AshLew93}. The analytic holonomy algebra is given by a commutative $C^*$-algebra $C(\Ab)$, which is an inductive limit of a family $\{C(\Ab_\Gamma)\}$ of commutative $C^*$-algebras associated to graphs. The inductive limit of $C^*$-algebras corresponds to a projective limit of a family $\{\Ab_\Gamma\}$ of configuration spaces. Due to the Tychnov-theorem the inductive limit space $\Ab$, which is constructed from the compact Hausdorff spaces, is a compact Hausdorff space, too. Each configuration space $\Ab_\Gamma$ is identified with $G^{\vert\Gamma\vert}$ where $G$ is the structure group of a principal fibre bundle. Usually this group is chosen to be a compact Lie group. On the configuration space $\Ab_\Gamma$ associated to each graph $\Gamma$ there exists a Haar measure $\mu_\Gamma$. The consistent family $\{\mu_\Gamma\}$ of measures defines a measure $\mu_{AL}$ on $\Ab$. Homeomorphisms on the compact Hausdorff space $\Ab$ leaving the measure $\mu_{AL}$ invariant corresponds one-to-one to unitary operators $U(g)$ for elements $g$, which are contained in the structure group $G$. These unitary operators implement the fluxes associated to a surface $S$. In particular, measure preserving transformations associated to a graph $\Gamma$ define $G^{\vert\Gamma\vert}$-invariant states on the $C^*$-algebra $C(\Ab_\Gamma)$. For a detailed investigation of this construction in the context of compact Hausdorff spaces and measures refer to Marolf and Mour\~{a}o \cite{MarolfMourao95} for graphs containing only analytic loops and Fleischhack \cite{DissFleisch} for general index sets. A study of the interplay of the projective structure of the configuration space and the inductive structure of the $C^*$-algebra is given in the article of Ashtekar and Lewandowski \cite{AshLew94}, Fleischhack \cite{Fleischhack07}, Velhinho \cite{Velhinho02,Velhinho04}. 

Although several other diffeomorphism invariant states on $C(\Ab)$ are available due to the work of Baez \cite{Baez93,Baez94}, or, Ashtekar and Lewandowski \cite{AshLew94,AshLewDG95}, only states which are $G^{\vert\Gamma\vert}$-invariant will allow to extend the quantum algebra by the flux operators for a surface $S$. This question has been analysed for example by Sahlmann in \cite{Sahlmann02}. In the context of Weyl algebras constructed from holonomies and quantum flux operators, which are exponentiated Lie algebra-valued operators, the first attempts are due to Sahlmann and Thiemann in \cite{SahlThiemSuper03}. The Weyl algebra of holonomies and (exponentiated) quantum fluxes, which are introduced by particular pull-backs of homeomorphisms on the configuration space, has been constructed by Fleischhack in \cite{Fleischhack06}. The developement of the Weyl algebra can be related to transformation groups associated to a flux group and the configuration space. First attempts in this direction has been presented by Velhinho in \cite{Velhinho08}. The irreducibility of the Weyl $C^*$-algebra has been studied first by Sahlmann and Thiemann in \cite{SahlThiem03}. Fleischhack has proved in \cite{Fleischhack06} irreducibility and under some technical assumptions that there is only one irreducible and diffeomorphism-invariant representation of his Weyl $C^*$-algebra on the Ashtekar-Lewandowski Hilbert space $\HS_{AL}$. For a short overview refer to Fleischhack \cite{Fleischhack071}. In comparison with the Weyl algebra presented in the project \textit{AQV}, Fleischhack has considered more general stratified objects, instead of $D-1$-dimensional surfaces in a $D$-dimensional manifold only, for the construction of his Weyl $C^*$-algebra. 

\subsection*{The operator-algebraic viewpoint}
\subsubsection*{The quantum configuration variables: holonomies along paths}

The fundamental geometric objects for a theory of Loop Quantum Gravity are (semi-) analytic paths and loops that form graphs. In section \ref{quantum variables} the basic quantum variables derived from these geometric objects will be introduced. A short overview will be presented in this section.

A \textit{graph} contains a finite set of independent edges. A set of edges is called \textit{independent} if the edges only intersect each other in the source or target vertices. A \textit{finite groupoid} is a finite set of paths equipped with a groupoid structure. The \textit{finite graph system} associated to a graph $\Gamma$ is given by all subgraphs of $\Gamma$. A \textit{finite path groupoid} associated to the graph $\Gamma$ is generated by all compositions of elements or their inverse elements of the set of edges that defines the graph $\Gamma$. Note that an element of a finite path groupoid is not necessarily an independent path. Clearly, for all these objects there exists an ordering such that 
\begin{enumerate}
 \item\label{indfam1} an \textit{inductive family of graphs}
 \item\label{indfam2} an \textit{inductive family of finite path groupoids} and
 \item\label{indfam3} an \textit{inductive family of finite graph systems} can be studied.
\end{enumerate} 

Furthermore, a \textit{holonomy map} is a groupoid morphism from the path groupoid to the compact structure group $G$. If a graph is considered, then the holonomy map maps each edge of the graph to an element of the structure group $G$. For generality it is assumed that $G$ is a compact  group. In \ref{subsec holmapsfinpath} two ways of an identification of the holonomy map evaluated for a subgraph of $\Gamma$ with elements in $G^{\vert\Gamma\vert}$ will be presented. One distinguishes between the \textit{natural} or \textit{the non-standard identification of the configuration space} $\Ab_\Gamma$ with $G^{\vert\Gamma\vert}$. Recall that a subgraph of $\Gamma$ is a set of independent paths, which are generated by the edges of the graph $\Gamma$. In the natural identification these paths are decomposed into the edges, which define the graph $\Gamma$. In the non-standard identification only graphs that contain only non-composable paths are considered. In both cases the holonomy maps evaluated on a subgraph $\Gp$ of $\Gamma$ are elements of $G^{M}$, where $M$ is the number of paths in $\Gp$. One obtains a product group $G^M$ for $M\leq \vert\Gamma\vert$, and which is embedded into $G^{\vert\Gamma\vert}$ by $G^M\times\{e_G\}\times ...\times \{e_G\}$. Hence, in both cases the holonomy evaluated on a subgraph of a graph $\Gamma$ is an element of $G^{\vert\Gamma\vert}$. In LQG \cite{AshLew93,AshLew94,Thiembook07} a holonomy map evaluated at the graph $\Gamma$ is an element of $G^{\vert\Gamma\vert}$, too.

The \textit{analytic holonomy $C^*$-algebra restricted to a finite graph system associated to a graph} is given by the commutative unital $C^*$-algebra $C(\Ab_\Gamma)$ of continuous functions on the configuration space $\Ab_\Gamma$ vanishing at infinity and supremum norm.

In the project \textit{AQV} the inductive limit $C^*$-algebra is constructed from an inductive family of $C^*$-algebras, which depend on finite graph systems. The reason is the following: Consider \textit{graph-diffeomorphisms} of the finite graph system associated to a graph $\Gamma$. These objects are pairs of maps and will be presented in more detail in section \ref{subsubsec bisections}. For short such a pair consists of a bijective map from vertices to vertices, which are situated in the manifold $\Sigma$, and a map that maps subgraphs to subgraphs of $\Gamma$. Then there are actions of these graph-diffeomorphisms on the analytic holonomy $C^*$-algebra restricted to a finite graph system associated to the graph $\Gamma$. There is no well-defined action of these graph-diffeomorphisms on the analytic holonomy $C^*$-algebra restricted to a fixed graph in general. This can be verified as follows. Assume that $\Gamma:=\{\gamma_1,\gamma_2,\gamma_3\}$ is a graph and $\Gp:=\{\gamma_1\}$, $\Gpp:=\{\gamma_1\circ\gamma_3\}$ are subgraphs of $\Gamma$ . Then consider a graph-diffeomorphism $(\varphi,\Phi)$ such that $\Phi(\Gp)=\Gpp$. Now the action $\zeta_{(\varphi,\Phi)}$ on the analytic holonomy $C^*$-algebra restricted to the graph $\Gamma$, which is defined by 
\beqs (\zeta_{(\varphi,\Phi)}f_\Gamma)(\ho_\Gamma(\Gamma))=f_{\Phi(\Gamma)}(\ho_{\Phi(\Gamma)}(\Phi(\Gamma)))= f_{\Gppp}(\ho_{\Gppp}(\Gppp))
\eqs whenever $\Phi(\Gamma)=\Gppp=\{\gamma_1\circ\gamma_3,\gamma_2,\gamma_3\}$ is not well-defined. The reason is: $\Gppp$ is not a graph. If $\Phi(\Gamma)$ is a subgraph of $\Gamma$, then in particluar $f_{\Phi(\Gamma)}$ is an element of the analytic holonomy $C^*$-algebra restricted to the subgraph $\Phi(\Gamma)$. The analytic holonomy $C^*$-algebra restricted to every subgraph of $\Gamma$ is a $C^*$-subalgebra of the analytic holonomy $C^*$-algebra restricted to the graph $\Gamma$. Hence, the last $C^*$-algebra is in particular a $C^*$-algebra, which is characterised by the finite graph system associated  to $\Gamma$. An action of graph-diffeomorphisms is an automorphism of the analytic holonomy $C^*$-algebra restricted to finite graph system associated  to $\Gamma$. Summarising, the concepts of the limit of $C^*$-algebras restricted to finite graph systems, and actions of graph-diffeomorphisms on the holonomy $C^*$-algebra restricted to finite graph systems engage with each other.

Finally note that the inductive limit $C^*$-algebra of the inductive family of $C^*$-algebras $\{C(\Ab_\Gamma),\beta_{\Gamma,\Gp}\}$ defines the \textit{projective limit configuration space} $\Ab$. The inductive limit $C^*$-algebra $C(\Ab)$ is called the \textit{analytic holonomy $C^*$-algebra} in the project \textit{AQV}.

The idea of using families of graph systems is influenced by the work of Giesel and Thiemann \cite{ThiemGiesel} in the LQG framework. They use particular cubic graphs instead of sets of paths in a groupoid and their inductive limit is constructed from families of cubic graph systems. 
In the project \textit{AQV} the \textit{inductive limit Hilbert space} $\HS_\infty$ will be derived from the natural or non-standard identified configuration spaces, the Haar measure on the structure group $G$ and an inductive limit of finite graph systems. It will be assumed that the inductive limit graph system only contains a countable set of subgraphs of an inductive limit graph $\Gamma_\infty$. This is contrary to the Hilbert space used in LQG literature \cite{Thiembook07}, which is the Ashtekar-Lewandowski Hilbert space $\HS_{\text{AL}}$. The Hilbert space $\HS_{\text{AL}}$ is manifestly non-separable, since the limit is taken over all sets of paths in $\Sigma$ and, hence over an infinite and uncountable set of all graphs. Clearly, the Hilbert space $\HS_\infty$ is constructed by using certain identification of the configuration space and the countable set of subgraph. In this simplified formulation some important aspects of the theory can be studied. It is possible to generalise partly the results for the Ashtekar-Lewandowski Hilbert space. 

The \textit{classical configuration space} in the context of LQG and Ashtekar variables is the space of smooth connections $\breve\A_s$ on an arbitrary principal fibre bundle $P(\Sigma,G)$. In this project the quantum operator $\QD(A)$ of the infinitesimal connection $A$ is given by the holonomy $\ho$ along a path $\gamma$. The operator $\QD(A)$ is represented as a multiplication operator on the inductive limit Hilbert space $\HS_\infty$.

\subsubsection*{The quantum momentum variables: group-valued flux operators}\label{subsec fluxop}

In the project \textit{AQV} the quantum operator $\QD( E^i)$ of the classical fluxes $E^i$ is either a group- or Lie algebra-valued operator, which depend on a surface $S$ and a path $\gamma$ or a graph $\Gamma$.  The idea of this definition is the following: Consider a surface $S$ and a path $\gamma$ that intersets each other in the source vertex of $\gamma$ and the path lies below the orientated surface $S$. Let $G$ be compact  group. The \textit{group-valued quantum flux operator} $\rho_S(\gamma)$ is given by the value of a map $\rho_S: P\Sigma\rightarrow G$ evaluated for a path $\gamma$ in the set $P\Sigma$ of paths in $\Sigma$. This definition does not coincide with the usual definition presented in LQG literature completely.

In general the idea is to obtain algebras, which are generated by 
\begin{enumerate}
 \item\label{item LG1} the group-valued quantum flux operators and the holonomies along paths in a graph, or
\item\label{item LG2} the group-valued quantum flux operators and certain functions depending on holonomies along paths in a graph, or
\end{enumerate}
In this work an algebra derived from the operators given in \ref{item LG2} satisfying some canonical commutator relations will be presented in section \ref{sec analholalg}.

Until now, a suitable set of surfaces in $\Sigma$ and a path $\gamma$ in the finite path groupoid $\PD_\Gamma\Sigma$ are fixed. For a general situation the following maps are studied in section \ref{subsec fluxdef} and \ref{subsec admfluxdef}:
\begin{enumerate}
 \item\label{maps3} a certain map $\rho_S:\PD_\Gamma\Sigma\rightarrow G$ 
\item\label{maps4} a cetrain map $\rho_S:\PD_\Gamma\Sigma\rightarrow \ZD$, where $\ZD$ denotes the center of the group $G$, and
\item\label{maps5} a certain map $\varrho:\PD_\Gamma\Sigma\rightarrow G$ and this map $\varrho$ is called \textit{admissible} in analogy to Fleischhack \cite{Fleischhack06}.
\end{enumerate}
Then the maps $\rho_S$ given by \ref{maps3} (or \ref{maps4}) define a group, which depend on the fixed path $\gamma$ and a suitable fixed surface set $\breve S$. Note that the surface set always contains at least one surface in $\Sigma$. This group is called \textit{flux group $\bar G_{\breve S,\gamma}$ associated to a surface set $\breve S$ and a path $\gamma$}. Clearly, for each suitable surface set there exist a flux group associated to this surface set. The maps  of the form $\varrho$ given by \ref{maps5} are used to define a more complicated structure. Furthermore, this concept generalises to holonomies of a graph $\Gamma$, which are maps from graphs to products of the structure group $G$. Then for example the \textit{flux group $\bar G_{\breve S,\Gamma}$ associated to a surface set and a graph} exists. 

Now, for the group-valued quantum flux operators different actions on the configuration space will be explicitly considered in section \ref{subsec dynsysfluxgroup}. In particular the left, right and inner actions are studied independently from each other and are denoted by $L$,$R$ or $I$. Furthermore, only the maps \ref{maps4} and \ref{maps5} define groupoid morphisms by composition of the action $L$ (or $R$, or $I$) and the holonomy map. For an overview about which maps define groupoid morphisms consider \citetableB. Note that using admissible maps (maps of the form \ref{maps5}) particular morphisms are defined. These morphisms are called \textit{equivalent groupoid morphisms} in analogy to Mackenzie \cite{Mack05} and are related to gauge transformations on the configuration space. The flux groups constructed from the maps \ref{maps3} and \ref{maps4}, the analytic holonomy $C^*$-algebra $C(\Ab_\Gamma)$ and the actions $L$, $R$ or $I$ define $C^*$-dynamical systems. If admissible maps are taken into account, the $C^*$-dynamical systems are very complicated. 

In general the parameter group of automorphism, which is defined from arbitrary group-valued quantum flux operators $\rho_S(\gamma)$ for every surface $S$ and a fixed path $\gamma$ to the group of automorphisms, i.e. $\rho_S(\gamma)\mapsto\alpha(\rho_S(\gamma))\in\Aut(C(\Ab_\gamma))$, does not define a group homomorphism to the group of automorphisms in $C(\Ab_\gamma)$. This is only true for certain group-valued quantum flux operators, which form a flux group associated to a certain surface set. Therefore, $C^*$-dynamical systems are defined by the analytic holonomy $C^*$-algebra $C(\Ab_\Gamma)$ restricted to the finite graph system and actions of the flux group associated to a certain surface set and the graph $\Gamma$ on this $C^*$-algebra. 

Furthermore, the analytic holonomy $C^*$-algebra can be restricted to certain subgraphs of a graph $\Gamma$. Therefore, the following object is important.
A \textit{finite orientation preserved graph system} is a set of certain subgraphs of a graph $\Gamma$ such that all paths in a subgraph are generated by compositions of the edges that generate the graph $\Gamma$. Note that in this definition the composition of edges and inverses of this edges are excluded. Then clearly there is an action of the flux group associated to the graph $\Gamma$ and a surface set on the analytic holonomy $C^*$-algebra restricted to the finite orientation preserved graph system $\PD_\Gamma^{\op}$. Furthermore, there is an action of the flux group associated to every subgraph of the finite orientation preserved graph system $\PD_\Gamma^{\op}$ and a surface set on the analytic holonomy $C^*$-algebra restricted to a finite orientation preserved graph system. There is a set of exceptional $C^*$-dynamical systems, which is defined by these automorphisms of the flux groups associated to suitable surface sets and graphs on the analytic holonomy algebras restricted to finite orientation preserved graph systems. The restriction to orientation preserved subgraphs is necessary to obtain either a purely left or right action of the flux group associated to a fixed surface set and subgraphs of a particular graph system on the holonomy $C^*$-algebra restricted to suitable graph systems. 

The Gelfand-Na\u{\i}mark theorem implies that there is an isomorphism between commutative $C^*$-algebras and continuous function algebras on configuration spaces. If other in particular non-abelian $C^*$-algebras are studied, then automorphisms of the algebras do not correspond to certain homeomorphisms on the configuration spaces. More generally, covariant representations of the $C^*$-dynamical systems replace the construction of Fleischhack. A covariant representation is a pair of maps, which is given by a representation of the $C^*$-algebra on the Hilbert space and a unitary representation of the flux group, and these maps satisfy a certain canonical commutator relation. In this project the \textit{Weyl $C^*$-algebra for surfaces} is constructed from all $C^*$-dynamical systems, which contains all actions of the flux groups associated to all different surface sets on the analytic holonomy $C^*$-algebra. In particular an element of the \textit{Weyl algebra of a surface set $\breve S$ restricted to a finite graph system $\PD_\Gamma$} is for example of the form
\beqs &\sum_{l=1}^L\idf_\Gamma U_{S_1}(\rho^l_{S,\Gamma}(\Gamma)) + \sum_{k=1}^K\sum_{i=1}^Mf^k_{\Gamma}U_{S_2}(\rho^i_{S,\Gamma}(\Gamma))  + \sum_{k=1}^K\sum_{i=1}^MU_{S_3}(\rho^i_{S,\Gamma}(\Gamma)) f^l_{\Gamma}U_{S_3}(\rho^i_{S,\Gamma}(\Gamma))^*
+\sum_{p=1}^Pf^p_{\Gamma}
\eqs whenever $f^k_{\Gamma},f^l_{\Gamma},f^p_{\Gamma}\in C(\Ab_\Gamma)$, $U_{S_i}\in \Rep(\bar G_{\breve S,\Gamma},\KD(\HS_\Gamma))$. The notion $U_{S_i}\in \Rep(\bar G_{\breve S,\Gamma},\KD(\HS_\Gamma))$ means that the unitary operators are represented on the $C^*$-algebra $\KD(\HS_\Gamma)$ of compact operators on the Hilbert space $\HS_\Gamma$. Furthermore, the unitaries and products of these unitaries, which satisfy the canonical commutator relation, are called \textit{Weyl elements} in this project. 
\subsubsection*{The quantum spatial diffeomorphisms}\label{subsec quandiffeo}

In the project of \textit{Algebras of Quantum Variables in LQG} the classical spatial diffeomorphisms are replaced by new quantum diffeomorphism operators. The classical diffeomorphisms are certain diffeomorphisms in the spatial hypersurface $\Sigma$.  In Mackenzie \cite{Mack05} a concept of translations in a general Lie groupoid has been presented. The ideas are used in section \ref{subsubsec bisections} for a redefinition of the classical diffeomorphisms. The new operators are called bisections. The idea of the definition of a bisection is presented in the next paragraph.

In the theory of groupoids the following object is often used: the groupoid isomorphism in a path groupoid, which consists of the classical diffeomorphism in $\Sigma$ and an additional bijective map from paths to paths in the path groupoid over $\Sigma$. This pair of maps is called the \textit{path-diffeomorphisms of a path groupoid}. The path-diffeomorphisms extend the notion of  graphomorphisms, which have been introduced by Fleischhack \cite{Fleischhack06}. There is only a slight difference betweeen these objects: A graphomorphism is a map from $\Sigma$ to $\Sigma$ that preserves additionally the path groupoid structure, whereas a path-diffeomorphism is a pair of maps. In particular \textit{finite path-diffeomorphisms} are given by a pair of maps, which contains a map that maps paths to paths in a finite path groupoid $\PD_\Gamma\Sigma$ and a bijective map that maps vertices to vertices of the vertex set of the graph $\Gamma$. Moreover, since graph systems are used in this project, a pair of maps that contains a map, which maps subgraphs to subgraphs, plays a fundamental role and is called \textit{finite graph-diffeomorphism}. Graphomorphisms define in particular groupoid isomorphisms and, hence, they transform non-trivial paths to non-trivial paths. To define maps that transform a trivial path at a vertex in $\Sigma$ to a non-trivial path other objects have to be considered. Translations in a finite path groupoid are naturally given by adding or deleting edges, which generate the graph $\Gamma$. One distinguishes between three translations in a path groupoid. One is given by adding a path $\gamma$ to a path $\theta$ at the source vertex $s(\theta)$ of the path $\theta$. The other case is given by composition of a path $\gamma$ to a path $\theta$ at the target vertex $t(\theta)$ of the path $\theta$. Finally, two paths can be composed with a path at the source and target vertices simultaneously. Hence, there is a natural map from the set of vertices to the set of paths, which is called a \textit{bisection of a finite path-groupoid}. For such a bisection $\sigma$ the map $t\circ\sigma$ is assumed to be bijective, where $t$ denotes the target map of the finite path groupoid. In the definition of a \textit{bisection of a path groupoid} the map $t\circ\sigma$ is required to be a diffeomorphism from $\Sigma$ to $\Sigma$ and the map $\sigma$ maps vertices to paths in a path groupoid.  Furthermore a \textit{right-translation $R_\sigma$ of a bisection $\sigma$} is a map that composes a path $\gamma$ with the path $\sigma(t(\gamma))$, i.e. $R_\sigma(\gamma)=\gamma\circ\sigma(t(\gamma))$. Furthermore a \textit{left-translation $L_\sigma$} and an \textit{inner-translation $I_\sigma$ of a bisection $\sigma$} can be defined similarly. The pair consisiting of the composition $t\circ\sigma$ of the bisection and the target map and the right translation $R_\sigma$ define in general no groupoid isomorphism. Nevertheless there are particular translations of suitable bisections that define path-diffeomorphisms. There is no doubt that the notion of a bisection can be generalised to a \textit{bisection of a path groupoid} or a \textit{bisection of a finite graph system}.
Moreover, the bisections of a path groupoid form a group and there is a group homomorphism between this group and the group of diffeomorphisms in $\Sigma$. Moreover, the bisections of a finite path groupoid or a finite graph system equipped with a sophisticated group multiplication form groups, too. Finally, a quantum diffeomorphism is assumed to be an element of the group of bisections of a path groupoid, a finite path groupoid or a finite graph system. 

Now, actions of the group of bisections on the analytic holonomy $C^*$-algebra restricted to a finite graph system will be used in section \ref{subsec dynsysfluxgroup2} to construct $C^*$-dynamical systems. If the group $\mathfrak{B}(\PD_\Gamma)$ of bisections of a finite graph system $\PD_\Gamma$ is considered, then the right-, left- or inner-translation of the bisections define three different $C^*$-dynamical systems. For example, there is a $C^*$-dynamical system $(C(\Ab_\Gamma),\mathfrak{B}(\PD_\Gamma),\zeta)$, where the action $\zeta$ is defined by the right-translation of the bisections. For each $C^*$-dynamical system there exists a covariant representation on the Hilbert space $\HS_\Gamma$.   Hence, the right- , left- or inner-translation of the bisections define unitary operators on the Hilbert space $\HS_\Gamma$ associated to a graph. The main advantage of translations of bisections is that they define graph changing operators. In particular these maps transfrom subgraphs into subgraphs of a graph $\Gamma$ such that the number of edges of the subgraphs can change.

Both actions, which are the action of the group of bisections of a finite graph system and the action of the flux group on the configuration space, lead to automorphisms on the analytic holonomy $C^*$-algebra. A comparison of the actions can be found in \citetableB. Similarly to actions of the flux group, the actions of the group of bisections composed with holonomy maps do not define groupoid morphisms in general. This causes no problems, since the configuration space restricted to a finite graph system $\PD_\Gamma$ is identified (naturally or in non-standard way) with $G^{\vert\Gamma\vert}$ and the right-, left- or inner-translation in the finite path groupoid transfer to \textit{right-translation $R_\sigma$, left-translation $L_\sigma$} or \textit{inner-translation $I_\sigma$ in the groupoid $G$ over $\{e_G\}$}. Finally, notice that only actions of certain bisections preserve the flux operators associated to a surface $S$. For example consider the bisection $\sigma$ of a path groupoid and recall the diffeomorphism $t\circ\sigma$. Then for example the diffeomorphism $t\circ\sigma$ is required to preserve the surface $S$. This particular bisection is called the \textit{surface-preserving bisection of a path groupoid}. There exists a similar description for a \textit{surface-preserving bisection for a finite path groupoid or a finite graph system}. Then the concept can be extended to bisections of a finite graph system that map surfaces to surfaces in a certain surface set and preserve the orientation of the surfaces with respect to the transformed subgraph. In this situation the bisections are called \textit{surface-orientation-preserving bisections for a finite graph system} and they form a subgroup of the group of bisection of a finite graph system.
Finally both actions on the analytic holonomy $C^*$-algebra restricted to a finite graph system:
\begin{enumerate}
\item the action of the group of surface-orientation-preserving bisections for a finite graph system and
\item the action of the center of the flux group associated to a surface set 
\end{enumerate} commute. In analogy to the surface-orientation-preserving bisections of a finite graph system the \textit{surface-orientation-preserving graph-diffeomorphisms} can be constructed.

Finally there is an action of bisections of the path groupoid $\PD$ over $\Sigma$ or the inductive limit graph system $\PD_{\Gamma_\infty}$ on the analytic holonomy $C^*$-algebra $C(\Ab)$. This automorphism is not point-norm continuous. Consequently, the infinitesimal diffeomorphism constraint is not implemented as a Hilbert space operator.

\subsection*{Representations for the Weyl $C^*$-algebras for surfaces on a Hilbert space}\label{subsec Weyl}

The main objects, which are introduced in this project \textit{AQV}, are given by
\begin{itemize}
 \item the flux groups  or the Lie flux algebras of Lie flux groups associated to surface sets,
 \item the analytic holonomy $C^ *$-algebra, which is given by the inductive limit $C^ *$-algebra of an inductive family of analytic holonomy $C^ *$-algebras restricted to finite graph systems.
\end{itemize} 
In the previous subsection the construction of the Weyl algebra has been introduced briefly. The Weyl algebra of Quantum Geometry \cite{Fleischhack06} has been constructed from the analytic holonomy $C^*$-algebra and unitary operators, which are defined by weakly continuous one-parameter unitary groups of $\R$ on the Hilbert space $\HS_{\text{AL}}$. The unitaries have been called Weyl operators by Fleischhack. 
The Weyl $C^*$-algebra for a surface set and restricted to a finite graph system is generated by the analytic holonomy $C^*$-algebra restricted to a finite graph system and Weyl elements. Assume for a moment that $G$ is a compact connected Lie group. Then consider a strongly continuous one-parameter unitary group of $\R$, which is given by $\R\ni t\mapsto U(\exp(t E_S(\gamma))$, on the Hilbert space $\HS_\infty$. Then each unitary $U(\exp(t E_S(\gamma))$ defines a \textit{Weyl element}, too. 

To obtain a uniqueness result of a representation of a $C^*$-algebra the following general facts will be used.
Since irreducible representations of a $C^*$-algebra on a Hilbert space correspond one-to-one to pure states on the $C^*$-algebra, the uniqueness of a particular representation of the $C^*$-algebra on a Hilbert space corresponds to a unique state. The inductive limit of an inductive family of $C^*$-algebras corresponds one-to-one to a projective limit on the projective family of state spaces of the $C^*$-algebras. The GNS-representation associated to a state of a $C^ *$-algebra consists of a cyclic vector $\Omega$ on a Hilbert space and a representation of the $C^*$-algebra on the Hilbert space. 

The uniqueness of a finite surface-orientation-preserving graph-diffeomorphism invariant pure state of the commutative Weyl $C^*$-algebra for surfaces is obtained in Theorem \ref{prop invstateweylalg} by several steps. The \textit{commutative Weyl $C^*$-algebra for surfaces} is constructed similarly to the Weyl $C^*$-algebra for surfaces with the difference that the group $G$ is replaced by the center of the group $G$. Then graph-diffeomorphism invariant states of the commutative Weyl algebra for surfaces restricted to a graph system $\PD_\Gamma$ are analysed. It turns out that a difference occur, if either the natural or if the non-standard identification of the configuration space $\Ab_\Gamma$ is taken into account. In particular, for the natural identification, the state is a sum over states, which are indexed by bisections. For the commutative Weyl algebra for surfaces the difference disappears. There exists a pure and unique state, which is invariant under finite graph-diffeomorphisms. This result is similar to the uniqueness of the representation of the Weyl algebra of Quantum Geometry and it is obtained in a complete new operator algebraic formulation. 

There is a problem of finding other representations of the Weyl $C^*$-algebra for surfaces. For example for the Weyl $C^*$-algebra for surfaces the important fact is that the flux group associated to a surface set is represented on the Hilbert space $\HS_\Gamma$ by a unitary representation in $\Rep(\bar G_{\breve S,\Gamma}, \KD(\HS_\Gamma))$. The only naturally or satisfactory representations of the group-valued quantum flux operators are given by the Weyl elements, which are given by unitary representation of the flux group on the Hilbert space $\HS_\Gamma$. 
\section{The basic quantum operators}\label{quantum variables}
\subsection{Finite path groupoids and graph systems}\label{subsec fingraphpathgroup}
Let $c:[0,1]\rightarrow\Sigma$ be continuous curve in the domain $\bra 0,1\ket$, which is (piecewise) $C^k$-differentiable ($1\leq k\leq \infty$), analytic ($k=\omega$) or semi-analytic ($k=s\omega$) in $\bra 0,1\ket$ and oriented such that the source vertex is $c(0)=s(c)$ and the target vertex is $c(1)=t(c)$. Moreover assume that, the range of each subinterval of the curve $c$ is a submanifold, which can be embedded in $\Sigma$. An \textbf{edge} is given by a \hyperlink{rep-equiv}{reparametrisation invariant} curve of class (piecewise) $C^k$. The maps $s_{\Sigma},t_{\Sigma}:P\Sigma\rightarrow\Sigma$ where $P\Sigma$ is the path space are surjective maps and are called the source or target map.    

A set of edges $\{e_i\}_{i=1,...,N}$ is called \textbf{independent}, if the only intersections points of the edges are source $s_{\Sigma}(e_i)$ or $t_{\Sigma}(e_i)$ target points. Composed edges are called \textbf{paths}. An \textbf{initial segment} of a path $\gamma$ is a path $\gamma_1$ such that there exists another path $\gamma_2$ and $\gamma=\gamma_1\circ\gamma_2$. The second element $\gamma_2$ is also called a \textbf{final segment} of the path $\gamma$.

\begin{defi}
A \textbf{graph} $\Gamma$ is a union of finitely many independent edges $\{e_i\}_{i=1,...,N}$ for $N\in\N$. The set $\{e_1,...,e_N\}$ is called the \textbf{generating set for $\Gamma$}. The number of edges of a graph is denoted by $\vert \Gamma\vert$. The elements of the set $V_\Gamma:=\{s_{\Sigma}(e_k),t_{\Sigma}(e_k)\}_{k=1,...,N}$ of source and target points are called \textbf{vertices}.
\end{defi}

A graph generates a finite path groupoid in the sense that, the set $\PD_\Gamma\Sigma$ contains all independent edges, their inverses and all possible compositions of edges. All the elements of $\PD_\Gamma\Sigma$ are called paths associated to a graph. Furthermore the surjective source and target maps $s_{\Sigma}$ and $t_{\Sigma}$ are restricted to the maps $s,t:\PD_\Gamma\Sigma\rightarrow V_\Gamma$, which are required to be surjective.

\begin{defi}\label{path groupoid} Let $\Gamma$ be a graph. Then a \textbf{finite path groupoid} $\PD_\Gamma\Sigma$ over $V_\Gamma$ is a pair $(\PD_\Gamma\Sigma, V_\Gamma)$ of finite sets equipped with the following structures: 
\begin{enumerate}
 \item two surjective maps \(s,t:\PD_\Gamma\Sigma\rightarrow V_\Gamma\), which are called the source and target map,
\item the set \(\PD_\Gamma\Sigma^2:=\{ (\gamma_i,\gamma_j)\in\PD_\Gamma\Sigma\times\PD_\Gamma\Sigma: t(\gamma_i)=s(\gamma_j)\}\) of finitely many composable pairs of paths,
\item the  composition \(\circ :\PD_\Gamma^2\Sigma\rightarrow \PD_\Gamma\Sigma,\text{ where }(\gamma_i,\gamma_j)\mapsto \gamma_i\circ \gamma_j\), 
\item the inversion map \(\gamma_i\mapsto \gamma_i^{-1}\) of a path,
\item the object inclusion map \(\iota:V_\Gamma\rightarrow\PD_\Gamma\Sigma\) and
\item $\PD_\Gamma\Sigma$ is defined by the set $\PD_\Gamma\Sigma$ modulo the algebraic equivalence relations generated by
\beq\label{groupoid0} \gamma_i^{-1}\circ \gamma_i\simeq \idf_{s(\gamma_i)}\text{ and }\gamma_i\circ \gamma_i^{-1}\simeq \idf_{t(\gamma_i)}
\eq 
\end{enumerate}
Shortly write $\fPSGm$. 
\end{defi} 
Clearly, a graph $\Gamma$ generates freely the paths in $\PD_\Gamma\Sigma$. Moreover the map $s \times t: \PD_\Gamma\Sigma\rightarrow V_\Gamma\times V_\Gamma$ defined by $(s\times t)(\gamma)=(s(\gamma),t(\gamma))$ for all $\gamma\in\PD_\Gamma\Sigma$ is assumed to be surjective ($\PD_\Gamma\Sigma$ over $V_\Gamma$ is a transitive groupoid), too. 

A general groupoid $\GG$ over $\GG^{0}$ defines a small category where the set of morphisms is denoted in general by $\GG$ and the set of objects is denoted by $\GG^{0}$. Hence in particular the path groupoid can be viewed as a category, since,
\begin{itemize}
\item the set of morphisms is identified with $\PD_\Gamma\Sigma$,
\item the set of objects is given by $V_\Gamma$ (the units) 
\end{itemize}

From the condition (\ref{groupoid0}) it follows that, the path groupoid satisfies additionally 
\begin{enumerate}
 \item $ s(\gamma_i\circ \gamma_j)=s(\gamma_i)\text{ and } t(\gamma_i\circ \gamma_j)=t(\gamma_j)\text{ for every } (\gamma_i,\gamma_j)\in\PD_\Gamma^2\Sigma$
\item $s(v)= v= t(v)\text{ for every } v\in V_\Gamma$
\item\label{groupoid1} $ \gamma \circ\idf_{s(\gamma)} = \gamma = \idf_{t(\gamma)}\circ \gamma\text{ for every } \gamma\in \PD_\Gamma\Sigma\text{ and }$
\item $\gamma \circ (\gamma_i\circ \gamma_j)=(\gamma \circ \gamma_i) \circ \gamma_j$
\item $\gamma \circ (\gamma^{-1}\circ \gamma_1)=\gamma_1= (\gamma_1 \circ \gamma) \circ \gamma^{-1}$
\end{enumerate}

The condition \ref{groupoid1} implies that the vertices are units of the groupoid. 

\begin{defi}
Denote the set of all finitely generated paths by
\beqs \PD_\Gamma\Sigma^{(n)}:=\{(\gamma_1,...,\gamma_n)\in \PD_\Gamma\times ...\PD_\Gamma: (\gamma_i,\gamma_{i+1})\in\PD^{(2)}, 1\leq i\leq n-1 \}\eqs
The set of paths with source point $v\in V_\Gamma$ is given by
\beqs \PD_\Gamma\Sigma^{v}:=s^{-1}(\{v\})\eqs
The set of paths with target  point $v\in V_\Gamma$ is defined by
\beqs \PD_\Gamma\Sigma_{v}:=t^{-1}(\{v\})\eqs
The set of paths with source point $v\in V_\Gamma$ and target point $u\in V_\Gamma$ is 
\beqs \PD_\Gamma\Sigma^{v}_u:=\PD_\Gamma\Sigma^{v}\cap \PD_\Gamma\Sigma_{u}\eqs
\end{defi}

A graph $\Gamma$ is said to be \hypertarget{disconnected}{\textbf{disconnected}} if it contains only mutually pairs $(\gamma_i,\gamma_j)$ of non-composable independent paths $\gamma_i$ and $\gamma_j$ for $i\neq j$ and $i,j=1,...,N$. In other words for all $1\leq i,l\leq N$ it is true that $s(\gamma_i)\neq t(\gamma_l)$ and $t(\gamma_i)\neq s(\gamma_l)$ where $i\neq l$ and $\gamma_i,\gamma_l\in\Gamma$.

\begin{defi}
Let $\Gamma$ be a graph. A \textbf{subgraph $\Gp$ of $\Gamma$} is a given by a finite set of independent paths in $\PD_\Gamma\Sigma$. 
\end{defi}
For example let $\Gamma:=\{\gamma_1,...,\gamma_N\}$ then $\Gp:=\{\gamma_1\circ\gamma_2,\gamma_3^{-1},\gamma_4\}$ where $\gamma_1\circ\gamma_2,\gamma_3^{-1},\gamma_4\in\PD_\Gamma\Sigma$ is a subgraph of $\Gamma$, whereas the set $\{\gamma_1,\gamma_1\circ\gamma_2\}$ is not a subgraph of $\Gamma$. Notice if additionally $(\gamma_2,\gamma_4)\in\PD_\Gamma^{(2)}$ holds, then $\{\gamma_1,\gamma_3^{-1},\gamma_2\circ\gamma_4\}$ is a subgraph of $\Gamma$, too. Moreover for $\Gamma:=\{\gamma\}$ the graph $\Gamma^{-1}:=\{\gamma^{-1}\}$ is a subgraph of $\Gamma$. As well the graph $\Gamma$ is a subgraph of $\Gamma^{-1}$. A subgraph of $\Gamma$ that is generated by compositions of some paths, which are not reversed in their orientation, of the set $\{\gamma_1,...,\gamma_N\}$ is called an \textbf{orientation preserved subgraph of a graph}. For example for $\Gamma:=\{\gamma_1,...,\gamma_N\}$ orientation preserved subgraphs are given by $\{\gamma_1\circ\gamma_2\}$, $\{\gamma_1,\gamma_2,\gamma_N\}$ or $\{\gamma_{N-2}\circ\gamma_{N-1}\}$ if $(\gamma_1,\gamma_2)\in\PD_\Gamma\Sigma^{(2)}$ and $(\gamma_{N-2},\gamma_{N-1})\in\PD_\Gamma\Sigma^{(2)}$.   

\begin{defi} A \textbf{finite graph system $\PD_\Gamma$ for $\Gamma$} is a finite set of subgraphs of a graph $\Gamma$. A finite graph system $\PD_{\Gp}$ for $\Gp$ is a \hypertarget{finite graph subsystem}{\textbf{finite graph subsystem}} of $\PD_\Gamma$ for $\Gamma$ if the set $\PD_{\Gp}$ is a subset of $\PD_{\Gamma}$ and $\Gp$ is a subgraph of $\Gamma$. Shortly write $\PD_{\Gp}\leq\PD_{\Gamma}$.

A \hypertarget{finite orientation preserved graph system}{\textbf{finite orientation preserved graph system}} $\PD^{\op}_\Gamma$ for $\Gamma$ is a finite set of orientation preserved subgraphs of a graph $\Gamma$. 
\end{defi}

Recall that, a finite path groupoid is constructed from a graph $\Gamma$, but a set of elements of the path groupoid need not be a graph again. For example let $\Gamma:=\{\gamma_1\circ\gamma_2\}$ and $\Gp=\{\gamma_1\circ\gamma_3\}$, then $\Gpp=\Gamma\cup\Gp$ is not a graph, since this set is not independent. Hence only appropriate unions of paths, which are elements of a fixed finite path groupoid, define graphs. The idea is to define a suitable action on elements of the path groupoid, which corresponds to an action of diffeomorphisms on the manifold $\Sigma$. The action has to be transfered to graph systems. But the action of bisection, which is defined by the use of the groupoid multiplication, cannot easily generalised for graph systems. 

\begin{problem}\label{problem group structure on graphs systems}
Let $\breve\Gamma:=\{\Gamma_i\}_{i=1,..,N}$ be a finite set such that each $\Gamma_i$ is a set of not necessarily independent paths such that 
\begin{enumerate}
\item the set contains no loops and
\item each pair of paths satisfies one of the following conditions
\begin{itemize}
\item the paths intersect each other only in one vertex,
\item the paths do not intersect each other or
\item one path of the pair is a segment of the other path.
\end{itemize}
\end{enumerate}

Then there is a map $\circ:\breve\Gamma\times \breve\Gamma\rightarrow\breve\Gamma$ of two elements $\Gamma_1$ and $\Gamma_2$ defined by
\beqs \{\gamma_1,...,\gamma_M\}\circ\{\tg_1,...,\tg_M\}:= &\Big\{ \gamma_i\circ\tg_j:t(\gamma_i)=s(\tg_j)\Big\}_{1\leq i,j\leq M}\\
\eqs for $\Gamma_1:=\{\gamma_1,...,\gamma_M\},\Gamma_2:=\{\tg_1,...,\tg_M\}$. 
Moreover define a map $^{-1}:\breve\Gamma\rightarrow\breve\Gamma$ by
\beqs  \{\gamma_1,...,\gamma_M\}^{-1}:= \{\gamma^{-1}_1,...,\gamma^{-1}_M\}\eqs 

Then the following is derived
\beqs \{\gamma_1,...,\gamma_M\}\circ\{\gamma^{-1}_1,...,\gamma^{-1}_M\}&=\Big\{ \gamma_i\circ\gamma^{-1}_j: t(\gamma_i)=t(\gamma_j)\Big\}_{1\leq i,j\leq M}\\
&=\Big\{ \gamma_i\circ\gamma^{-1}_j:t(\gamma_i)=t(\gamma_j)\text{ and }i\neq j\Big\}_{1\leq i,j\leq M}\\
&\quad\cup\{\idf_{s_{\gamma_j}}\}_{1\leq j\leq M}\\
\neq &\quad\cup\{\idf_{s_{\gamma_j}}\}_{1\leq j\leq M}
\eqs The equality is true, if the set $\breve\Gamma$ contains only graphs such that all paths are mutually non-composable. Consequently this does not define a well-defined multiplication map. Notice that, the same is discovered if a similar map and inversion operation are defined for a finite graph system $\PD_\Gamma$. 
\end{problem}

Consequently the property of paths being independent need not be dropped for the definition of a suitable multiplication and inversion operation. In fact the independence property is a necessary condition for the construction of the holonomy algebra for analytic paths. Only under this circumstance each analytic path is decomposed into a finite product of independent piecewise analytic paths again. 

\begin{defi}
A finite path groupoid $\PD_{\Gp}\Sigma$ over $V_{\Gp}$ is a \textbf{finite path subgroupoid} of $\PD_{\Gamma}\Sigma$ over $V_\Gamma$ if the set $V_{\Gp}$ is contained in $V_\Gamma$ and the set $\PD_{\Gp}\Sigma$ is a subset of $\PD_{\Gamma}\Sigma$. Shortly write $\PD_{\Gp}\Sigma\leq\PD_{\Gamma}\Sigma$.
\end{defi}

Clearly for a subgraph $\Gamma_1$ of a graph $\Gamma_2$, the associated path groupoid $\PD_{\Gamma_1}\Sigma$ over $V_{\Gamma_1}$ is a subgroupoid of $\PD_{\Gamma_2}\Sigma$ over $V_{\Gamma_2}$.  This is a consequence of the fact that, each path in $\PD_{\Gamma_1}\Sigma$ is a composition of paths or their inverses in $\PD_{\Gamma_2}\Sigma$. 

\begin{defi}
A \textbf{family of finite path groupoids} $\{\PD_{\Gamma_i}\Sigma\}_{i=1,...,\infty}$, which is a set of finite path groupoids $\PD_{\Gamma_i}\Sigma$ over $V_{\Gamma_i}$, is said to be \textbf{inductive} if for any $\PD_{\Gamma_1}\Sigma,\PD_{\Gamma_2}\Sigma$ exists a $\PD_{\Gamma_3}\Sigma$ such that $\PD_{\Gamma_1}\Sigma,\PD_{\Gamma_2}\Sigma\leq\PD_{\Gamma_3}\Sigma$.

A \textbf{family of graph systems} $\{\PD_{\Gamma_i}\}_{i=1,...,\infty}$, which is a set of finite path systems $\PD_{\Gamma_i}$ for $\Gamma_i$, is said to be \textbf{inductive} if for any $\PD_{\Gamma_1},\PD_{\Gamma_2}$ exists a $\PD_{\Gamma_3}$ such that $\PD_{\Gamma_1},\PD_{\Gamma_2}\leq \PD_{\Gamma_3}$.
\end{defi}

\begin{defi}
Let $\{\PD_{\Gamma_i}\Sigma\}_{i=1,...,\infty}$ be an inductive family of path groupoids and $\{\PD_{\Gamma_i}\}_{i=1,...,\infty}$ be an inductive family of graph systems.

The \textbf{inductive limit path groupoid $\PD$ over $\Sigma$} of an inductive family of finite path groupoids such that $\PD:=\limi\PD_{\Gamma_i}\Sigma$ is called the \textbf{(algebraic) path groupoid} $\PGoS$.

Moreover there exists an \textbf{inductive limit graph $\Gamma_\infty$} of an inductive family of graphs such that $\Gamma_\infty:=\limi \Gamma_i$.

The \textbf{inductive limit graph system} $\PD_{\Gamma_\infty}$ of an inductive family of graph systems such that $\PD_{\Gamma_\infty}:=\limi \PD_{\Gamma_i}$
\end{defi}

Assume that, the inductive limit $\Gamma_\infty$ of a inductive family of graphs is a graph, which consists of an infinite countable number of independent paths. The inductive limit $\PD_{\Gamma_\infty}$ of a inductive family $\{\PD_{\Gamma_i}\}$ of finite graph systems contains an infinite countable number of subgraphs of $\Gamma_\infty$ and each subgraph is a finite set of arbitrary independent paths in $\Sigma$. 

\subsection{Holonomy maps for finite path groupoids, graph systems and transformations}\label{subsec holmapsfinpath}
In section \ref{subsec fingraphpathgroup} the concept of finite path groupoids for analytic paths has been given. Now the holonomy maps are introduced for finite path groupoids and finite graph systems. The ideas are familar with those presented by Thiemann \cite{Thiembook07}. But for example the finite graph systems have not been studied before. Ashtekar and Lewandowski \cite{AshLew93} have defined the analytic holonomy $C^*$-algebra, which they have based on a finite set of independent hoops. The hoops are generalised for path groupoids and the independence requirement is implemented by the concept of finite graph systems. 

\subsubsection{Holonomy maps for finite path groupoids}\label{subsubsec holmap}

\paragraph*{Groupoid morphisms for finite path groupoids}\hspace{10pt} 

Let $\GGim, \GGiim$ be two arbitrary groupoids.

\begin{defi}
A \hypertarget{groupoid-morphism}{\textbf{groupoid morphism}} between two groupoids $\GG_1$ and $\GG_2$ consists of two maps  $\ho:\GG_1\rightarrow\GG_2$  and $h:\GG_1^0\rightarrow\GG_2^0$ such that
\beqs (\hypertarget{G1}{G1})\qquad \ho(\gamma\circ\gp)&= \ho(\gamma)\ho(\gp)\text{ for all }(\gamma,\gp)\in \GG_1^{(2)}\eqs
\beqs (\hypertarget{G2}{G2})\qquad s_{2}(\ho(\gamma))&=h(s_{1}(\gamma)),\quad t_2(\ho(\gamma))=h(t_{1}(\gamma))\eqs 
 
A \textbf{strong groupoid morphism} between two groupoids $\GG_1$ and $\GG_2$ additionally satisfies
\beqs (\hypertarget{SG2}{SG})\qquad \text{ for every pair }(\ho(\gamma),\ho(\gp))\in\GG_2^{(2)}\text{ it follows that }(\gamma,\gp)\in \GG_1^{(2)}\eqs
\end{defi}

Let $G$ be a Lie group. Then $G$ over $e_G$ is a groupoid, where the group multiplication $\cdot: G^2\rightarrow G$ is defined for all elements  $g_1,g_2,g\in G$ such that $g_1\cdot g_2 = g$. A groupoid morphism between a finite path groupoid $\PD_\Gamma\Sigma$ to $G$ is given by the maps
\[\ho_\Gamma: \PD_\Gamma\Sigma\rightarrow G,\quad h_\Gamma:V_\Gamma\rightarrow e_G \] Clearly
\beq \ho_\Gamma(\gamma\circ\gp)&= \ho_\Gamma(\gamma)\ho_\Gamma(\gp)\text{ for all }(\gamma,\gp)\in \PD_\Gamma\Sigma^{(2)}\\
s_G(\ho_\Gamma(\gamma))&=h_\Gamma(s_{\PD_\Gamma\Sigma}(\gamma)),\quad t_G(\ho_\Gamma(\gamma))=h_\Gamma(t_{\PD_\Gamma\Sigma}(\gamma))
\eq But for an arbitrary pair $(\ho_\Gamma(\gamma_1),\ho_\Gamma(\gamma_2))=:(g_1,g_2)\in G^{(2)}$ it does not follows that, $(\gamma_1,\gamma_2)\in \PD_\Gamma\Sigma^{(2)}$ is true. Hence $\ho_\Gamma$ is not a strong groupoid morphism.

\begin{defi}\label{def sameholanal}Let $\fPG$ be a finite path groupoid.

Two paths $\gamma$ and $\gp$ in $\PD_\Gamma\Sigma$ have the \textbf{same-holonomy for all connections} iff 
\beqs \ho_\Gamma(\gamma)=\ho_\Gamma(\gp)\text{ for all }&(\ho_\Gamma,h_\Gamma)\text{ groupoid morphisms }\\ & \ho_\Gamma:\PD_\Gamma\Sigma\rightarrow G, h:V_\Gamma\rightarrow\{e_G\}
\eqs Denote the relation by $\sim_{\text{s.hol.}}$.
\end{defi}
\begin{lem}
The same-holonomy for all connections relation is an equivalence relation. 
\end{lem}
Notice that, the quotient of the finite path groupoid and the same-holonomy relation for all connections replace the hoop group, which has been used in \cite{AshLew93}.
\begin{defi}\label{genrestgroupoidforgraph}
Let $\fPG$ be a finite path groupoid modulo same-holonomy for all connections equivalence.

A \hypertarget{holonomy map for a finite path groupoid}{\textbf{holonomy map for a finite path groupoid}} $\PD_\Gamma\Sigma$ over $V_\Gamma$ is a groupoid morphism consisting of the maps $(\ho_\Gamma,h_\Gamma)$, where
\(\ho_\Gamma:\PD_\Gamma\Sigma\rightarrow G,h_\Gamma:V_\Gamma\rightarrow \{e_G\}\). 
The set of all holonomy maps is abbreviated by $\Hom(\PD_\Gamma\Sigma,G)$.
\end{defi}

For a short notation observe the following.
In further sections it is always assumed that, the finite path groupoid $\fPG$ is considered modulo same-holonomy for all connections equivalence although it is not stated explicitly.

\paragraph*{Admissable maps and equivalent groupoid morphisms}\hspace{10pt}

Now consider a finite path groupoid morphism $(\ho_\Gamma,h_\Gamma)$ from a finite path groupoid $\PD_\Gamma\Sigma$ over $V_\Gamma$ to the groupoid $G$ over $\{e_G\}$, which is contained in $\Hom(\PD_\Gamma\Sigma,G)$.

Consider an arbitrary map $\go_\Gamma: \PD_\Gamma\Sigma\rightarrow G$. Then there is a groupoid morphism defined by
\beq\label{eq similarity_2b} \Go_\Gamma(\gamma)&:= \go_\Gamma(\gamma)\ho_\Gamma(\gamma)\go_\Gamma(\gamma^{-1})^{-1}\text{ for all }\gamma\in\PD_\Gamma\Sigma\\
\eq if and only if 
\beqs \Go_\Gamma(\gamma_1\circ\gamma_2)&=\Go_\Gamma(\gamma_1)\Go_\Gamma(\gamma_2)\text{ for all }(\gamma_1,\gamma_2)\in\PD_\Gamma\Sigma^{(2)}\eqs holds. Then $\Go_\Gamma\in \Hom(\PD_\Gamma\Sigma,G)$.

Hence for all $(\gamma_1,\gamma_2)\in\PD_\Gamma\Sigma^{(2)}$ it is necessary that
\beqs  \Go_\Gamma(\gamma_1\circ\gamma_2)
&= \go_\Gamma(\gamma_1\circ\gamma_2)\ho_\Gamma(\gamma_1\circ\gamma_2)\go_\Gamma(\gamma_2^{-1}\circ\gamma_1^{-1})^{-1}\\
&=\go_\Gamma(\gamma_1\circ\gamma_2)\ho_\Gamma(\gamma_1)\ho_\Gamma(\gamma_2)\go_\Gamma(\gamma_2^{-1}\circ\gamma_1^{-1})^{-1}\\
&\overset{!}{=}\go_\Gamma(\gamma_1)\ho_\Gamma(\gamma_1)\go_\Gamma(\gamma_1^{-1})^{-1}\go_\Gamma(\gamma_2)\ho_\Gamma(\gamma_2)\go_\Gamma(\gamma_2^{-1})^{-1}
\eqs is satisfied. Therefore the map is required to fulfill
\beq\label{equ2 g_gamma} &\go_\Gamma(\gamma_1)=\go_\Gamma(\gamma_1\circ\gamma_2)\text{, }\go_\Gamma(\gamma_2^{-1})=\go_\Gamma((\gamma_1\circ\gamma_2)^{-1})\text{ and }\\
&\go_\Gamma(\gamma_1^{-1})^{-1}\go_\Gamma(\gamma_2)=e_G\text{ for all }(\gamma_1,\gamma_2)\in\PD_\Gamma\Sigma^{(2)}\text{ in particular, }\\
&\go_\Gamma(\gamma^{-1})^{-1}\go_\Gamma(\gamma)=e_G\text{ for all }(\gamma^{-1},\gamma)\in\PD_\Gamma\Sigma^{(2)}
\eq
for every refinement $\gamma_1\circ\gamma_2$ of each $\gamma$ in $\PD_\Gamma\Sigma$ and $\gamma_1$ being an initial segment of $\gamma_1\circ \gamma_2$ and $\gamma_2^{-1}$ an final segment of $(\gamma_1\circ\gamma_2)^{-1}$.
In comparison with Fleischhack's definition in \cite[Def. 3.7]{Fleischhack06} such maps are called admissible.
\begin{defi}\label{def admiss}
The set of maps $\go_\Gamma:\PD_\Gamma\Sigma\rightarrow G$ satisfying \eqref{equ2 g_gamma} for all pairs of decomposable paths in $\PD_\Gamma^{(2)}\Sigma$ is called the \textbf{set of admissible maps} and is denoted by $\Map^{\adm}(\PD_\Gamma\Sigma,G)$. 
\end{defi} 

Consider a map $g_\Gamma:V_\Gamma\rightarrow G$ such that 
\beqs(g_\Gamma,\ho_\Gamma)\in \Map(V_\Gamma,G)\times \Hom(\PD_\Gamma\Sigma,G)
\eqs which is also called a local gauge map.
Then the map $ \tilde\Go_\Gamma$ defined by
\beq\label{eq similarity1} \tilde\Go_\Gamma(\gamma)&:= g_\Gamma(s(\gamma))\ho_\Gamma(\gamma)g_\Gamma(s(\gamma^{-1}))^{-1}\text{ for all }\gamma\in\PD_\Gamma\Sigma\eq is a groupoid morphism. This is a result of the computation: 
\beqs \tilde\Go_\Gamma(\gamma_1\gamma_2)&= g_\Gamma(s(\gamma_1))\ho_\Gamma(\gamma_1\gamma_2)g_\Gamma(t(\gamma_2))^{-1}\\&= g_\Gamma(s(\gamma_1))\ho_\Gamma(\gamma_1)g_\Gamma(t(\gamma_1))^{-1}
g_\Gamma(s(\gamma_2))\ho_\Gamma(\gamma_2)g_\Gamma(t(\gamma))^{-1}\eqs
since $t(\gamma_1)=s(\gamma_2)$.  

\begin{defi}\label{def similargroupoidhom}
Two groupoid morphisms $(\ho_\Gamma,h_\Gamma)$ and $(\Go_\Gamma,h_\Gamma)$, or respectively $(\tilde\Go_\Gamma,h_\Gamma)$, between the groupoids $\PD_\Gamma$ over $V_\Gamma$ and the groupoid $G$ over $\{e_G\}$, which are defined for $(\go_\Gamma,\ho_\Gamma)\in\Map(\PD_\Gamma\Sigma,G)\times\Hom(\PD_\Gamma\Sigma,G)$ by \eqref{eq similarity_2b}, or respectively for $(g_\Gamma,\ho_\Gamma)\in \Map(V_\Gamma,G)\times \Hom(\PD_\Gamma\Sigma,G)$ by \eqref{eq similarity1}, are said to be \textbf{similar or equivalent groupoid morphisms}.
\end{defi}

\subsubsection{Holonomy maps for finite graph systems}\label{subsec graphhol}

Ashtekar and Lewandowski \cite{AshLew93} have presented the loop decomposition into a finite set of independent hoops (in the analytic category). This structure is replaced by a graph, since a graph is a set of independent edges. Notice that, the set of hoops that is generated by a finite set of independent hoops, is generalised to the set of finite graph systems. A finite path groupoid is generated by the set of edges, which defines a graph $\Gamma$, but a set of elements of the path groupoid need not be a graph again. The appropriate notion for graphs constructed from sets of paths is the finite graph system, which is defined in section \ref{subsec fingraphpathgroup}. Now the concept of holonomy maps is generalised for finite graph systems. Since the set, which is generated by a finite number of independent edges, contains paths that are composable, there are two possibilities to identify the image of the holonomy map for a finite graph system on a fixed graph with a subgroup of $G^{\vert\Gamma\vert}$. One way is to use the generating set of independend edges of a graph, which has been also used in \cite{AshLew93}. On the other hand, it is also possible to identify each graph with a disconnected subgraph of a fixed graph, which is generated by a set of independent edges. Notice that, the author implements two situations. One case is given by a set of paths that can be composed further and the other case is related to paths that are not composable. This is necessary for the definition of an action of the flux operators. Precisely the identification of the image of the holonomy maps along these paths is necessary to define a well-defined action of a flux element on the configuration space. This issue is studied in remark \ref{rem fluxlikeoperators} in section \ref{subsec dynsysfluxgroup}.

First of all consider a graph $\Gamma$ that is generated by the set $\{\gamma_1,...,\gamma_N\}$ of edges. Then each subgraph of a graph $\Gamma$ contain paths that are composition of edges in $\{\gamma_1,...,\gamma_N\}$ or inverse edges. For example the following set $\Gp:=\{\gamma_1\circ\gamma_2\circ\gamma_3,\gamma_4\}$ defines a subgraph of $\Gamma:=\{\gamma_1,\gamma_2,\gamma_3,\gamma_4\}$. Hence there is a natural identification available.

\begin{defi}
A subgraph $\Gp$ of a graph $\Gamma$ is always generated by a subset $\{\gamma_1,...,\gamma_M\}$ of the generating set $\{\gamma_1,...,\gamma_N\}$ of independent edges that generates the graph $\Gamma$. Hence each subgraph is identified with a subset of $\{\gamma_1^{\pm 1},...,\gamma_N^{\pm 1}\}$. This is called the \hypertarget{natural identification}{\textbf{natural identification of subgraphs}}.
\end{defi}

\begin{exa}\label{exa natidentif}
For example consider a subgraph $\Gp:=\{\gamma_1\circ\gamma_2,\gamma_3\circ\gamma_4,...,\gamma_{M-1}\circ\gamma_M\}$, which is identified naturally with a set $\{\gamma_1,...,\gamma_M\}$. The set $\{\gamma_1,...,\gamma_M\}$ is a subset of $\{\gamma_1,...,\gamma_N\}$ where $N=\vert \Gamma\vert$ and $M\leq N$. 

Another example is given by the graph $\Gpp:=\{\gamma_1,\gamma_2\}$ such that $\gamma_2=\gpe\circ\gpz$, then $\Gpp$ is identified naturally with $\{\gamma_1,\gpe,\gpz\}$. This set is a subset of $\{\gamma_1,\gpe,\gpz,\gamma_3,...,\gamma_{N-1}\}$. 
\end{exa}

\begin{defi}
Let $\Gamma$ be a graph, $\PD_\Gamma$ be the finite graph system. Let $\Gp:=\{\gamma_1,...,\gamma_M\}$be a subgraph of $\Gamma$.

A \hypertarget{holonomy map for a finite graph system}{\textbf{holonomy map for a finite graph system}} $\PD_\Gamma$ is a given by a pair of maps $(\ho_\Gamma,h_\Gamma)$ such that there exists a holonomy map\footnote{In the work the holonomy map for a finite graph system and the holonomy map for a finite path groupoid is denoted by the same pair $(\ho_\Gamma,h_\Gamma)$.} $(\ho_\Gamma,h_\Gamma)$ for the finite path groupoid $\fPG$ and
\beqs &\ho_\Gamma:\PD_\Gamma\rightarrow G^{\vert \Gamma\vert},\quad \ho_\Gamma(\{\gamma_1,...,\gamma_M\})=(\ho_\Gamma(\gamma_1),...,\ho_\Gamma(\gamma_M), e_G,...,e_G)\\
&h_\Gamma:V_\Gamma\rightarrow \{e_G\}
\eqs 
The set of all holonomy maps for the finite graph system is denoted by $\Hom(\PD_\Gamma,G^{\vert \Gamma\vert})$.

The image of a map $\ho_\Gamma$ on each subgraph $\Gp$ of the graph $\Gamma$ is given by
\beqs (\ho_\Gamma(\gamma_1),...,\ho_\Gamma(\gamma_M),e_G,...,e_G)
\eqs is an element of $G^{\vert \Gamma\vert}$. The set of all images of maps on subgraphs of $\Gamma$ is denoted by $\Ab_\Gamma$.
\end{defi}
The idea is now to study two different restrictions of the set $\PD_\Gamma$ of subgraphs. For a short notation of a ''set of  holonomy maps for a certain restricted set of subgraphs of a graph'' in this article the following notions are introduced.
\begin{defi}
If the subset of all disconnected subgraphs of the finite graph system $\PD_\Gamma$ is considered, then the restriction of $\Ab_\Gamma$, which is identified with $G^{\vert \Gamma\vert}$ appropriately, is called the \hypertarget{non-standard identification}{\textbf{non-standard identification of the configuration space}}. If the subset of all natural identified subgraphs of the finite graph system $\PD_\Gamma$ is considered, then the restriction of $\Ab_\Gamma$, which is identified with $G^{\vert \Gamma\vert}$ appropriately, is called the \hypertarget{natural identification}{\textbf{natural identification of the configuration space}}.
\end{defi}

A comment on the non-standard identification of $\Ab_\Gamma$ is the following. If $\Gp:=\{\gamma_1\circ\gamma_2\}$ and $\Gpp:=\{\gamma_2\}$ are two subgraphs of $\Gamma:=\{\gamma_1,\gamma_2,\gamma_3\}$. The graph $\Gp$ is a subgraph of $\Gamma$. Then evaluation of a map $\ho_\Gamma$ on a subgraph $\Gp$ is given by
\beqs \ho_\Gamma(\Gp)=(\ho_\Gamma(\gamma_1\circ\gamma_2),\ho_\Gamma(s(\gamma_2)),\ho_\Gamma(s(\gamma_3)))=(\ho_\Gamma(\gamma_1)\ho_\Gamma(\gamma_2),e_G,e_G)\in G^3
\eqs and the holonomy map of the subgraph $\Gpp$ of $\Gp$ is evaluated by
\beqs \ho_\Gamma(\Gpp)=(\ho_\Gamma(s(\gamma_1)),\ho_\Gamma(s(\gamma_2))\ho_\Gamma(\gamma_2),\ho_\Gamma(s(\gamma_3)))=(\ho_\Gamma(\gamma_2),e_G,e_G)\in G^3
\eqs

\begin{exa}
Recall example \thesubsection.\ref{exa natidentif}.
For example for a subgraph $\Gp:=\{\gamma_1\circ\gamma_2,\gamma_3\circ\gamma_4,...,\gamma_{M-1}\circ\gamma_M\}$, which is naturally identified with a set $\{\gamma_1,...,\gamma_M\}$. Then the holonomy map is evaluated at $\Gp$ such that \[\ho_\Gamma(\Gp)=(\ho_\Gamma(\gamma_1),\ho_\Gamma(\gamma_2),....,\ho_\Gamma(\gamma_M),e_G,...,e_G)\in G^N\] where $N=\vert \Gamma\vert$. For example, let $\Gp:=\{\gamma_1,\gamma_2\}$ such that $\gamma_2=\gpe\circ\gpz$ and which is naturally identified with $\{\gamma_1,\gpe,\gpz\}$. Hence \[\ho_\Gamma(\Gp)=(\ho_\Gamma(\gamma_1),\ho_\Gamma(\gpe),\ho_\Gamma(\gpz),e_G,...,e_G)\in G^N\] is true.

Another example is given by the disconnected graph $\Gp:=\{\gamma_1\circ\gamma_2\circ\gamma_3,\gamma_4\}$, which is a subgraph of $\Gamma:=\{\gamma_1,\gamma_2,\gamma_3,\gamma_4\}$. Then the non-standard identification is given by
\[\ho_\Gamma(\Gp)=(\ho_\Gamma(\gamma_1\circ\gamma_2\circ\gamma_3),\ho_\Gamma(\gamma_4),e_G,e_G)\in G^4\]

If the natural identification is used, then $\ho_\Gamma(\Gp)$ is idenified with 
\[(\ho_\Gamma(\gamma_1),\ho_\Gamma(\gamma_2),\ho_\Gamma(\gamma_3),\ho_\Gamma(\gamma_4))\in G^4\]

Consider the following example. Let $\Gppp:=\{\gamma_1,\alpha,\gamma_2,\gamma_3\}$ be a graph such that 
 \begin{center}
\includegraphics[width=0.2\textwidth]{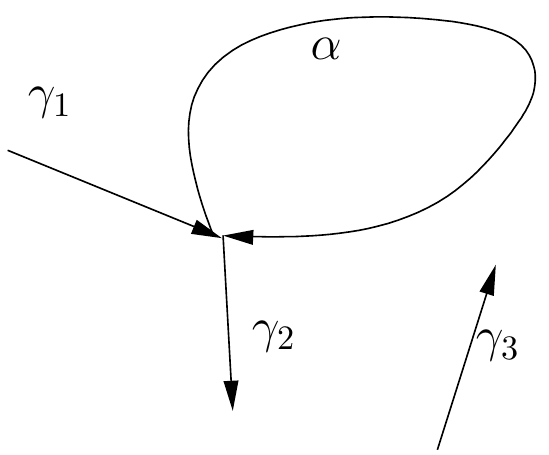}
\end{center}
Then notice the sets $\Gamma_1:=\{\gamma_1\circ\alpha,\gamma_3\}$ and $\Gamma_2:=\{\gamma_1\circ\alpha^{-1},\gamma_3\}$. In the non-standard identification of the configuration space $\Ab_{\Gppp}$ it is true that,
\beqs \ho_{\Gppp}(\Gamma_1)=(\ho_{\Gppp}(\gamma_1\circ\alpha),\ho_{\Gppp}(\gamma_3),e_G,e_G)\in G^4,\\
\ho_{\Gppp}(\Gamma_2)=(\ho_{\Gppp}(\gamma_1\circ\alpha^{-1}),\ho_{\Gppp}(\gamma_3),e_G,e_G)\in G^4
\eqs holds. Whereas in the natural identification of $\Ab_{\Gppp}$
 \beqs \ho_{\Gppp}(\Gamma_1)=(\ho_{\Gppp}(\gamma_1),\ho_{\Gppp}(\alpha),\ho_{\Gppp}(\gamma_3),e_G)\in G^4,\\
\ho_{\Gppp}(\Gamma_2)=(\ho_{\Gppp}(\gamma_1),\ho_{\Gppp}(\alpha^{-1}),\ho_{\Gppp}(\gamma_3),e_G)\in G^4
\eqs yields.
\end{exa}

The equivalence class of similar or equivalent groupoid morphisms defined in definition \ref{def similargroupoidhom} allows to define the following object.
The set of images of all holonomy maps of a finite graph system modulo the similar or equivalent groupoid morphisms equivalence relation is denoted by $\Ab_\Gamma/\bar\SimGroup_\Gamma$. 

\subsubsection{Transformations in finite path groupoids and finite graph systems}\label{subsubsec bisections}

The aim of this section is to clearify the graph changing operators in LQG framework and the role of finite diffeomorphisms in $\Sigma$. 
Therefore operations, which add, delete or transform paths, are introduced.  In particular translations in a finite path graph groupoid and in the groupoid $G$ over $\{e_G\}$ are studied. 

\paragraph*{Transformations in finite path groupoid\\[5pt]}

\begin{defi}
Let $\varphi$ be a $C^k$-diffeomorphism on $\Sigma$, which maps surfaces into surfaces. 

Then let $(\Phi_\Gamma,\varphi_\Gamma)$ be a pair of bijective maps, where $\varphi\vert_{V_\Gamma}=\varphi_\Gamma$ and 
\beq\Phi_\Gamma:\PD_\Gamma\Sigma\rightarrow\PD_\Gamma\Sigma\text{ and }\varphi_\Gamma:V_\Gamma\rightarrow V_\Gamma\eq 
such that 
\beq (s\circ\Phi_\Gamma)(\gamma)=(\varphi_\Gamma\circ s)(\gamma),\quad (t\circ \Phi_\Gamma)(\gamma)=(\varphi_\Gamma\circ t)(\gamma)\text{ for all }\gamma\in\PD_\Gamma\Sigma\eq holds such that $(\Phi_\Gamma,\varphi_\Gamma)$ defines a groupoid morphism.

Call the pair $(\Phi_\Gamma,\varphi_\Gamma)$ a \textbf{path-diffeomorphism of a finite path groupoid} $\PD_\Gamma\Sigma$ over $V_\Gamma$. Denote the set of finite path-diffeomorphisms by $\Diff(\PD_\Gamma\Sigma)$.
\end{defi}

Notice that, for $(\gamma,\gp)\in\PD_\Gamma\Sigma^{(2)}$ it is true that
\beq\label{eq requcombi0} \Phi_\Gamma(\gamma\circ\gp)=\Phi_\Gamma(\gamma)\circ\Phi_\Gamma(\gp)
\eq requires that
\beq\label{eq requcombi} (t\circ\Phi_\Gamma)(\gamma)=(s\circ\Phi_\Gamma)(\gp)
\eq Hence from \eqref{eq requcombi0} and \eqref{eq requcombi} it follows that, $\Phi_\Gamma(\idf_v)=\idf_{\varphi_\Gamma(v)}$ is true.

A path-diffeomorphism $(\Phi_\Gamma,\varphi_\Gamma)$ is lifted to $\Hom(\PD_\Gamma\Sigma,G)$. \\
The pair $(\ho_\Gamma\circ\Phi_\Gamma,h_\Gamma\circ\varphi_\Gamma)$ defined by
\beqs \ho_\Gamma\circ\Phi_\Gamma&: \PD_\Gamma\Sigma\rightarrow G,\quad \gamma\mapsto (\ho_\Gamma\circ\Phi_\Gamma)(\gamma)\\
h_\Gamma\circ\varphi_\Gamma&: V_\Gamma\rightarrow \{e_G\},\quad (h_\Gamma\circ\varphi_\Gamma)(v)=e_G
\eqs such that
\beqs &s_{\Hol}((\ho_\Gamma\circ\Phi_\Gamma)(\gamma))=(h_\Gamma\circ\varphi_\Gamma)(s(\gamma))=e_G,\\
&t_{\Hol}(\ho_\Gamma\circ\Phi_\Gamma(\gamma))=(h_\Gamma\circ\varphi_\Gamma)(t(\gamma))=e_G\text{ for all }\gamma\in\PD_\Gamma\Sigma
\eqs whenever $(\ho_\Gamma,h_\Gamma)\in\Hom(\PD_\Gamma\Sigma,G)$ and $(\Phi_\Gamma,\varphi_\Gamma)$ is a path-diffeomorphism, is a \hyperlink{holonomy map for a finite path groupoid}{holonomy map for a finite path groupoid} $\PD_\Gamma\Sigma$ over $V_\Gamma$.

\begin{defi}
A \textbf{left-translation in the finite path groupoid} $\PD_\Gamma\Sigma$ over $V_\Gamma$ at a vertex $v$ is a map defined by
\beqs L_\theta:\PD_\Gamma\Sigma^v\rightarrow \PD_\Gamma\Sigma^{w},\quad \gamma\mapsto L_{\theta}(\gamma):=\theta\circ\gamma
\eqs
for some $\theta\in\PD_\Gamma\Sigma_{v}^{w}$ and all $\gamma\in\PD_\Gamma\Sigma^v$.
\end{defi} 
In analogy a right-translation $R_\theta$ and an inner-translation $I_{\theta,\theta^\prime}$ in the finite path groupoid $\PD_\Gamma\Sigma$ over $V_\Gamma$ at a vertex $v$ can be defined.
\begin{rem}
Let $(\Phi_\Gamma,\varphi_\Gamma)$ be a path-diffeomorphism on a finite path groupoid $\PD_\Gamma\Sigma$ over $V_\Gamma$. Then a left-translation in the finite path groupoid $\PD_\Gamma\Sigma$ over $V_\Gamma$ at a vertex $v$ is defined by a path-diffeomorphism $(\Phi_\Gamma,\varphi_\Gamma)$ and the following object
\beq L_{\Phi_\Gamma}:\PD_\Gamma\Sigma^v\rightarrow \PD_\Gamma\Sigma^{\varphi_\Gamma(v)},\quad \gamma\mapsto L_{\Phi_\Gamma}(\gamma):=\Phi_\Gamma(\gamma)\text{ for }\gamma\in\PD_\Gamma\Sigma^v
\eq
Furthermore a right-translation in the finite path groupoid $\PD_\Gamma\Sigma$ over $V_\Gamma$ at a vertex $v$ is defined by a path-diffeomorphism $(\Phi_\Gamma,\varphi_\Gamma)$ and the following object
\beq R_{\Phi_\Gamma}:\PD_\Gamma\Sigma_v\rightarrow \PD_\Gamma\Sigma_{\varphi_\Gamma(v)},\quad \gamma\mapsto R_{\Phi_\Gamma}(\gamma):=\Phi_\Gamma(\gamma)\text{ for }\gamma\in\PD_\Gamma\Sigma_v
\eq

Finally an inner-translation in the finite path groupoid $\PD_\Gamma\Sigma$ over $V_\Gamma$ at the vertices $v$ and $w$ is defined by
\beqs  I_{\Phi_\Gamma}:\PD_\Gamma\Sigma^v_w\rightarrow \PD_\Gamma\Sigma^{\varphi_\Gamma(v)}_{\varphi_\Gamma(w)},\quad \gamma\mapsto I_{\Phi_\Gamma}(\gamma)=\Phi_\Gamma(\gamma)\text{ for }\gamma\in\PD_\Gamma\Sigma^v_w
\eqs where $(s\circ\Phi_\Gamma)(\gamma)=\varphi_\Gamma(v)$ and $(t\circ\Phi_\Gamma)(\gamma)=\varphi_\Gamma(w)$.
\end{rem}

In the following considerations the right-translation in a finite path groupoid is focused, but there is a generalisation to left-translations and inner-translations.
\begin{defi}
A \hypertarget{bisection of a finite path groupoid}{\textbf{bisection of a finite path groupoid}} $\PD_\Gamma\Sigma$ over $V_\Gamma$ is a map $\sigma:V_\Gamma\rightarrow\PD_\Gamma\Sigma$, which is right-inverse to the map $s:\PD_\Gamma\Sigma\rightarrow V_\Gamma$ (i.o.w. $s\circ\sigma=\id_{V_\Gamma}$) and such that $t\circ\sigma:V_\Gamma\rightarrow V_{\Gamma}$ is a bijective map\footnote{Note that in the infinite case of path groupoids an additional condition for the map $t\circ\sigma:\Sigma\rightarrow\Sigma$ has to be required. The map has to be a diffeomorphism. Observe that, the map $t\circ\sigma$ defines the finite diffeomorphism $\varphi_\Gamma:V_\Gamma\rightarrow V_\Gamma$.}. The set of bisections on $\PD_\Gamma\Sigma$ over $V_\Gamma$ is denoted $\mathfrak{B}(\PD_\Gamma\Sigma)$.
\end{defi}
\begin{rem}\label{rem defiofright}
Discover that, a bisection $\sigma\in\mathfrak{B}(\PD_\Gamma\Sigma)$ defines a path-diffeomorphism $(\varphi_\Gamma,\Phi_\Gamma)\in \Diff(\PD_\Gamma\Sigma)$, where $\varphi_\Gamma=t\circ\sigma$ and $\Phi_\Gamma$ is given by the right-translation $R_{\sigma(v)}:\PD_\Gamma\Sigma_v\rightarrow\PD_\Gamma\Sigma_{\varphi_\Gamma(v)}$ in $\fPG$, where $R_{\sigma(v)}(\gamma)=\Phi_\Gamma(\gamma)$ for all $\gamma\in\PD_\Gamma\Sigma_v$ and for a fixed $v\in V_\Gamma$. The right-translation is defined by 
\beq\label{eq Rendv}R_{\sigma(v)}(\gamma):=
\left\{\begin{array}{ll}
 \gamma\circ\sigma(v) & v=t(\gamma)\\
\gamma\circ\idf_{t(\gamma)} & v\neq t(\gamma)\\
\end{array}\right.\\
\eq whenever $t(\gamma)$ is the target vertex of a non-trivial path $\gamma$ in $\Gamma$. For a trivial path $\idf_v$ the right-translation is defined by $R_{\sigma(v)}(\idf_v)=\idf_{(t\circ\sigma)(v)}$ and $R_{\sigma(v)}(\idf_w)=\idf_{w}$ whenever $v\neq w$. The right-translation $R_{\sigma(v)}$ is required to be bijective. Before this result is proven in lemma \ref{lem path-diffeom} notice the following considerations. 
\end{rem}

Note that, $(R_{\sigma(v)},t\circ \sigma)$ transfers to the holonomy map such that
\beq\label{eq righttransl} (\ho_\Gamma\circ R_{\sigma(t(\gp))}(\gamma\circ\gp)&=\ho_\Gamma(\gamma\circ\gp\circ\sigma(t(\gp)))\\
&=\ho_\Gamma(\gamma)\ho_\Gamma(\gp\circ\sigma(t(\gp)))
\eq is true.
There is a bijective map between a right-translation $R_{\sigma(v)}:\PD_{\Gamma}\Sigma_v\rightarrow\PD_\Gamma\Sigma_{(t\circ\sigma)(v)}$ and a path-diffeomorphism $(\varphi_\Gamma,\Phi_\Gamma)$. In particular observe that, $\sigma\in\mathfrak{B}(\PD_\Gamma\Sigma_v)$ and $(\varphi_\Gamma,\Phi_\Gamma)\in\Diff(\PD_\Gamma\Sigma_v)$. Simply speaking the path-diffeomorphism does not change the source and target vertex at the same time. The path-diffomorphism changes the target vertex by a (finite) diffeomorphism and, therefore, the path is transformed. 

Bisections $\sigma$ in a finite path groupoid can be transfered, likewise path-diffeomorphisms, to holonomy maps. The pair $(\ho_\Gamma\circ \Phi_\Gamma, h_\Gamma\circ\varphi_\Gamma)$ of the maps
defines a pair of maps $(\ho_\Gamma\circ \Phi_\Gamma,h_\Gamma\circ\varphi_\Gamma)$ by 
\beq \ho_\Gamma\circ \Phi_\Gamma:\PD_\Gamma\Sigma_v\rightarrow G\text{ and } h_\Gamma\circ\varphi_\Gamma: V_\Gamma\rightarrow\{e_G\}
\eq which is a \hyperlink{holonomy map for a finite path groupoid}{holonomy map for a finite path groupoid} $\PD_\Gamma\Sigma$ over $V_\Gamma$.

\begin{lem}\label{lem groupbisection}The set $\mathfrak{B}(\PD_\Gamma\Sigma)$ of bisections on the finite path groupoid $\PD_\Gamma\Sigma$ over $V_\Gamma$ forms a group.
\end{lem}
\begin{proofs}The group multiplication is given by
\beqs (\sigma\ast\sigma^\prime)(v) =\sigma^\prime(v)\circ\sigma(t(\sigma^\prime(v)))\text{ for }v\in V_\Gamma
\eqs whenever $\sigma^\prime(v)\in\PD_\Gamma\Sigma^v_{\varphi_\Gamma^\prime(v)}$ and $\sigma(t(\sigma^\prime(v)))\in\PD_\Gamma\Sigma^{(t\circ\sigma^\prime)(v)}_{\varphi_\Gamma(v)}$.

Clearly the group multiplication is associative.
The unit $\id$ is equivalent to the object inclusion $v\mapsto\idf_v$ of the groupoid $\fPG$, where $\idf_v$ is the constant loop at $v$, and the inversion is given by
\beqs \sigma^{-1}(v)=\sigma((t\circ\sigma)^{-1}(v))^{-1}\text{ for }v\in V_\Gamma
\eqs 
\end{proofs}

The group property of bisections $\mathfrak{B}(\PD_\Gamma\Sigma)$ carries over to holonomy maps. Using the group multiplication $\cdot $ of $G$ conclude that
\beqs (\ho_\Gamma\circ R_{(\sigma\ast\sigma^\prime)(v)})(\idf_{v})=\ho_\Gamma\circ (R_{\sigma^\prime(v)}\circ R_{\sigma(t(\sigma^\prime(v)))})(\idf_{v})
= \ho_\Gamma(\sigma^\prime(v))\cdot\ho_\Gamma(\sigma(t(\sigma^\prime(v))))\text{ for }v\in V_\Gamma
\eqs is true. 
\begin{rem}
Moreover right-translations define path-diffeomorphisms, i.e. $R_{(\sigma)(v)}=\Phi_\Gamma$ and $\varphi_\Gamma=t\circ\sigma$ whenever $v\in V_\Gamma$.
But for two bisections $\sigma_\Gp,\breve\sigma_\Gp\in\mathfrak{B}(\PD_\Gamma\Sigma)$ the object $\sigma_\Gp(v)\circ\breve\sigma_\Gp(v)$ is not comparable with $(\sigma_\Gp\ast\breve\sigma_\Gp)(v)$. Then for the composition $\Phi_1(\gamma)\circ\Phi_2(\gamma)$, there exists no path-diffeomorphism $\Phi$ such that $\Phi_1(\gamma)\circ\Phi_2(\gamma)=\Phi(\gamma)$ yields in general. Moreover generally the object $\Phi_1(\gamma)\circ\Phi_2(\gp)=\Phi(\gamma\circ\gp)$ is not well-defined.

But the following is defined
\beq\label{def compositionofdiffeo} R_{(\sigma\ast\sigma^\prime)(v)}(\gamma)=\Phi_\Gamma^\prime(\gamma)\circ\Phi_\Gamma(\idf_{\varphi_\Gamma^\prime(v)})=:(\Phi_\Gamma^\prime\ast\Phi_\Gamma)(\gamma)
\eq whenever $\gamma\in\PD_\Gamma\Sigma_v$, $(\varphi_\Gamma,\Phi_\Gamma)\in\Diff(\PD_\Gamma\Sigma_v)$ and $(\varphi_\Gamma^\prime,\Phi_\Gamma^\prime)\in\Diff(\PD_\Gamma\Sigma_{\varphi_\Gamma^\prime(v)})$ are path-diffeomorphisms  such that $\varphi_\Gamma=t\circ\sigma$, $\Phi_\Gamma=R_{\sigma(\varphi_\Gamma^\prime(v))}$ and  $\varphi_\Gamma^\prime=t\circ\sigma^\prime$, $\Phi_\Gamma^\prime=R_{\sigma^\prime(v)}$.

Moreover for $(\gamma,\gp)\in\PD_\Gamma\Sigma^{(2)}$ and $\gp\in\PD_\Gamma\Sigma_v$ it is true that
\beqs (\Phi_\Gamma^\prime\ast\Phi_\Gamma)(\gamma\circ\gp)=\Phi_\Gamma^\prime(\gamma\circ\gp)\circ\Phi_\Gamma(\idf_{\varphi_\Gamma^\prime(v)}) = \Phi_\Gamma^\prime(\gamma)\circ\Phi_\Gamma^\prime(\gp)\circ\Phi_\Gamma(\idf_{\varphi_\Gamma^\prime(v)}) = \Phi_\Gamma^\prime(\gamma)\circ(\Phi_\Gamma^\prime\ast\Phi_\Gamma)(\gp)
\eqs holds.
\end{rem}

Then the following lemma easily follows.

\begin{lem}\label{lem path-diffeom}Let $\sigma$ be a bisection contained in $\mathfrak{B}(\PD_\Gamma\Sigma)$ and $v\in V_\Gamma$.

The pair $(R_{\sigma(v)},t\circ\sigma)$ of maps such that
\beqs &R_{\sigma(v)}:\PD_\Gamma\Sigma_v\rightarrow\PD_\Gamma\Sigma_{(t\circ \sigma)(v)},\quad &
s\circ R_{\sigma(v)} = (t\circ\sigma)\circ s\\
&t\circ\sigma:V_\Gamma\rightarrow V_\Gamma,\quad &t\circ R_{\sigma(v)}= (t\circ\sigma)\circ t
\eqs defined in remark \ref{rem defiofright} is a path-diffeomorphism in $\fPG$.
\end{lem}
\begin{proofs}This follows easily from the derivation
\beqs R_{\sigma(t(\gp))}(\gamma\circ\gp)&=\gamma\circ\gp\circ\sigma(t(\gp))
= R_{\sigma(t(\gp))}(\gamma)\circ R_{\sigma(t(\gp))}(\gp)
\eqs 
\beqs R_{\sigma(t(\gamma))}(\idf_{s(\gamma)}\circ\gamma)&=R_{\sigma(t(\gamma))}(\idf_{s(\gamma)})\circ R_{\sigma(t(\gamma))}(\gamma)=\idf_{s(\gamma)}\circ\gamma\circ\sigma(t(\gamma))\\
R_{\sigma(t(\gamma))}(\gamma\circ\idf_{t(\gamma)})&=R_{\sigma(t(\gamma))}(\gamma)\circ R_{\sigma(t(\gamma))}(\idf_{t(\gamma)})=\gamma\circ\sigma(t(\gamma))\circ\idf_{(t\circ\sigma)(t(\gamma))}
\eqs The inverse map satisfies 
\beqs R^{-1}_{\sigma(v)}(\gamma\circ\sigma(v))=R_{\sigma^{-1}(v)}(\gamma\circ\sigma(v))
=\gamma\circ\sigma(v)\circ\sigma^{-1}(v)=\gamma
\eqs whenever $v=t(\gamma)$, 
\beqs R^{-1}_{\sigma(v)}(\gamma)=\gamma
\eqs whenever $v\neq t(\gamma)$ and
\beqs R^{-1}_{\sigma(v)}(\idf_{(t\circ\sigma)(v)})=\idf_{v}
\eqs

Moreover derive
\beqs ( s\circ R_{\sigma(v)})(\gp)= ((t\circ\sigma)\circ s)(\gp)
\eqs for all $\gp\in\PD_\Gamma\Sigma_v$ and a fixed bisection $\sigma\in \mathfrak{B}(\PD_\Gamma\Sigma)$.
 \end{proofs}

Notice that, $L_{\sigma(v)}$ and $I_{\sigma(v)}$ similarly to the pair $(R_{\sigma(v)},t\circ\sigma)$ can be defined. Summarising the pairs $(R_{\sigma(v)},t\circ\sigma)$, $(L_{\sigma(v)},t\circ\sigma)$ and $(I_{\sigma(v)},t\circ\sigma)$ for a bisection $\sigma \in\mathfrak{B}(\PD_\Gamma\Sigma)$ are path-diffeomorphisms of a finite path groupoid $\fPG$.

In general a right-translation $(R_\sigma,t\circ\sigma)$ in the finite path groupoid $\PD_\Gamma\Sigma$ over $\Sigma$ for a bisection $\sigma\in\mathfrak{B}(\PD_\Gamma\Sigma)$ is defined by the bijective maps $R_\sigma$ and $t\circ\sigma$, which are given by 
\beqs 
&R_{\sigma}:\PD_\Gamma\Sigma\rightarrow\PD_\Gamma\Sigma,\quad &
s\circ R_{\sigma} =  s\qquad\quad\text{ }\\
&t\circ\sigma:V_\Gamma\rightarrow V_\Gamma,\quad &t\circ R_{\sigma}= (t\circ\sigma)\circ t\\
&R_\sigma(\gamma):=\gamma\circ\sigma(t(\gamma))\quad\forall \gamma\in \PD_\Gamma\Sigma;\quad R^{-1}_\sigma:=R_{\sigma^{-1}}
\eqs For example for a fixed suitable bisection $\sigma$ the right-translation is $R_\sigma(\idf_v)=\gamma$, then $R^{-1}_\sigma(\gamma)=\gamma\circ\gamma^{-1}=\idf_v$ for $v=s(\gamma)$. Clearly the right-translation $(R_\sigma,t\circ\sigma)$ is not a groupoid morphism in general. 

\begin{defi}
Define for a given bisection $\sigma$ in $\mathfrak{B}(\PD_\Gamma\Sigma)$, the \textbf{right-translation in the groupoid $G$ over $\{e_G\}$} through
\beqs &\ho_\Gamma\circ R_{\sigma}:\PD_\Gamma\Sigma\rightarrow G, \quad \gamma\mapsto (\ho_\Gamma\circ R_\sigma)(\gamma):=\ho_\Gamma(\gamma\circ\sigma(t(\gamma)))= \ho_\Gamma(\gamma)\cdot\ho_\Gamma(\sigma(t(\gamma))) \\
&h_\Gamma\circ t\circ\sigma:V_\Gamma\rightarrow e_G 
\eqs 

Furthermore for a fixed $\sigma\in\mathfrak{B}(\PD_\Gamma\Sigma)$ define
the \textbf{left-translation in the groupoid $G$ over $\{e_G\}$} by
\beqs &\ho_\Gamma\circ L_\sigma:\PD_\Gamma\Sigma\rightarrow G,\quad \gamma\mapsto \ho_\Gamma(\sigma((t\circ\sigma)^{-1}(s(\gamma)))\circ\gamma)=\ho_\Gamma(\sigma((t\circ\sigma)^{-1}(s(\gamma))))\cdot\ho_\Gamma(\gamma)\\
&h_\Gamma\circ t\circ\sigma:V_\Gamma\rightarrow e_G 
\eqs
and the \textbf{inner-translation in the groupoid $G$ over $\{e_G\}$}
\beqs &\ho_\Gamma\circ I_\sigma:\PD_\Gamma\Sigma\rightarrow G,\quad \gamma\mapsto \ho_\Gamma(\sigma((t\circ\sigma)^{-1}(s(\gamma)))\circ\gamma\circ\sigma(t(\gamma)))=\ho_\Gamma(\sigma((t\circ\sigma)^{-1}(s(\gamma))))\cdot\ho_\Gamma(\gamma)\cdot\ho_\Gamma(\sigma(t(\gamma)))\\
&h_\Gamma\circ t\circ\sigma:V_\Gamma\rightarrow e_G 
\eqs such that $I_\sigma=L_{\sigma^{-1}}\circ R_\sigma$.
\end{defi}

The pairs $(R_\sigma,t\circ\sigma)$ and $(L_\sigma,t\circ\sigma)$ are not groupoid morphisms. Whereas the pair $(I_\sigma,t\circ\sigma)$ is a groupoid morphism, since for all pairs $(\gamma,\gp)\in\PD_\Gamma\Sigma^{(2)}$ such that $t(\gamma)=s(\gp)$ it is true that $\sigma(t(\gamma)) \circ\sigma((t\circ\sigma)^{-1}(t(\gamma)))^{-1}=\idf_{t(\gamma)}$ holds. Notice that, in this situation $\sigma(t(\gamma))=\sigma(t(\gamma\circ\gp))$ is satisfied.

\begin{prop}
The map $\sigma\mapsto R_\sigma$ is a group isomorphism, i.e. $R_{\sigma\ast\sigma^\prime}=R_\sigma\circ R_{\sigma^\prime}$ and where $\sigma\mapsto t\circ\sigma$ is a group isomorphism from $\mathfrak{B}(\PD_\Gamma\Sigma)$ to the group of finite diffeomorphisms $\Diff(V_\Gamma)$ in a finite subset $V_\Gamma$ of $\Sigma$. 

The maps $\sigma\mapsto L_\sigma$ and $\sigma\mapsto I_\sigma$ are group isomorphisms.
\end{prop}

There is a generalisation of path-diffeomorphisms in the finite path groupoid, which coincide with the graphomorphism presented by Fleischhack in \cite{Fleischhack06}. In this approach the diffeomorphism $\varphi:\Sigma\rightarrow\Sigma$ changes the source and target vertex of a path $\gamma$. Consequently the path-diffeomorphism $(\Phi,\varphi)$, which implements the inner-translation $I_{\Phi}$ in the path groupoid $\PGs$, is a graphomorphism in the context of Fleischhack. Some element of the set of graphomorphisms is directly related to a right-translation $R_\sigma$ in the path groupoid.  Precisely for every $v\in\Sigma$ and $\sigma\in\mathfrak{B}(\PD\Sigma)$ the pairs $(R_{\sigma(v)},t\circ\sigma)$, $(L_{\sigma(v)},t\circ\sigma)$ and $(I_{\sigma(v)},t\circ\sigma)$ define graphomorphism. Furthermore the right-translation $R_{\sigma(v)}$,  the left-translation $L_{\sigma(v)}$ and the inner-translation $I_{\sigma(v)}$ are required to be bijective maps, and hence the maps cannot map non-trivial paths to trivial paths. This property restricts the set of all graphomorphism, which is generated by these translations. In particular in this article graph changing operations, which change the number of edges of a graph, are studied. Hence the left- or right-translation in a finite path groupoid is used in the further development. Notice that in general, these objects do not define graphomorphism.
Finally notice that, in particular for the graphomorphism $(R_{\sigma(v)},t\circ\sigma)$ and a holonomy map for the path groupoid $\PGs$ a similar relation \eqref{eq righttransl} holds. The last equation is fundamental for the construction of $C^*$-dynamical systems, which contain the analytic holonomy $C^*$-algebra restricted to a finite path groupoid $\fPG$ and a point norm continuous action of the finite path-diffeomorphism group $\Diff(V_\Gamma)$ on this algebra. Clearly the right-, left- and inner-translations $R_\sigma$, $L_\sigma$ and $I_\sigma$ are constructed such that \eqref{eq righttransl} generalises. But note that, in the infinite case considered by Fleischhack the action of the bisections $\mathfrak{B}(\PD\Sigma)$ are not point-norm continuous implemented. The advantage of the usage of bisections is that, the map $\sigma\mapsto t\circ\sigma$ is a group morphism between the group $\mathfrak{B}(\PD\Sigma)$ of bisections in $\PGs$ and the group $\Diff(\Sigma)$ of diffeomorphisms in $\Sigma$. Consequently there is an action of the group of diffeomorphisms in $\Sigma$ on the finite path groupoid, which is used to define an action of the group of diffemorphisms in $\Sigma$ on the analytic holonomy $C^*$-algebra. 

\paragraph*{Transformations in finite graph systems\\[5pt]}
To proceed it is necessary to transfer the notion of bisections and right-translations to finite graph systems. A right-translation $R_{\sigma_\Gamma}$  is a mapping that maps graphs to graphs. Each graph is a finite union of independent edges. This causes problems. Since the definition of right-translation in a finite graph system $\PD_\Gamma$ is often not well-defined for all bisections in the finite graph system and all graphs. For example if the graph $\Gamma:=\{\gamma_1,\gamma_2\}$ is disconnected and the bisection $\tilde\sigma$ in the finite path groupoid $\PD_{\Gamma}\Sigma$ over $V_\Gamma$ is defined by $\tilde\sigma(s(\gamma_1))=\gamma_1$, $\tilde\sigma(s(\gamma_2))=\gamma_2$, $\tilde\sigma(t(\gamma_1))=\gamma_1^{-1}$ and $\tilde\sigma(t(\gamma_2))=\gamma_2^{-1}$ where $V_\Gamma:=\{s(\gamma_1),t(\gamma_1),s(\gamma_2),t(\gamma_2)\}$. Let $\idf_\Gamma$ be the set given by the elements $\idf_{s(\gamma_1)}$,$\idf_{s(\gamma_2)}$,$\idf_{t(\gamma_1)}$ and $\idf_{t(\gamma_2)}$. Then notice that, a bisection $\sigma_\Gamma$, which maps a set of vertices in $V_\Gamma$ to a set of paths in $\PD_\Gamma\Sigma$, is given for example by $\sigma_\Gamma(V_\Gamma)= \{\gamma_1,\gamma_2,\gamma_1^{-1},\gamma_2^{-1}\}$. In this case the right-translation $R_{\sigma_\Gamma(V_\Gamma)}(\idf_\Gamma)$ is equivalent to $\{\gamma_1,\gamma_2,\gamma_1^{-1},\gamma_2^{-1}\}$, which is not a set of independent edges and hence not a graph. Loosely speaking the graph-diffeomorphism acts on all vertices in the set $V_\Gamma$ and hence implements four new edges. But a bisection $\sigma_\Gamma$, which maps a subset $V:= \{s(\gamma_1),s(\gamma_2)\}$ of $V_\Gamma$ to a set of paths,  leads to a translation $R_{\sigma_\Gamma(V)}(\{\idf_{s(\gamma_1)},\idf_{s(\gamma_2)}\})=\{\gamma_1,\gamma_2\}$, which is indeed a graph. Set $\Gp:=\{\gamma_1\}$ and $V^\prime=\{s(\gamma_1)\}$. Then observe that, for a restricted bisection, which maps a set $V^\prime$ of vertices in $V_\Gamma$ to a set of paths in $\PD_\Gp\Sigma$, the right-translation become $R_{\sigma_\Gp(V^\prime)}(\{\idf_{s(\gamma_1)}\})=\{\gamma_1\}$, which defines a graph, too. Notice that $\idf_{s(\gamma_1)}$ is a subgraph of $\Gp$. Hence in the simpliest case new edges are emerging. The next definition of the right-tranlation shows that composed paths arise, too.

\begin{defi}\label{defi bisecongraphgroupioid}
Let $\Gamma$ be a graph, $\fPG$ be a finite path groupoid and let $\PD_\Gamma$ be a finite graph system. Moreover the set $V_\Gamma$ is given by $\{v_1,...,v_{2N}\}$.

A \hypertarget{bisection of a finite graph system}{\textbf{bisection of a finite graph system}} $\PD_\Gamma$ is a map $\sigma_\Gamma:V_\Gamma\rightarrow\PD_\Gamma$ such that there exists a bisection $\tilde \sigma\in\mathfrak{B}(\PD_\Gamma\Sigma)$ such that $\sigma_\Gamma(V)=\{\tilde\sigma(v_i):v_i\in V\}$ whenever $V$ is a subset of $V_\Gamma$.

Define a restriction $\sigma_\Gp:V_\Gp\rightarrow\PD_\Gp$ of a bisection $\sigma_\Gamma$ in $\PD_\Gamma$ by
\beqs \sigma_\Gp(V):=\{ \tilde\sigma(w_k) : & w_k\in V\}
\eqs for each subgraph $\Gp$ of $\Gamma$ and $V\subseteq V_\Gp$. 

A \textbf{right-translation in the finite graph system} $\PD_\Gamma$ is a map $R_{\sigma_\Gp}: \PD_\Gp\rightarrow \PD_\Gp$, which is given by a  bisection $\sigma_\Gp:V_\Gp\rightarrow \PD_\Gp$ such that
\beqs
&R_{\sigma_\Gp}(\Gpp)= R_{\sigma_\Gp}(\{\gppe,...,\gppm,\idf_{w_i}:w_i\in\{s(\gpe),...,s(\gpk)\in V^s_{\Gp}:s(\gpi)\neq s(\gpj) \forall i\neq j\}\setminus V_{\Gpp}\})\\[4pt]
&:=\left\{
\begin{array}{ll}
\gppe,...,\gppj,\gppje\circ\tilde\sigma(t(\gppje)),...,\gppm\circ\tilde\sigma(t(\gppm)),\idf_{w_i}\circ\tilde\sigma(w_i) : & \\[4pt]
w_i\in \{s(\gpe),...,s(\gpk)\in V^s_{\Gp}:s(\gpi)\neq s(\gpj) \forall i\neq j\}\setminus V_{\Gpp}, \quad t(\gppi)\neq t(\gppl)\quad\forall i\neq l; i,l\in\bra j+1,M\ket &\\
\end{array}\right\}\\
&=\Gamma^{\prime\prime}_\sigma\\
\eqs where $\tilde\sigma\in\mathfrak{B}(\PD_{\Gamma}\Sigma)$, $K:=\vert\Gp\vert$ and $M:=\vert\Gpp\vert$, $V^s_{\Gp}$ is the set of all source vertices of $\Gp$ and such that $\Gpp:=\{\gppe,...,\gppm\}$ is a subgraph of $\Gp:=\{\gpe,...,\gpk\}$ and $\Gamma^{\prime\prime}_\sigma$ is a subgraph of $\Gp$.
\end{defi} 

Derive that, for $\tilde\sigma(t(\gamma_i))=\gamma_i^{-1}$ it is true that
$(t\circ\tilde\sigma)(s(\gamma_i^{-1}))=s(\gamma_i)=(t\circ\tilde\sigma)(t(\gamma_i))$ holds.

\begin{exa}
Let $\Gamma$ be a disconnected graph.
Then for a bisection $\tilde\sigma\in\mathfrak{B}(\PD_\Gamma\Sigma)$ such that $\sigma(t(\gamma_i))=\gamma_i^{-1}$ for all $1\leq i\leq \vert\Gamma\vert$ it is true that 
\beqs R_{\sigma_\Gamma}(\Gamma)&=\Big\{\gamma_1\circ\tilde\sigma(t(\gamma_1)),...,\gamma_N\circ\tilde\sigma(t(\gamma_N)),\idf_{s(\gamma_1)}\circ\tilde\sigma(s(\gamma_1)),...,\idf_{s(\gamma_N)}\circ\tilde\sigma(s(\gamma_N))\Big\}\\
&=\{\idf_{s(\gamma_1)},...,\idf_{s(\gamma_N)}\}
\eqs yields. Set $\Gp:=\{\gpe,...,\gpm\}$, then derive
\beqs R_{\sigma_\Gamma}(\Gp)&=\Big\{\gpe\circ\tilde\sigma(t(\gpe)),...,\gpm\circ\tilde\sigma(t(\gpm)),\idf_{s(\gamma_1)}\circ\tilde\sigma(s(\gamma_1)),...,\idf_{s(\gamma_{N-M})}\circ\tilde\sigma(s(\gamma_{N-M}))\Big\}\\
&=\{\idf_{s(\gpe)},...,\idf_{s(\gpm)},\gamma_{1},...,\gamma_{N-M}\}
\eqs if $\Gamma=\Gp\cup\{\gamma_1,...,\gamma_{N-M}\}$.
\end{exa}
To understand the definition of the right-translation notice the following problem.

\begin{problem}\label{prob withoutcond} Consider a subgraph $\Gamma$ of $\tilde\Gamma:=\{\gamma_1,\gamma_2,\gamma_3,\gamma_4\}$, a map $\tilde\sigma:V_{\tilde\Gamma}\rightarrow \PD_{\tilde\Gamma}\Sigma$. Then the map 
\beqs R_{\sigma_{\tilde\Gamma}}(\Gamma)=\{\gamma_1\circ\gamma_1^{-1},\gamma_2\circ\idf_{t(\gamma_2)},\gamma_3\circ\idf_{t(\gamma_3)},\idf_{s(\gamma_1)}\circ\gamma_4 \}=:\Gamma_\sigma
\eqs 
 \begin{center}
\includegraphics[width=0.5\textwidth]{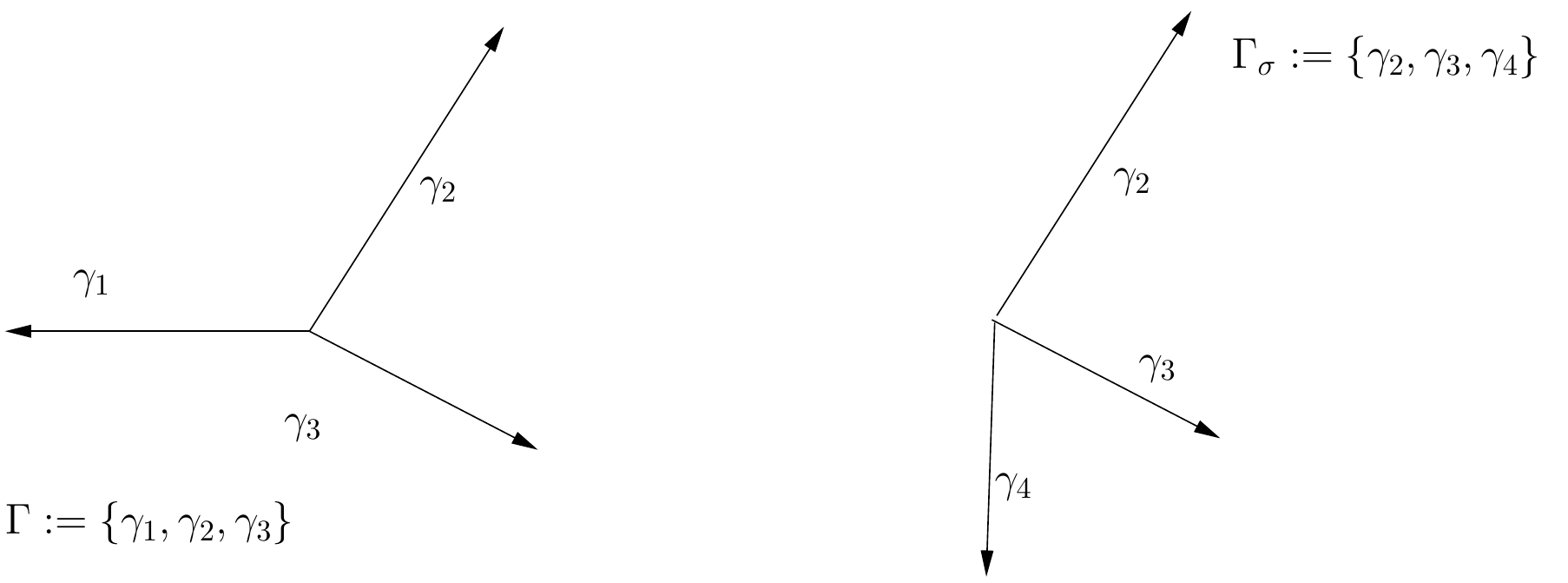}
\end{center} is not a right-translation. This follows from the following fact. Notice that,  the map $\sigma$ maps $t(\gamma_1)\mapsto s(\gamma_1)$, $t(\gamma_2)\mapsto t(\gamma_2)$, $t(\gamma_3)\mapsto t(\gamma_3)$ and $s(\gamma_1)\mapsto t(\gamma_4)$. Then the map $\tilde\sigma$ is not a bisection in the finite path groupoid $\PD_{\tilde\Gamma}\Sigma$ over $V_{\tilde\Gamma}$ and does not define a right-translation $R_{\sigma_{\tilde\Gamma}}$ in the finite graph system $\PD_{\tilde\Gamma}$.

This is a general problem. For every bisection $\tilde\sigma$ in a finite path groupoid such that a graph $\Gamma:=\{\gamma\}$ is translated to $\{\gamma\circ\tilde\sigma(t(\gamma),\tilde\sigma(s(\gamma))\}$. Hence either such translations in the graph system are excluded or the definition of the bisections has to be restricted to maps such that the map $t\circ\tilde\sigma$ is not bijective.  Clearly, the restriction of the right-translation such that $\Gamma$ is mapped to $\{\gamma\circ\tilde\sigma(t(\gamma),\idf_{s(\gamma)}\}$ implies that a simple path orientation transformation is not implemented by a right-translation.

Furthermore there is an ambiguity for graph containing to paths $\gamma_1$ and $\gamma_2$ such that $t(\gamma_1)=t(\gamma_2)$. Since in this case a bisection $\sigma$, which maps $t(\gamma_1)$ to $t(\gamma_3)$, the right-translation is $\{\gamma_1\circ\gamma_3,\gamma_2\circ\gamma_3\}$, is not a graph anymore. 
\end{problem}

\begin{exa}
Otherwise there is for example a subgraph $\Gp$ of $\tilde\Gamma:=\{\gamma_1,\gamma_2,\gamma_3,\gamma_4\}$ and a bisection $\tilde\sigma_{\tilde\Gamma}$ such that 
\beqs \Gamma^\prime_\sigma:=\{\gamma_1\circ\gamma_1^{-1},\gamma_2\circ\idf_{s(\gamma_2)},\gamma_3\circ\idf_{s(\gamma_3)}
\}
\eqs
 \begin{center}
\includegraphics[width=0.5\textwidth]{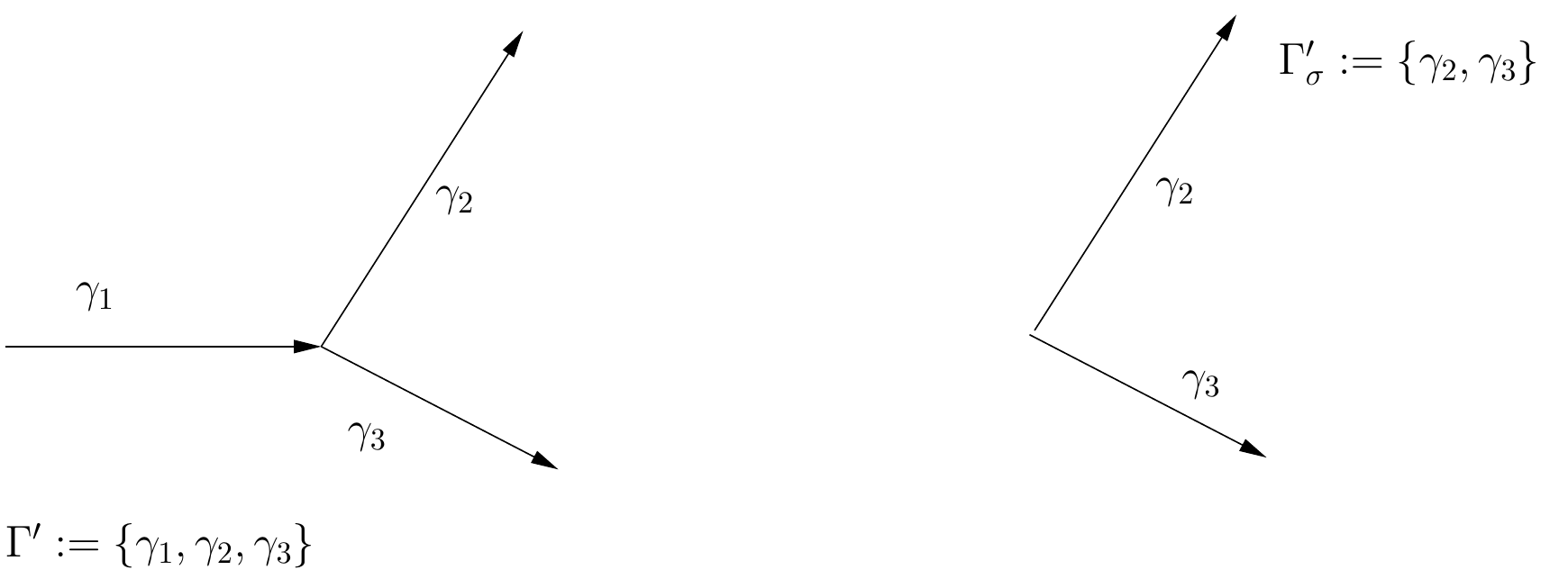}
\end{center} Notice that, $t(\gamma_1)\mapsto s(\gamma_1)$, $t(\gamma_2)\mapsto t(\gamma_2)$, $t(\gamma_3)\mapsto t(\gamma_3)$ and $t(\gamma_4)\mapsto t(\gamma_4)$. Hence the the map $\tilde\sigma_{\tilde\Gamma}:V_{\tilde\Gamma}\rightarrow\PD_{\tilde\Gamma}\Sigma$ is bijective map and consequently a bisection. The bisection $\sigma_{\tilde\Gamma}$ in the graph system $\PD_{\tilde\Gamma}$ defines a right-translation $R_{\sigma_{\tilde\Gamma}}$ in $\PD_{\tilde\Gamma}$.

Moreover for a subgraph $\Gpp:=\{\gamma_2,\gamma_3\}$ of the graph $\breve\Gamma:=\{\gamma_1,\gamma_2,\gamma_3\}$ there exists a map $\sigma_{\breve \Gamma}:V_{\breve\Gamma}\rightarrow\PD_{\breve\Gamma}$ such that
\beqs R_{\sigma_{\breve\Gamma}}(\Gpp)=\{\gamma_2,\gamma_3,\tilde\sigma(s(\gamma_1))\}=\{\gamma_2,\gamma_3,\gamma_1\}
\eqs
\begin{center}
\includegraphics[width=0.5\textwidth]{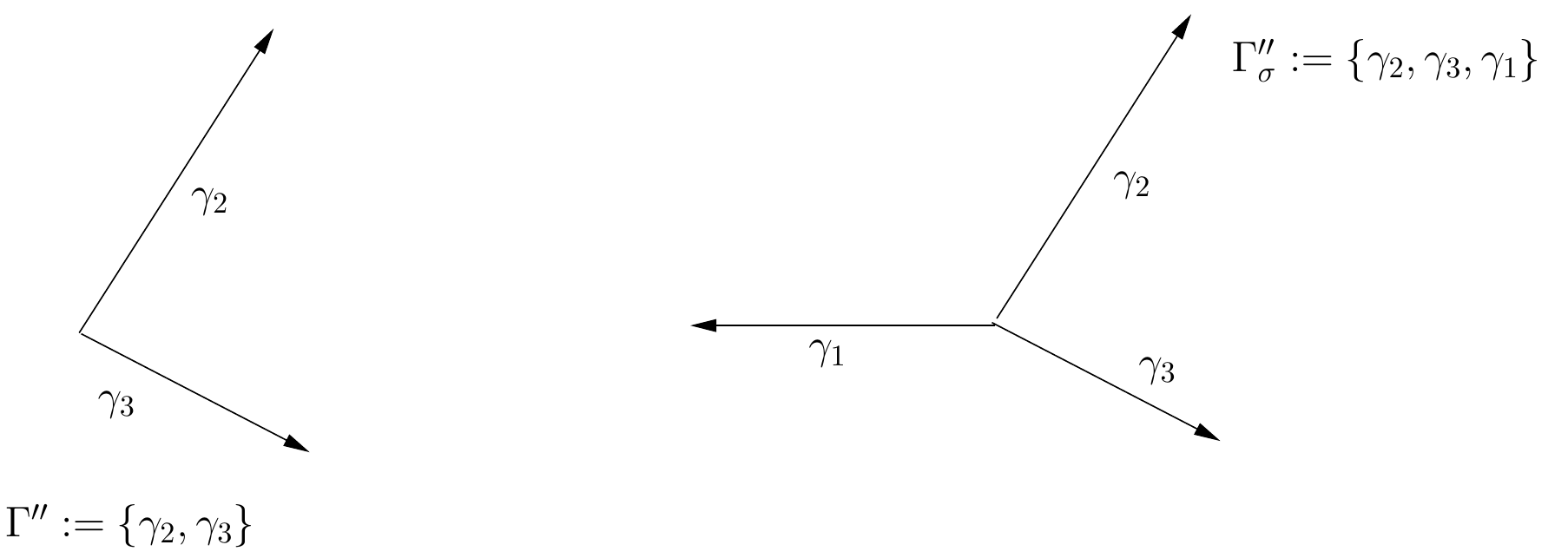}
\end{center}
where $t(\gamma_2)\mapsto t(\gamma_2)$, $t(\gamma_3)\mapsto t(\gamma_3)$ and $s(\gamma_1)\mapsto t(\gamma_1)$. Consequently in this example the map $\tilde\sigma_{\breve \Gamma}$ is a bisection, which defines a right-translation in $\PD_{\breve\Gamma}$.

Note that, for a graph $\Gamma$ such that $\tilde\Gamma$ and $\breve\Gamma$ are subgraphs the bisection $\sigma_{\tilde\Gamma}$ extends to a bisection $\sigma$ in $\PD_\Gamma$ and $\sigma_{\breve\Gamma}$ extends to a bisection $\breve\sigma$ in $\PD_\Gamma$.
\end{exa}

Moreover the bisections of a finite graph system are transfered, analogously, to bisections of a finite path groupoid $\fPG$ to the group $G^{\vert\Gamma\vert}$. Let $\sigma\in\mathfrak{B}(\PD_\Gamma)$ and $(\ho_\Gamma,h_\Gamma)\in\Hom(\PD_\Gamma,G^ {\vert\Gamma\vert})$. Thus there are two maps
\beq \ho_\Gamma\circ R_\sigma:\PD_\Gamma\rightarrow G^{\vert\Gamma\vert}\text{ and }h_\Gamma\circ(t\circ\sigma):V_\Gamma\rightarrow \{e_G\}
\eq which defines a \hyperlink{holonomy map for a finite graph system}{holonomy map for a finite graph system} if $\sigma$ is suitable.

Now, a similar right-translation in a finite graph system in comparison to the right-translation $R_{\sigma(v)}$ in a finite path groupoid is studied.
Let $\sigma_\Gp:V_\Gp\rightarrow\PD_\Gp$ be a restriction of $\sigma_\Gamma\in\mathfrak{B}(\PD_\Gamma)$.  Moreover let $V$ be a subset of $V_ \Gp$, let $\Gpp$ be a subgraph of $\Gp$ and $\Gppp$ be a subgraph of $\Gpp$. Then a right-translation is given by
\beqs &R_{\sigma_\Gp(V)}(\Gpp)\\
&:= 
\left\{\begin{array}{ll}
 R_{\sigma_\Gp}(\{\gppe,...,\gppm,\idf_{w_i}:w_i\in\{s(\gpe),...,s(\gpk)\in V_{\Gp}:s(\gpi)\neq s(\gpj) \forall i\neq j\}\setminus V_{\Gpp}\}& : V_{\Gpp}\subset V\\
R_{\sigma_\Gp}(\{\gppe,...,\gpppn,\idf_{w_i}:w_i\in\{s(\gpe),...,s(\gpk)\in V_{\Gpp}:s(\gpi)\neq s(\gpj) \forall i\neq j\}\setminus V_{\Gppp}\})& \\
\quad \cup\{\idf_{x_i}:x_i\in V\setminus  V_{\Gpp}\}\cup\{\Gpp\setminus\Gppp\}
&: V_{\Gpp}\not\subset V,V_{\Gppp}\subset V\\
       \end{array}\right.
\eqs Loosely speaking, the action of a path-diffeomorphism is somehow localised on a fixed vertex set $V$. 

For example note that for a subgraph $\Gp:=\{\gamma\circ\gp\}$ of $\Gamma:=\{\gamma,\gp\}$ and a subset $V:=\{t(\gp)\}$ of $V_\Gamma$, it is true that
\beqs (\ho_\Gamma\circ R_{\sigma_{\Gamma}(V)})(\gamma\circ\gp)= (\ho_\Gamma\circ R_{\sigma_{\Gamma}(V)})(\gamma)\cdot (\ho_\Gamma\circ R_{\sigma_{\Gamma}(V)})(\gp)= \ho_\Gamma(\gamma)\cdot (\ho_\Gamma\circ R_{\sigma_{\Gamma}(V)})(\gp)
=\ho_\Gamma(\gamma\circ\gp\circ\sigma(t( \gp)))
\eqs yields whenever $\sigma_{\Gamma}\in\mathfrak{B}(\PD_\Gamma\Sigma)$. For a special bisection $\breve\sigma_\Gamma$ it is true that,
\beqs (\ho_\Gamma\circ R_{\breve\sigma_{\Gamma}})(\gamma)= \ho_\Gamma(\gamma\circ\gp)= (\ho_\Gamma\circ R_{\breve\sigma_{\Gamma}})(\gamma)\cdot (\ho_\Gamma\circ R_{\breve\sigma_{\Gamma}})(\gp)
\eqs holds whenever $\breve\sigma_\Gamma\in\mathfrak{B}(\PD_\Gamma\Sigma)$, $\breve\sigma_\Gamma(t(\gp))=\idf_{t(\gp)}$ and $\breve\sigma_\Gamma(t(\gamma))=\gp$. Let $\tilde\sigma$ be the bisection in the finite path groupoid $\PD_\Gamma\Sigma$ that defines the bisection $\breve\sigma$ in $\PD_\Gamma$. Then the last statement is true, since $R_{\breve\sigma_{\Gamma}}(\gp)=\gp\circ\gp^{-1}$ requires $\tilde\sigma_\Gamma:t(\gp)\mapsto s(\gp)$ and $R_{\breve\sigma_{\Gamma}}(\gamma)=\gamma\circ\gp$ needs $\tilde\sigma_\Gamma: t(\gamma)\mapsto t(\gp)$, where $s(\gp)=t(\gamma)$. Then
$R_{\breve\sigma_\Gamma}(\gamma)$ and $R_{\sigma_\Gamma(t(\gp))}(\gamma)$ coincide if $\breve\sigma_\Gamma(t(\gamma))=\sigma_\Gamma(t(\gamma))$ and $\breve\sigma_\Gamma(t(\gp))=\idf_{t(\gp)}$ holds. 
 
\begin{problem}\label{prob righttranslgrouoid}Let $\Gp$ be a subgraph of the graph $\Gamma$, $\sigma_\Gamma$ be a bisection in $\PD_\Gamma$, $\sigma_\Gp:V_\Gp\rightarrow\PD_\Gp$ be a restriction of $\sigma_\Gamma\in\mathfrak{B}(\PD_\Gamma)$.  Moreover let $V$ be a subset of $V_ \Gp$, let $\Gpp:=\{\gamma\circ\gp\}$ be a subgraph of $\Gp$. Let $(\gamma,\gp)\in\PD_\Gp\Sigma^{(2)}$.

Then even for a suitable bisection $\sigma_\Gp$ in $\PD_\Gamma$ it follows that,
\beq\label{eq ineqsolv} R_{\sigma_{\Gp}(V)}(\gamma\circ\gp)\neq R_{\sigma_{\Gp}(V)}(\gamma)\circ R_{\sigma_{\Gp}(V)}(\gp) 
\eq yields. This is a general problem. In comparison with problem \ref{subsec fingraphpathgroup}.\ref{problem group structure on graphs systems} the multiplication map $\circ$ is not well-defined and hence
 \beqs R_{\sigma_{\Gp}(V)}(\gamma)\circ R_{\sigma_{\Gp}(V)}(\gp)
\eqs  is not well-defined. Recognize that, $R_{\sigma_{\Gp}(V)}:\PD_{\Gamma}\rightarrow\PD_\Gamma$.

Consequently in general it is not true that, 
 \beq\label{eq ineqsolv2}(\ho_\Gamma\circ R_{\sigma_{\Gp}(V)})(\gamma\circ\gp)=\ho_\Gamma(R_{\sigma_{\Gp}(V)}(\gamma)\circ R_{\sigma_{\Gp}(V)}(\gp))=(\ho_\Gamma\circ R_{\sigma_{\Gp}(V)})(\gamma)\cdot (\ho_\Gamma\circ R_{\sigma_{\Gp}(V)})(\gp)
\eq yields.
\end{problem}

With no doubt the left-tranlation $L_{\sigma_\Gp}$ and the inner autmorphisms $I_{\sigma_\Gp}$ in a finite graph system $\PD_\Gamma$ for every $\Gp\in\PD_\Gamma$ are defined similarly.

\begin{defi}Let $\sigma_\Gamma\in \mathfrak{B}(\PD_\Gamma)$ be a bisection in the finite graph system $\PD_\Gamma$. Let $R_{\sigma_\Gamma(V)}$ be a right-tranlation, where $V$ is a subset of $V_\Gamma$.

Then the pair $(\Phi_\Gamma,\varphi_\Gamma)$ defined by $\Phi_\Gamma=R_{\sigma_\Gamma(V)}$ (or, respectively, $\Phi_\Gamma=L_{\sigma_\Gamma(V)}$, or $\Phi_\Gamma=I_{\sigma_\Gamma(V)}$) for a subset $V\subseteq V_\Gamma$ and $\varphi_\Gamma=t\circ \sigma_\Gamma$ is called a \textbf{graph-diffeomorphism of a finite graph system}. 
Denote the set of finite graph-diffeomorphisms by $\Diff(\PD_\Gamma)$.
\end{defi}

Let $\Gp$ be a subgraph of $\Gamma$ and $\sigma_\Gp$ be a restriction of  bisection $\sigma_\Gamma$ in $\PD_\Gamma$. Then for example another graph-diffeomorphism $(\Phi_\Gp,\varphi_\Gp)$ in $\Diff(\PD_\Gamma)$ is defined by $\Phi_\Gp=R_{{\sigma_\Gp}(V)}$ for a subset $V\subseteq V_\Gp$ and $\varphi_\Gp=t\circ \sigma_\Gp$.

Remembering that the set of bisections of a finite path groupoid forms a group (refer \ref{lem groupbisection}) one may ask if the bisections of a finite graph system form a group, too. 

\begin{prop}\label{lemma bisecform}The set of bisections $\mathfrak{B}(\PD_\Gamma)$ in a finite graph system $\PD_\Gamma$ forms a group.
\end{prop}
\begin{proofs}Let $\Gamma$ be a graph and let $V_\Gamma$ be equivalent to the set $\{v_1,...,v_{2N}\}$.
  
First two different multiplication operations are studied. The studies are comparable with the results of the definition \ref{defi bisecongraphgroupioid} of a right-translation in a finite graph system. The easiest multiplication operation is given by $\ast_1$, which is defined by 
\beqs (\sigma\ast_1\sigma^\prime)(V_\Gamma)&:=\{
(\tilde\sigma\ast \tilde\sigma^\prime)(v_1),...,(\tilde\sigma\ast \tilde\sigma^\prime)(v_{2N}):v_i\in V_\Gamma\}
\eqs where $\ast$ denotes the multiplication of bisections on the finite path groupoid $\fPG$. Notice that, this operation is not well-defined in general. In comparison with the definition of the right-translation in a finite graph system one has to take care. First the set of vertices doesn't contain any vertices twice, the map $\sigma$ in the finite path system is bijective, the mapping $\sigma$ maps each set to a set of vertices containing no  
vertices twice and the situation in problem \thesection.\ref{prob withoutcond} has to be avoided.

Fix a bisection $\tilde\sigma$ in a finite path groupoid $\fPG$.  Let $V_{\sigma^\prime}$ be a subset of $V_\Gamma$ where $\Gamma:=\{\gamma_1,...,\gamma_N\}$ and for each $v_i$ in $V_{\sigma^\prime}$ it is true that $v_i\neq v_j$ and $v_i\neq (t\circ\tilde\sigma^\prime)(v_j)$ for all $i\neq j$. Define the set $V_{\sigma,\sigma^\prime}$ to be equal to a subset of the set of all vertices $\{v_k\in V_{\sigma^\prime}: 1\leq k\leq 2N\}$ such that each pair $(v_i,v_j)$ of vertices in $V_{\sigma,\sigma^\prime}$ satisfies $(t\circ(\tilde\sigma\ast\tilde\sigma^\prime))(v_i)\neq (t\circ\tilde\sigma^\prime)(v_j)$ and $(t\circ\tilde\sigma^\prime)(v_i)\neq (t\circ\tilde\sigma^\prime)(v_j)$ for all $i\neq j$.
Define
\beqs W_{\sigma,\sigma^\prime}:=\Big\{
w_i\in \{V_{\sigma}\cap V_{\sigma^\prime}\}\setminus V_{\sigma,\sigma^\prime}:
& (t\circ\tilde\sigma)(w_j)\neq (t\circ\tilde\sigma^\prime)(w_i)\quad\forall i\neq j,\quad 1\leq i,j\leq l\Big\}
\eqs
The set $V_{\sigma,\sigma^\prime,\breve \sigma}$ is a subset of all vertices $\{v_k\in V_{\sigma,\sigma^\prime}: 1\leq k\leq 2N\}$ such that each pair $(v_i,v_j)$ of vertices in $V_{\sigma,\sigma^\prime,\breve\sigma}$ satisfies $(t\circ(\breve\sigma\ast\tilde\sigma\ast\tilde\sigma^\prime))(v_i)\neq (t\circ(\tilde\sigma\ast\tilde\sigma^\prime))(v_j)$ and $(t\circ(\tilde\sigma\ast\tilde\sigma^\prime))(v_i)\neq (t\circ(\tilde\sigma\ast\tilde\sigma^\prime))(v_j)$ for all $i\neq j$.

Consequently define a second multiplication on $\mathfrak{B}(\PD_\Gamma)$ similarly to the operation $\ast_1$. This is done by the following definition. Set
\beqs (\sigma\ast_2\sigma^\prime)(V_{\sigma^\prime})
:=&\left\{
(\tilde\sigma\ast \sigma^\prime)(v_1),...,(\tilde\sigma\ast\sigma^\prime)(v_{k}):
v_1,...,v_k\in V_{\sigma,\sigma^\prime}, 1\leq k\leq 2N\right\}\\
&\cup\left\{ \tilde\sigma(w_1),\sigma^\prime(w_1)...,\tilde\sigma(w_{l}),\sigma^\prime(w_{l}): w_1,...,w_l\in W_{\sigma,\sigma^\prime}, 1\leq l\leq 2N
\right\}\\
&\cup\left\{\idf_{p_1},...,\idf_{p_n}:p_1,...,p_n\in V_{\sigma^\prime}\setminus \{V_{\sigma,\sigma^\prime}\cup W_{\sigma,\sigma^\prime}\}, 1\leq n\leq 2N\right\}
\eqs 

Hence the inverse is supposed to be $\sigma^{-1}(V_\Gamma)=\sigma((t\circ\sigma)^{-1}(V_\Gamma))^{-1}$ such that
\beqs (\sigma\ast_2\sigma^{-1})(V_{\sigma^{-1}})=&\{ (\tilde\sigma\ast \tilde\sigma^{-1})(v_1),...,(\tilde\sigma\ast \tilde\sigma^{-1})(v_{2N}):v_i\in V_{\sigma,\sigma^{-1}}\}\\
&\cup\left\{ \tilde\sigma(w_1),\sigma^\prime(w_1)^{-1}...,\tilde\sigma(w_{l}),\sigma^\prime(w_{l})^{-1}: w_1,...,w_l\in W_{\sigma,\sigma^{-1}}, 1\leq l\leq 2N
\right\}\\
&\cup\left\{\idf_{p_1},...,\idf_{p_n}:p_1,...,p_n\in V_{\sigma^\prime}\setminus \{V_{\sigma,\sigma^{-1}}\cup W_{\sigma,\sigma^{-1}}\}, 1\leq n\leq 2N\right\}
\eqs
\end{proofs} 
Notice that, the problem \thesection.\ref{prob righttranslgrouoid} is solved by a multiplication operation $\circ_2$, which is defined similarly to $\ast_2$. Hence the equality of \eqref{eq ineqsolv} is available and consequently \eqref{eq ineqsolv2} is true. Furthermore a similar remark to \ref{def compositionofdiffeo} can be done.

\begin{exa}
Now consider the following example. Set $\Gp:=\{\gamma_1,\gamma_3\}$, let $\Gamma:=\{\gamma_1,\gamma_2,\gamma_3\}$ and\\ $V_\Gamma:=\{s(\gamma_1),t(\gamma_1),s(\gamma_2),t(\gamma_2),s(\gamma_3),t(\gamma_3): s(\gamma_i)\neq s(\gamma_j),t(\gamma_i)\neq t(\gamma_j)\text{ }\forall i\neq j\}$. \\
Set $V$ be equal to $\{s(\gamma_1),s(\gamma_2),s(\gamma_3)\}$. Take two maps $\sigma$ and $\sigma^\prime$ such that $\sigma^\prime(V)=\{\gamma_1,\gamma_3\}$, $\sigma(V)=\{ \gamma_2\}$, where $(t\circ\tilde\sigma)(s(\gamma_3))=t(\gamma_3)$, $\tilde\sigma^\prime(s(\gamma_3))=\gamma_3$, $\tilde\sigma^\prime(s(\gamma_1))=\gamma_1$ and  $\tilde\sigma(t(\gamma_3))=\gamma_2$. Then $s(\gamma_3)\in V_{\sigma_{\Gp},\sigma^\prime_\Gp}$ and $s(\gamma_1)\in W_{\sigma_{\Gp},\sigma^\prime_\Gp}$.
Derive 
\beqs (\sigma\ast_1\sigma^\prime)(V)= \{\gamma_3\circ\gamma_2,\gamma_1\}
\eqs 
Then conclude that,
\beqs (\sigma\ast_2\sigma^\prime)(V_\Gamma)= \{\gamma_3\circ\gamma_2,\gamma_1\}
\eqs holds. Notice that
\beqs (\sigma\ast_2\sigma^\prime)(V)\neq (\sigma^\prime\ast_2\sigma)(V)= \{\gamma_2,\gamma_1,\gamma_3\}
\eqs is true.
Finally obtain
\beqs (\sigma\ast_2\sigma^{-1})(V_\Gamma)= \{\gamma_3\circ\gamma_3^{-1},\gamma_1\circ\gamma_1^{-1}\}=\{\idf_{s(\gamma_3)},\idf_{s(\gamma_1)}\}
\eqs 
Let $\sigma^\prime(V_\Gamma)=\{\gamma_1,\gamma_3\}$ and $\breve \sigma(V_\Gamma)=\{\gamma_2,\gamma_4\}$. Then notice that,
\beqs (\breve \sigma\ast_1\sigma^\prime)(V_\Gamma)=\{\gamma_3\circ\gamma_2,\gamma_1\}
\eqs and 
\beqs (\breve \sigma\ast_2\sigma^\prime)(V_\Gamma)=\{\gamma_3\circ\gamma_2,\gamma_1,\gamma_4\}
\eqs yields.

Furthermore assume supplementary that $t(\gamma_3)=t(\gamma_1)$ holds. Then calculate the product of the maps $\sigma$ and $\sigma^\prime$:
\beqs (\sigma\ast_1\sigma^\prime)(V)= \{\gamma_3\circ\gamma_2,\gamma_1\circ\gamma_2\}\notin\PD_\Gamma
\eqs  and
\beqs (\sigma\ast_2\sigma^\prime)(V_\Gamma)= \{\idf_{t(\gamma_1)}, \idf_{t(\gamma_3)}\}\in\PD_\Gamma
\eqs
\end{exa}

The group structure of $\mathfrak{B}(\PD_\Gamma)$ transferes to $G$. Let $\tilde\sigma$ be a bisection in the finite path groupoid $\fPSGm$, which defines a bisection $\sigma$ in $\PD_\Gamma$ and let $\tilde\sigma^\prime$ be a bisection in $\fPSGm$, which defines another bisection $\sigma^\prime$ in $\PD_\Gamma$. Let $V_{\sigma,\sigma^\prime}$ be equal to $V_\Gamma$, then derive
\beq &\ho_\Gamma\left((\sigma\ast_2\sigma^\prime)(V_\Gamma)\right) 
= \{\ho_\Gamma((\tilde\sigma\ast\sigma^\prime)(v_1)),...,\ho_\Gamma((\tilde\sigma\ast\sigma^\prime)(v_{2N})) \}\\ 
&=\ho_\Gamma(\sigma^\prime(V_\Gamma)\circ\sigma(t(\sigma^\prime(V_\Gamma))))
=\{\ho_\Gamma(\sigma^\prime(v)\circ\tilde\sigma(t(\sigma^\prime(v_1)))),...,\ho_\Gamma(\sigma^\prime(v_{N})\circ\tilde\sigma(t(\sigma^\prime(v_{N}))))\}\\
&=\{\ho_\Gamma(\sigma^\prime(v))\ho_\Gamma(\tilde\sigma(t(\sigma^\prime(v_1)))),...,\ho_\Gamma(\sigma^\prime(v_{N}))\ho_\Gamma(\tilde\sigma(t(\sigma^\prime(v_{N}))))\}\\
&= \ho_\Gamma(\sigma^\prime(V_\Gamma))\ho_\Gamma(\sigma(V_\Gamma))
\eq
Consequently the right-translation in the finite product $G^{\vert\Gamma\vert}$ is definable.

\begin{defi}Let $\sigma_\Gp$ be in $\mathfrak{B}(\PD_\Gamma)$, $\Gp$ a subgraph of $\Gamma$, $\Gpp$ a subgraph of $\Gp$ and $R_{\sigma_\Gp}$ a right-translation, $L_{\sigma_\Gp}$ a left-translation and $I_{\sigma_\Gp}$ an inner-translation in $\PD_\Gamma$.

Then the \textbf{right-translation in the finite product $G^{\vert\Gamma\vert}$} is given by
\beqs \ho_\Gamma\circ R_{\sigma_\Gp}:\PD_\Gamma\rightarrow G^{\vert\Gamma\vert}, \quad \Gpp\mapsto (\ho_\Gamma\circ R_{\sigma_\Gp})(\Gpp)
\eqs
Furthermore define the \textbf{left-translation in the finite product $G^{\vert\Gamma\vert}$} by
\beqs \ho_\Gamma\circ L_{\sigma_\Gp}:\PD_\Gamma\rightarrow G^{\vert\Gamma\vert},\quad \Gpp\mapsto (\ho_\Gamma\circ L_{\sigma_\Gp})(\Gpp)
\eqs
and the \textbf{inner-translation in the finite product $G^{\vert\Gamma\vert}$}
\beqs \ho_\Gamma\circ I_{\sigma_\Gp}:\PD_\Gamma\rightarrow G^{\vert\Gamma\vert},\quad \Gpp\mapsto (\ho_\Gamma\circ I_{\sigma_\Gp})(\Gpp)
\eqs such that $I_{\sigma_\Gp}=L_{\sigma_\Gp^{-1}}\circ R_{\sigma_\Gp}$.
\end{defi}
\begin{lem}It is true that $R_{\sigma_\Gp\ast_2\sigma_\Gp^\prime}=R_{\sigma_\Gp}\circ R_{\sigma_\Gp^\prime}$, $L_{\sigma_\Gp\ast_2\sigma_\Gp^\prime}=L_{\sigma_\Gp}\circ L_{\sigma_\Gp^\prime}$ and $I_{\sigma_\Gp\ast_2\sigma_\Gp^\prime}=I_{\sigma_\Gp}\circ I_{\sigma_\Gp^\prime}$ for all bisections $\sigma_\Gp$ and $\sigma^\prime_\Gp$ in $\mathfrak{B}(\PD_\Gamma)$.
\end{lem}

There is an action of $\mathfrak {B}(\PD_\Gamma)$ on $G^{\vert \Gamma\vert}$ by
\beqs (\zeta_{\sigma_\Gp}\circ\ho_\Gamma)(\Gpp):= (\ho_\Gamma\circ R_{\sigma_\Gp})(\Gpp)
\eqs whenever $\sigma_\Gp\in \mathfrak {B}(\PD_\Gamma)$, $\Gpp\in\PD_\Gp$ and $\Gp\in \PD_\Gamma$.
Then for another $\breve\sigma\in \mathfrak {B}(\PD_\Gamma)$ it is true that,
\beqs ((\zeta_{\breve\sigma_\Gp}\circ\zeta_{\sigma_\Gp})\circ\ho_\Gamma)(\Gpp)= (\ho_\Gamma\circ R_{\breve\sigma\ast_2\sigma_\Gp})(\Gpp)=(\zeta_{\breve\sigma_\Gp\ast_2\sigma_\Gp}\circ\ho_\Gamma)(\Gpp)
\eqs yields.

Recall that, the map $\tilde\sigma\mapsto t\circ\tilde\sigma$ is a group isomorphism between the group of bisections $\mathfrak {B}(\PD_\Gamma\Sigma)$ and the group $\Diff(V_\Gamma)$ of finite diffeomorphisms in $V_\Gamma$. Therefore if the graphs $\Gp=\Gpp$ contain only the path $ \gamma$, then the action $\zeta_{\sigma_\Gp}$ is equivalent to an action of the finite diffeomorphism group $\Diff(V_\Gamma)$. Loosely speaking, the graph-diffeomorphisms $(R_{\sigma_\Gp(V)},t\circ\sigma_\Gp)$ on a subgraph $\Gpp$ of $\Gp$ transform graphs and respect the graph structure of $\Gp$. The diffeomorphism $t\circ\tilde\sigma$ in the finite path groupoid only implements the finite diffeomorphism in $\Sigma$, but it doesn't adopt any path groupoid or graph preserving structure. Summarising the bisections of a finite graph system respect the graph structure and implement the finite diffeomorphisms in $\Sigma$. There is another reason why the group of bisections is more fundamental than the path- or graph-diffeomorphism group. In section \ref{subsec dynsysfluxgroup} the concept of  $C^*$-dynamical systems is studied. It turns out that, there are three different $C^*$-dynamical systems, each is build from the analytic holonomy $C^*$-algebra and a point-norm continuous action of the group of bisections of a finite graph system. The actions are implemented by one of the three translations, i.e. the left-, right- or inner-translation in the finite product $G^{\vert\Gamma\vert}$. Furthermore the actions are related to each other by an unitary $1$-cocycle. Hence the automorphisms are exterior equivalent \cite[Def.:2.66]{Williams07}.  In this case the full information is contained in the $C^*$-dynamical systems constructed from the right-translation in a finite graph system. Moreover the path- (or graph-) diffeomorphisms define particular right-, left- or inner-translations associated to suitable bisections. This is why in the introduction and later in section \ref{sec analholalg} the path- (or graph-) diffeomorphism group are often indentified with the group of bisections in a finite path groupoid (or graph system).

Finally the left or right-translations in a finite path groupoid can be studied in the context of natural or non-standard identification of the configuration space. This new concept leads to two different notions of diffeomorphism-invariant states. The actions of path- and graph-diffeomorphism and the concepts of natural or non-standard identification of the configuration space was not used in the context of LQG before.
\subsection{The group-valued quantum flux operators associated to surfaces and graphs}\label{subsec fluxdef}

Let $G$ be the structure group of a principal fibre bundle $P(\Sigma,G,\pi)$. Then the quantum flux operators, which are associated to a fixed surface $S$, are $G$-valued operators. For the construction of the quantum flux operator $\rho_S(\gamma)$ different maps from a graph $\Gamma$ to a direct product $G\times G$ are considered. This is related to the fact that, one distinguishes between paths that are ingoing and paths that are outgoing with resepect to the surface orientation of $S$. If there are no intersection points of the surface $S$ and the source or target vertex of a path $\gamma_i$ of a graph $\Gamma$, then the map maps the path $\gamma_i$ to zero in both entries. For different surfaces or for a fixed surface different maps refer to different quantum flux operators.  
 
\begin{defi}
Let $\breve S$ be a finite set $\{S_i\}$ of surfaces in $\Sigma$, which is closed under a flip of orientation of the surfaces. Let $\Gamma$ be a graph such that each path in $\Gamma$ satisfies one of the following conditions 
\begin{itemize}
 \item the path intersects each surface in $\breve S$ in the source vertex of the path and there are no other intersection points of the path and any surface contained in $\breve S$,
 \item the path intersects each surface in $\breve S$ in the target vertex of the path and there are no other intersection points of the path and any surface contained in $\breve S$,
 \item the path intersects each surface in $\breve S$ in the source and target vertex of the path and there are no other intersection points of the path and any surface contained in $\breve S$,
 \item the path does not intersect any surface $S$ contained in $\breve S$.
\end{itemize} Finally let $\PD_\Gamma$ denotes the finite graph system associated to $\Gamma$. 

Then define the intersection functions $\iota_L:\breve S\times \Gamma\rightarrow \{\pm 1,0\}$ such that
\beqs \iota_L(S,\gamma):=
\left\{\begin{array}{ll}
1 &\text{ for a path }\gamma\text{ lying above and outgoing w.r.t. }S\\
-1 &\text{ for a path }\gamma\text{ lying below and outgoing w.r.t. }S\\
0 &\text{ the path }\gamma\text{ is not outgoing w.r.t. }S
\end{array}\right.
\eqs
and the intersection functions $\iota_R:\breve S\times \Gamma\rightarrow\{\pm 1,0\}$ such that
\beqs \iota_L(S,\gamma):= \left\{\begin{array}{ll}
-1 &\text{ for a path }\gp\text{ lying above and ingoing w.r.t. }S\\
1 &\text{ for a path }\gp\text{ lying below and ingoing w.r.t. }S\\
0 &\text{ the path }\gp\text{ is not ingoing w.r.t. }S
\end{array}\right.
\eqs whenever $S\in\breve S$ and $\gamma\in\Gamma$.

Define a map $o_L:\breve S\rightarrow G$ such that
\beqs o_L(S)&=o_L(S^ {-1})
\eqs whenever $S\in\breve S$ and $S^ {-1}$ is the surface $S$ with reversed orientation. Denote the set of such maps by $\breve o_L$. Respectively the map $o_R:\breve S\rightarrow G$ such that
\beqs o_R(S)&=o_R(S^ {-1})
\eqs whenever $S\in\breve S$. Denote the set of such maps by $\breve o_R$.
Moreover there is a map $o_L\times o_R:\breve S\rightarrow G\times G$ such that
\beqs (o_L,o_R)(S)&=(o_L,o_R)(S^ {-1})
\eqs whenever $S\in\breve S$. Denote the set of such maps by $\breve o$.

Then define the \textbf{group-valued quantum flux set for paths}
\beqs  \Gop_{\breve S,\Gamma}
:=\bigcup_{o_L\times o_R\in\breve o}\bigcup_{S\in\breve S}\Big\{& (\rho^L,\rho^R)\in\Map(\Gamma,G\times G): 
&(\rho^L, \rho^R)(\gamma):=(o_L(S)^{\iota_L(S,\gamma)},o_R(S)^{\iota_R(S,\gamma)})\Big\}
\eqs
where $\Map(\Gamma,G\times G)$ denotes the set of all maps from the graph $\Gamma$ to the direct product $G\times G$.

Define the \textbf{set of group-valued quantum fluxes for graphs}
\beqs G_{\breve S,\Gamma}:= \bigcup_{o_L\times o_R\in\breve o}\bigcup_{S\in\breve S}\Big\{ \rho_{S,\Gamma}\in\Map(\PD^{\op}_\Gamma,G^{\vert\Gamma\vert}\times G^{\vert\Gamma\vert}):\quad 
&\rho_{S,\Gamma}:=\rho_S\times...\times \rho_S\\&\text{ where }\rho_S(\gamma):=(o_L(S)^{\iota_L(\gamma,S)},o_R(S)^{\iota_R(\gamma,S)}),\\
&\rho_S\in\Gop_{\breve S,\Gamma},S\in\breve S,\gamma\in\Gamma\Big\}\eqs 
\end{defi}
Notice if $H$ is a closed subgroup of $G$, then $H_{\breve S,\Gamma}$ can be defined in analogy to $G_{\breve S,\Gamma}$.
In particular if the group $H$ is replaced by the center $\ZD(G)$ of the group $G$, then the set $\Gop_{\breve S,\Gamma}$ is replaced by $\Zop_{\breve S,\Gamma}$ and $G_{\breve S,\Gamma}$ is changed to $\ZD_{\breve S,\Gamma}$. 

Furthermore observe that, $(\iota_L\times \iota_R)(S^{-1},\gamma)=(-\iota_L\times -\iota_R)(S,\gamma)$ for every $\gamma\in\Gamma$ holds. Remark that, the condition $\rho^L(\gamma)=\rho^R(\gamma^{-1})$ is not required. 

\begin{exa}\label{exa Exa1}
For example the following example can be analysed. Consider a graph $\Gamma$ and two disjoint surface sets $\breve S$ and $\breve T$.
\begin{center}
\includegraphics[width=0.45\textwidth]{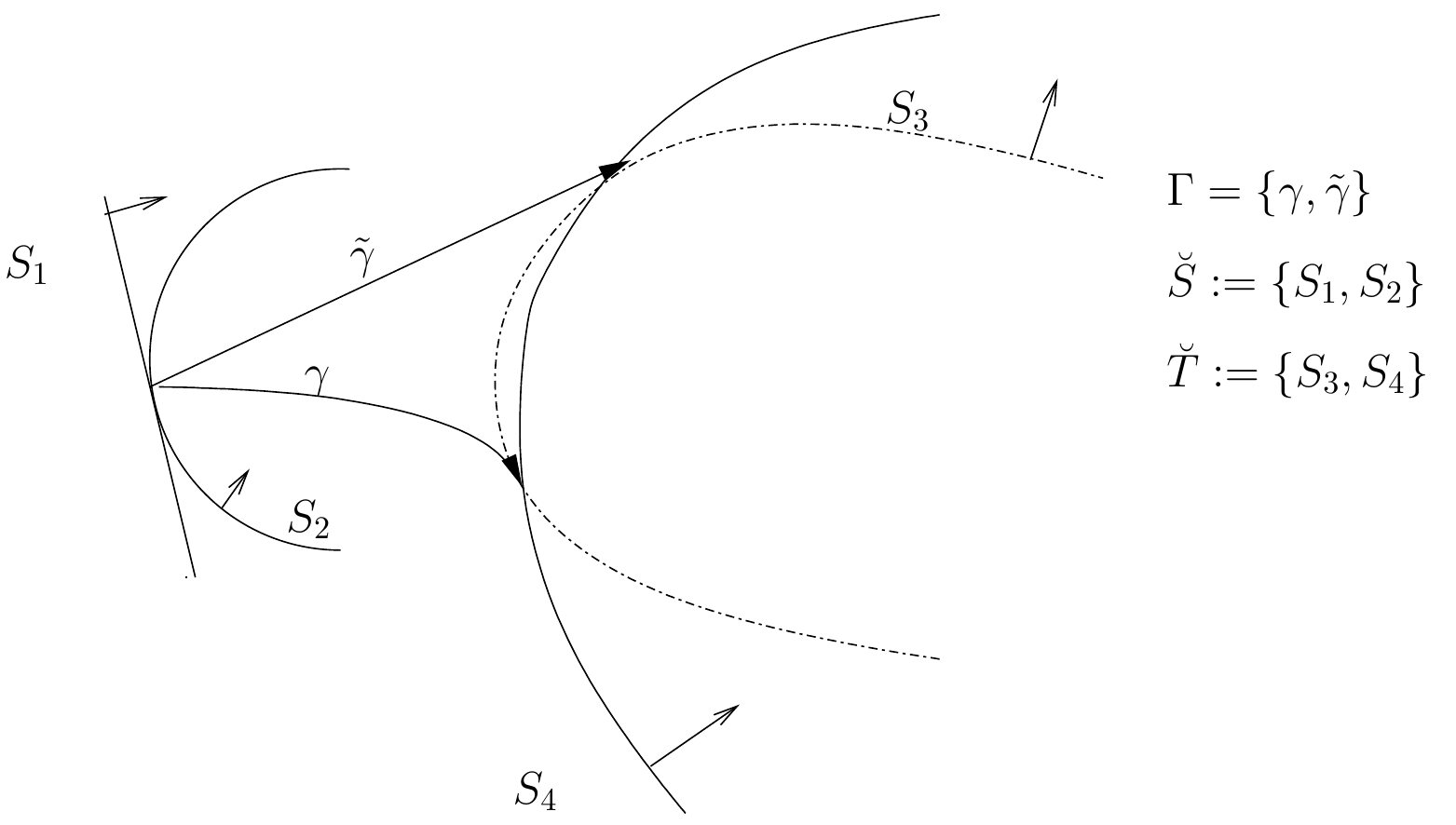}
\end{center}
Then the elements of $\Gop_{\breve S,\Gamma}$ are for example the maps $\rho^L_{i}\times \rho^R_{i}$ for $i=1,2$ such that 
\beqs 
\rho_1(\gamma)&:= (\rho^L_{1}, \rho^R_{1})(\gamma)=(\sigma_L(S_1)^{\iota_L(S_1,\gamma)},\sigma_R(S_1)^{\iota_R(S_1,\gamma)})=(g_{1},0)\\
\rho_1(\tg)&:= (\rho^L_{1}, \rho^R_{1})(\tg)=(\sigma_L(S_1)^{\iota_L(S_1,\tg)},\sigma_R(S_1)^{\iota_R(S_1,\tg)})= (g_1,0)\\
\rho_2(\gamma)&:= (\rho^L_{2}, \rho^R_{2})(\gamma)=(\sigma_L(S_2)^{\iota_L(S_2,\gamma)},\sigma_R(S_2)^{\iota_R(S_2,\gamma)})
=(g_{2},0)\\
\rho_2(\tg)&:= (\rho^L_{2}, \rho^R_{2})(\tg)=(\sigma_L(S_2)^{\iota_L(S_2,\tg)},\sigma_R(S_2)^{\iota_R(S_2,\tg)})
=(g_{2},0)\\
\rho_3(\gamma)&:= (\rho^L_{3}, \rho^R_{3})(\gamma)=(\sigma_L(S_3)^{\iota_L(S_3,\gamma)},\sigma_R(S_3)^{\iota_R(S_3,\gamma)})
=(0,h_{3}^{-1})\\
\rho_3(\tg)&:= (\rho^L_{3}, \rho^R_{3})(\tg)=(\sigma_L(S_3)^{\iota_L(S_3,\tg)},\sigma_R(S_3)^{\iota_R(S_3,\tg)})= (0,h_3^{-1})\\
\rho_4(\gamma)&:= (\rho^L_{4}, \rho^R_{4})(\gamma)=(\sigma_L(S_4)^{\iota_L(S_4,\gamma)},\sigma_R(S_4)^{\iota_R(S_4,\gamma)})
=(0,h_{4})\\
\rho_4(\tg)&:= (\rho^L_{4}, \rho^R_{4})(\tg)=(\sigma_L(S_4)^{\iota_L(S_4,\tg)},\sigma_R(S_4)^{\iota_R(S_4,\tg)})= (0,h_4)
\eqs 

This example shows that, the surfaces $\{S_1,S_2\}$ are similar, whereas the surfaces $\{T_1,T_2\}$ produce different signatures for different paths. Moreover the set of surfaces are chosen such that one component of the direct sum is always zero. 
\end{exa}

For a particular surface set $\breve S$, the following set is defined
\beqs\bigcup_{\sigma_L\times\sigma_R\in\breve\sigma}\bigcup_{S\in\breve S}
\Big\{ (\rho^L,\rho^R)\in\Map(\Gamma,G\times G): \quad(\rho^L, \rho^R)(\gamma):=(\sigma_L(S)^{\iota_L(S,\gamma)},0)\Big\}\eqs can be identified with 
\beqs\bigcup_{\sigma_L\in\breve\sigma_L}\bigcup_{S\in\breve S}\Big\{\rho\in\Map(\Gamma,G): \quad
\rho(\gamma):=\sigma_L(S)^{\iota_L(S,\gamma)}\Big\}
\eqs 
The same is observed for another surface set $\breve T$ and the set $\Gop_{\breve T,\Gamma}$ is identifiable with 
\beqs\bigcup_{\sigma_R\in\breve\sigma_R}\bigcup_{T\in\breve T}
\Big\{\rho\in\Map(\Gamma,G): \quad
\rho(\gamma):=\sigma_R(T)^{\iota_R(T,\gamma)}\Big\}
\eqs

The intersection behavoir of paths and surfaces plays a fundamental role in the definition of the flux operator. There are exceptional configurations of surfaces and paths in a graph. One of them is the following.

\begin{defi}
A surface $S$ has the \textbf{surface intersection property for a graph} $\Gamma$, if the surface intersects each path of $\Gamma$ once in the source or target vertex of the path and there are no other intersection points of $S$ and the path. 
\end{defi}

This is for example the case for the surface $S_1$ or the surface $S_3$, which are presented in example \thesection.\ref{exa Exa1}. Notice that in general, for the surface $S$ there are $N$ intersection points with $N$ paths of the graph. In the example the evaluated map $\rho_1(\gamma)=(g_1,0)=\rho_1(\tg)$ for $\gamma,\tg\in\Gamma$ if the surface $S_1$ is considered.

The property of a path lying above or below is not important for the definition of the surface intersection property for a surface. This indicates that the surface $S_4$ in the example \thesection.\ref{exa Exa1} has the surface intersection property, too.

Let a surface $S$ does not have the surface intersection property for a graph $\Gamma$, which contains only one path $\gamma$. Then for example the path $\gamma$ intersects the surface $S$ in the source and target vertices such that the path lies above the surface $S$. Then the map $\rho^ L\times \rho^ R$ is evaluated for the path $\gamma$ by
\beqs (\rho^ L\times \rho^ R)(\gamma)=(g,h^{-1})
\eqs
Hence simply speaking the surface intersection property reduces the components of the map $\rho^ L\times \rho^ R$, but for different paths to different components.

Now, consider a bunch of sets of surfaces such that for each surface there is only one intersection point.
\begin{defi}\label{def intprop}
A set $\breve S$ of $N$ surfaces has the \textbf{surface intersection property for a graph $\Gamma$} with $N$ independent edges, if it contain only surfaces, for which each path $\gamma_i$ of a graph $\Gamma$ intersects each surface $S_i$ only once in the same source or target vertex of the path $\gamma_i$, there are no other intersection points of each path $\gamma_i$ and each surface in $\breve S$ and there is no other path $\gamma_j$ that intersects the surface $S_i$ for $i\neq j$ where $1 \leq i,j\leq N$.
\end{defi}
Then for example consider the following configuration.

\begin{exa} 
\begin{center}
\includegraphics[width=0.45\textwidth]{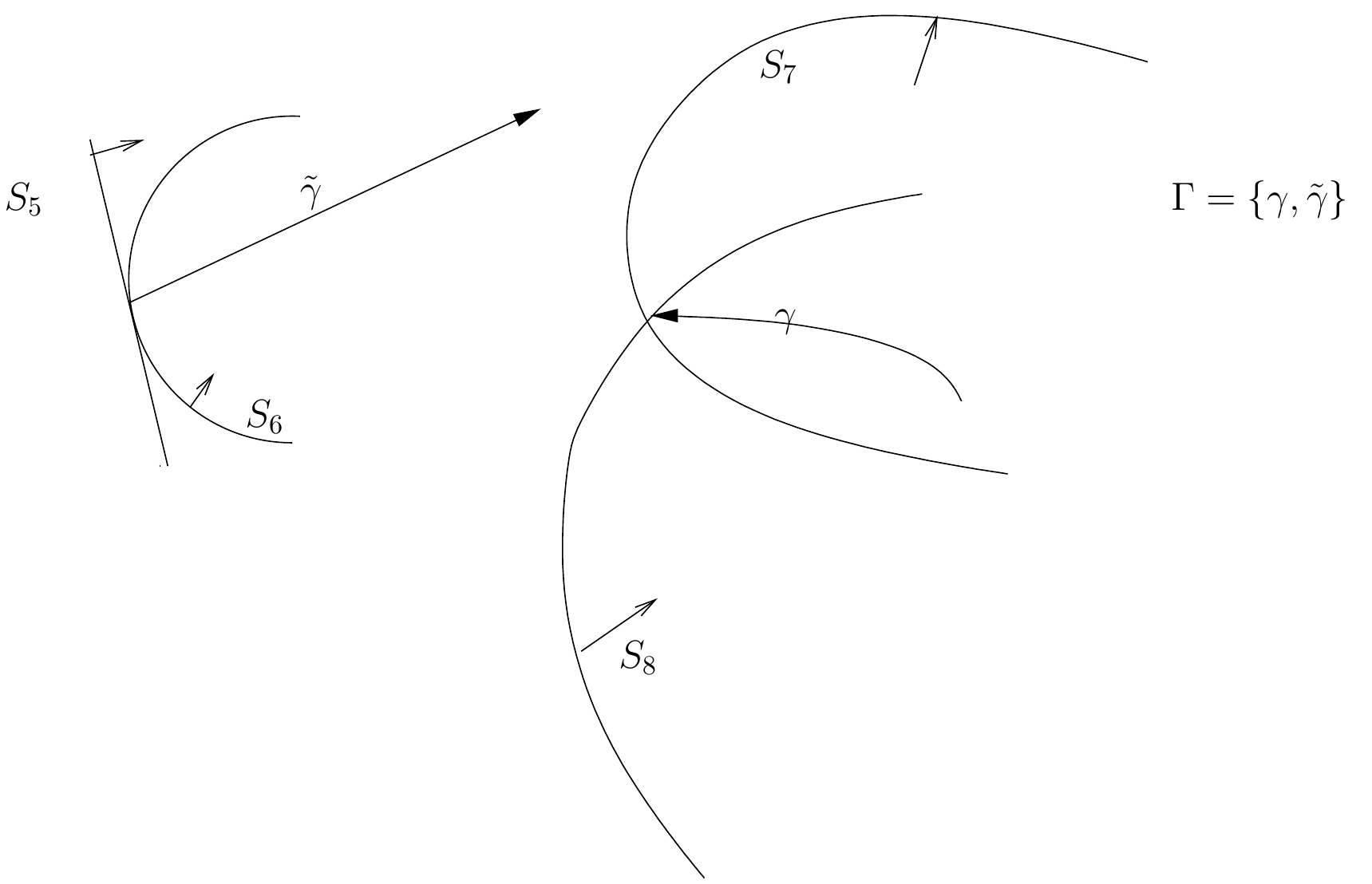}
\end{center} 
The sets $\{S_{6},S_{7}\}$ or $\{S_5,S_{8}\}$ have the surface intersection property for the graph $\Gamma$. 
The images of a map $E$ is
\beqs \rho_5(\tg)=(g_5,0),\quad \rho_{8}(\gamma)=(0,h_{8})
\eqs
\end{exa}
Note that simply speaking the property indicates that each map reduces to a component of $\rho^ L\times \rho^R$.

A set of surfaces that has the surface intersection property for a graph is further specialised by restricting the choice to paths lying ingoing and below with respect to the surface orientations. 
\begin{defi}
A set $\breve S$ of $N$ surfaces has the \hypertarget{simple surface intersection property for a graph}{\textbf{simple surface intersection property for a graph $\Gamma$}} with $N$ independent edges, if it contains only surfaces, for which each path $\gamma_i$ of a graph $\Gamma$ intersects only one surface $S_i$ only once in the target vertex of the path $\gamma_i$, the path $\gamma_i$ lies above and there are no other intersection points of each path $\gamma_i$ and each surface in $\breve S$. 
\end{defi}
\begin{exa}Consider the following example.
\begin{center}
\includegraphics[width=0.45\textwidth]{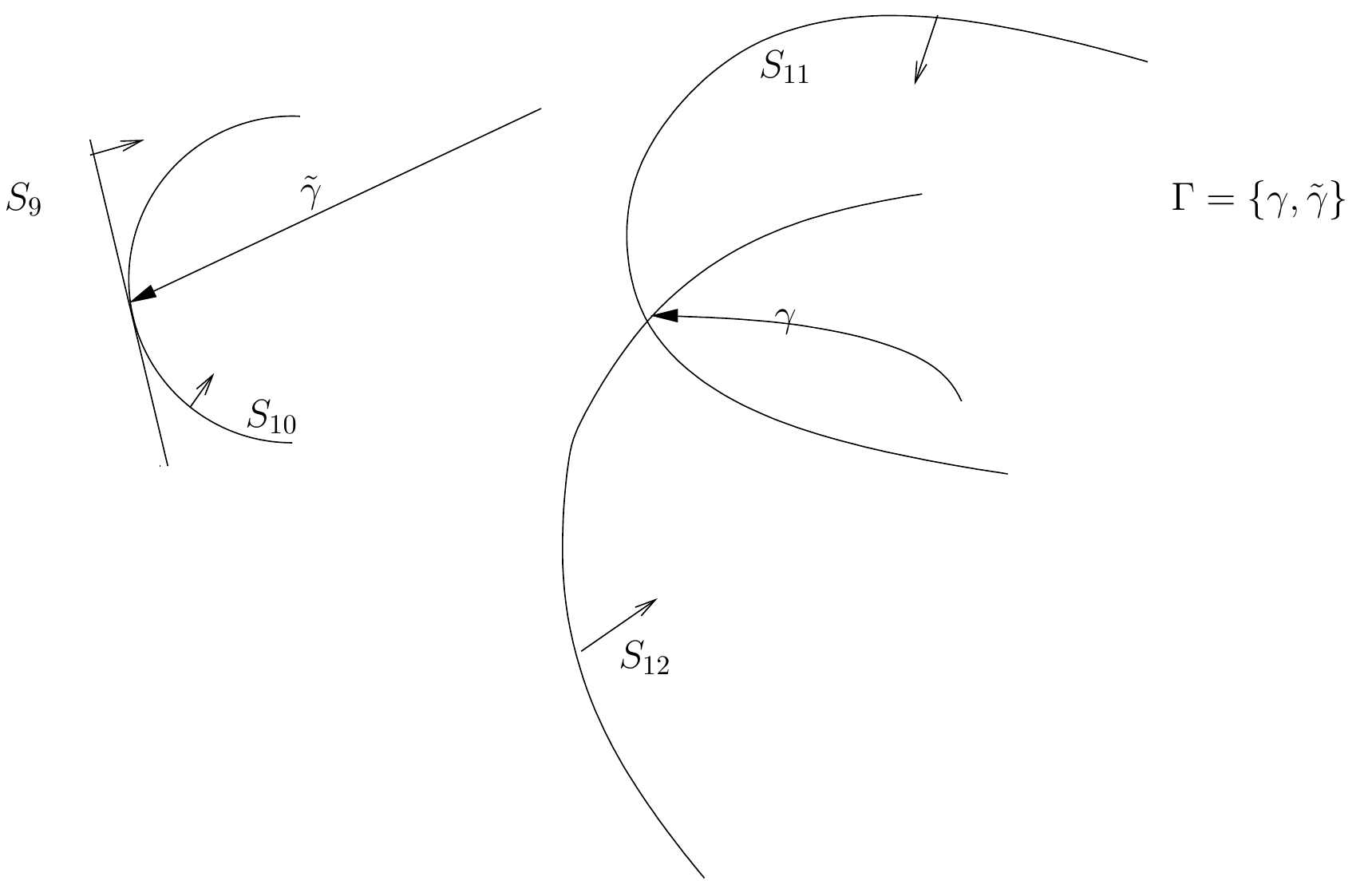}
\end{center} The sets $\{S_{9},S_{11}\}$ or $\{S_{10},S_{12}\}$ have the simple surface intersection property for the graph $\Gamma$.
Calculate
\beqs \rho_{9}(\tg)=(0,h_{9}^{-1}),\quad \rho_{11}(\gamma)=(0,h_{11}^{-1})
\eqs
\end{exa}
In this case the set $\Gop_{\breve S,\Gamma}$ reduces to
\beqs\bigcup_{\sigma_R\in\breve\sigma_R}\bigcup_{S\in\breve S}\Big\{\rho\in\Map(\Gamma,\go): \quad
\rho(\gamma):=\sigma_R(S)^{-1}\text{ for }\gamma\cap S= t(\gamma)\Big\}
\eqs Notice that, the set $\Gamma\cap \breve S=\{t(\gamma_i)\}$ for a surface $S_i\in\breve S$ and $\gamma_i\cap S_j\cap S_i =\{\varnothing\}$ for a path $\gamma_i$ in $\Gamma$ and $i\neq j$.

On the other hand, there exists a set of surfaces such that each path of a graph intersects all surfaces of the set in the same vertex. This contradicts the assumption that each path of a graph intersects only one surface once. 
\begin{defi}Let $\Gamma$ be a graph that contains no loops.

A set $\breve S$ of surfaces has the \hypertarget{same intersection property}{\textbf{same surface intersection property for a graph}} $\Gamma$ iff each path $\gamma_i$ in $\Gamma$ intersects with all surfaces of $\breve S$ in the same source vertex $v_i\in V_\Gamma$ ($i=1,..,N$), all paths are outgoing and lie below each surface $S\in\breve S$ and there are no other intersection points of each path $\gamma_i$ and each surface in $\breve S$. 

A surface set $\breve S$ has the \textbf{same right surface intersection property for a graph} $\Gamma$ iff each path $\gamma_i$ in $\Gamma$ intersects with all surfaces of $\breve S$ in the same target vertex $v_i\in V_\Gamma$ ($i=1,..,N$), all paths are ingoing and lie above each surface $S\in\breve S$ and there are no other intersection points of each path $\gamma_i$ and each surface in $\breve S$. 
\end{defi}

Recall the example \thesection.\ref{exa Exa1}. Then the set $\{S_{1},S_{2}\}$ has the same surface intersection property for the graph $\Gamma$.

Then the set $\Gop_{\breve S,\Gamma}$ reduces to
\beqs\bigcup_{\sigma_L\in\breve\sigma_L}\bigcup_{S\in\breve S}\Big\{\rho\in\Map(\Gamma,\go): \quad
\rho(\gamma):= \sigma_L(S)^{-1}\text{ for }\gamma\cap S= s(\gamma)\Big\}
\eqs Notice that, $\gamma\cap S_1\cap ...\cap S_N=s(\gamma)$ for a path $\gamma$ in $\Gamma$ whereas $\Gamma\cap\breve S=\{s(\gamma_i)\}_{1\leq i\leq N}$. Clearly $\Gamma\cap S_i=s(\gamma_i)$ for a surface $S_i$ in $\breve S$ holds. 
Simply speaking the physical intution behind that is given by fluxes associated to different surfaces that should act on the same path.
 
A very special configuration is the following.
\begin{defi}
A set $\breve S$ of surfaces has the \textbf{same surface intersection property for a graph $\Gamma$ containing only loops} iff each loop $\gamma_i$ in $\Gamma$ intersects with all surfaces of $\breve S$ in the same vertices $s(\gamma_i)=t(\gamma_i)$ in $V_\Gamma$ ($i=1,..,N$), all loops lie below each surface $S\in\breve S$ and there are no other intersection points of each loop in $\Gamma$ and each surface in $\breve S$. 
\end{defi}

Notice that, both properties can be restated for other surface and path configurations. Hence a surface set have the simple or same surface intersection property for paths that are outgoing and lie above (or ingoing and below, or outgoing and below). The important fact is related to the question if the intersection vertices are the same for all surfaces or not.

Finally for the definition of the quantum flux operators notice the following objects.
\begin{defi}
 The set of all images of maps in $\Gop_{\breve S,\Gamma}$ for a fixed surface set $\breve S$ and a fixed path $\gamma$ in $\Gamma$ is denoted by $\bar\Gop_{\breve S,\gamma}$. 

The set of all finite products of images of maps in $G_{\breve S,\Gamma}$ for a fixed surface set $\breve S$ and a fixed graph $\Gamma$ is denoted by $\bar G_{\breve S,\Gamma}$. 
\end{defi}
The product $\cdot$ on $\bar G_{\breve S,\Gamma}$ is given by
\beqs \rho_{S_1,\Gamma}(\Gamma)\cdot \rho_{S_2,\Gamma}(\Gamma)&=(\rho_{S_1}(\gamma_1)\cdot \rho_{S_2}(\gamma_1),...,\rho_{S_1}(\gamma_N)\cdot \rho_{S_2}(\gamma_N))\\
&=(o_L(S_1)^{-1}o_L(S_2)^{-1},...,o_L(S_2)^{-1}o_L(S_1)^{-1})\\
&=((o_L(S_2)o_L(S_1))^{-1},...,(o_L(S_2)o_L(S_1))^{-1})\\
&=\rho_{S_3,\Gamma}(\Gamma)
\eqs

\begin{defi}
Let $S$ be a surface and $\Gamma$ be a graph such that the only intersections of the graph and the surface in $S$ are contained in the vertex set $V_\Gamma$. Moreover let $\fPG$ be a finite path groupoid associated to $\Gamma$. 

Then define the set for a fixed surface $S$ by
\beqs &\Map_S(\PD_\Gamma\Sigma,G\times G)\\&
:= \bigcup_{o_L\times o_R\in\breve o}\bigcup_{S\in\breve S}\Big\{& (\rho^L,\rho^R)\in\Map(\PD_\Gamma\Sigma,G\times G): 
&(\rho^L, \rho^R)(\gamma):=(o_L(S)^{\iota_L(S,\gamma)},o_R(S)^{\iota_R(S,\gamma)})\Big\}
\eqs
\end{defi}

Then the quantum flux operators are elements of the following group.
\begin{prop}Let $\breve S$ a set of surfaces and $\Gamma$ be a fixed graph, which contains no loops, such that the set $\breve S$ has the same surface intersection property for the graph $\Gamma$. 

The set $\bar\Gop_{\breve S,\gamma}$ has the structure of a group.

The group $\bar\Gop_{\breve S,\gamma}$ is called the \textbf{flux group associated a path and a finite set of surfaces}.
\end{prop}
\begin{proofs}This follows easily from the observation that in this case 
$\Gop_{\breve S,\gamma}$ reduces to
\beqs\bigcup_{o_L\in\breve o_L}\bigcup_{S\in\breve S}\Big\{& \rho^L\in\Map(\Gamma,G): 
&\rho^L(\gamma):=o_L(S)^{-1}\text{ for }\gamma\cap S=s(\gamma)\Big\}
\eqs

There always exists a map $\rho^L_{S,3}\in\Gop_{\breve S,\gamma}$ such that the following equation defines a  multiplication operation 
\beqs \rho^L_{S,1}(\gamma)\cdot\rho^L_{S,2}(\gamma)=g_{1}g_{2}:=\rho^L_{S,3}(\gamma)\in\bar\Gop_{\breve S,\gamma}\eqs with inverse $(\rho^L_S(\gamma))^{-1}$ such that
\beqs \rho^L_S(\gamma)\cdot (\rho^L_S(\gamma))^{-1}=(\rho^L_S(\gamma))^{-1}\cdot \rho^L_S(\gamma)=e_G\quad\forall\gamma\in\Gamma\eqs
\end{proofs}

Notice that for a loop $\alpha$ an element $\rho_S(\alpha)\in\bar\Gop_{\breve S,\gamma}$ is defined by \beqs\rho_S(\alpha):=(\rho_S^L\times\rho_S^R)(\alpha)= (g, h)\in G^2 
\eqs In the case of a path $\gp$ that intersects a surface $S$ in the source and target vertex there is also an element $\rho_S(\gp)\in\bar\Gop_{\breve S,\gamma}$ defined by \beqs\rho_S(\gp):=(\rho_S^L\times\rho_S^R)(\gp)= (g, h)\in G^2 
\eqs

\begin{prop}
Let $\breve S$ be a set of surfaces and $\Gamma$ be a fixed graph, which contains no loops, such that the set $\breve S$ has the same surface intersection property for the graph $\Gamma$. Let $\PD^{\op}_\Gamma$ be a finite orientation preserved graph system such that the set $\breve S$.

The set $\bar G_{\breve S,\Gamma}$ has the structure of a group.

The  set $\bar G_{\breve S,\Gamma}$ is called the \textbf{flux group associated a graph and a finite set of surfaces}.
\end{prop}
\begin{proofs}
This follows from the observation that the set $G_{\breve S,\Gamma}$ is identified with
\beqs  \bigcup_{\sigma_L\in\breve\sigma_L}\bigcup_{S\in\breve S}\Big\{ \rho_{S,\Gamma}\in\Map(\PD^{\op}_\Gamma,G^{\vert E_\Gamma\vert}):\quad 
&\rho_{S,\Gamma}:=\rho_S\times...\times \rho_S\\&\text{ where }\rho_S(\gamma):=o_L(S)^{-1},
\rho_S\in \Gop_{\breve S,\Gamma},S\in\breve S,\gamma\in\Gamma\Big\}\eqs 

Let $\breve S$ be a surface set having the same intersection property for a fixed graph $\Gamma:=\{\gamma_1,...,\gamma_N\}$. Then for two surfaces $S_1,S_2$ contained in $\breve S$ define
\beqs \rho_{S_1,\Gamma}(\Gamma)\cdot \rho_{S_2,\Gamma}(\Gamma)
&=(\rho_{S_1}(\gamma_1)\cdot \rho_{S_2}(\gamma_1),...,\rho_{S_1}(\gamma_N)\cdot \rho_{S_2}(\gamma_N))\\
&= (g_{S_1},...,g_{S_1})\cdot (g_{S_2},..., g_{S_2}) =( g_{S_1}g_{S_2},..., g_{S_1}g_{S_2})
\eqs where $\Gamma=\{\gamma_1,...,\gamma_N\}$.
Note that, since the maps $o_L$ are arbitrary maps from $\breve S$ to $G$, it is assumed that the maps satisfy $o_L(S_i):=g^{-1}_{S_i}\in G$ for $i=1,2$. 
Clearly this is related to in this particular case of the graph $\Gamma$ and can be generalised. 

The inverse operation is given by
\beqs (\rho_{S,\Gamma}(\Gamma))^{-1}=((\rho_S(\gamma_1))^{-1},...,(\rho_S(\gamma_N))^{-1})
\eqs
where $N=\vert \Gamma\vert$ and $\rho_S\in\Gop_{\breve S,\gamma}$ for $S\in\breve S$. Since it is true that
\beqs \rho_{S,\Gamma}(\Gamma)\cdot \rho_{S,\Gamma}(\Gamma)^{-1}&= (g_{S},...,g_{S})\cdot (g_{S}^{-1},..., g_{S}^{-1})\\
&=(\rho_{S}(\gamma_1)\cdot \rho_{S}(\gamma_1)^{-1},...,\rho_{S}(\gamma_N)\cdot \rho_{S}(\gamma_N)^{-1})\\
& =( g_{S}g_{S}^{-1},..., g_{S}g_{S}^{-1})=(e_G,...,e_G)
\eqs yields.
\end{proofs}
Notice that, it is not defined that
\beqs &\rho_{S_1,\Gamma}(\Gamma)\bullet_R \rho_{S_2,\Gamma}(\Gamma)\\&=(o_L(S_2)^{-1}o_L(S_1)^{-1},...,o_L(S_2)^{-1}o_L(S_1)^{-1})
=((o_L(S_1)o_L(S_2))^{-1},...,(o_L(S_1)o_L(S_2))^{-1})\\
&=\rho_{S_3,\Gamma}(\Gamma)
\eqs is true.
Moreover observe that, if all subgraphs of a finite orientation preserved graph system are \hyperlink{natural identification}{naturally identified}, then $\bar G_{\breve S,\Gp\leq\Gamma}$ is a subgroup of $\bar G_{\breve S,\Gamma}$ for all subgraphs $\Gp$ in $\PD_\Gamma$. If $G$ is assumed to be a compact Lie group, then the flux group $\bar G_{\breve S,\Gamma}$ is called the Lie flux group.

There is another group, if another surface set is considered.
\begin{prop}Let $\breve T$ be a set of surfaces and $\Gamma$ be a fixed graph such that the set $\breve T$ has the simple surface intersection property for the graph $\Gamma$. Let $\PD^{\op}_\Gamma$ be a finite orientation preserved graph system.

The set $\bar G_{\breve T,\Gamma}$ has the structure of a group.
\end{prop}
The same arguments using the identification of $\bar G_{\breve T,\Gamma}$ with
\beqs  \bigcup_{\sigma_R\in\breve\sigma_R}\Big\{ \rho_{T,\Gamma}\in\Map(\PD^{\op}_\Gamma,G^{\vert E_\Gamma\vert}):\quad 
&\rho_{\breve T,\Gamma}:=\rho_{T_1}\times...\times \rho_{T_N}\\&\text{ where }\rho_{T_i}(\gamma):=o_R(T_i)^{-1},
\rho_{T_i}\in \Gop_{\breve T,\Gamma},T_i\in\breve T,\gamma\in\Gamma\Big\}\eqs 
which is given by
\beqs \rho_{T_1,\Gamma}(\Gamma)\cdot ...\cdot \rho_{T_N,\Gamma}(\Gamma)
&=(\rho_{T_1}(\gamma_1)e_G, e_G,...,e_G)\cdot 
(e_G, \rho_{T_2}(\gamma_2) e_G, e_G,...,e_G)\cdot 
 ...\cdot (e_G,...,e_G,\rho_{T_N}(\gamma_N)e_G)\\
&=(\rho^1_{T_1}(\gamma_1),...,\rho^1_{T_N}(\gamma_N))= (g_{1},...,g_{N})\in G^N\\
&=:\rho_{\breve T,\Gamma}(\Gamma)
\eqs
Then the multiplication operation is presented by
\beqs \rho^1_{\breve T,\Gamma}(\Gamma)\cdot \rho^2_{\breve T,\Gamma}(\Gamma)
&=(\rho^1_{T_1}(\gamma_1)\cdot \rho^2_{T_1}(\gamma_1),...,\rho^1_{T_N}(\gamma_N)\cdot \rho^2_{T_N}(\gamma_N))\\
&= (g_{1,1},...,g_{1,N})\cdot (g_{2,1},..., g_{2,N}) =( g_{1,1}g_{2,1},..., g_{1,N}g_{2,N})\in G^N
\eqs where $\Gamma=\{\gamma_1,...,\gamma_N\}$.

It is also possible that, the fluxes are located only in a vertex and do not depend on ingoing or outgoing, above or below orientation properties.
\begin{defi}\label{def Gloc}Let $\PD_\Gamma$ be a finite graph groupoid associated to a graph $\Gamma$ and let $N$ be the number of edges of the graph $\Gamma$. 

Define the set of maps 
\beqs G^{\loc}_{\Gamma}:=
\Big\{\textbf{g}_\Gamma\in\Map(\PD_\Gamma,G^{\vert \Gamma\vert}):& \textbf{g}_\Gamma:=g^1_\Gamma\circ s\times ...\times g^N_\Gamma\circ s\\  
&g^i_\Gamma\in\Map(\Gamma,G)\Big\}
\eqs 

Then $\bar G^{\loc}_{\Gamma}$ is the set of all images of maps in $G^{\loc}_{\Gamma}$ for all graphs in $\PD_\Gamma$ and $\bar G^{\loc}_{\Gamma}$ is called the \textbf{local flux group associated a finite graph system}. 
\end{defi}

\subsection{The group-valued quantum flux operators associated to surfaces and finite path groupoids}\label{subsec admfluxdef}

Recall the set of admissible maps $\Map^{\adm}(\PD_\Gamma\Sigma,G)$ presented in definition \ref{def admiss}. 

\begin{defi}\label{def adm}Let $\breve S$ be a finite set of surfaces which is closed under a flip of orientation of the surfaces. Let $\fPG$ be a finite path groupoid associated to a graph $\Gamma$ such that each path in $\PD_\Gamma\Sigma$ satisfies one of the following conditions 
\begin{itemize}
 \item the path intersects each surface in $\breve S$ in the source vertex of the path and there are no other intersection points of the path and any surface contained in $\breve S$,
 \item the path intersects each surface in $\breve S$ in the target vertex of the path and there are no other intersection points of the path and any surface contained in $\breve S$,
 \item the path intersects each surface in $\breve S$ in the source and target vertex of the path and there are no other intersection points of the path and any surface contained in $\breve S$,
 \item the path does not intersect any surface $S$ contained in $\breve S$.
\end{itemize} Finally, let $\PD_\Gamma$ denote the finite graph system associated to $\Gamma$. 

Then the \textbf{set of admissible maps associated to a graph and surfaces} $\breve S$ are defined by
\beqs \Gop^{\adm}_{\breve S,\Gamma}
:= \bigcup_{o_L\times o_R\in\breve o}\bigcup_{S\in\breve S}\Big\{& (\varrho^L,\varrho^R)\in\Map^{\adm}(\PD_\Gamma\Sigma,G\times G): 
&(\varrho^L, \varrho^R)(\gamma):=(o_L(S)^{\iota_L(S,\gamma)},o_R(S)^{\iota_R(S,\gamma)})\Big\}
\eqs

Define the \textbf{set of admissible maps associated to a finite graph system and surfaces} $\breve S$ is presented by
\beqs G^{\adm}_{\breve S,\Gamma}:=
\bigcup_{o_L\times o_R\in\breve o}\bigcup_{S\in\breve S}\Big\{\varrho_{S,\Gamma}\in\Map^{\adm}(\PD_\Gamma,G^{\vert \Gamma\vert}\times G^{\vert \Gamma\vert}):\quad
&\varrho_{S,\Gamma}:=\varrho_S\times ...\times\varrho_{S}\\
&\text{ where }\varrho_S(\gamma_i)=(o_L(S)^{\iota_L(S,\gamma)},o_R(S)^{\iota_R(S,\gamma_i)})\\
&\varrho_S\in\Gop^{\adm}_{\breve S,\Gamma}, S\in\breve S,\gamma_i\in\Gp,\Gp\in\PD_\Gamma\Big\}.\eqs
\end{defi}
Observe that, these maps have the following properties. For all elements of $\PD_\Gamma\Sigma_{v}$ (or $\PD_\Gamma\Sigma^{v}$) that intersect the surface $S$ only in their target (or source) vertex $v$ the maps $\varrho_{S}^L$ (or $\varrho_{S}^R$) in $\Gop^{\adm}_{\breve S,\Gamma}$ satisfies
\beq &\varrho^L_{S}(\gamma)=\varrho^L_{S}(\gamma\circ\gp)= \varrho^L_{S}(\gpp) =g_{S,L}\quad\forall\gamma,\gamma\circ\gp,\gpp\in\PD_\Gamma\Sigma_{v}\text{ and }v=s(\gamma)=S\cap\gamma\\
&\varrho^L_{S}(\gp^{-1})=\varrho^L_{S}((\gamma\circ\gp)^{-1})= k_{S,L}\quad\forall\gp^{-1},(\gamma\circ\gp)^{-1}\in\PD_\Gamma\Sigma_{v}\text{ and }v=t(\gp)=S\cap\gp^{-1}
\eq
Furthermore for paths $\gamma$ and $\gp$ that compose and intersect $S$ in the common vertex $t(\gamma)=s(\gp)$ it is true that
\beq\label{eq admmapsleftright} (\varrho^R_{S}(\gamma^{-1}))^{-1}\varrho^L_{S}(\gp)=e_G\qquad\text{ and }
\qquad  (\varrho^R_{S}(\gamma))\varrho^L_{S}(\gp^{-1})^{-1}=e_G
\eq whenever $(\gamma,\gp)\in\PD_\Gamma\Sigma^{(2)}$ and for all maps $(\varrho^L_{S},\varrho^R_{S})\in\Gop^{\adm}_{\breve S,\Gamma}$.

In both definitions of the sets $\Gop_{\breve S,\Gamma}$ or $\Gop^{\adm}_{\breve S,\Gamma}$ of maps, there is a mapping $\rho_S$ or, respectively, $\varrho_S$, which maps all paths in $\PD_\Gamma\Sigma^v$ to one element $X_S$, i.e. $\rho_S(\gamma)=\rho_S(\gamma\circ\gp)$ for all $\gamma,\gamma\circ\gp\in\PD_\Gamma\Sigma^v$ where $v=s(\gamma)$. But the equalities \eqref{eq admmapsleftright} are required only for maps in $\Gop^{\adm}_{\breve S,\Gamma}$.   

Notice that, if the group $G$ is replaced by the center $\ZD(G)$ of the group $G$, then the set $\Gop^{\adm}_{\breve S,\Gamma}$ is replaced by $\Zop^{\adm}_{\breve S,\Gamma}$ and $G^{\adm}_{\breve S,\Gamma}$ is changed to $\ZD^{\adm}_{\breve S,\Gamma}$.
\section{The analytic holonomy $C^*$-algebra and Weyl $C^*$-algebra}\label{sec analholalg}
\subsection{Dynamical systems of actions of the flux group on the analytic holonomy $C^*$-algebra}\label{subsec dynsysfluxgroup}
\subsubsection*{The analytic holonomy $C^*$-algebra for finite graph systems}

In this article the analytic holonomy algebra $C(\Ab_\Gamma)$ restricted to a graph system $\PD_\Gamma$ is given by the set $C(\Ab_\Gamma)$ of continuous functions on $\Ab_\Gamma$, pointwise multiplication, complex conjugation and the completion is taken with respect to the $\sup$-norm. The \textbf{analytic holonomy $C^*$-algebra $C(\Ab)$} is given by the inductive limit of the family of unital commutative $C^*$-algebras $\{(C(\Ab_{\Gamma_i}),\beta_{\Gamma_i,\Gamma_j}):\PD_{\Gamma_i}\leq\PD_{\Gamma_j},i,j\in\N\}$ for an inductive family $\{\PD_{\Gamma_i}\}$ of finite graph systems, where $\beta_{\Gamma_i,\Gamma_j}$ is a unit-preserving injective $^*$-homomorphism from the analytic holonomy algebra $C(\Ab_{\Gamma_i})$ to $C(\Ab_{\Gamma_j})$. The maps $\beta_{\Gamma,\Gp}$ satisfy the consistency conditions
\beqs \beta_{\Gamma,\Gpp}=\beta_{\Gamma,\Gp}\circ\beta_{\Gp,\Gpp}
\eqs whenever $\PD_{\Gamma}\leq\PD_{\Gp}\leq \PD_{\Gpp}$.  The configuration space $\Ab_{\Gamma_i}$ associated to a graph $\Gamma_i$ is derived from the set of all holonomy maps from the finite graph system $\PD_{\Gamma_i}$ to the product group $G^{\vert\Gamma_i\vert}$. Recall the notion of natural or non-standard identification of the configuration space $\Ab_\Gamma$, which is presented in subsection \ref{subsec graphhol}. The elements of $\Ab_\Gamma$ are identified naturally or in a non-standard way with $G^{\vert\Gamma\vert}$ by the evaluation of the holonomy map for a subset of the finite graph system $\PD_\Gamma$. Simply speaking the choice of the identification is a matter of the labeling of the configuration variables. In this article the identifications are needed for the definition of graph changing automorphisms acting on the analytic holonomy $C^*$-algebra.

An element of the analytic holonomy $C^*$-algebra $C(\Ab)$ is of the form
\beqs f=f_{\Gamma_i} \circ\pi_{{\Gamma_i}}=\beta_{\Gamma_i}\circ f_{\Gamma_i}
\eqs where $f\in C(\Ab)$, $\pi_{\Gamma_i}:\Ab\rightarrow \Ab_{\Gamma_i}$, $f_{\Gamma_i}\in C(G^{\vert{\Gamma_i}\vert})$ and the map $\beta_{\Gamma_i}:C(\Ab_{\Gamma_i})\rightarrow C(\Ab)$ is an unit-preserving injective $^*$-homomorphisms.
Furthermore the maps $\beta_\Gamma:C(\Ab_\Gamma)\longrightarrow C(\Ab)$ are isometries, since
\beqs \|f\|=\|\beta_{\Gamma_i}f_{\Gamma_i}\|=\sup\vert f_{\Gamma_i}\vert
\eqs yields whenever $f\in C(\Ab)$ and $f_{\Gamma_i}\in C(\Ab_{\Gamma_i})$ for all graphs $\Gamma_i$.

The idea is to define actions of groups on the $C^*$-algebra $C(\Ab_\Gamma)$ of continuous functions on the compact compact Hausdorff space $\Ab_\Gamma$ associated to a graph $\Gamma$, which can be extended to actions on the inductive limit algebra $C(\Ab)$. 

\subsubsection*{Group actions on the configuration space}

Let $\Gamma$ be a graph, $\PD_\Gamma$ be the associated finite graph system.
Assume that the subgraphs in a finite graph system $\PD_\Gamma$ are identified naturally and hence the configuration space $\Ab_\Gamma$ is identified in the natural way with $G^{\vert\Gamma\vert}$.  

Then there is a group action \[G^{N}\times \Ab_\Gamma\ni((g_1,..,g_N),(\ho_\Gamma(\gamma_1),...,\ho_\Gamma(\gamma_N)))\mapsto (g_1\ho_\Gamma(\gamma_1),...,g_N\ho_\Gamma(\gamma_N))\in \Ab_\Gamma\] of a finite product of a compact group $G$ on the compact Hausdorff space $\Ab_\Gamma$ where $N:=\vert\Gamma\vert$. For each $\textbf{g}:=(g_1,...,g_N)\in G^{\vert\Gamma\vert}$ the map $L(\textbf{g})$ given by
\[\Ab_\Gamma\ni (\ho_\Gamma(\gamma_1),...,\ho_\Gamma(\gamma_N))\mapsto L(\textbf{g})(\ho_\Gamma(\gamma_1),...,\ho_\Gamma(\gamma_N)):= (g_1\ho_\Gamma(\gamma_1),...,g_N\ho_\Gamma(\gamma_N))\in\Ab_\Gamma\] is a homeomorphism $L(\textbf{g}):\Ab_\Gamma\longrightarrow \Ab_\Gamma$. Moreover 
\[L(\textbf{g})(L(\textbf{h})(\ho_\Gamma(\gamma_1),...,\ho_\Gamma(\gamma_N)))=(L(\textbf{gh}))(\ho_\Gamma(\gamma_1),...,\ho_\Gamma(\gamma_N))\] for all $\textbf{g},\textbf{h}\in G^{\vert\Gamma\vert}$ and $(\ho_\Gamma(\gamma_1),...,\ho_\Gamma(\gamma_N))\in \Ab_\Gamma$ yields. Clearly there is a right action presented by the map $R(\textbf{g})$, which is defined by 
\[\Ab_\Gamma\ni (\ho_\Gamma(\gamma_1),...,\ho_\Gamma(\gamma_N))\mapsto R(\textbf{g})(\ho_\Gamma(\gamma_1),...,\ho_\Gamma(\gamma_N)):= (\ho_\Gamma(\gamma_1)g_1^{-1},...,\ho_\Gamma(\gamma_N)g_N^{-1})\in\Ab_\Gamma\] such that 
\beqs R(\textbf{g}\textbf{h})(\ho_\Gamma(\gamma_1),...,\ho_\Gamma(\gamma_N))
&= (\ho_\Gamma(\gamma_1)(g_1h_1)^{-1},...,\ho_\Gamma(\gamma_N)(g_Nh_N)^{-1})\\
&= (\ho_\Gamma(\gamma_1)h_1^{-1}g_1^{-1},...,\ho_\Gamma(\gamma_N)h_N^{-1}g_N^{-1})\\
&=R(\textbf{g})(R(\textbf{h})(\ho_\Gamma(\gamma_1),...,\ho_\Gamma(\gamma_N)))
\eqs

Consider a \hyperlink{finite orientation preserved graph system}{finite orientation preserved graph system} $\PD_\Gamma^{\op}$ associated to a graph $\Gamma$ and a finite set of surfaces $\breve S$ such that the set $\breve S$ has the \hyperlink{same intersection for ori}{same surface intersection property for the graph} $\Gamma$. Then the flux group $\bar G_{\breve S,\Gamma}$ is a subgroup of $G^{\vert\Gamma\vert}$. Since each subgraph $\Gp$ of $\Gamma$, for example, consists only paths that intersect each surface in $\breve S$ in the source vertex of the path and lie above. The evaluation of a map $\rho_{S,\Gamma}$ in $G_{\breve S,\Gp\leq\Gamma}$ for a subgraph $\Gp$ in $\PD_\Gamma^{\op}$ is given by $\rho_{S,\Gamma}(\Gp)=(\rho_S(\gamma_1),....,\rho_S(\gamma_M))$. The element $\rho_{S,\Gamma}(\Gp)$ is contained in $\bar G_{\breve S,\Gp\leq\Gamma}$. Furthermore the element $(\rho_S(\gamma_1),....,\rho_S(\gamma_M),e_G,...,e_G)$ is contained in $G_{\breve S,\Gamma}$.

Consider a graph $\Gamma:=\{\gamma_1,...,\gamma_N\}$ and a subgraph $\Gp:=\{\gamma_1,...,\gamma_M\}$ of $\Gamma$,  a finite graph system $\PD_\Gamma$ and a finite orientation preserved graph system $\PD_\Gamma^{\op}$ exists. Then there is a surfaces set $\breve S$, which has the same surface intersection property for $\Gamma$. Moreover assume that $\Gp\in\PD_\Gamma^{\op}$. Then for a map $\rho_{S,\Gamma}\in G_{\breve S,\Gamma}$ there exists a left action $L:\bar G_{S,\Gamma}\rightarrow\Ab_\Gamma$, which is given by
\beq L(\rho_{S,\Gamma}(\Gamma))(\ho_\Gamma(\gamma_1),...,\ho_\Gamma(\gamma_N))&=
L(\rho_S(\gamma_1),...,\rho_S(\gamma_N))(\ho_\Gamma(\gamma_1),...,\ho_\Gamma(\gamma_N))\\
&:=(\rho_S(\gamma_1)\ho_\Gamma(\gamma_1),...,\rho_S(\gamma_N)\ho_\Gamma(\gamma_N))
\eq and which defines a homeomorphism on $\Ab_\Gamma$. Certainly, if the surface set $\breve S$ has the \hyperlink{same intersection property for ori}{same surface intersection property for} $\Gamma$, then there is a right action $R$ of $\bar G_{\breve S,\Gamma}$ on $\Ab_\Gamma$. This action $R$ is of the form 
\beq &R(\rho_{S,\Gamma}(\Gp))(\ho_\Gamma(\gamma_1),...,\ho_\Gamma(\gamma_N))
= R(\rho_S(\gamma_1),...,\rho_S(\gamma_M))(\ho_\Gamma(\gamma_1),...,\ho_\Gamma(\gamma_N))\\
&:=(\ho_\Gamma(\gamma_1)\rho_S(\gamma_1)^{-1},...,\ho_\Gamma(\gamma_M)\rho_S(\gamma_M)^{-1},\ho_\Gamma(\gamma_{M+1}),...,\ho_\Gamma(\gamma_N))
\eq for $\rho_{S,\Gamma}(\Gp)\in \bar G_{S,\Gamma}$. This action defines a homeomorphism of $\Ab_\Gamma$, too. Notice that, the flux operator given by $\rho_{S,\Gamma}(\Gp)$ is for example restricted to a subgraph $\Gp$, whereas the holonomies are computated on the whole graph $\Gamma$. Mathematically, this is well-defined. Physically, the flux operators are somehow localised on a subgraph.

In subsection \ref{subsec fluxdef} the flux operators are constructed from maps, which map a graph $\Gamma$ to the structure group $G$. If the flux operators would be defined by groupoid morphisms between the finite path groupoid $\fPG$ and the groupoid $G$ over $\{e_G\}$ then the following difficulties arise.

\begin{rem}\label{rem fluxlikeoperators}
Let $\Hom_S(\PD_\Gamma\Sigma,G)$ be the set of groupoid morphisms between the finite path groupoid $\fPG$ and the groupoid $G$ over $\{e_G\}$ associated to a surface $S$ such that each groupoid morphism $p_{S}$ is an element of $\Map_S(\PD_\Gamma\Sigma,G)$. Then every groupoid morphism $p_{S}$ in $\Hom_S(\PD_\Gamma\Sigma,G)$ satisfies
\beqs p_{S}(\gamma^{-1})=p_{S}(\gamma)^{-1},\quad p_{S}(\gamma\circ\gp)=p_{S}(\gamma)p_{S}(\gp)\quad\forall\gamma\in\PD_\Gamma\Sigma,(\gamma, \gp)\in\PD_\Gamma\Sigma^{(2)}
\eqs and the maps $p_{S}$ have the special structure of $\Map_S(\PD_\Gamma\Sigma,G)$, which implements the intersection behavior of the paths of $\Gamma$ and the surface $S$. Observe that for a surface $S$ and a path $\gamma\circ\gp$ that intersects $S$ only in $s(\gamma)$ the maps $p_{S}$ satisfy $p_S(\gamma\circ\gp)=p_S(\gamma)$, since $p_S(\gp)=e_G$.

Due to the specific structure of the groupoid homomorphisms $\ho_\Gamma:\PD_\Gamma\Sigma\longrightarrow G$ there is in general no groupoid morphism $\mathfrak{H}$ defined by 
\[\Hom(\PD_\Gamma\Sigma,G)\ni\ho_\Gamma(\gamma)\mapsto\mathfrak{H}(\gamma):= p_{S}(\gamma)\ho_\Gamma(\gamma)\notin\Hom(\PD_\Gamma\Sigma,G)\]
for $p_{S}\in\Hom_S(\PD_\Gamma\Sigma,G)$. Let $\Gamma=\{\gamma,\gp\}$, then $\gamma,\gamma^\prime\in \PD_\Gamma\Sigma$ and assume that $(\gamma,\gamma^\prime)\in \PD_\Gamma\Sigma^{(2)}$. Then this is shown by the computation
\beq\label{eq problgrouphom1} & \mathfrak{H}(\gamma\circ\gamma^\prime)=p_{S}(\gamma\circ\gamma^\prime)\ho_\Gamma(\gamma\circ\gamma^\prime)=p_{S}(\gamma)p_{S}(\gamma^\prime)\ho_\Gamma(\gamma)\ho_\Gamma(\gamma^\prime)\\
&\neq p_{S}(\gamma)\ho_\Gamma(\gamma)p_{S}(\gamma^\prime)\ho_\Gamma(\gamma^\prime)\\
&=\mathfrak{H}(\gamma)\mathfrak{H}(\gamma^\prime)
\eq for $p_{S}\in\Hom_S(\PD_\Gamma\Sigma,G)$.

The equality holds for all $p_{S}\in\Map_{S}(\PD_\Gamma\Sigma,\ZD(G))$, where $\ZD(G)$ is the center of the group $G$. Hence $L(p_S)\circ\ho_\Gamma\in\Hom(\PD_\Gamma\Sigma,G)$ for $(L(p_S)\circ\ho_\Gamma)(\gamma):=p_{S}(\gamma)\ho_\Gamma(\gamma)$ and $p_S\in\Map_{S}(\PD_\Gamma\Sigma,\ZD(G))$ and $\gamma\in\PD_\Gamma\Sigma$. This indicate that for $p_{S}\in\Hom_S(\PD_\Gamma\Sigma,\ZD(G))$ the following properties are true
\begin{enumerate}
 \item\label{centercond1} $p_{S}(\gamma^\prime)\ho_\Gamma(\gamma)=\ho_\Gamma(\gamma)p_{S}(\gamma^\prime)$ for all paths $\gamma$ and $\gp$ contained in $\PD_\Gamma\Sigma$, 
 \item\label{centercond2} $p_{S}(\gamma^{-1})=p_{S}(\gamma)^{-1}$ and $p_{S}(\gamma\circ\gamma^\prime)=p_{S}(\gamma)p_{S}(\gamma^\prime)$ for all paths $\gamma,\gp\in\PD_\Gamma\Sigma$, 
 \item $p_{S}(\gamma)=p_{S}(\gpp)$ for all $\gamma,\gpp\in\PD_\Gamma\Sigma^v$ where $v=s(\gamma)=s(\gpp)$,
 \item $p_{S}(\gamma\circ\gamma^{-1})=e_G$ for each path $\gamma$ in $\PD_\Gamma\Sigma$ and
 \item $p_{S}(\gamma)=e_G$ if $\gamma$ in $\PD_\Gamma\Sigma$ does not intersect $S$ in the source or target vertex.
\end{enumerate}
Moreover consider $\ZD_{S,\Gamma}$ to be the the set
\beqs \ZD_{S,\Gamma}:=\{p_{S,\Gamma}\in \Hom_S(\PD_\Gamma,\ZD(G)^{\vert\Gamma\vert}): \quad&\exists\text{ } p_S\in\Hom_S(\PD_\Gamma\Sigma,\ZD(G))\text{ s.t. }\\
&p_{S,\Gamma}(\Gp)=(p_S(\gamma_1),...,p_S(\gamma_M))\\
&\forall\Gp\in\PD_\Gamma\}
\eqs and  $\ZD_{\breve S,\Gamma}:=\times_{S\in\breve S} \ZD_{S,\Gamma}$.
\end{rem}
\begin{rem}
Otherwise, in subsection \ref{subsec admfluxdef} the set of maps $\Gop^{\adm}_{\breve S,\Gamma}$ associated to a set of surfaces $\breve S$ is presented. For a pair of maps $\varrho_S^L$ and $\varrho_S^R$ in $\Gop^{\adm}_{\breve S,\Gamma}$ and for each surface $S$ in $\breve S$ it is true that
\begin{enumerate}\setcounter{enumi}{6}
 \item $S\cap\gamma=\{s(\gamma),t(\gamma)\}$ for all $\gamma\in\PD_\Gamma\Sigma$,
 \item\label{admiscond1} $\varrho_S^L(\gamma\circ\gp)=\varrho_S^L(\gamma)$,  $\varrho_S^R(\gamma\circ\gp)=\varrho_S^R(\gamma)$ for all $\gamma,\gamma\circ\gp\in\PD_\Gamma\Sigma^{v}$ and $S\cap\gamma\circ\gp=\{s(\gamma)\}$,
 \item\label{admiscond2} $\varrho_S^L((\gamma\circ\gp)^{-1})=\varrho_S^L(\gp^{-1})$, $\varrho_S^R((\gamma\circ\gp)^{-1})=\varrho_S^R(\gp^{-1})$ for all $\gp^{-1}\circ\gamma^{-1}\in\PD_\Gamma\Sigma^{v}$ and $S\cap\gp^{-1}\circ\gamma^{-1}=\{t(\gp^{-1})\}$,
 \item $\varrho_S^L(\gamma)^{-1}=\varrho_S^L(\gamma^{-1})$, $\varrho_S^R(\gamma)^{-1}=\varrho_S^R(\gamma^{-1})$ for $\gamma\in\PD_\Gamma\Sigma$,
 \item $\varrho^L_S(\gamma)=\varrho^L_S(\gpp)$,  $\varrho^R_S(\gamma)=\varrho^R_S(\gpp)$ for all $\gamma,\gpp\in\PD_\Gamma\Sigma^v$ where $v=s(\gamma)=s(\gpp)$,
 \item $\varrho_S^R(\gamma^{-1})^{-1}\varrho_S^L(\gp)=e_G$ for all $(\gamma,\gp)\in\PD_\Gamma\Sigma^{(2)}$
 \item\label{admiscond3} $\varrho_S^L(\gamma\circ\gamma^{-1})=e_G$, $\varrho_S^R(\gamma\circ\gamma^{-1})=e_G$ for a path $\gamma$ in $\PD_\Gamma\Sigma$ and
 \item $\varrho_S^L(\gamma)=e_G$, $\varrho_S^R(\gamma)=e_G$ if $\gamma$ in $\PD_\Gamma\Sigma$ does not intersect $S$ in the source or target vertex.
\end{enumerate}
Then for example, for a path $\gamma$ that intersects a fixed surface $S$ in the source vertex $s(\gamma)$ the map $\Go^L_\Gamma$ defined by
\beq\label{eq problgrouphom2}\Hom(\PD_\Gamma\Sigma,G)\ni\ho_\Gamma(\gamma)\mapsto \Go^L_\Gamma(\gamma):= \varrho_S^L(\gamma)\ho_\Gamma(\gamma)\notin\Hom(\PD_\Gamma\Sigma,G)\eq 
is not a groupoid morphism. This is verified by 
\beqs \Go^L_\Gamma(\gamma\circ\gp)&=\varrho_S^L(\gamma\circ\gp)\ho_\Gamma(\gamma\circ\gp)
=\varrho_S^L(\gamma)\ho_\Gamma(\gamma)\varrho_S^L(\gamma^{-1})^{-1}\varrho_S^L(\gp)\ho_\Gamma(\gp)\\
&\neq \Go^L_\Gamma(\gamma)\Go^L_\Gamma(\gp)=\varrho_S^L(\gamma)\ho_\Gamma(\gamma)\varrho_S^L(\gp)\ho_\Gamma(\gp)
\eqs Otherwise, if the path $\gamma$ intersects the surface $S$ in the target vertex $t(\gamma)$ then the map $\Go_\Gamma^R$ given by
\[\ho_\Gamma(\gamma)\mapsto \Go_\Gamma^R(\gamma):=\ho_\Gamma(\gamma)\rho_S^R(\gamma^{-1})^{-1}\] 
is not a groupoid morphism, too.

Recall in definition \ref{def similargroupoidhom} a groupoid morphism $\Go_\Gamma$ has been defined for a path $\gamma$ that intersects the surface $S$ in the source vertex $s(\gamma)$ and the target $t(\gamma)$ by
\[\ho_\Gamma(\gamma)\mapsto \Go_\Gamma(\gamma):=\varrho_S^L(\gamma)\ho_\Gamma(\gamma)\varrho_S^R(\gamma^{-1})^{-1}\]

Notice that, there is a difference between $\Hom_{S}(\PD_\Gamma\Sigma,\ZD(G))$ and $\Gop^{\adm}_{\breve S,\Gamma}$. For example in condition \ref{centercond1}, which solve the problem \eqref{eq problgrouphom1} of the groupoid multiplication by using the center of $G$ and condition \ref{centercond2} for maps in $\Hom_{S}(\PD_\Gamma\Sigma,\ZD(G))$. Otherwise the maps in $\Gop^{\adm}_{\breve S,\Gamma}$ satisfy in particularly condition \ref{admiscond3} for composable paths.
\end{rem}

\subsubsection*{A dynamical system of actions of the flux group on the analytic holonomy algebra for a fixed finite graph system and surfaces}

Let $\breve S$ be a suitable surface set for a finite orientation preserved graph system $\PD_\Gamma^{\op}$ and the object $\bar G_{\breve S,\Gamma}$ is the flux group. Then equivalently to a group action of $\bar G_{\breve S,\Gamma}$ on the configuration space $\Ab_\Gamma$ an action $\alpha$ of $\bar G_{\breve S,\Gamma}$ on the analytic holonomy $C^*$-algebra $C(\Ab_\Gamma)$ associated to a graph $\Gamma$ is studied as follows. The action is for example of the form
\beqs (\alpha(\rho_{S,\Gamma}(\Gamma))(f_\Gamma))(\ho_\Gamma(\Gamma)):= f_\Gamma(L(\rho_{S,\Gamma}(\Gamma))(\ho_\Gamma(\Gamma)))
\eqs where $\rho_{S,\Gamma}\in G_{\breve S,\Gamma}$ and $f_\Gamma\in C(\Ab_\Gamma)$. 

Notice that, there is a state on $C(\Ab_\Gamma)$, which is $\bar G_{\breve S,\Gamma}$-invariant. In this subsection a bunch of different actions of this form is constructed.

Before the investigations start the following remark on the nature of the definition of the flux operators has to be done. In subsection \ref{subsec fluxdef} the flux operators are constructed from maps, which map a graph $\Gamma$ to the structure group $G$. If the flux operators would be defined by groupoid morphisms between the finite path groupoid $\fPG$ and the groupoid $G$ over $\{e_G\}$ then difficulties arise, which have been studied in \cite[Sect.: 6.1]{KaminskiPHD}.

First restrict the surface and graph configuration to the simpliest case. Consider a surface $S$, which has the same surface intersection property for a graph $\Gamma$. This equivalent to consider a surface $S$ that intersects each path of $\Gamma$ in the source vertex of the path such that each path lies below the surface and is outgoing w.r.t. the surface orientation of $S$. Furthermore there are no other intersection points of the surface $S$ with paths of the graph $\Gamma$.

In this subsection representations and actions of the flux group $\bar G_{\breve S,\Gamma}$ in the $C^*$-algebra $C(\Ab_\Gamma)$ are studied instead of analysing group actions on the configuration space or transformation groups. These representations are maps from the flux group $\bar G_{\breve S,\Gamma}$ to the multiplier algebra of the $C^*$-algebra having several properties presented in \cite[Appendix 12.2]{KaminskiPHD}. In the following different actions of the flux group $\bar G_{S,\Gamma}$ on $C^*$-algebra $C(\Ab_\Gamma)$ are investigated first.

Assume that, $G$ is a compact  group.
Therefore the configuration space $\Ab_\Gamma$ for a finite graph system $\PD_\Gamma$ associated to a graph $\Gamma$ is a compact Hausdorff space. Let $\Alg_\Gamma$ be the quantum algebra generated by the configuration variables, which is isomorphic to $C(\Ab_\Gamma)$. Notice that, the elements of $\Ab_\Gamma$ are identified with $G^{\vert\Gamma\vert}$ by the \hyperlink{natural identification}{natural identification}. The evaluation of the holonomy map $\ho_\Gamma$ for a finite graph system $\PD_\Gamma$ on a subgraph $\Gp$ of $\Gamma$ is $\ho_\Gamma(\Gp)=(\ho_\Gamma(\gamma_1),....,\ho_\Gamma(\gamma_M),e_G,...,e_G)$ an element in $G^{\vert\Gamma\vert}$.

An important property of actions on commutative $C^*$-algebras is the the following.
\begin{defi}\label{def automorphic}Let $\Alg$ be an unital commutative $C^*$-algebra isometrically isomorphic to $C(X)$ where $X$ is a compact space, $G$ be an arbitrary group and $\alpha$ be an automorphism of $\Alg$.

Then the action $\alpha$ of $G$ on $\Alg$ is \textbf{automorphic} if the following conditions are satisfied 
\begin{enumerate}
\item\label{def automorphic1} $\alpha(gh)(f)=\alpha(g)(\alpha(h)(f))$ for any $f\in\Alg$, $g,h\in G$
\item\label{def automorphic2} $\alpha(g)(f_1f_2)=\alpha(g)(f_1)\alpha(g)(f_2)$ for any $f_1,f_2\in\Alg$, $g\in G$
\item\label{def automorphic3} $\alpha(g)(f^*)=\alpha(g)(f)^*$ for any $f\in\Alg$, $g\in G$
\end{enumerate}
\end{defi}

In this article the flux operators w.r.t. a surface $S$ are implemented as group actions of $\bar G_{\breve S,\Gamma}$ on the configuration space $\Ab_\Gamma$ or equivalently as group actions on $C^*$-algebras. 
Consequently for each surface set and graph system configuration an action on the holonomy algebra for a suitable finite graph system is defined. Investigate the following actions on the holonomy $C^*$-algebra $C(\Ab_\Gamma)$, where the configuration space $\Ab_\Gamma$ is identified with $G^{\vert \Gamma\vert}$ naturally.

\begin{lem} Let $\Gamma$ be a graph and $\PD^{\op}_\Gamma$ be the \hypertarget{finite orientation preserved graph system}{finite orientation preserved graph system} associated to $\Gamma$. Furthermore let $S$ be a fixed surface in $\Sigma$ such that $S$ intersects each path of $\Gamma$ in the source vertex of the path such that each path lies below the surface and is outgoing w.r.t. the surface orientation of $S$. There are no other intersection points of the surface $S$ with paths of the graph $\Gamma$.

Let $\bar G_{\breve S,\Gamma}$ denotes the flux group and $\Ab_\Gamma$ denotes the configuration space for the finite orientation preserved graph system $\PD_\Gamma^{\op}$, where all elements of $\PD_\Gamma^{\op}$ are identified in the natural way with a subset of the set of generators of $\Gamma$.

Then there is an action $\alpha$ of $\bar G_{S,\Gamma}$ on $C(\Ab_\Gamma)$ defined by
\beqs
(\alpha(\rho_{S,\Gamma}(\Gamma))f_\Gamma)(\ho_{\Gamma}(\Gamma)) 
&:= f_\Gamma(\rho_{S}(\gamma_1)\ho_{\Gamma}(\gamma_1),...,\rho_{S}(\gamma_N)\ho_{\Gamma}(\gamma_N))
\eqs for $\rho_{S,\Gamma}\in G_{\breve S,\Gamma}$ and $\rho_{S}\in\Gop_{S,\gamma}$, which is automorphic. 
\end{lem} 

\begin{proofs}Observe
\beqs (\alpha(\rho_{S,\Gamma}(\Gamma)\tilde\rho_{S,\Gamma}(\Gamma))f_\Gamma)(\ho_{\Gamma}(\gamma_1),...,\ho_{\Gamma}(\gamma_N))&= f_\Gamma(\rho_{S}(\gamma_1)\tilde\rho_{S}(\gamma_1)\ho_{\Gamma}(\gamma_1),...,\rho_{S}(\gamma_N)\tilde\rho_{S}(\gamma_N)\ho_{\Gamma}(\gamma_N))\\
&= (\alpha(\rho_{S,\Gamma}(\Gamma))(\alpha(\tilde\rho_{S,\Gamma}(\Gamma))f_\Gamma))(\ho_{\Gamma}(\gamma_1),...,\ho_{\Gamma}(\gamma_N))
\eqs The multiplication between two functions in $C(\Ab_\Gamma)$ is pointwise:
\beqs &\alpha(\rho_{S,\Gamma}(\Gamma))(f_\Gamma(\ho_{\Gamma}(\gamma_1),...,\ho_{\Gamma}(\gamma_N))f^\prime_\Gamma(\ho_{\Gamma}(\gamma_1),...,\ho_{\Gamma}(\gamma_N)))\\
&= \alpha(\rho_{S,\Gamma}(\Gamma))(\tilde f_\Gamma)(\ho_{\Gamma}(\gamma_1),...,\ho_{\Gamma}(\gamma_N))\\
&= \tilde f_\Gamma(\rho_{S}(\gamma_1)\ho_{\Gamma}(\gamma_1),...,\rho_{S}(\gamma_N)\ho_{\Gamma}(\gamma_N))\\
&= f_\Gamma(\rho_{S}(\gamma_1)\ho_{\Gamma}(\gamma_1),...,\rho_{S}(\gamma_N)\ho_{\Gamma}(\gamma_N))f^\prime_\Gamma(\rho_{S}(\gamma_1)\ho_{\Gamma}(\gamma_1),...,\rho_{S}(\gamma_N)\ho_{\Gamma}(\gamma_N))\\
&= (\alpha(\rho_{S,\Gamma}(\Gamma))f_\Gamma)(\ho_{\Gamma}(\gamma_1),...,\ho_{\Gamma}(\gamma_N))(\alpha(\rho_{S,\Gamma}(\Gamma))f_\Gamma^\prime)(\ho_{\Gamma}(\gamma_1),...,\ho_{\Gamma}(\gamma_N))
\eqs 
Finally, recover
\beqs (\alpha(\rho_{S,\Gamma}(\Gamma))f^*_\Gamma)(\ho_{\Gamma}(\gamma_1),...,\ho_{\Gamma}(\gamma_N)) &= \overline{f_\Gamma(\rho_{S}(\gamma_1)\ho_{\Gamma}(\gamma_1),...,\rho_{S}(\gamma_N)\ho_{\Gamma}(\gamma_N))}\\
&= (\alpha(\rho_{S,\Gamma}(\Gamma))f_\Gamma)(\ho_{\Gamma}(\gamma_1),...,\ho_{\Gamma}(\gamma_N))^*
\eqs 
\end{proofs}
In particular there is a map $\tilde J:\bar \Gop_{\breve S,\Gamma}\longrightarrow \bar \Gop_{\breve S,\Gamma}$, $J:\rho_S(\gamma)\mapsto\rho_{S^{-1}}(\gamma)$. Then for $\rho_S(\gamma)=:g$, where the surface $S$ and the path $\gamma$ are suitable and $\breve S:=\{S,S^{-1}\}$, the map satisfies $(\tilde J(\rho_S))(\gamma)=\rho_{S^{-1}}(\gamma)=g^{-1}$. With no doubt there is a general map  $J:\bar G_{\breve S,\Gamma}\longrightarrow \bar G_{\breve S,\Gamma}$, $J:\rho_{S,\Gamma}(\Gamma)\mapsto\rho_{S^{-1},\Gamma}(\Gamma)$. Then derive
\beq &(\alpha(\rho_{S,\Gamma}(\Gamma)\big(J(\rho_{S,\Gamma})(\Gamma)\big))f_\Gamma)(\ho_{\Gamma}(\gamma_1),...,\ho_{\Gamma}(\gamma_N))\\
&= f_\Gamma(\rho_{S}(\gamma_1)\rho_{S^{-1}}(\gamma_1)\ho_{\Gamma}(\gamma_1),...,\rho_{S}(\gamma_N)\rho_{S^{-1}}(\gamma_N)\ho_{\Gamma}(\gamma_N))\\
&= f_\Gamma(gg^{-1}\ho_{\Gamma}(\gamma_1),...,gg^{-1}\ho_{\Gamma}(\gamma_N))\\
&=f_\Gamma(\ho_{\Gamma}(\gamma_1),...,\ho_{\Gamma}(\gamma_N))
\eq

It is necessary that, the action is defined for a finite orientation preserved graph system, since the following observation is made. Let $\Gamma$ be equivalent to a path $\gamma$. Then it is true that $\ho_\Gamma(\gamma^{-1})=(\ho_\Gamma(\gamma))^{-1}$ and consequently 
\beqs
(\alpha(\rho_{S,\Gamma}(\Gamma))f_\Gamma)((\ho_{\Gamma}(\gamma))^{-1}) 
&= f_\Gamma(\rho_{S}^L(\gamma)(\ho_{\Gamma}(\gamma))^{-1})
= f_\Gamma(\rho_{S}^L(\gamma)\ho_{\Gamma}(\gamma^{-1}))
\eqs holds for $\rho_{S,\Gamma}\in G_{\breve S,\Gamma}$ and $\rho_{S}^L\in\Gop_{S,\gamma}$. Although the holonomy is evaluated for the path $\gamma^{-1}$ the action is defined by the left action $L$. Later it is analysed, how a left action $L$ can be transfered to a right action $R$. The right action is given by  
\beqs
(\alpha(\rho_{S,\Gamma^{-1}}(\Gamma^{-1}))f_\Gamma)(\ho_{\Gamma}(\gamma^{-1}))
&= f_\Gamma(R(\rho_{S}^R(\gamma^{-1}))(\ho_{\Gamma}(\gamma^{-1})))
= f_\Gamma(\ho_{\Gamma}(\gamma^{-1})\rho_{S}^R(\gamma^{-1})^{-1})
\eqs for $\rho_{S,\Gamma^{-1}}\in G_{\breve S,\Gamma^{-1}}$ and $\rho_{S}^R\in\Gop_{S,\gamma^{-1}}$.

Furthermore an automorphic action has another interesting property, which allows to speak about $C^*$-dynamical systems.
\begin{cor}Let $\Gamma$ be a graph and $\PD^{\op}_\Gamma$ be the \hypertarget{finite orientation preserved graph system}{finite orientation preserved graph system} associated to $\Gamma$. Furthermore let $S$ be a fixed surface in $\Sigma$ such that $S$ intersects each path of $\Gamma$ in the source vertex of the path such that each path lies below the surface and is outgoing w.r.t. the surface orientation of $S$. There are no other intersection points of the surface $S$ with paths of the graph $\Gamma$.

Let $\bar G_{\breve S,\Gamma}$ denotes the flux group and $\Ab_\Gamma$ denotes the configuration space for the finite orientation preserved graph system $\PD_\Gamma^{\op}$, where all elements of $\PD_\Gamma^{\op}$ are identified in the natural way with a subset of the set of generators of $\Gamma$, where all elements of $\PD_\Gamma^{\op}$ are identified in the natural way with a subset of the set of generators of $\Gamma$.

The triple $(\bar G_{S,\Gamma},C(\Ab_\Gamma),\alpha)$ consisting of a compact group $\bar G_{S,\Gamma}$, a $C^*$-algebra $C(\Ab_\Gamma)$ and an automorphic action $\alpha$ of $\bar G_{S,\Gamma}$ on $C(\Ab_\Gamma)$ such that for each $f_\Gamma\in C(\Ab_\Gamma)$ the function \(\bar G_{S,\Gamma}\ni\rho_{S,\Gamma}(\Gamma)\mapsto\|\alpha(\rho_{S,\Gamma}(\Gamma))(f_\Gamma)\|\) is continuous\footnote{I.o.w. for every $f_\Gamma\in \Alg_\Gamma$ the map $\alpha: \rho_{S,\Gamma}(\Gamma)\mapsto \alpha(\rho_{S,\Gamma}(\Gamma))(f_\Gamma)$ is a continuous map from the $\bar G_{S,\Gamma}$-open set topology on $\bar G_{S,\Gamma}$ to the norm topology on $\Alg_\Gamma$ ($\alpha$ is point-norm continuous).}, is a \textbf{$C^*$-dynamical system for a finite orientation preserved graph system associated to a graph} $\Gamma$. 
\end{cor}
\begin{proofs}Let $f_\Gamma\in C(\Ab_\Gamma)$ and $\Gamma:=\{\gamma\}$ then for a fixed suitable surface $S$ it is true that
\beqs &\lim_{\rho_{S,\Gamma}(\Gamma)\longrightarrow \id_{S,\Gamma}(\Gamma)}\|\alpha(\rho_{S,\Gamma}(\Gamma))(f_\Gamma)- f_\Gamma\|\\
\quad&= \lim_{\rho_{S}(\gamma)\longrightarrow \id_{S,\Gamma}(\gamma)}\|f_\Gamma(L(\rho_{S,\Gamma}(\gamma))(\ho_\Gamma(\gamma))) -f_\Gamma(\ho_\Gamma(\gamma))\|\\\quad&=0
\eqs yields for $\rho_{S,\Gamma},\id_{S,\Gamma}\in G_{S,\Gamma}$ and $\rho_S,\id_{S}\in\Gop_{S,\gamma}$, if $\rho_{S}(\gamma)=g\in G$ and  $\id_S(\gamma)=e_G$ for all $\gamma\in\Gamma$.
\end{proofs}

Observe if $\Alg_\Gamma$ is a $C^*$-subalgebra of the $C^*$-algebra $\LD(\HS_\Gamma)$ of bounded operators on a Hilbert space $\HS_\Gamma$, then $\Alg_\Gamma$ is non-degeneratly represented on $\HS_\Gamma$ if the inclusion map of $\Alg_\Gamma$ into $\LD(\HS_\Gamma)$ is a non-degenerate representation of $\Alg_\Gamma$. Set $\HS_\Gamma$ be equal to $L^2(\Ab_\Gamma,\mu_\Gamma)$. Then the multiplication representation $\Phi_M$ of $C(\Ab_\Gamma)$ in $\HS_\Gamma$ defined by
\beq \Phi_M(f_\Gamma) \psi_\Gamma=f_\Gamma\cdot  \psi_\Gamma\text{ for } \psi_\Gamma\in\HS_\Gamma\text{ and } f_\Gamma\in \Alg_\Gamma\eq is non-degenerate. 

\begin{defi}Let $S$ be a fixed surface in $\Sigma$, $\Gamma$ be a graph, $\PD^{\op}_\Gamma$ be the \hypertarget{finite orientation preserved graph system}{finite orientation preserved graph system} associated to $\Gamma$ such that the surface $S$ has the \hyperlink{same surface intersection property for a ori}{surface intersection property for a finite orientation preserved graph system} $\PD^{\op}_\Gamma$. Let $\bar G_{\breve S,\Gamma}$ denotes the flux group and $\Ab_\Gamma$ denotes the configuration space for the finite orientation preserved graph system $\PD_\Gamma^{\op}$, where all elements of $\PD_\Gamma^{\op}$ are identified in the natural way with a subset of the set of generators of $\Gamma$.

A \textbf{covariant representation} of the $C^*$-dynamical system $(\bar G_{S,\Gamma},C(\Ab_\Gamma),\alpha)$ in a $C^*$-algebra $\LD(\HS_\Gamma)$ consists of a pair $(\Phi_M,U)$ where $\Phi_M$ is a non-degenerate representation of $C(\Ab_\Gamma)$ on a Hilbert space $\HS_\Gamma$ (i.e. $\Phi_M\in\Mor(C(\Ab_\Gamma),\LD(\HS_\Gamma))$) and $U$ is a (strongly continuous) unitary representation of $\bar G_{S,\Gamma}$ on $\HS_\Gamma$ such that for all $\rho_{S,\Gamma}(\Gamma)\in \bar G_{S,\Gamma}$ and $f_\Gamma\in C(\Ab_\Gamma)$ it is true that
\beq\label{eq weyl relationU} U(\rho_{S,\Gamma}(\Gamma))\Phi_M(f_\Gamma)U^*(\rho_{S,\Gamma}(\Gamma))=\Phi_M(\alpha(\rho_{S,\Gamma}(\Gamma))(f_\Gamma))\eq holds.
The pair $(\Phi_M,U)$ is also called a \textbf{covariant Hilbert space representation of the $C^*$-dynamical system} $(\bar G_{S,\Gamma},C(\Ab_\Gamma),\alpha)$.
\end{defi}

There is a covariant representation of $( \bar G_{S,\Gamma},C(\Ab_\Gamma),\alpha)$ with respect to the automorphic action $\alpha$ of $\bar G_{S,\Gamma}$ on $C(\Ab_\Gamma)$ given by $M$ and the unitary operator $U$, which is a map $U: \bar G_{S,\Gamma}\longrightarrow U(\HS_\Gamma)$, where $U(\HS_\Gamma)$ is the unitary group of $\LD(\HS_\Gamma)$, since
\beq U(\rho_{S,\Gamma}(\Gamma)) U(\rho_{S,\Gamma}(\Gamma))^*\psi_\Gamma(\ho_\Gamma(\Gamma))
&= U(\rho_{S,\Gamma}(\Gamma)) U(\rho_{S,\Gamma}(\Gamma)^{-1})\psi_\Gamma(\ho_\Gamma(\Gamma))\\
&= U(\rho_{S,\Gamma}(\Gamma))\psi_\Gamma(L(\rho_{S,\Gamma}(\Gamma)^{-1})(\ho_\Gamma(\Gamma)))\\
&=\psi_\Gamma(L(\rho_{S,\Gamma}(\Gamma)\rho_{S,\Gamma}(\Gamma)^{-1})(\ho_\Gamma(\Gamma)))\\
&=\psi_\Gamma(\ho_\Gamma(\Gamma))
\eq holds for $\psi_\Gamma\in\HS_\Gamma$ and which satisfies \eqref{eq weyl relationU}. Notice that, $ U(\rho_{S,\Gamma}(\Gamma))^*=U(\rho_{S^{-1},\Gamma}(\Gamma))$.

In the following it is often assumed that, a strongly continuous unitary representation of the flux group on a Hilbert space is a representation of the group on the $C^*$-algebra of compact operators on the Hilbert space. Therefore this relation is explicity given in the following lemma.
\begin{lem}Let  $(\Phi_M,U)$ is a covariant representation $U$ of a $C^*$-dynamical system $(\bar G_{S,\Gamma},C(\Ab_\Gamma),\alpha)$.
Then the strongly continuous unitary representation of $\bar G_{S,\Gamma}$ on $\HS_\Gamma$ is a representation of the group $\bar G_{S,\Gamma}$ on the $C^*$-algebra of compact operators on the Hilbert space $\HS_\Gamma$, i.e. $U\in\Rep(\bar G_{S,\Gamma},\KD(\HS_\Gamma))$.
\end{lem}

\begin{prop}\label{prop invstate} Let $\Gamma$ be a graph and $\PD^{\op}_\Gamma$ be the \hypertarget{finite orientation preserved graph system}{finite orientation preserved graph system} associated to $\Gamma$. Furthermore let $S$ be a fixed surface in $\Sigma$ such that $S$ intersects each path of $\Gamma$ in the source vertex of the path such that each path lies below the surface and is outgoing w.r.t. the surface orientation of $S$. There are no other intersection points of the surface $S$ with paths of the graph $\Gamma$.

Let $\bar G_{\breve S,\Gamma}$ denotes the flux group and $\Ab_\Gamma$ denotes the configuration space for a identified in the natural way finite orientation preserved graph system $\PD_\Gamma^{\op}$. Moreover let $(\bar G_{S,\Gamma},C(\Ab_\Gamma),\alpha)$ be a $C^*$-dynamical system and $(M,U)$ a covariant representation of $(\bar G_{S,\Gamma},C(\Ab_\Gamma),\alpha)$ on a Hilbert space $\HS_\Gamma$. 

Then there exists a GNS-triple $(\HS_\Gamma,\Phi_M,\Omega_\Gamma)$, where $\Omega_\Gamma$ is the cyclic vector for $\Phi_M$ on $\HS_\Gamma$. Moreover the associated GNS-state $\omega_M^\Gamma$ on $C(\Ab_\Gamma)$ is $\bar G_{S,\Gamma}$-invariant, i.e. 
\[\omega^\Gamma_M(\alpha(\rho_{S,\Gamma}(\Gamma))(f_\Gamma))=\omega^\Gamma_M(f_\Gamma):=\langle \Omega_\Gamma,\Phi_M(f_\Gamma)\Omega_\Gamma\rangle_\Gamma\] holds for all $\rho_{S,\Gamma}\in  G_{S,\Gamma}$ and $f_\Gamma\in C(\Ab_\Gamma)$.
\end{prop}

In general automorphic actions on $C^*$-algebras define $C^*$-dynamical systems, since they are related to covariant representations of groups on $C^*$-algebras (refer to \cite[\refapprepgroupalg]{KaminskiPHD}). This is connected to the definition of inner automorphisms.
\begin{defi}Let $G$ be group and $\Alg$ be a $C^*$-algebra.

An automorphic action $\alpha$ of a group $G$ on $\Alg$ is called \textbf{inner}, if there is a representation of the group $G$ on $\Alg$, i.e. $U\in\Rep(G,\Alg)$ such that
\beqs \alpha_{g}(A)=U(g)A U(g)^*\eqs  whenever $A\in\Alg$ and $g\in G$. Otherwise $\alpha$ is called outer.
\end{defi}
But since the holonomy $C^*$-algebra $C(\Ab_\Gamma)$ for a finite graph system is commutative, there is only one inner automorphic action of $\bar G_{S,\Gamma}$ on $C(\Ab_\Gamma)$ given by the trivial one.

Due to the different intersection behavior of surfaces and paths, there are a lot of different  automorphic actions on the holonomy $C^*$-algebra for finite graph systems. 

\subsubsection*{Dynamical systems of the analytic holonomy algebra and actions of the flux goup w.r.t. different graph and surface configurations}

A graph is a set of paths, each path and a surface $S$ has a specific intersection behavior. In general a graph does not contain only paths that are ingoing and lie above w.r.t. the surface orientation of $S$. In this subsection different actions for general graphs are studied. In the interesting configurations the corresponding actions on the analytic holonomy algebra turn out to be automorphic and point-norm continuous. There are only few which have to be excluded.

\paragraph*{Purely left or right actions of the flux group\\}\hspace*{10pt}

In the construction of dynamical systems the following surface and graph configurations play a particular role. This implies that, particular actions, which are for example purely left or right actions of a group on a $C^*$-algebra, are analysed. 

During the whole subsection the set $\bar G_{\breve S,\Gamma}$ and the multiplication operation 
\beqs (\rho_x^{S_1}(\gamma_1),...,\rho_x^{S_N}(\gamma_N))\cdot(\tilde\rho_x^{S_1}(\gamma_1),...,\tilde\rho^{S_N}_x(\gamma_N))=(\rho_x^{S_1}(\gamma_1)\tilde\rho^{S_1}_x(\gamma_1),...,\rho_x^{S_N}(\gamma_N)\tilde\rho_x^{S_N}(\gamma_N))
\eqs where $x$ is equal to $L$ or $R$ and $\breve S=\{S_i\}_{1\leq i\leq \vert\Gamma\vert}$, is considered. If $\breve S$ contains for example only one surface, then $S_i=S$ for all $1\leq i\leq \vert\Gamma\vert$. The right group multiplication is explicity assumed in the context. If both left and right multiplication would be used, the set $\bar G_{\breve A,\Gamma}$ does not form a group and consequently there are problems in the definition of automorphic actions, which are stated in the problem \thesubsection.\ref{problem I}.

For a summary recall the definition of the last subsection.
\begin{lem}\label{lem allpaths}
Let $\Gamma$ be a graph and $\PD^{\op}_\Gamma$ be the \hypertarget{finite orientation preserved graph system}{finite orientation preserved graph system} associated to $\Gamma$. Furthermore let $S$ be a fixed surface in $\Sigma$ such that $S$ intersects each path of $\Gamma$ in the source vertex of the path such that each path lies below the surface and is outgoing w.r.t. the surface orientation of $S$. There are no other intersection points of the surface $S$ with paths of the graph $\Gamma$.

Then redefine the action
\beq\label{eq actionsurfgraph1} (\alpha^1_{\overleftarrow{L}}(\rho^1_{S,\Gamma})f_\Gamma)(\ho_{\Gamma}(\Gamma))&:= f_\Gamma(\rho_{S}(\gamma_1)\ho_{\Gamma}(\gamma_1),...,\rho_{S}(\gamma_N)\ho_{\Gamma}(\gamma_N))\\
&= f_\Gamma(g_S\ho_{\Gamma}(\gamma_1),...,g_S\ho_{\Gamma}(\gamma_N))
\eq for $\rho^1_{S,\Gamma}=(\rho_S(\gamma_1),...,\rho_S(\gamma_N))=(g_S,...,g_S)$, $\rho^1_{S,\Gamma}\in \bar G_{S,\Gamma}$ such that $\rho_{S}\in\Gop_{S,\gamma}$ and $f_\Gamma\in C(\Ab_\Gamma)$.

Then the action $\alpha_{\overleftarrow{L}}^1$ of $\bar G_{S,\Gamma}$ on $C(\Ab_\Gamma)$ is automorphic and point-norm continuous.
\end{lem} 

For a simplification of the following considerations allways assume that, there exists a \hypertarget{finite orientation preserved graph system}{finite orientation preserved graph system} $\PD^{\op}_\Gamma$ associated to $\Gamma$. Furthermore there are no other intersection points of the surface $S$ with paths of the graph $\Gamma$ except the intersections, which are required in the different lemmata. The proofs can be found in \cite[Section 6.1]{KaminskiPHD}.
\begin{lem}\label{lem actonlygammaN}
Let only the path $\gamma_N$ in $\Gamma$ intersect in (source) vertex of the set $V_\Gamma$ with a surface $S$ such that $\gamma_N$ is outgoing and lies above the surface $S$. 

Then define an action
\beq\label{eq actionsurfgraph1} (\alpha_{\overrightarrow{L}}^{1,1}(\rho^{1,1}_{S,\Gamma})f_\Gamma)(\ho_{\Gamma}(\gamma_1),...,\ho_{\Gamma}(\gamma_N))
&:= f_\Gamma(\rho_{S}(\gamma_1)\ho_{\Gamma}(\gamma_1),...,\rho_{S}(\gamma_N)\ho_{\Gamma}(\gamma_N))\\
&= f_\Gamma(\ho_{\Gamma}(\gamma_1),...,g_S^{-1}\ho_{\Gamma}(\gamma_N))
\eq for $\rho_{S,\Gamma}^{1,1}\in \bar G_{\breve S,\Gamma}$ and $\rho_{S}\in\Gop_{S,\gamma}$ and $\rho^{1,1}_{S,\Gamma}=(\rho_S(\gamma_1),...,\rho_S(\gamma_N))=(e_G,...,e_G,g_S^{-1})$.

Then the action $\alpha_{\overrightarrow{L}}^{1,1}$ of $\bar G_{S,\Gamma}$ on $C(\Ab_\Gamma)$ is automorphic and point-norm continuous.
\end{lem}

\begin{lem}\label{lem allpaths1N-1}Let $\Gamma$ be a graph given by $\{\gamma_1,...,\gamma_N\}$.
Moreover let only the paths $\gamma_1,...,\gamma_{1,N-1}$ intersect in (source) vertices of the set $V_\Gamma$ with a surface $S$ such that all paths are outgoing and lie below. 

Then define the action
\beq\label{eq actionsurfgraph2}  (\alpha_{\overleftarrow{L}}^{1,N-1}(\rho^{1,N-1}_{S,\Gamma})f_\Gamma)(\ho_{\Gamma}(\gamma_1),...,\ho_{\Gamma}(\gamma_N))
&:= f_\Gamma(\rho_{S}(\gamma_1)\ho_{\Gamma}(\gamma_1),...,\rho_{S}(\gamma_N)\ho_{\Gamma}(\gamma_N))\\
&= f_\Gamma(g_S\ho_{\Gamma}(\gamma_1),...,g_S\ho_{\Gamma}(\gamma_{1,N-1}),\ho_{\Gamma}(\gamma_N))
\eq for $\rho_{S,\Gamma}^{1,N-1}\in \bar G_{S,\Gamma}$ and $\rho_S\in\Gop_{S,\gamma}$ and $\rho^{1,N-1}_{S,\Gamma}=(\rho_S(\gamma_1),...,\rho_S(\gamma_N))=(g_S,...,g_S,e_G)$.

The action $\alpha_{\overleftarrow{L}}^{1,N-1}$ of $\bar G_{S,\Gamma}$ on $C(\Ab_\Gamma)$ is automorphic and point-norm continuous.
\end{lem}
In the following an action is defined, which does not lead to a point-norm continuous automorphic action on the analytic holonomy algebra.
 \begin{lem}Let $\Gamma$ be a graph given by $\{\gamma_1,...,\gamma_N\}$.
Moreover let all paths intersect in (source) vertices of the set $V_\Gamma$ with a surface $S$ such that all paths $\gamma_1,...,\gamma_{N-1}$ are outgoing and lie below, $\gamma_N$ is outgoing and lies above the surface $S$. There are no other intersection point of paths and the surface $S$.

Then define the action
\beq\label{eq actionsurfgraph2}  (\alpha_L^1(\rho^1_{S,\Gamma})f_\Gamma)(\ho_{\Gamma}(\gamma_1),...,\ho_{\Gamma}(\gamma_N))
&:= f_\Gamma(\rho_{S}(\gamma_1)\ho_{\Gamma}(\gamma_1),...,\rho_{S}(\gamma_N)\ho_{\Gamma}(\gamma_N))\\
&= f_\Gamma(g_S\ho_{\Gamma}(\gamma_1),...,g_S^{-1}\ho_{\Gamma}(\gamma_N))
\eq for $\rho_{S,\Gamma}^1\in  G_{\breve S,\Gamma}$ and $\rho_S\in\Gop_{S,\gamma}$ and $\rho^1_{S,\Gamma}=(\rho_S(\gamma_1),...,\rho_S(\gamma_N))=(g_S,...,g_S^{-1})$.

The action $\alpha_L^1$ of $\bar G_{S,\Gamma}$ on $C(\Ab_\Gamma)$ is automorphic and the action is not point-norm continuous. 
\end{lem}

\begin{lem}Let $\Gamma:=\{\gamma_1,...,\gamma_N\}$ be a graph. The paths $\{\gamma_1,...,\gamma_{1,N-1}\}$ intersect in their source vertices with a surface $S$ such that the paths are outgoing and lie below the surface $S$. The path $\gamma_N$ in $\Gamma$ intersects in the source vertex with the surface $S^\prime$ such that $\gamma_N$ is outgoing and lies above the surface $S^\prime$. There are no other intersection point of paths and the surface $S$ and $S^\prime$.

Then the action defined by
\beq\label{eq actionsurfgraph2}  (\alpha_{L}^2(\rho_{S,\Gamma}(\Gamma),\rho_{S^\prime,\Gamma}(\Gamma))f_\Gamma)(\ho_{\Gamma}(\gamma_1),...,\ho_{\Gamma}(\gamma_N))
&:= f_\Gamma(\rho_{S}(\gamma_1)\ho_{\Gamma}(\gamma_1),...,\rho_{S^\prime}(\gamma_N)\ho_{\Gamma}(\gamma_N))\\
&= f_\Gamma(g_S\ho_{\Gamma}(\gamma_1),...,g_{S^\prime}^{-1}\ho_{\Gamma}(\gamma_N))
\eq for $\rho_{\breve S,\Gamma}^1\in \bar G_{\breve S,\Gamma}$ and $\rho_S,\rho_{S^\prime}\in\Gop_{\breve S,\gamma}$ and $\rho^{1}_{\breve S,\Gamma}=(\rho_S(\gamma_1),...,\rho_{S^\prime}(\gamma_N))=(g_S,...,g_{S^\prime}^{-1})$ where $g_S\neq g_{S^\prime}$.

The action $\alpha_{L}^2$ of $\bar G_{\breve S,\Gamma}$ on $C(\Ab_\Gamma)$ is automorphic and point-norm continuous.
\end{lem}

\begin{lem}\label{lem diffactionsLN}Let all paths in $\Gamma:=\{\gamma_1,...,\gamma_N\}$ intersect in vertices of the set $V_\Gamma$ with a surface $S_N$ such that $\gamma_1,...,\gamma_{N}$ are outgoing and lie below the surface $S_N$. Let the paths $\gamma_1,...,\gamma_{N-1}$ in $\Gamma$ intersect in vertices of the set $V_\Gamma$ with a surface $S_{N-1}$ such that $\gamma_1,...,\gamma_{1,N-1}$ are outgoing and lie below the surface $S_{N-1}$. The same is true for a surface $S_{N-2}$ and paths $\{\gamma_1,...,\gamma_{N-2}\}$, and so on, til $S_{1}$ and $\{\gamma_1\}$. There are no other intersections between the paths and surfaces $S_{1},...,S_{N-1}$ and $S_N$. Moreover let each path $\gamma_i$ in $\Gamma$ intersects in vertices of the set $V_\Gamma$ with a surface $S_{1,i}$ such that $\gamma_i$ is outgoing and lies below the surface $S_i$ for $i=1,...,N$. There are no other intersections between the paths in $\Gamma$ and surfaces $S_{1,1},...,S_{1,N}$. The surfaces $S_N$ and $S_{1,N}$ coincide. The set $\breve S:=\{S_{1,i}\}_{1\leq i\leq N}$ has the \hyperlink{simple surface intersection property for a graph}{simple surface intersection property} for $\Gamma$.

The action for two different maps $\rho_{S_{N-1},\Gamma}$ in $G_{S_{N-1},\Gamma}$ and $\tilde\rho_{S_{N},\Gamma}$ in $G_{S_{N},\Gamma}$, such that there is an action of $\bar G_{S_{1,N-1},\Gamma}\times \bar G_{S_{N},\Gamma}$ on $C(\Ab_\Gamma)$ given by
\beq\label{eq actionsurfgraph3} &(\alpha^2_{\overleftarrow{L}}(\rho_{S_{N-1},\Gamma},\tilde\rho_{S_{N},\Gamma})f_\Gamma)(\ho_{\Gamma}(\gamma_1),...,\ho_{\Gamma}(\gamma_N))\\
&:= f_\Gamma(\rho_{S_{N-1}}(\gamma_1)\ho_{\Gamma}(\gamma_1),\rho_{S_{N-1}}(\gamma_2)\ho_{\Gamma}(\gamma_2),...,\rho_{S_{N-1}}(\gamma_{N-1})\ho_{\Gamma}(\gamma_{N-1}),\tilde\rho_{S_{N}}(\gamma_N)\ho_{\Gamma}(\gamma_N))\\
&= f_\Gamma(g_{S_{N-1}}\ho_{\Gamma}(\gamma_1),g_{S_{N-1}}\ho_{\Gamma}(\gamma_2),...,g_{S_{N-1}}\ho_{\Gamma}(\gamma_{N-1}),h_{S_{N}}\ho_{\Gamma}(\gamma_N))\eq

The action for $(N-1)$-different maps $\rho_{S_2,\Gamma}$ in $G_{S_2,\Gamma}$ , $\tilde\rho_{S_{1,3},\Gamma}$ in $G_{S_{1,3},\Gamma}$ til $\breve\rho_{S_{1,N-2},\Gamma}$ in $G_{S_{1,N-2},\Gamma}$, such that there is an action of $\bar G_{S_{2},\Gamma}\times \bar G_{S_{1,3},\Gamma}\times\cdots \times \bar G_{S_{1,N-2},\Gamma}$ on $C(\Ab_\Gamma)$ given by
\beq\label{eq actionsurfgraph4} &(\alpha_{\overleftarrow{L}}^{N-1}((\rho_{S_2,\Gamma},\rho_{S_2,\Gamma},\tilde\rho_{S_{1,3},\Gamma}...,\breve\rho_{S_{1,N-1},\Gamma}))f_\Gamma)(\ho_{\Gamma}(\gamma_1),...,\ho_{\Gamma}(\gamma_N))\\ &:= f_\Gamma(\rho_{S_2}(\gamma_1)\ho_{\Gamma}(\gamma_1),\rho_{S_2}(\gamma_2)\ho_{\Gamma}(\gamma_2),\tilde\rho_{S_{1,3}}(\gamma_3)\ho_{\Gamma}(\gamma_3),...,\breve\rho_{1,S_{N-2}}(\gamma_N)\ho_{\Gamma}(\gamma_N))\\
&= f_\Gamma(g_{S_2}\ho_{\Gamma}(\gamma_1),g_{S_2}\ho_{\Gamma}(\gamma_2),k_{S_{1,3}}\ho_{\Gamma}(\gamma_3),..., h_{S_{1,N-2}}\ho_{\Gamma}(\gamma_N))\eq
whenever $N>5$.

Respectively the action of $N$-different maps is equivalent to an action of $\bar G_{S_{1,1},\Gamma}\times\cdots \times \bar G_{S_{1,N},\Gamma}$ on $C(\Ab_\Gamma)$, which is defined by
\beq\label{eq actionsurfgraph5} &(\alpha^N_{\overleftarrow{L}}(\rho_{S_{1,1},\Gamma}(\Gamma),...,\breve\rho_{S_{1,N},\Gamma}(\Gamma))f_\Gamma)(\ho_{\Gamma}(\gamma_1),...,\ho_{\Gamma}(\gamma_N)) \\
&:= f_\Gamma(\rho_{S_{1,1}}(\gamma_1)\ho_{\Gamma}(\gamma_1),...,\rho_{S_{1,N}}(\gamma_N)\ho_{\Gamma}(\gamma_N))\\
&= f_\Gamma(g_{S_{1,1}}\ho_{\Gamma}(\gamma_1),...,h_{S_{1,N}}\ho_{\Gamma}(\gamma_N))\eq

Then the actions $\alpha_{\overleftarrow{L}}^2$, ..., $\alpha_{\overleftarrow{L}}^{1,N-1}$ and $\alpha_{\overleftarrow{L}}^N$ of $\bar G_{\breve S,\Gamma}$ on $C(\Ab_\Gamma)$ are automorphic and point-norm continuous actions.
\end{lem}
\begin{lem}Let all paths of a graph $\Gamma$ intersect in vertices of the set $V_\Gamma$ with a surface $S$ such that all paths are ingoing and lie above the surface. Moreover let the set $\breve S:=\{S_{1,i}\}_{1\leq i\leq N}$ has the \hyperlink{simple surface intersection property for a graph}{simple surface intersection property} for $\Gamma$.

Then there is an action such that
\beq\label{eq actionsurfgraph6} &(\alpha^{\overleftarrow{R}}_1(\rho^1_{S,\Gamma})f_\Gamma(\ho_{\Gamma}(\gamma_1),...,\ho_{\Gamma}(\gamma_N))f_\Gamma)(\ho_{\Gamma}(\gamma_1),...,\ho_{\Gamma}(\gamma_N)) \\
&\quad:= f_\Gamma(\ho_{\Gamma}(\gamma_1)\rho_S(\gamma_1)^{-1},...,\ho_{\Gamma}(\gamma_N)\rho_S(\gamma_N)^{-1})\\
&\quad= f_\Gamma(\ho_{\Gamma}(\gamma_1)g^{-1}_S,...,\ho_{\Gamma}(\gamma_N)g^{-1}_S)
\eq 
This action is changed such that
\beq\label{eq actionsurfgraph7} &(\alpha^{\overleftarrow{R}}_N(\rho^N_{\breve S,\Gamma})f_\Gamma)(\ho_{\Gamma}(\gamma_1),...,\ho_{\Gamma}(\gamma_N))\\
&\quad:= f_\Gamma(\ho_{\Gamma}(\gamma_1)\rho_{S_{1,1}}(\gamma_1)^{-1},...,\ho_{\Gamma}(\gamma_N)\rho_{S_{1,N}}(\gamma_N)^{-1})\\
&\quad= f_\Gamma(\ho_{\Gamma}(\gamma_1)g^{-1}_{S_{1,1}},...,\ho_{\Gamma}(\gamma_N)h^{-1}_{S_{1,N}})
\eq 

Then the actions $\alpha^{\overleftarrow{R}}_1$ of $\bar G_{S,\Gamma}$ on $C(\Ab_\Gamma)$,... and $\alpha^{\overleftarrow{R}}_N$ of $\bar G_{\breve S,\Gamma}$ on $C(\Ab_\Gamma)$ are automorphic and point-norm continuous actions.
\end{lem}

\begin{lem}
Let all paths intersect in their target vertices contained in the set $V_\Gamma$ with a surface $S$ such that all paths are ingoing and lie below the surface $S$. 

Then there is an action such that
\beq\label{eq actionsurfgraph8} &(\alpha^{\overrightarrow{R}}_1(\rho^1_{S,\Gamma})f_\Gamma)(\ho_{\Gamma}(\gamma_1),...,\ho_{\Gamma}(\gamma_N))\\
&\quad:= f_\Gamma(\ho_{\Gamma}(\gamma_1)\rho^{-1}_S(\gamma_1),...,\ho_{\Gamma}(\gamma_N)\rho^{-1}_S(\gamma_N))\\
&\quad= f_\Gamma(\ho_{\Gamma}(\gamma_1)g_S,...,\ho_{\Gamma}(\gamma_N)g_S)
\eq 
holds.

Then the action $\alpha^{\overrightarrow{R}}_1$ of $\bar G_{S,\Gamma}$ on $C(\Ab_\Gamma)$ is an automorphic and point-norm continuous action.
\end{lem} 

The actions in the last paragraphs are constructed such that there is a always decomposition of left and right structures. For example, if graphs are considered such that all paths have the same intersection behavior w.r.t a fixed surface set $\breve S$, then for elements of a finite orientation preserved graph system $\PD_\Gamma^{\op}$ an action is defined. On the other hand for every graph $\Gamma$ in $\PD_\Gamma^{\op}$ there always exists a graph $\Gamma^{-1}$, which referes to the set $\{\gamma_1^{-1},...,\gamma_N^{-1}\}$, which is obviously not an element of $\PD_\Gamma^{\op}$. But this graph of reversed path orientations forms a second finite orientation preserved graph system $\PD_{\Gamma^{-1}}^{\op}$. Moreover there is an action of the corresponding flux group $\bar G_{\breve S,\Gamma^{-1}}$ on $C(\Ab_{\Gamma})$, where the configuration space is constructed from the finite graph groupoid $\PD_\Gamma$. Recall that, $\ho_\Gamma(\gamma^{-1})=\ho_\Gamma(\gamma)^{-1}$ yields for an arbitrary $\gamma\in\PD_\Gamma\Sigma$. Hence it is easy to verify that, 
\beq &(\alpha_{\overleftarrow{L}}^1(\rho_{S,\Gamma}^1)f_\Gamma)(\ho_{\Gamma}(\gamma_1),...,\ho_{\Gamma}(\gamma_N))\\
&\quad=f_\Gamma(\rho_S^L(\gamma_1)\ho_{\Gamma}(\gamma_1),...,\rho_S^L(\gamma_N)\ho_{\Gamma}(\gamma_N))\\
&(\alpha_{\overleftarrow{L}}^1(\rho_{S,\Gamma}^1)f_\Gamma)(\ho_{\Gamma}(\gamma_1)^{-1},...,\ho_{\Gamma}(\gamma_N)^{-1})\\
&\quad=f_\Gamma(\rho_S^L(\gamma_1)\ho_{\Gamma}(\gamma_1)^{-1},...,\rho_S^L(\gamma_N)\ho_{\Gamma}(\gamma_N)^{-1})\\
&\quad=f_\Gamma(\rho_S^L(\gamma_1)\ho_{\Gamma}(\gamma_1^{-1}),...,\rho_S^L(\gamma_N)\ho_{\Gamma}(\gamma_N^{-1}))\\
&(\alpha^{\overleftarrow{R}}_1(\rho_{S,\Gamma^{-1}}^1))f_\Gamma)(\ho_{\Gamma}(\gamma_1^{-1}),...,\ho_{\Gamma}(\gamma_N^{-1}))\\
&\quad=f_\Gamma(\ho_{\Gamma}(\gamma_1^{-1})\rho_S^R(\gamma_1^{-1})^{-1},...,\ho_{\Gamma}(\gamma_N^{-1})\rho_S^R(\gamma_N^{-1})^{-1})\\
&\quad=f_\Gamma(\ho_{\Gamma}(\gamma_1)^{-1}\rho_S^R(\gamma_1^{-1})^{-1},...,\ho_{\Gamma}(\gamma_N)^{-1}\rho_S^R(\gamma_N^{-1})^{-1})
\eq is fulfilled, whenever $\rho_S^L$ and $\rho_S^R$ denote the maps in $\Gop_{\breve S,\Gamma}$.

\begin{defi}\label{eq modularop} Define the map $I: C(\Ab_\Gamma)\rightarrow C(\Ab_\Gamma)$
\beqs I:f_\Gamma \mapsto \breve f_\Gamma, \text{ where } \breve f_\Gamma(\ho_{\Gamma}(\gamma_1),...,\ho_{\Gamma}(\gamma_N):= f_{\Gamma}(\ho_{\Gamma}(\gamma_1)^{-1},...,\ho_{\Gamma}(\gamma_N)^{-1})
\eqs such that $I^2=\idf$ where $\idf$ is the identical automorphism on $C(\Ab_\Gamma)$. 
\end{defi}
Then deduce 
\beq&I(\alpha_{\overleftarrow{L}}^1(\rho_{S,\Gamma}^1)f_\Gamma)(\ho_{\Gamma}(\gamma_1)^{-1},...,\ho_{\Gamma}(\gamma_N)^{-1})\\
&\quad=(If_\Gamma(\rho^L_S(\gamma_1)\ho_{\Gamma}(\gamma_1)^{-1},...,\rho^L_S(\gamma_N)\ho_{\Gamma}(\gamma_N)^{-1})\\
&\quad= f_\Gamma((\rho^L_S(\gamma_1)\ho_{\Gamma}(\gamma_1)^{-1})^{-1},...,(\rho^L_S(\gamma_N)\ho_{\Gamma}(\gamma_N)^{-1})^{-1})\\
&\quad=( f_\Gamma(\ho_{\Gamma}(\gamma_1^{-1})\rho^L_S(\gamma_1)^{-1},...,\ho_{\Gamma}(\gamma_N^{-1})\rho^L_S(\gamma_N)^{-1})\\
&=(\alpha^{\overleftarrow{R}}_1(\rho_{S,\Gamma^{-1}}^1)f_\Gamma)(\ho_{\Gamma}(\gamma_1^{-1}),...,\ho_{\Gamma}(\gamma_N^{-1}))\\
\eq if $\rho^L_S(\gamma_i)=\rho^R_S(\gamma_i^{-1})$ for all $i=1,...,N$, $\rho_S^L$ and $\rho_S^R$ are maps in $\Gop_{\breve S,\Gamma}$. Consequently the actions satisfy
\beq\label{eq defJ2} I \alpha_{\overleftarrow{L}}^1(\rho^1_{S,\Gamma}) If_{\Gamma}
&=\alpha^{\overleftarrow{R}}_1(\rho_{S,\Gamma^{-1}}^1)I f_\Gamma\\
I \alpha_{\overleftarrow{L}}^1(\rho^1_{S,\Gamma}) f_{\Gamma^{-1}}&=\alpha^{\overleftarrow{R}}_1(\rho_{S,\Gamma^{-1}}^1) f_{\Gamma^{-1}}  
\eq 
Notice that, if $\rho^L_S(\gamma_i)^{-1}=\rho^R_{S^{-1}}(\gamma_i^{-1})$ for all $i=1,...,N$, $\rho_S^L$ is satisfied and $\rho_S^R$ are maps in $\Gop_{\breve S,\Gamma}$ then 
\beq
I \alpha_{\overleftarrow{L}}^1(J(\rho^1_{S,\Gamma})) f_{\Gamma^{-1}}
&=\alpha^{\overleftarrow{R}}_1(\rho_{S^{-1},\Gamma^{-1}}^1) f_{\Gamma^{-1}}  
\eq yields, whenever $J$ is the map $J:\bar G_{\breve S,\Gamma}\longrightarrow \bar G_{\breve S,\Gamma}$, $J:\rho_{S,\Gamma}(\Gamma)\mapsto\rho_{S^{-1},\Gamma}(\Gamma)$ where $\breve S:=\{S,S^{-1}\}$.

\begin{lem}Let $S_1,...,S_N$ form a set $\breve S$, where the set $\breve S$ has the simple surface intersection property for a graph $\Gamma$, let $ \PD_\Gamma^{\op}$ be a finite orientation preserved graph system for $\Gamma$. 

Then there is an action of $\bar G_{\breve S,\Gamma}$ on $C(\Ab_\Gamma)$ given by
\beq\label{eq actionsurfgraph17} \alpha_{\overleftarrow{R}}(\rho_{\breve S,\Gamma}(\Gp))f_\Gamma(\ho_{\Gamma}(\Gp)) 
&:= f_\Gamma(\ho_{\Gamma}(\gamma_1)\rho_{S_1}(\gamma_1)^{-1},...,\ho_{\Gamma}(\gamma_M)\rho_{S_M}(\gamma_M)^{-1})\\
&= f_\Gamma(\ho_{\Gamma}(\gamma_1)g_{S_1},...,\ho_{\Gamma}(\gamma_M )g_{S_M})\\
\eq whenever $\Gp:=\{\gamma_1,...,\gamma_M\}$ is an element of $\PD_\Gamma^{\op}$, for all $\rho_{\breve S,\Gamma}\in G_{\breve S,\Gamma}$. This action is point-norm continuous and automorphic. 
\end{lem}

Notice that, in this case of a suitbale surface set $\breve S$ instead of $\bar G_{\breve S,\Gamma}$ one can use $\times_{i=1}^N\bar G_{S_i,\Gamma}$ equivalently.

\paragraph*{Left and right actions of the flux group\\}\hspace*{10pt}

Left and right actions are defined on the same level for some configurations of the surfaces and paths. Therefore recall the maps contained in the set $\Gop_{\breve S,\Gamma}$ with left multiplication operation, which decomposes into $\rho_S^L\times\rho_S^R$.

\begin{lem}
Let all paths in $\Gamma$ intersect in vertices of the set $V_\Gamma$ with a surface $S$ such that $\gamma_1,...,\gamma_{1,N-1}$ are ingoing paths and lie above the surface $S$, whereas $\gamma_N$ is a outgoing path lying below w.r.t. the surface orientation of $S$. 

Then the action defined by
\beq\label{eq actionsurfgraph9} (\alpha^{\overleftarrow{R},1}_{\overleftarrow{L}}(\rho^1_{S,\Gamma})f_\Gamma)(\ho_{\Gamma}(\gamma_1),...,\ho_{\Gamma}(\gamma_N))
&:= f_\Gamma(\ho_{\Gamma}(\gamma_1)\rho_S^R(\gamma_1)^{-1},...,\rho_{S}^L(\gamma_N)\ho_{\Gamma}(\gamma_N))\\
&= f_\Gamma(\ho_{\Gamma}(\gamma_1)g^{-1}_S,...,g_S\ho_{\Gamma}(\gamma_N))
\eq for $\rho^1_{S,\Gamma}=(\rho_S^R(\gamma_1),...,\rho_S^L(\gamma_N))=(g_S,...,g_S)$, $\rho^1_{S,\Gamma}\in \bar G_{S,\Gamma}$ and $f_\Gamma\in C(\Ab_\Gamma)$.

Then the action $\alpha^{\overleftarrow{R},1}_{\overleftarrow{L}}$ of $\bar G_{S,\Gamma}$ on $C(\Ab_\Gamma)$ is automorphic and point-norm continuous.
\end{lem}

In \cite[Section 6.1]{KaminskiPHD} similar actions build from left and right actions have been presented. For example, consider a set $\{S_i\}_{1\leq i\leq N}=:\breve S$ of surfaces.
Furthermore let each path $\gamma_i$ in $\Gamma$ intersects in one vertex of the set $V_\Gamma$ with a surface $S_i$ and there are no other intersections with any other surface. In particular for $i=1,...,N-1$ each the path $\gamma_i$ is an ingoing path and lies above the surface $S_i$, whereas $\gamma_N$ is a outgoing path lying below w.r.t. the surface orientation of $S_N$.

Then the action of $\bar G_{\breve S,\Gamma}$ on $C(\Ab_\Gamma)$ is
\beq\label{eq actionsurfgraph10} 
(\alpha^{\overleftarrow{R},N}_{\overleftarrow{L}}(\rho^N_{\breve S,\Gamma})f_\Gamma)(\ho_{\Gamma}(\gamma_1),...,\ho_{\Gamma}(\gamma_N))
&:= f_\Gamma(\ho_{\Gamma}(\gamma_1)\rho^R_{S_{1}}(\gamma_1)^{-1},...,\rho^L_{S_{N}}(\gamma_N)\ho_{\Gamma}(\gamma_N))\\
&= f_\Gamma(\ho_{\Gamma}(\gamma_1)g^{-1}_{S_{1}},...,g_{S_{N}}\ho_{\Gamma}(\gamma_N))
\eq holds for $\rho^N_{\breve S,\Gamma}=(\rho^R_{S_1}(\gamma_1),...,\rho^L_{S_N}(\gamma_N))=(g_{S_1},...,g_{S_N})$, $\rho^N_{\breve S,\Gamma}\in \breve G_{\breve S,\Gamma}$ and $f_\Gamma\in C(\Ab_\Gamma)$. Then the action $\alpha^{\overleftarrow{R},N}_{\overleftarrow{L}}$ is an automorphic and point-norm continuous action of $\bar G_{S,\Gamma}$ on $C(\Ab_\Gamma)$.

Clearly in the same way the actions $\alpha^{\overleftarrow{R},N}_{\overleftarrow{L}}$,
$\alpha^{\overleftarrow{R},N}_{\overrightarrow{L}}$,
$\alpha_{\overleftarrow{L}}^{\overrightarrow{R},N}$ and so on are automorphic and point-norm continuous actions of the flux group associated to a suitable surface set on $C(\Ab_\Gamma)$.

Furthermore the actions defined above can be easily generalised to finite orientation preserved graph systems.

\begin{problem}\label{problem I}
Observe that there is a problem if the following actions\footnote{Notice that, there is a difference between a left action of a group on a space and a left action of a group on a $C^*$-algebra.} on $\Ab_\Gamma$ for a surface $S$ are defined. 
Assume that, the set $\bar G_{\breve S,\Gamma}$ is equipped with a left and right multiplication on the same time. This could be the case if a surface $S$ is considered such that $S$ intersects the paths $\gamma_1,...,\gamma_{N-1}$ in the source vertices, hence the paths are outgoing and lie below.  Furthermore $S$ intersects the path $\gamma_N$ such that the path lies outgoing and above. Then a left action, which is defined by the set $\bar G_{\breve S,\Gamma}$ with a left multiplication for the paths $\gamma_1,...,\gamma_{N-1}$ and a right multiplication for the path $\gamma_N$, is not automorphic. This is verified by the following computation:
\beq\label{eq problem} &(\alpha_{\overrightarrow{L},\overleftarrow{L}}^1(\rho_{S,\Gamma}^1\tilde\rho_{S,\Gamma}^1)f_\Gamma)(\ho_{\Gamma}(\gamma_1),...,\ho_{\Gamma}(\gamma_N))\\ 
&= f_\Gamma(\rho_{S}(\gamma_1)\tilde\rho_{S}(\gamma_1)\ho_{\Gamma}(\gamma_1),...,\tilde\rho_{S}(\gamma_N)\rho_{S}(\gamma_N)\ho_{\Gamma}(\gamma_N))\\
&= f_\Gamma(g_S\tilde g_S\ho_{\Gamma}(\gamma_1),...,\tilde g_S g_S\ho_{\Gamma}(\gamma_N))\\
&\neq (\alpha_{\overrightarrow{L},\overleftarrow{L}}^1(\rho^1_{S,\Gamma})(\alpha_{\overrightarrow{L},\overleftarrow{L}}^1(\tilde\rho^1_{S,\Gamma}f_\Gamma(\ho_{\Gamma}))))(\gamma_1),...,\ho_{\Gamma}(\gamma_N))
\eq This is the reason why all paths that lie outgoing w.r.t. a surface $S$ are defined by left actions of $\bar G_{\breve S,\Gamma}$ with a left multiplication on $\Ab_\Gamma$. Hence this problem is absent by definition of the actions, which are stated before. A left action of $\bar G_{\breve S,\Gamma}$ with a left multiplication on $\Ab_\Gamma$ is called the \textbf{left action} of $\bar G_{\breve S,\Gamma}$ on $\Ab_\Gamma$. Whereas a left action of $\bar G_{\breve S,\Gamma}$ with a right multiplication on $\Ab_\Gamma$ is called the \textbf{inverse left action} of $\bar G_{\breve S,\Gamma}$ on $\Ab_\Gamma$.  A right action of $\bar G_{\breve S,\Gamma}$ with a left multiplication on $\Ab_\Gamma$ is called the \textbf{right action} of $\bar G_{\breve S,\Gamma}$ on $\Ab_\Gamma$. Whereas a right action of $\bar G_{\breve S,\Gamma}$ with a right multiplication on $\Ab_\Gamma$ is called the \textbf{inverse right action} of $\bar G_{\breve S,\Gamma}$ on $\Ab_\Gamma$. Hence the action defined in equation \eqref{eq problem} corresponds to a left action of $\bar G_{\breve S,\Gamma}$ on $\Ab_\Gamma$ for the paths $\gamma_1,...,\gamma_{1,N-1}$ and a left inverse action of $\bar G_{\breve S,\Gamma}$ on $\Ab_\Gamma$ for the path $\gamma_N$.
 
Recognize that, there is no homeomorphism $H$ on $\Ab_\Gamma$ correponding to a group action of $\bar G_{S,\Gamma}$ with a left and right multiplication on $\Ab_\Gamma$ defined in \eqref{eq problem}, i.o.w. 
\beq H(g_S\tilde g_S)(\ho_{\Gamma}(\gamma_1),...\ho_{\Gamma}(\gamma_N)))&=(g_S\tilde g_S\ho_{\Gamma}(\gamma_1),...,\tilde g_S^{-1}g_S^{-1}\ho_{\Gamma}(\gamma_N))\\
&\neq H(g_S)(H(\tilde g_S)(\ho_{\Gamma}(\gamma_1),...\ho_{\Gamma}(\gamma_N)))\\[3pt]
H(g_S)(H(\tilde g_S)(\ho_{\Gamma}(\gamma_1),...\ho_{\Gamma}(\gamma_N)))&=H(g_S)((\tilde g_S\ho_{\Gamma}(\gamma_1),...,\tilde g_S^{-1}\ho_{\Gamma}(\gamma_N)))\\
&=(g_S\tilde g_S\ho_{\Gamma}(\gamma_1),...,g_S^{-1}\tilde g_S^{-1}\ho_{\Gamma}(\gamma_N))
\eq

In the same way the problem occurs if there is 
a inverse right action of $\bar G_{\breve S,\Gamma}$ on $\Ab_\Gamma$ for the paths $\gamma_1,...,\gamma_{1,N-1}$ intersecting a surface $S$ such that they are ingoing and lie above, and a right action of $\bar G_{\breve S,\Gamma}$ on $\Ab_\Gamma$ for the path $\gamma_N$ intersecting $S$ such that the path is ingoing and lies below, is studied. Then it is true that  
\beq\label{eq problem2} &(\alpha_{\overrightarrow{R},\overleftarrow{R}}^1(\rho_{S,\Gamma}^1\tilde\rho_{S,\Gamma}^1)f_\Gamma)(\ho_{\Gamma}(\gamma_1),...,\ho_{\Gamma}(\gamma_N))\\ 
&= f_\Gamma(\ho_{\Gamma}(\gamma_1)\rho_{S}(\gamma_1)^{-1}\tilde\rho_{S}(\gamma_1)^{-1},...,\ho_{\Gamma}(\gamma_N)(\rho_{S}(\gamma_N)\tilde\rho_{S}(\gamma_N))^{-1})\\
&= f_\Gamma(\ho_{\Gamma}(\gamma_1)g_S^{-1}\tilde g_S^{-1},...,\ho_{\Gamma}(\gamma_N)\tilde g_S^{-1} g_S^{-1})\\
&\neq (\alpha_{\overrightarrow{R},\overleftarrow{R}}^1(\rho^1_{S,\Gamma})(\alpha_{\overrightarrow{R},\overleftarrow{R}}^1(\tilde\rho^1_{S,\Gamma}f_\Gamma(\ho_{\Gamma}))))(\gamma_1),...,\ho_{\Gamma}(\gamma_N))
\eq yields.

There is also a problem if there is 
a left action of $\bar G_{\breve S,\Gamma}$ on $\Ab_\Gamma$ for the paths $\gamma_1,...,\gamma_{1,N-1}$, which intersect a surface $S$ such that they are outgoing and lie above, and a inverse right group action on $\Ab_\Gamma$ for the path $\gamma_N$ lying ingoing and below, is considered. Then derive  
\beq\label{eq problem3} &(\hat\alpha_{\overleftarrow{L}}^{\overrightarrow{R},1}(\rho_{S,\Gamma}^1\tilde\rho_{S,\Gamma}^1)f_\Gamma)(\ho_{\Gamma}(\gamma_1),...,\ho_{\Gamma}(\gamma_N))\\ 
&= f_\Gamma(\rho_{S}(\gamma_1)\tilde\rho_{S}(\gamma_1)\ho_{\Gamma}(\gamma_1),...,\ho_{\Gamma}(\gamma_N)(\tilde\rho_{S}(\gamma_N)\rho_{S}(\gamma_N))^{-1})\\
&= f_\Gamma(g_S\tilde g_S\ho_{\Gamma}(\gamma_1),...,\ho_{\Gamma}(\gamma_N)\tilde g_S^{-1}  g_S^{-1})\\
&\neq \Big(\hat\alpha_{\overleftarrow{L}}^{\overrightarrow{R},1}(\rho^1_{S,\Gamma})\big(\hat\alpha_{\overleftarrow{L}}^{\overrightarrow{R},1}(\tilde\rho^1_{S,\Gamma}f_\Gamma)\big)\Big)(\ho_{\Gamma}(\gamma_1),...,\ho_{\Gamma}(\gamma_N))
\eq 

Finally there is also a problem if there is 
a left inverse action of $\bar G_{\breve S,\Gamma}$ on $\Ab_\Gamma$ for the paths $\gamma_1,...,\gamma_{1,N-1}$ lying outgoing and below, and a right action of $\bar G_{\breve S,\Gamma}$ on $\Ab_\Gamma$ for the path $\gamma_N$ lying ingoing and above is considered. In this case,    
\beq\label{eq problem4} &(\hat\alpha_{\overrightarrow{L}}^{\overleftarrow{R},1}(\rho_{S,\Gamma}^1\tilde\rho_{S,\Gamma}^1)f_\Gamma)(\ho_{\Gamma}(\gamma_1),...,\ho_{\Gamma}(\gamma_N))\\ 
&= f_\Gamma(\tilde\rho_{S}(\gamma_1)\rho_{S}(\gamma_1)\ho_{\Gamma}(\gamma_1),...,\ho_{\Gamma}(\gamma_N)(\rho_{S}(\gamma_N) \tilde\rho_{S}(\gamma_N))^{-1})\\
&= f_\Gamma(\tilde g_Sg_S\ho_{\Gamma}(\gamma_1),...,\ho_{\Gamma}(\gamma_N)\tilde g_S^{-1} g_S^{-1})\\
&\neq \Big(\hat\alpha_{\overrightarrow{L}}^{\overleftarrow{R},1}(\rho^1_{S,\Gamma})\big(\hat\alpha_{\overrightarrow{L}}^{\overleftarrow{R},1}(\tilde\rho^1_{S,\Gamma})(f_\Gamma)\big)\Big)(\ho_{\Gamma}(\gamma_1),...,\ho_{\Gamma}(\gamma_N))
\eq holds. All these problems are excluded by the definition of the actions. Since for example there is an action defined by $\alpha_{\overrightarrow{L}}^{\overleftarrow{R},1}$ instead of the action $\hat\alpha_{\overrightarrow{L}}^{\overleftarrow{R},1}$ given in \eqref{eq problem4}. Respectively there is an action defined by $\alpha_{\overleftarrow{L}}^{\overrightarrow{R},1}$ instead of the action $\hat\alpha_{\overleftarrow{L}}^{\overrightarrow{R},1}$ given in \eqref{eq problem3}.
\end{problem}

\paragraph*{Non-standard identification of the configuration space\\}\hspace*{10pt}\label{problem of admissibility}

At the beginning of this considerations one assumes that the subgraphs of the finite graph system are identified naturally. If instead the \hyperlink{non-standard identification}{non-standard identification of the configuration space} is used, the following observation is made.

If the set $\Hom(\PD_\Gamma,G^{\vert \Gamma\vert})$ of holonomy maps for a finite graph system $\PD_\Gamma$ is considered and the configuration space is identified with $G^N$ by the non-standard way. Then there exists a situation such that $\ho_\Gamma(\gamma\circ\gp)=\ho_\Gamma(\gamma)\ho_\Gamma(\gp)$ for arbitrary $(\gamma,\gp)\in\PD_\Gamma\Sigma^{(2)}$ holds. Despite the property there is an action of $\bar G_{S,\Gamma}$ on $C(\Ab_\Gamma)$ derivable.

First observe that, for all paths of $\Gamma$ intersecting $S$ in their source vertices and all paths lie below one concludes that
\beq &(\alpha_{\overleftarrow{L}}(\rho_{S,\Gamma}(\Gamma))\circ \alpha^{\overleftarrow{R}}(\rho_{S,\Gamma^{-1}}(\Gamma^{-1}))f_\Gamma)(\ho_{\Gamma}(\gamma_1)\ho_{\Gamma}(\gamma_1^{-1}),...,\ho_{\Gamma}(\gamma_N)\ho_{\Gamma}(\gamma_N^{-1}))\\
&=f_\Gamma(\rho_S^L(\gamma_1)\ho_{\Gamma}(\gamma_1)\ho_{\Gamma}(\gamma_1^{-1})\rho_S^R(\gamma_1^{-1}),...,\rho_S^L(\gamma_N)\ho_{\Gamma}(\gamma_N)\ho_{\Gamma}(\gamma_N^{-1})\rho_S^R(\gamma_N)^{-1})\\
&=f_\Gamma(g_S\ho_{\Gamma}(\gamma_1)\ho_{\Gamma}(\gamma_1^{-1})g_S^{-1},...,g_S\ho_{\Gamma}(\gamma_N)\ho_{\Gamma}(\gamma_N^{-1})g_S^{-1})\\
&=f_\Gamma(e_G,...,e_G)
\eq holds whenever $\rho_{S,\Gamma}\in G_{S,\Gamma}$, $\rho_{S,\Gamma^{-1}}\in G_{S,\Gamma^{-1}}$ and $f_\Gamma\in C(\Ab_\Gamma)$, where $\Gamma^{-1}$ referes to the set $\{\gamma_1^{-1},...,\gamma_N^{-1}\}$ and if it is assumed that $\rho^L_S(\gamma_i)=\rho^R_S(\gamma_i^{-1})=g_S$ for all $\gamma_i\in\Gamma$ is fulfilled. In this case there is an action $\alpha_{\overleftarrow{L}}^{\overleftarrow{R}}$ of $\bar G_{S,\vert\Gamma\vert}$ on $C(\Ab_\Gamma)$ defined by
\beq &(\alpha_{\overleftarrow{L}}^{\overleftarrow{R}}(\rho_S(\Gamma))f_\Gamma)(\ho_{\Gamma}(\gamma_1)\ho_{\Gamma}(\gamma_1^{-1}),...,\ho_{\Gamma}(\gamma_N)\ho_{\Gamma}(\gamma_N^{-1}))\\&
:=(\alpha_{\overleftarrow{L}}(\rho_{S,\Gamma}(\Gamma))\circ \alpha^{\overleftarrow{R}}(\rho_{S,\Gamma^{-1}}(\Gamma^{-1}))f_\Gamma)(\ho_{\Gamma}(\gamma_1)\ho_{\Gamma}(\gamma_1^{-1}),...,\ho_{\Gamma}(\gamma_N)\ho_{\Gamma}(\gamma_N^{-1}))
\eq

Moreover for a graph $\Gamma$ consisting of two paths $\gamma$ and $\gp$ such that $\gamma$ and $S$ intersect in the target vertex $t(\gamma)$ of $\gamma$ and the path lies above and resepectively $\gp$ and $S$ intersect in $s(\gp)$ such that $\gp$ lies above and is outgoing. Then $(\gamma,\gp)\in\PD_\Gamma^{(2)}\Sigma$.
Set $\Gamma=\{\gamma,\gp\}$, $\Gp:=\{\gamma\circ\gp\}$ and assume $S\cap\{\gamma,\gp\}=\{t(\gamma)\}$. Then there is an action of $\bar G_{S,\Gamma}$ on $C(\Ab_\Gamma)$ given by
\beq (\alpha_{\overrightarrow{L}}^{\overleftarrow{R}}(\rho_{S,\Gamma}(\Gamma))f_\Gamma)(\ho_{\Gamma}(\Gamma))
&=(\alpha_{\overrightarrow{L}}^{\overleftarrow{R}}(\rho_{S,\Gamma}(\Gamma))f_\Gamma)(\ho_{\Gamma}(\gamma),\ho_{\Gamma}(\gp))\\
&=f_\Gamma(\ho_{\Gamma}(\gamma)\rho_S^R(\gamma)^{-1},\rho_S^L(\gp)\ho_{\Gamma}(\gp))\\
&=f_\Gamma(\ho_{\Gamma}(\gamma)g_S^{-1},g_S^{-1}\ho_{\Gamma}(\gp))
\eq where it is assumed that $\rho^L_S(\gamma)=\rho^R_S(\gp)^{-1}=g_S$ for all $(\gamma,\gp)\in\PD_\Gamma\Sigma^{(2)}$. For the definition of an action of $\bar G_{S,\Gamma}$ on $C(\Ab_\Gamma)$, whenever the configuration space is identified with $G^N$ in the non-standard way, it is necessary to define a the following map.

Let $D_S: C(\Ab_\Gamma)\longrightarrow C(\Ab_\Gamma)$ be a map such that 
\beqs (D_Sf_\Gamma)(\ho_\Gamma(\gamma)g,h\ho_\Gamma(\gp))= f_\Gamma(\ho_\Gamma(\gamma)gh\ho_\Gamma(\gp))
\eqs whenever $g,h\in G$ and if $(\gamma,\gp)\in\PD_\Gamma\Sigma^{(2)}$ and $S\cap\{\gamma,\gp\}=\{t(\gamma)\}$.
There is an ambiguity in the definition of the inverse of $D_S$, since it is possible that
\beqs &(D_S^{-1}f_\Gamma)(\ho_\Gamma(\gamma)gh\ho_\Gamma(\gp))= f_\Gamma(\ho_\Gamma(\gamma),gh\ho_\Gamma(\gp))
\text{ for }g\in G\\
&\qquad\text{ or }\\
&(D_S^{-1}f_\Gamma)(\ho_\Gamma(\gamma)gh\ho_\Gamma(\gp))= f_\Gamma(\ho_\Gamma(\gamma)g,h\ho_\Gamma(\gp))\\
&\qquad\text{ or }\\
&(D_S^{-1}f_\Gamma)(\ho_\Gamma(\gamma)gh\ho_\Gamma(\gp))= f_\Gamma(\ho_\Gamma(\gamma)gh,\ho_\Gamma(\gp))
\eqs holds whenever $g,h\in G$. Recall that, $\ho_\Gamma(\gamma)=h$ is an element of $G$. Let $g\in\ZD(G)$ and $\ho_\Gamma(\gamma)\neq\ho_\Gamma(\gp)$. Then there always exists a groupoid morphism $\tilde\ho_\Gamma$ such that for a pair $(\gamma,\gp)\in\PD_\Gamma\Sigma^{(2)}$ it is true that, $\tilde\ho_\Gamma(\gamma\circ\gp)=\ho_\Gamma(\gamma)\ho_\Gamma(\gp)$ and $\tilde\ho_\Gamma(\gamma)=\tilde\ho_\Gamma(\gp)$ yields. Then 
\beqs (D_S^{-1}f_\Gamma)(\ho_\Gamma(\gamma)g^2\ho_\Gamma(\gp))
&= f_\Gamma(\tilde\ho_\Gamma(\gamma)g,g\tilde\ho_\Gamma(\gp))
\eqs is well-defined. The reason is given by the following property. Since for example for $\ho_\Gamma(\gamma)=h=\ho_\Gamma(\gp)$ and $g\in\ZD(G)$, it is true that $hg(hg)=hggh=(hg)^2$. 

More generally consider all $g\in G$ such that there exists a $k\in G$ and $\ho_\Gamma(\gamma)gg\ho_\Gamma(\gp)=k^2$ for all $\ho_\Gamma(\gamma),\ho_\Gamma(\gp)\in \Ab_\Gamma$ for a fixed pair $(\gamma,\gp)\in\PD_\Gamma\Sigma^{(2)}$.

\begin{defi}Let $S$ be a surface and $\Gamma$ be a graph, which consists of two paths $\gamma$ and $\gp$ such that $\gamma$ and $S$ intersect in the target vertex $t(\gamma)$ of $\gamma$ and lies above and resepectively the path $\gp$ and $S$ intersect in $s(\gp)$ such that the path $\gp$ lies above and outgoing.

Then define the action of $\bar \ZD_{S,\Gamma}$ on $C(\Ab_\Gamma)$ by
\beqs (D_S\alpha_{\overrightarrow{L}}^{\overleftarrow{R},1}(\rho_{S,\Gamma}(\Gamma))D_S^{-1}f_\Gamma)(\ho_{\Gamma}(\Gp))
&:=f_\Gamma(\ho_{\Gamma}(\gamma)\rho_S^R(\gamma)^{-1}\rho_S^L(\gp)\ho_{\Gamma}(\gp))
\eqs whenever $\rho_{S,\Gamma}(\Gamma)\in \bar \ZD_{S,\Gamma}$ and it is assumed that 
\beq\label{eq diffactions1}&\rho^L_S(\gamma)=\rho^R_S(\gp)^{-1}=g_S^{-1}\text{ for }\rho^L_S,\rho^R_S\in\Zop_{S,\Gamma},\\
&\ho_{\Gamma}(\gamma)=\ho_{\Gamma}(\gp)\text{ for }\ho_\Gamma\in\Hom(\PD_\Gamma,G^{\vert\Gamma\vert}) 
\eq holds for the pair $(\gamma,\gp)\in\PD_\Gamma\Sigma^{(2)}$ such that $\Gp:=\{\gamma\circ\gp\}$.
The action is redefined by
\beqs (D_S\alpha_{\overrightarrow{L}}^{\overleftarrow{R},1}(\rho_{S,\Gamma}(\Gamma))D_S^{-1}f_\Gamma)(\ho_{\Gamma}(\Gp))&:=
(\alpha_{\overrightarrow{L}}^{\overleftarrow{R},1}(\rho_{S,\Gamma}(\Gp))f_\Gamma)(\ho_{\Gamma}(\Gp))
\eqs 
\end{defi}Notice that, the action is computed in the following way
\beqs (D_S\alpha_{\overrightarrow{L}}^{\overleftarrow{R},1}(\rho_{S,\Gamma}(\Gamma))D_S^{-1}f_\Gamma)(\ho_{\Gamma}(\Gp))&=
(D_S\alpha_{\overrightarrow{L}}^{\overleftarrow{R},1}(\rho_{S,\Gamma}(\Gamma))D_S^{-1}f_\Gamma)(\ho_{\Gamma}(\gamma)\ho_{\Gamma}(\gp))\\
&=(D_S\alpha_{\overrightarrow{L}}^{\overleftarrow{R},1}(\rho_{S,\Gamma}(\Gamma))f_\Gamma)(\ho_{\Gamma}(\gamma),\ho_{\Gamma}(\gp))\\
&=(D_Sf_\Gamma)(\ho_{\Gamma}(\gamma)\rho_S^R(\gamma)^{-1},\rho_S^L(\gp)\ho_{\Gamma}(\gp))\\
&=(D_Sf_\Gamma)(\ho_{\Gamma}(\gamma)g_S^{-1},g_S^{-1}\ho_{\Gamma}(\gp))\\
&=f_\Gamma(\ho_{\Gamma}(\gamma)g_S^{-1}g_S^{-1}\ho_{\Gamma}(\gp))
\eqs whenever $\rho_{S,\Gamma}(\Gamma)\in \bar \ZD_{S,\Gamma}$ and it is assumed that $\rho^L_S(\gamma)=\rho^R_S(\gp)^{-1}=g_S^{-1}$ holds for the pair $(\gamma,\gp)\in\PD_\Gamma\Sigma^{(2)}$ such that $\Gp:=\{\gamma\circ\gp\}$.

Finally derive that, for this action it is true that
\beq (\alpha_{\overrightarrow{L}}^{\overleftarrow{R},1}(\rho_{S,\Gamma}(\Gp)\hat\rho_{S,\Gamma}(\Gp))f_\Gamma)(\ho_{\Gamma}(\Gp))
&= (\alpha_{\overrightarrow{L}}^{\overleftarrow{R},1}(\rho_{S,\Gamma}(\Gp))\alpha_{\overrightarrow{L}}^{\overleftarrow{R},1}(\hat\rho_{S,\Gamma}(\Gp))f_\Gamma)(\ho_{\Gamma}(\Gp))
\eq holds whenever $\hat\rho_{S,\Gamma},\rho_{S,\Gamma}\in\ZD_{S,\Gamma}$. Hence the action $\alpha_{\overrightarrow{L}}^{\overleftarrow{R},1}$ is automorphic. One can show that the action $\alpha_{\overrightarrow{L}}^{\overleftarrow{R},1}$ is point-norm continuous.

If the graph is changed only slightly, then recognize the following. 
\begin{defi}
Let $S$ be a surface and $\Gp$ be a graph, which is given by the composed path of a path $\gamma$ and $\gp$ such that $\gamma$ and $S$ intersects in the target vertex $t(\gamma)$ of $\gamma$ and lies below and resepectively $\gp$ and $S$ intersect in $s(\gp)$ such that $\gp$ lies above and outgoing. Then $(\gamma,\gp)\in\PD_\Gamma^{(2)}\Sigma$. 

The action of $\bar G_{S,\Gamma}$ on $C(\Ab_\Gamma)$ in this case is presented by
\beqs(D_S\alpha_{\overrightarrow{L}}^{\overleftarrow{R},1}(\rho_{S,\Gamma}(\Gamma))D_S^{-1}f_\Gamma)(\ho_{\Gamma}(\Gp))
&:=f_\Gamma(\ho_{\Gamma}(\gamma)\rho_S^R(\gamma)^{-1}\rho_S^L(\gp)\ho_{\Gamma}(\gp))\\
&=f_\Gamma(\ho_{\Gamma}(\gamma)g_Sg_S^{-1}\ho_{\Gamma}(\gp))\\
&=f_\Gamma(\ho_{\Gamma}(\gamma\circ\gp))
\eqs whenever $f_\Gamma\in C(\Ab_\Gamma)$, $\rho_{S,\Gamma}\in G_{S,\Gamma}$ and it is assumed that 
\beq\label{eq diffactions2}&\rho^L_S(\gamma)=\rho^R_S(\gp)=g_S^{-1}\text{ for }\rho^L_S,\rho^R_S\in\Zop_{S,\Gamma},\\
&\ho_{\Gamma}(\gamma)=\ho_{\Gamma}(\gp)\text{ for }\ho_\Gamma\in\Hom(\PD_\Gamma,G^{\vert\Gamma\vert}) 
\eq for the pair $(\gamma,\gp)\in\PD_\Gamma\Sigma^{(2)}$ such that $\Gp:=\{\gamma\circ\gp\}$. Set
\beqs(D_S\alpha_{\overrightarrow{L}}^{\overleftarrow{R},1}(\rho_{S,\Gamma}(\Gamma))D_S^{-1}f_\Gamma)(\ho_{\Gamma}(\Gp)):=(\alpha_{\overrightarrow{L}}^{\overleftarrow{R},1}(\rho_{S,\Gamma}(\Gp))f_\Gamma)(\ho_{\Gamma}(\Gp))
\eqs yields.
\end{defi}
Since in this case
\beqs (D_S^{-1}f_\Gamma)(\ho_\Gamma(\gamma)gg^{-1}\ho_\Gamma(\gp))
&= f_\Gamma(\ho_\Gamma(\gamma),\ho_\Gamma(\gp))
\text{ for }g\in G
\eqs is well-defined. Clearly this action can be restricted to those maps that are elements of $\ZD_{S,\Gamma}$. The action $\alpha_{\overrightarrow{L}}^{\overleftarrow{R},1}$ is automorphic and point-norm continuous.

Notice that, there are two actions $\alpha_{\overrightarrow{L}}^{\overleftarrow{R},1}$ and $\alpha_{\overrightarrow{L}}^{\overleftarrow{R},1}$ of $\bar\ZD_{S,\Gamma}$ are restricted by the requirement \eqref{eq diffactions1} or \eqref{eq diffactions2}. The actions depend on the orientation of both paths of the pair $(\gamma,\gp)\in\PD_\Gamma\Sigma^{(2)}$ w.r.t. the surface orientation of $S$. Clearly the actions can be generalised to graphs containg a set of pairs $\{(\gamma,\gp)\}$ of paths in $\PD_\Gamma\Sigma^{(2)}$.

In \cite[Section 6.1]{KaminskiPHD} actions of fluxes, which are constructed from admissible maps associated to surfaces, have been introduced. These actions on the analytic holonomy $C^*$-algebra restricted to a finite graph system are more complicated.
A summary of the results is given by the following. If the non-standard identification of the configuration space is used, then the holonomy maps are defined on arbitrary elements of the finite graph groupoid. In particular, a graph $\Gp$, which is a subgraph of $\Gamma$ and contains a path $\gamma\circ\gp$. Then consider a surface $S$ that intersects only the paths $\gamma$ and $\gp$ of the graph $\Gamma$ in the vertex $t(\gamma)$. Then different actions of an element $\varrho_{S,\Gamma}(\Gp)$ on a function in $C(\Ab_\Gamma)$ are considered. There are two different kinds of actions. One referes to a translation of the center of the flux group $\bar\ZD_{S,\Gamma}$. The other is related to a translation of the fluxes related to admissible maps. Moreover each actions depend on the orientation of the paths $\gamma$ and $\gp$ w.r.t. the surface $S$.

\paragraph*{The set of actions of $\bar G_{\breve S,\Gamma}$ on $C(\Ab_\Gamma)$\\}\hspace*{10pt}

Assume that, the configuration space $\Ab_\Gamma$ of generalised connections is identified in the natural way with $G^{\vert\Gamma\vert}$. Certainly there are a lot of different actions on $C(\Ab_\Gamma)$ corresponding to different surfaces and graph configurations, which are build from left and right actions of the group $\bar G_{\breve S,\Gamma}$ on the $C^*$-algebra. In general there is an exceptional set of all well-defined point-norm continuous automorphic actions. 
\begin{defi}\label{def actGA}
Denote the set of all point-norm continuous automorphic actions of $\bar G_{\breve S,\Gamma}$ on $C(\Ab_\Gamma)$ by $\Act(\bar G_{\breve S,\Gamma},C(\Ab_\Gamma))$ for every suitable set $\breve S$ of surfaces, a graph $\Gamma$ and a finite graph system $\PD_\Gamma$.
\end{defi}

Let $\Phi_M$ be the multiplication representation of $C(\Ab_\Gamma)$ on $\HS_\Gamma$.
For all automorphic and point-norm continuous actions presented above there are unitary representations $U$ of $\bar G_{S,\Gamma}$ or $\bar G_{\breve S,\Gamma}$ on the Hilbert space $\HS_\Gamma$, which satisfy for a suitable surface $S$ or surface set $\breve S$ and graph configuration one of the following or some equivalent Weyl relations 
\beq\label{cancomrel I} &U_{L}(\rho_{S_1,\Gamma}(\Gamma),...,\rho_{S_p,\Gamma}(\Gamma))\Phi_M(f_\Gamma)U_{L}(\rho_{S_1,\Gamma}(\Gamma),...,\rho_{S_p,\Gamma}(\Gamma))^{-1}=\Phi_M(\alpha_{L}^p(\rho_{S_1,\Gamma}(\Gamma),...,\rho_{S_p,\Gamma}(\Gamma))f_\Gamma)\\ &\quad\qquad\quad\qquad\quad\qquad\quad\qquad\quad\qquad\quad\qquad\quad\qquad\quad\qquad\quad\qquad\quad\qquad\quad\qquad\quad\qquad\quad\qquad\forall p\in\N\text{, }\\
&U_{\overleftarrow{L}}^k(\rho^k_{S,\Gamma})\Phi_M(f_\Gamma)U_{\overleftarrow{L}}^k(\rho^k_{S,\Gamma})^{-1}=\Phi_M(\alpha_{\overleftarrow{L}}^k(\rho^k_{S,\Gamma})f_\Gamma) \quad\forall k\in\N\text{, }\\
&U^{\overleftarrow{R}}_k(\rho^k_{S,\Gamma})\Phi_M(f_\Gamma)U^{\overleftarrow{R}}_k(\rho^k_{S,\Gamma})^{-1}=\Phi_M(\alpha^{\overleftarrow{R}}_k(\rho^k_{S,\Gamma})f_\Gamma) \quad\forall k\in\N\text{, } \\
&U^{\overleftarrow{R}}_{p,k}(\rho^{p,k}_{S,\Gamma})\Phi_M(f_\Gamma)U^{\overleftarrow{R}}_{p,k}(\rho^{p,k}_{S,\Gamma})^{-1}=\Phi_M(\alpha^{\overleftarrow{R}}_{p,k}(\rho^{p,k}_{S,\Gamma})f_\Gamma)\quad\text{ for }p\leq k\leq N-1\text{ or } \\
&U_{\overleftarrow{L}}^k(\rho^k_{S,\Gamma})U^{\overleftarrow{R}}_k(\rho^k_{S,\Gamma})\Phi_M(f_\Gamma)U^{\overleftarrow{R}}_k(\rho^k_{S,\Gamma})^{-1}U_{\overleftarrow{L}}^k(\rho^k_{S,\Gamma})^{-1}=\Phi_M(\alpha^{\overleftarrow{R},k}_{\overleftarrow{L}}(\rho^k_{S,\Gamma})f_\Gamma) \quad\forall k\in\N
\eq and so on
for all $f_\Gamma\in C(\Ab_\Gamma)$, $\rho_{S,\Gamma}^k\in \bar G_{S,\Gamma}$, $\rho^{p,k}_{S,\Gamma}\in\bar G_{\breve S,\Gamma}$ and $\rho_{S_i,\Gamma}\in G_{\breve S,\Gamma}$ for $i=1,...,k$ where $\breve S:=\{S_i\}_{1\leq i\leq p}$ is suitable. 

Observe that, for each unitary $U$ defined above the pair $(U,\Phi_M)$ consisting of $U\in\Rep(\bar G_{\breve S,\Gamma},\KD(\HS_\Gamma))$ and $\Phi_M\in\Mor(C(\Ab_\Gamma),\LD(\HS_\Gamma))$ is a covariant pair of the dynamical $C^*$-system $(C(\Ab_\Gamma),\bar G_{\breve S,\Gamma},\alpha)$ for an action $\alpha\in\Act(\bar G_{\breve S,\Gamma},C(\Ab_\Gamma))$.

\begin{problem} There are no unitary operators, which satisfy the Weyl relations for the actions 
\beqs  \alpha^1_{\overrightarrow{L},\overleftarrow{L}}(\rho^1_{S,\Gamma}),\quad \alpha^1_{\overleftarrow{L},\overrightarrow{R}}(\rho^1_{S,\Gamma}),...
\eqs
which are presented in problem \thesubsubsection.\ref{problem I} and which are not automorphic actions of $\bar G_{S,\Gamma}$ on $C(\Ab_\Gamma)$. This is true, since there are unitary operators such that 
\beqs
&U^1_{\overrightarrow{L},\overleftarrow{L}}(\rho^1_{S,\Gamma}\tilde\rho^1_{S,\Gamma})(U^1_{\overrightarrow{L},\overleftarrow{L}}(\rho^1_{S,\Gamma}\tilde\rho^1_{S,\Gamma}))^*\\
&= U^1_{\overrightarrow{L},\overleftarrow{L}}(\rho_S(\gamma_1)\tilde\rho_S(\gamma_1),...,
\tilde\rho_S(\gamma_1)\rho_S(\gamma_1))U^1_{\overrightarrow{L},\overleftarrow{L}}(\tilde\rho_S(\gamma_1)^{-1}\rho_S(\gamma_1)^{-1},...,
\rho_S(\gamma_1)^{-1}\tilde\rho_S(\gamma_1)^{-1})\\
&=U^1_{\overrightarrow{L},\overleftarrow{L}}(\rho_S(\gamma_1)\tilde\rho_S(\gamma_1)\tilde\rho_S(\gamma_1)^{-1}\rho_S(\gamma_1)^{-1},...,
\rho_S(\gamma_1)^{-1}\tilde\rho_S(\gamma_1)^{-1}\tilde\rho_S(\gamma_1)\rho_S(\gamma_1))\\
&=\idf
\eqs holds. But $U^1_{\overrightarrow{L},\overleftarrow{L}}(\bar G_{\breve S,\Gamma})$ does not form a group. Consequently
\beqs
U^1_{\overrightarrow{L},\overleftarrow{L}}(\rho^1_{S,\Gamma})U^1_{\overrightarrow{L},\overleftarrow{L}}(\tilde\rho^1_{S,\Gamma})\Phi_M(f_\Gamma)(U^1_{\overrightarrow{L},\overleftarrow{L}}(\rho^1_{S,\Gamma}\tilde\rho^1_{S,\Gamma})U^1_{\overrightarrow{L},\overleftarrow{L}}(\rho^1_{S,\Gamma}\tilde\rho^1_{S,\Gamma}))^*\neq \Phi_M(\alpha^1_{\overrightarrow{L},\overleftarrow{L}}(\rho^1_{S,\Gamma}\tilde\rho^1_{S,\Gamma}))
\eqs holds.
\end{problem}

\begin{defi}
Let $\breve S$ be an arbitrary set of surfaces. The Hilbert space $\HS_\Gamma$ is identified with $L^2(\Ab_\Gamma,\mu_{\Gamma})$.

Each unitary $U(\rho_S(\Gamma))$ for $U\in\Rep(\bar G_{\breve S,\Gamma},\KD(\HS_\Gamma))$ and $\rho_S(\Gamma)\in\bar G_{\breve S,\Gamma}$ is called a \textbf{Weyl element}. The set of all linearly independent Weyl elements is denoted by $W(\bar G_{\breve S,\Gamma})$. The vector space of all finite complex linear combinations of Weyl elements, which are unitary operators satisfying the Weyl relations \eqref{cancomrel I} on their own, is symbolised by $\mathbf{W}(\bar G_{\breve S,\Gamma})$. 

Each unitary $U(\rho_S(\Gamma))$ for $U\in\Rep(\bar\ZD_{\breve S,\Gamma},\KD(\HS_\Gamma))$ is called the \textbf{commutative Weyl element}. The set of all linearly independent commutative Weyl elements is denoted by $W(\bar\ZD_{\breve S,\Gamma})$. The vector space $\mathbf{W}(\bar\ZD_{\breve S,\Gamma})$ is given by the set of all finite complex linear combinations of commutative Weyl elements such that the linear combinations are unitary operators and satisfy the Weyl relations \eqref{cancomrel I}. 
\end{defi}

In general finite linear combinations of unitary elements (elements such that $UU^*=U^*U=\idf$) form a unital $^*$-algebra. Otherwise, for each fixed suitable surface set and $\bar G_{\breve S,\Gamma}$, the set $U(\bar G_{\breve S,\Gamma})$ form a group with the usual left multiplication operation. The group $U(\bar G_{\breve S,\Gamma})$ is a subgroup of the group $U(\HS_\Gamma)$ of unitaries on a Hilbert space $\HS_\Gamma$. 
Notice that, the Weyl algebra associated to surfaces and a graph, which is introduced in the next subsection, is generated by the constant function $\idf_\Gamma$, the elements of the analytic holonomy $C^*$-algebra $C(\Ab_\Gamma)$ and the unitaries associated to surfaces. For example, an element of the Weyl algebra is of the form
\beqs \sum_{l=1}^L\idf_\Gamma U(\rho^l_{S,\Gamma}(\Gamma)) + \sum_{k=1}^K\sum_{i=1}^Mf^k_{\Gamma}(\ho_\Gamma(\Gamma))U(\rho^i_{S,\Gamma}(\Gamma))  + \sum_{k=1}^K\sum_{i=1}^MU(\rho^i_{S,\Gamma}(\Gamma)) f^k_{\Gamma}(\ho_\Gamma(\Gamma))U(\rho^i_{S,\Gamma}(\Gamma))^*
+\sum_{p=1}^Pf^p_{\Gamma}(\ho_\Gamma(\Gamma))
\eqs whenever $f^k_{\Gamma},f^p_{\Gamma}\in C(\Ab_\Gamma)$ and $\rho^l_{S,\Gamma},\rho^i_{S,\Gamma}\in G_{\breve S,\Gamma}$. Notice that, for a compact group $G$ the analytic holonomy $C^*$-algebra is unital.

In this context a reformulation of the Weyl $C^*$-algebra is given as follows. 

\begin{lem}
Let $\HS_\Gamma$ be the Hilbert space $L^2(\Ab_\Gamma,\mu_\Gamma)$ with norm $\|.\|_2$.

With the involution $^*$ and the natural product of unitaries the vector space $\mathbf{W}(\bar G_{\breve S,\Gamma})$ is a unital $^*$-algebra, where $\WD(\bar G_{\breve S,\Gamma})$ stands for the $^*$-algebra of Weyl elements. 

The $^*$-algebra $\WD(\bar G_{\breve S,\Gamma})$ of Weyl elements completed w.r.t. the  strong operator norm is a $C^*$-algebra. Denote this algebra by $\mathsf{W}(\bar G_{\breve S,\Gamma})$.
\end{lem}

Notice that, $\mathsf{W}(\bar G_{\breve S,\Gamma})$ is a $C^*$-subalgebra of the $C^*$-algebra $\LD(\HS_\Gamma)$ of bounded operators on the Hilbert space $\HS_\Gamma$, which is equal to $L^2(\Ab_\Gamma,\mu_\Gamma)$.

Recall that, $G$ is assumed to be a compact sgroup and $\Ab_\Gamma$ is identified in the natural way with $G^{\vert\Gamma\vert}$. Then remember the action $\beta_{\overrightarrow{L}}^{\overrightarrow{R},1}$, which defines a $C^*$-dynamical system $(C(\Ab_\Gamma),\bar G_{\breve S,\Gamma},\beta_{\overrightarrow{L}}^{\overrightarrow{R},1})$ for a fixed graph and a suitable surface set $\breve S$.

\begin{prop}\label{prop Ginvstate}Let $\breve S^\prime$ and $\breve S$ be two suitable surface sets, let $\Gamma$ be a graph and let $\PD_\Gamma$ be the finite graph system associated to $\Gamma$. 

Furthermore let $(U_{\overrightarrow{L}}^{\overrightarrow{R},1},\Phi_M)$ be a covariant pair of a dynamical $C^*$-system\\ $(C(\Ab_\Gamma),\bar G_{\breve S^\prime,\Gamma},\beta_{\overrightarrow{L}}^{\overrightarrow{R},1})$ associated to the surface set $\breve S^\prime$ and the orientation-preserved finite graph system $\PD_\Gamma^{\op}$ associated to the graph $\Gamma$. 

Denote a general covariant pair by $(U,\Phi_M)$ of a dynamical $C^*$-system $(C(\Ab_\Gamma),\bar G_{\breve S,\Gamma},\alpha)$ for an action $\alpha\in\Act(\bar G_{\breve S,\Gamma},C(\Ab_\Gamma))$ for the surface set $\breve S$ and the finite graph system $\PD_\Gamma$. The set $\surf$ of all suitable surface sets for $\Gamma$ contains all surface sets such that there exists an action in $\Act(\bar G_{\breve S,\Gamma},C(\Ab_\Gamma))$.\\[5pt]
 
Then there exists a GNS-triple $(\HS_\Gamma,\Phi_M,\Omega_\Gamma)$ where $\Omega_\Gamma$ is the cyclic vector for $\Phi_M$ on $\HS_\Gamma$. Moreover the associated GNS-state $\omega_M^\Gamma$ on $C(\Ab_\Gamma)$ is $\bar G_{\breve S,\Gamma}$- and $\bar G_{\breve S^\prime,\Gamma}$-invariant, i.e. 
\beqs\omega^\Gamma_M(\beta_{\overrightarrow{L}}^{\overrightarrow{R},1}(\rho_{S^\prime,\Gamma}^1)(f_\Gamma))
&=\omega^\Gamma_M(f_\Gamma)
:=\langle \Omega_\Gamma,\Phi_M(f_\Gamma)\Omega_\Gamma\rangle_\Gamma\\
&=\omega^\Gamma_M(\alpha(\rho_{S,\Gamma}(\Gamma))(f_\Gamma))
\eqs for $\alpha,\beta_{\overrightarrow{L}}^{\overrightarrow{R},1}\in\Act(\bar G_{\breve S,\Gamma},C(\Ab_\Gamma)) $, $\rho_{S,\Gamma}^1,\rho_{S,\Gamma}(\Gamma)\in \bar G_{\breve S,\Gamma}$ and $f_\Gamma\in C(\Ab_\Gamma)$. 

Moreover the set $\MD:=\Phi_M(C(\Ab_\Gamma))\cup \{U(\bar G_{\breve S,\Gamma}): U\in \Rep(\bar G_{\breve S,\Gamma},\Alg_\Gamma)\}$ is irreducible on $\HS_\Gamma$.
\end{prop}

Remark that, the state $\omega_M^\Gamma$ on $C(\Ab_\Gamma)$ is $\bar G_{\breve S,\Gamma}$-invariant for many different suitable surface sets. The surface sets are required to intersect the graph $\Gamma$ only in vertices and hence such that there exists an action in $\Act(\bar G_{\breve S,\Gamma},C(\Ab_\Gamma))$.

\begin{proofs}The GNS-triple is constructed on the Hilbert space $\HS_\Gamma$, which is given by $L^2(\Ab_\Gamma,\mu_\Gamma)$. The unitaries $U$ of the set $\Rep(\bar G_{\breve S,\Gamma },\KD(\HS_\Gamma))$ for every suitable surface set $\breve S$ in $\surf$ and the representation $\Phi_M$ are given on the same Hilbert space $\HS_\Gamma$. Consequently a state $\omega_M^\Gamma$ exists.

The crucial property is the irreducibility of the set $\Phi_M(C(\Ab_\Gamma))\cup \{U(\bar G_{\breve S,\Gamma})):U\in\Rep(\bar G_{\breve S,\Gamma},\LD(\HS_\Gamma))\}$. Notice that, $\MD^\prime =\Phi_M(C(\Ab_\Gamma))^\prime\cap\{U(\bar G_{\breve S,\Gamma})):U\in\Rep(\bar G_{\breve S,\Gamma},\LD(\HS_\Gamma))\}^\prime$. First notice that $\Phi_M(C(\Ab_\Gamma))\subset \Phi_M(C(\Ab_\Gamma))^\prime$. Then it is true that $U_{\overleftarrow{L}}^k(\bar G_{\breve S,\Gamma})\subset U^{\overleftarrow{R}}_k(\bar G_{\breve S,\Gamma})^\prime$ and  $U^{\overleftarrow{R}}_k(\bar G_{\breve S,\Gamma})\subset U_{\overleftarrow{L}}^k(\bar G_{\breve S,\Gamma})^\prime$ whenever $U_{\overleftarrow{L}}^k,U^{\overleftarrow{R}}_k\in\Rep(\bar G_{\breve S,\Gamma},\LD(\HS_\Gamma))$ for $1\leq k\leq N-1$. Clearly $U^{\overleftarrow{R}}_k(\bar G_{\breve S,\Gamma})\nsubseteq  U^{\overleftarrow{R}}_k(\bar G_{\breve S,\Gamma})^\prime$ is satisfied. Then observe that, $U_{\overleftarrow{L}}^k(\bar G_{\breve S,\Gamma})^\prime\cap U^{\overleftarrow{R}}_k(\bar G_{\breve S,\Gamma})^\prime=\{\lambda\cdot\idf:\lambda\in\R\}$ yields. Hence $\MD^\prime$ is given by $\{\lambda\cdot\idf:\lambda\in\R\}$ .
\end{proofs}

The proposition \ref{prop Ginvstate} is reformulated for the non-standard identification of the configuration space $\Ab_\Gamma$ if $\bar G_{\breve S,\Gamma}$ is replaced by $\bar \ZD_{\breve S,\Gamma}$.

\begin{cor}\label{cor uniqueness}Let $\PD_\Gamma$ be a finite graph system associated to the graph $\Gamma$. Let $\Ab_\Gamma$ be the set of generalised connections, which is identified in the natural way with $G^N$.

There is a unique measure on $\Ab_\Gamma$ given by the Haar measure $\mu_{\Gamma}$ on the product $G^N$ of a compact  group $G$. 

If $\bar G_{\breve S,\Gamma}$ is identified with $G^N$, then there is a unique state $\omega_M^\Gamma$ on $C(\Ab_\Gamma)$, which is $G^N$-invariant and which is given by
\beqs \omega_M^\Gamma(f_\Gamma)&=\int_{G^N}\dif\mu_N(\ho_\Gamma(\Gamma))f_\Gamma(\ho_\Gamma(\Gamma))=\int_{G^N}\dif\mu_N(\ho_\Gamma(\Gamma))f_\Gamma(R(\textbf{g})(\ho_\Gamma(\Gamma)))\\
&=\int_{G^N}\dif\mu_N(\ho_\Gamma(\Gamma))f_\Gamma(L(\textbf{g})(\ho_\Gamma(\Gamma)))
\eqs for all $f_\Gamma\in C(\Ab_\Gamma)$ and $\textbf{g}\in G^N$.
\end{cor}
The same result is obtained for the non-standard identification of the finite graph system and the configuration space.

\begin{rem}\label{rem statewithL1}Let $\PD_\Gamma$ be a finite graph system associated to the graph $\Gamma$s. Let $\Ab_\Gamma$ be the space of generalised connections identified in the natural way with $G^{\vert\Gamma\vert}$ and $G$ be compact  group.

Let $f$ be a function in the convolution $^*$-algebra $\CD(\Ab_\Gamma)$, which is given by the algebra $C_c(\Ab_\Gamma)$ of compactly supported functions on $\Ab_\Gamma$ equipped with the convolution as the multiplication operation. For example for two continuous compactly supported function $f,k$ on $\Ab_\Gamma$ the convolution product is illustrated by 
\beqs (f\ast k)(\ho_\Gamma(\gamma),\ho_\Gamma(\gp))=\int_{\Ab_\Gamma}\dif\mu_\Gamma(\hat\ho_\Gamma(\Gp))
f(\ho_\Gamma(\gamma)\hat\ho_\Gamma(\gamma)^{-1},\ho_\Gamma(\gp)\hat\ho_\Gamma(\gp)^{-1}) k(\hat\ho_\Gamma(\gamma),\hat\ho_\Gamma(\gp))
\eqs for $\Gp:=\{\gamma,\gp\}$. The involution is given by
\beqs f(\ho_\Gamma(\Gp))^*=f(\ho_\Gamma(\gamma),\ho_\Gamma(\gp))^*:=\overline{f(\ho_\Gamma(\gamma)^{-1},\ho_\Gamma(\gp)^{-1})}
\eqs
Notice that, $\ho_\Gamma(\Gamma)$ and $\ho_\Gamma(\Gp)$ are elements of $G^N$ and hence the convolution $^*$-algebra $\CD(\Ab_\Gamma)$ is identified with $\CD(G^N)$.

Then there exists a state on $C(\Ab_\Gamma)$ given by
\beqs \omega_{M,f}^\Gamma(f_\Gamma)=\int_{G^N}\dif\mu_\Gamma(\ho_\Gamma(\Gp))f_\Gamma(\ho_\Gamma(\Gp))f(\ho_\Gamma(\Gp))
\eqs where $f_\Gamma\in C(\Ab_\Gamma)$ and $f\in \CD(\Ab_\Gamma)$. 
Derive from
\beqs \omega_{M,f}^\Gamma(\alpha(\textbf{k})(f_\Gamma))&=\int_{G^N}\dif\mu_\Gamma(\ho_\Gamma(\Gp))f_\Gamma(R(\textbf{k}^{-1})(\ho_\Gamma(\Gp)))f(\ho_\Gamma(\Gp))\\
&=\int_{G^N}\dif\mu_\Gamma(\ho_\Gamma(\Gp))f_\Gamma(\ho_\Gamma(\Gp))f(R(\textbf{k})(\ho_\Gamma(\Gp)))\\
&\overset{!}{=}\int_{G^N}\dif\mu_\Gamma(\ho_\Gamma(\Gp))f_\Gamma(\ho_\Gamma(\Gp))f(\ho_\Gamma(\Gp))
\eqs that $ \omega_{M,f}^\Gamma$ is $G^N$-invariant iff 
\beq f(R(\textbf{k})(\ho_\Gamma(\Gp)))=f(\ho_\Gamma(\Gp))
\eq for any $\textbf{k}\in G^N$. If $\CD(\bar G_{\breve S,\Gamma})$ is identified with $\CD(G^N)$, then for $f_{\breve S}\in\CD(G^N)$ and  $f_{\breve S}(\textbf{g}\textbf{k}^{-1})=f_{\breve S}(\textbf{g})$ for all $\textbf{k}\in G^N$ there is another state on $C(G^N)$ defined by 
\beqs \omega_{M,f_{\breve S}}^\Gamma(f_\Gamma)=\int_{G^N}\dif\mu_N(\textbf{g})f_\Gamma(\textbf{g})f_{\breve S}(\textbf{g})
\eqs whenever $f_\Gamma\in C(G^N)$. But both states will be not finite path-and graph-diffeomorphism invariant. This is shown in problem \ref{subsec dynsysfluxgroup2}.\ref{probl statewithfdiff}. Clearly the states $\omega_{M,f}^\Gamma$ on $C(\Ab_\Gamma)$ are not invariant under all actions $\alpha\in\Act(\bar G_{\breve S,\Gamma},C(\Ab_\Gamma))$. 

Notice the function $f$ is also an element of $\CD(\bar G_{\breve S,\Gamma})$, but in this case the state $\omega_{M,f}^\Gamma$ is obviously not finite graph-diffeomorphism invariant.
\end{rem}

\begin{proofo}\textbf{of corollary \ref{cor uniqueness}}:\\
After the natural identication of $\Ab_\Gamma$ with $G^N$ there is a unique Haar measure $\mu_\Gamma$ on $G^{\vert\Gamma\vert}$. The dual of $C(\Ab_\Gamma)$ is given by the Banach space of all bounded complex Baire measures on $\Ab_\Gamma$. Furthermore there exists an extension of each Baire measure to a regular Borel measure on $\Ab_\Gamma$. The linear space of regular Borel measures equipped with the convolution operation form a Banach $^*$-algebra $\textbf{M}(G^N)$. The algebra decomposes into norm closed subspaces consisting of measures absoluely continuous with respect to the Haar measure of $G^N$, continuous measures singular with respect to the Haar measure and discrete measures (refer to \cite[chap.: 19]{HewittRoss}). The Banach $^*$-algebra generated by Dirac point measures is excluded in the following considerations, a closer look on this structure will be presented in \cite{Kaminski2} and \cite[Section 7.1]{KaminskiPHD}. The subspace $\textbf{M}_{s}(G^N)$ of all continuous measures singular with respect to the Haar measure is not a subalgebra of $\textbf{M}(G^N)$, in general. Consequently the space $\textbf{M}_{s}(G^N)$ is not considered. Notice the space $\textbf{M}_{a}(G^N)$ consisting of measures absoluely continuous with respect to the Haar measure of $G^N$ is identified with $L^1(G^N,\mu_N)$.

The norm-closed subspace of all regular Borel measures on $\Ab_\Gamma$, which are absolutely continuous to the uniquely defined Haar measure $\mu_\Gamma$ is given by $L^1(\Ab_\Gamma,\mu_\Gamma)$.
Hence a state on $C(\Ab_\Gamma)$ is given by
\beqs \omega_{M,f}^\Gamma(f_\Gamma)=\int_{G^N}\dif\mu_\Gamma(\ho_\Gamma(\Gp))f_\Gamma(\ho_\Gamma(\Gp))f(\ho_\Gamma(\Gp))
\eqs
for all $f\in L^1(\Ab_\Gamma,\mu_\Gamma)$ and $f_\Gamma\in C(\Ab_\Gamma)$. The Banach $^*$-algebra $L^1(\Ab_\Gamma,\mu_\Gamma)$ is the completition of $\CD(\Ab_\Gamma)$ w.r.t. the $\|.\|_1$-norm.  

But there is only one state on $C(\Ab_\Gamma)$ given by
\beqs \omega_{M}^\Gamma(f_\Gamma)=\int_{G^N}\dif\mu_\Gamma(\ho_\Gamma(\Gp))f_\Gamma(\ho_\Gamma(\Gp))
\eqs which is invariant under all actions $\alpha\in\Act(\bar G_{\breve S,\Gamma},C(\Ab_\Gamma))$ for any arbitrary set $\breve S$ of suitable surfaces. The invariance of the state under different actions of $\bar G_{\breve S,\Gamma}$ is derived from the fact that the product Haar measure of the product of the  compact group $G$ is left and right invariant. 
\end{proofo}

\begin{rem}\label{rem statewithL2}Let $\Ab_\Gamma$ be the space of generalised connections identified in the natural way with $G^N$. Let $H$ be a closed subgroup of the compact group $G$.

Let $f$ be a function in $\CD(\Ab_\Gamma)$ such that $f(\ho_\Gamma(\Gp)\textbf{k}^{-1})=f(\ho_\Gamma(\Gp))$ for any $\textbf{k}\in H^N$ and consider the state
\beqs \omega_{H,f}^\Gamma(f_\Gamma)=\int_{G^N}\dif\mu_N(\ho_\Gamma(\Gp))f_\Gamma(\ho_\Gamma(\Gp))f(\ho_\Gamma(\Gp))
\eqs where $f_\Gamma\in C(\Ab_\Gamma)$. 
Then this state $\omega_{H,f}^\Gamma$ is $H^N$-invariant. But this state will be not path- or graph-diffeomorphism invariant in general. 
\end{rem}

\subsection{Dynamical systems of actions of the group of bisections on two $C^*$-algebras}\label{subsec dynsysfluxgroup2}
\subsubsection*{Actions of the group of bisections of the analytic holonomy algebra for finite graph systems}

In this subsection the new concept for graph changing operations is introduced. On the level of finite path groupoids the bisections of the path groupoid $\PD_\Gamma\Sigma$ over $V_\Gamma$ implement path-diffeomorphisms. A path-diffeomorphism is a pair of maps such that one bijective mapping maps vertices to vertices and the second bijective mapping maps non-trivial paths to non-trivial paths. A fixed set of independent paths is a graph, on the other hand, each path of a graph is an element of a path groupoid. Hence there is a concept of bisections of finite graph systems. This concepts are presented in subsection \ref{subsubsec bisections}. 
The action of a bisection of finite path groupoids changes paths by adding or deleting segments of paths. In particular an action maps a non-trivial path to a trivial one or conversely. Hence an action of a bisection of a finite graph system transforms graphs to graphs. Actions of bisections of a finite path groupoid are either right-, left- or inner-translations in the finite path groupoid. Hence actions of bisections of a finite graph system are defined by right-, left- or inner-translations in the finite graph system. For a detail analysis refer to definition \ref{defi bisecongraphgroupioid} of a right translation in a finite graph system. 
Recall the lemma \ref{lemma bisecform}, which states that the set $\mathfrak{B}(\PD_{\Gamma})$ of bisections form a group w.r.t. a multiplication operation $\ast_2$ and an inverse $^{-1}$. Furthermore each bisection $\sigma$ define a right translation $R_{\sigma}$ on a finite graph system $\PD_ {\Gamma}$.

In the following investigations the non-standard identification of the configuration space is used, but it is also possible to derive results for the natural identification.
\begin{prop}\label{prop groupbisecdynsys}Let $\PD_\Gamma$ be a finite graph system associated to a graph $\Gamma$ and let $C(\Ab_\Gamma)$ be the analytic holonomy $C^*$-algebra associated to $\Gamma$.
  
There is an action $\zeta$ of the group $\mathfrak{B}(\PD_{\Gamma})$ of bisections  equipped with $\ast_2$ and an inverse $^{-1}$ on $C(\Ab_\Gamma)$ defined by
\beqs (\zeta_{\sigma} f_\Gamma)(\ho_\Gamma(\Gp)):=f_{\Gamma}((\ho_\Gamma\circ R_{\sigma})(\Gp))
=f_{\Gamma}(\ho_{\Gamma}(\Gamma_\sigma))
\eqs whenever $f_\Gamma\in C(\Ab_\Gamma)$, $\sigma\in\mathfrak{B}(\PD_{\Gamma})$ and for the subgraphs $\Gp,\Gamma_\sigma$ of $\Gamma$. The inverse action is given by
\beqs (\zeta_{\sigma}^{-1} f_\Gamma)(\ho_\Gamma(\Gp)):=f_{\Gamma}((\ho_\Gamma\circ R_{\sigma^{-1}})(\Gp))
=f_{\Gamma}(\ho_{\Gamma}(\Gamma_{\sigma^{-1}}))
\eqs whenever $f_\Gamma\in C(\Ab_\Gamma)$, $\sigma\in\mathfrak{B}(\PD_{\Gamma})$ and for the subgraphs $\Gp,\Gamma_{\sigma^{-1}}$ of $\Gamma$.

This action $\zeta$ is point-norm continuous and automorphic.

Hence $(\mathfrak{B}(\PD_{\Gamma}),C(\Ab_\Gamma),\zeta_\sigma)$ is a $C^*$-dynamical system.
\end{prop}
\begin{proofs} First it is proved that $\zeta$ defines an automorphic action. Moreover for simplicity reasons the arguments are verified for one particular graph. The proof generalises to arbitraray graphs. Assume that $\tilde\Gamma$ is a graph, $V_{\tilde\Gamma}=\{v_1^\prime,v_2^\prime,v_1,v_2,w_1,w_2\}$  and $\Gp:=\{\gamma_1\}$ is a subgraph of $\tilde\Gamma$. Let $\sigma$ and $\sigma^\prime$ be two bisections of $\PD_{\tilde\Gamma}$ such that $\sigma^\prime(z)=\idf_{z}$ if $z\neq v_1$ and $z\neq w_1$ for $z\in V_{\tilde\Gamma}$. Moreover let $v_i^\prime=s(\gamma_i)$, $v_i=t(\gamma_i)$ and $w_i=t(\sigma^\prime(v_i))$ ($v_i=s(\sigma^\prime(v_i))$) for $i=1,2$. Then derive 
\beqs &((\zeta_{\sigma}\circ\zeta_{\sigma^\prime}) f_{\tilde\Gamma})(\ho_{\tilde\Gamma}(\Gamma))
=f_{\tilde\Gamma}( \ho_{\tilde\Gamma}((\gamma_1\circ\sigma^\prime(v_1))\circ\sigma(w_1)))
\\
& =f_{\tilde\Gamma}\Big(\ho_{\tilde\Gamma} \big(\gamma_1\circ\big(\sigma^\prime(v_1)\circ\sigma(w_1) \big))\Big)
=f_{\tilde\Gamma}(\ho_{\tilde\Gamma}\left(\gamma_1\circ(\sigma\ast_2\sigma^\prime)(v_1)\right))\\
&=(\zeta_{\sigma\ast_2\sigma^\prime}f_{\tilde\Gamma})(\ho_{\tilde\Gamma}(\Gamma))
\eqs whenever $f_{\tilde\Gamma}\in C(\Ab_{\tilde\Gamma})$. This shows condition \ref{def automorphic1} of definition \ref{def automorphic}. The conditions condition \ref{def automorphic2} and condition \ref{def automorphic3} are obvious due to the properties of the $C^*$-algebra $C(\Ab_{\tilde\Gamma})$. Clearly this generalises to arbitrary subgraphs $\Gamma$ of an arbitrary graph $\tilde\Gamma$.

The point-norm continuity follows if the map $\zeta:\mathfrak{B}(\PD_\Gamma)\ni \sigma\mapsto \zeta_\sigma(f_\Gamma)$ is norm-continuous. Notice that, $\mathfrak{B}(\PD_\Gamma)$ is finite, since the number of paths $\vert\Gamma\vert$ is finite. Assume that, $\Gamma=\{\gamma\}$ is a subgraph of $\tilde\Gamma$. Let $f_{\tilde\Gamma}\in C(\Ab_{\tilde\Gamma})$ then
\beqs &\lim_{\sigma\rightarrow \id}\|\zeta_{\sigma}(f_{\tilde\Gamma})- f_{\tilde\Gamma}\|_{\tilde\Gamma}
=\lim_{\sigma\rightarrow \id}\|f_{\tilde\Gamma}(\ho_{\tilde\Gamma}(\gamma)\cdot\ho_{\tilde\Gamma}(\sigma(v)))- f_{\tilde\Gamma}(\ho_{\tilde\Gamma}(\gamma))\|_{\tilde\Gamma}\\
&=0
\eqs holds for a graph $\Gamma=\{\gamma\}$ and $\Gamma_{\sigma}=\{\gamma\circ \sigma(v)\}$ being subgraphs of $\tilde\Gamma$, and if for the bisection $\sigma\in\mathfrak{B}(\PD_\Gamma)$ the equality $\sigma(w)= \idf_w$ holds for any $w\in V_{\tilde\Gamma}\setminus\{v\}$ where $v=t(\gamma)$. Moreover the map $\id(w)=\idf_w$ for all $w\in V_{\tilde\Gamma}$, in particular for $v=t(\gamma)$, is the identity bisection.
Deduce the properties for arbitrary subgraphs $\Gamma$ of an arbitrary graph $\tilde\Gamma$. 
\end{proofs}

Clearly there are also $C^*$-dynamical systems $(\mathfrak{B}(\PD_{\Gamma}),C(\Ab_\Gamma),\zeta^L_\sigma)$ and $(\mathfrak{B}(\PD_{\Gamma}),C(\Ab_\Gamma),\zeta^I_\sigma)$ for the actions 
\beqs (\zeta^L_{\sigma} f_\Gamma)(\ho_\Gamma(\Gp)):=f_{\Gamma}((\ho_\Gamma\circ L_{\sigma})(\Gp))
\eqs and
\beqs (\zeta^I_{\sigma} f_\Gamma)(\ho_\Gamma(\Gp)):=f_{\Gamma}((\ho_\Gamma\circ I_{\sigma})(\Gp))
\eqs whenever $L_\sigma$ and $I_\sigma$ are translations in $\PD_\Gamma$. Refer to subsection \ref{subsubsec bisections} for the definition of these objects. 

\begin{prop}\label{prop CDynbisection}
The $C^*$-dynamical systems $(\mathfrak{B}(\PD_{\Gamma}),C(\Ab_\Gamma),\zeta_\sigma)$ and $(\mathfrak{B}(\PD_{\Gamma}),C(\Ab_\Gamma),\zeta^L_\sigma)$\\ (or $(\mathfrak{B}(\PD_{\Gamma}),C(\Ab_\Gamma),\zeta_\sigma)$ and $(\mathfrak{B}(\PD_{\Gamma}),C(\Ab_\Gamma),\zeta_\sigma^I)$, or $(\mathfrak{B}(\PD_{\Gamma}),C(\Ab_\Gamma),\zeta^L_\sigma)$ and $(\mathfrak{B}(\PD_{\Gamma}),C(\Ab_\Gamma),\zeta_\sigma^I)$) \\are exterior equivalent \cite[Def.: 2.66]{Williams07}. 
\end{prop}
\begin{proofs}The strictly continuous unitary-valued function $u:\mathfrak{B}(\PD_\Gamma)\rightarrow C(\Ab_\Gamma)$, which satisfies
\beqs &\zeta_{\sigma}(f_\Gamma)=u_\sigma \zeta_\sigma^L(f_\Gamma)u^*_\sigma\\
&u_{\sigma\ast\sigma^\prime}=u_ \sigma\zeta_\sigma^L(u_{\sigma^\prime})
\eqs for all $\sigma,\sigma^\prime\in\mathfrak{B}(\PD_\Gamma)$ is constructed as follows.
Recall the left- and right-translation $L_\sigma $ and $R_\sigma$ in $\PD_\Gamma$, which is presented in subsection \ref{subsubsec bisections}. Then derive that 
\beqs \zeta_{\sigma}(f_\Gamma)\cdot k_\Gamma&=
f_\Gamma\circ R_{\sigma}\cdot k_\Gamma= f_\Gamma\circ L_{\sigma^{-1}}\circ L_\sigma\cdot k_\Gamma\circ R_{\sigma^{-1}}
&=u_\sigma \zeta_\sigma^L(f_\Gamma)u^*_\sigma
\eqs holds whenever $f_\Gamma,k_\Gamma\in C(\Ab_\Gamma)$,  $u_ \sigma f_\Gamma:=f_\Gamma\circ L_{\sigma^{-1}}$, $u_ \sigma ^*f_\Gamma:=f_\Gamma\circ R_{\sigma^{-1}}$ and $\cdot$ is the multiplication in $C(\Ab_\Gamma)$. Notice that, 
$u_ \sigma u_ \sigma ^*= L_{\sigma^{-1}}\circ R_{\sigma^{-1}}=\id$.
\end{proofs}
For every graph-diffeomorphism in $\Diff(\PD_\Gamma)$ there exists a bisection $\sigma\in\mathfrak{B}(\PD_\Gamma)$ and either a left-, or right- or inner-translation such that $\Phi_\Gamma=X_\sigma$, where $X$ is equivalent to $L$, or $R$ or $I$,  and $\varphi_\Gamma=t\circ\sigma$.  The set $\Diff(\PD_\Gamma)$ does not form a group in general. If one of the actions is fixed, then loosely speaking, the group of graph-diffeomorphisms is the set of graph-diffeomorphism in $\Diff(\PD_\Gamma)$, which are defined by a bisection $\sigma\in\mathfrak{B}(\PD_\Gamma)$ and a left- translation such that $\Phi_\Gamma=L_\sigma$ and $\varphi_\Gamma=t\circ\sigma$. Clearly the left-translation can be replaced by right- or inner-translation.

Note that, the $C^*$-dynamical systems $(\mathfrak{B}(\PD_{\Gamma}),C(\Ab_\Gamma),\zeta_\sigma)$ and $(\bar G_{\breve S,\Gamma},C(\Ab_\Gamma),\alpha)$ are not exterior or equivariantly isomorphic \cite[Def.: 2.64]{Williams07}. 

\subsubsection*{Groups of surface or surface-orientation-preserving bisections for finite graph systems}

Up to now only actions of bisections of the holonomies are considered. Therefore in the next investigations the behavior of an action of bisections on the flux operators is analysed. Clearly this action is required to preserve the structure of the group-valued quantum flux operators.

\begin{prop}
Let $\breve S$ be a set of surfaces and $\Gamma$ be a graph such that $\breve S\cap \Gamma\subset V_\Gamma$. Let $\varphi$ be a diffeomorphism in $\Sigma$, which maps each surface $S\in \breve S$ to a surface $S_\sigma\in\breve S$.

A \textbf{surface-preserving bisection $\sigma$ of a finite path groupoid} and a set of surfaces $\breve S$ is defined by a bisection $\sigma:V_\Gamma\longrightarrow\PD_\Gamma\Sigma$ in $\PD_\Gamma$ such that 
\begin{itemize}                                                                                                                                                                                                                                 \item the map $\varphi_{\Gamma}:V_{\Gamma}\longrightarrow V_{\Gamma}$, which is given by $\varphi_{\Gamma}=t\circ\sigma$, is bijective and $(t\circ\sigma)(v)=v$ whenever $v\in S\cap V_\Gamma$ and each $S\in\breve S$ yields,
 \item for each path $\gamma\in \PD_\Gamma\Sigma$ that intersects a surface $S$ and, such that $\gamma\cap S=\{s(\gamma)\}$ holds, the non-trivial transformed path is presented by $\gamma\circ\sigma(t(\gamma))$ for $S\in\breve S$,\\ 
\item for each path $\gamma\in \PD_\Gamma\Sigma$ that intersects a surface $S$ and, such that  $\gamma\cap S=\{t(\gamma)\}$ is satisfied, then $\sigma(t(\gamma))=\idf_{t(\gamma)}$ and $\gamma\circ\sigma(t(\gamma))=\gamma$ yields for a surface $S\in\breve S$,\\ 
\item the bisection $\sigma:V\rightarrow \PD_\Gamma\Sigma$, where $V=V_\Gamma\setminus V_{\breve S}$ and $V_{\breve S}=V_\Gamma\cap \{S_i:S_i\in\breve S\}$, is such that $(t\circ\sigma)(v)\in V$ for all $v\in V$ holds. Then for each $\gamma\in\PD_\Gamma\Sigma$ such that $\gamma\cap S=\{\varnothing\}$, the transformed path is given by $\gamma\circ\sigma(t(\gamma))$.
\end{itemize}
The set of all surface-preserving bisections for a path groupoid forms a group equipped with $\ast$ and $^{-1}$ and is called the \textbf{group $\mathfrak{B}_{\breve S,\diff}(\PD_{\Gamma}\Sigma)$ of surface-preserving bisections of a finite path groupoid}. 
\end{prop}

Observe for a certain bisection $\sigma\in\mathfrak{B}_{\breve S,\diff}(\PD_{\Gamma}\Sigma_v)$ the right translation $R_{\sigma(v)}(\gamma)$ in $\fPG$ for a $v\in V_\Gamma$ defines a path-diffeomorpism $\Phi_\Gamma(\gamma)$ such that $(\varphi_\Gamma,\Phi_\Gamma)\in\Diff(\PD_\Gamma\Sigma_v)$. Consequently similar to surface-preserving bisections of a finite path groupoid the corresponding surface-preserving path-diffeomorphisms is defined. Refer to \cite[Section 6.2]{KaminskiPHD} for a detailed investigation.

In general it follows that, $R_\sigma(\idf_v)=\sigma(v)$ for $v\in V$, where $V=V_\Gamma\setminus V_S$ and $V_S=V_\Gamma\cap S$, and $R_\sigma(\idf_w)=\idf_w$ for $w\in V_S$ for each $S\in\breve S$ is satisfied. Since $\idf_v$ and $\idf_w$  for $v\in V$ and $w\in V_S$ are elements of $\PD_\Gamma\Sigma$.

\begin{prop}
A \textbf{surface-preserving bisection $\sigma$ of a finite graph system} is defined as a bisection $\sigma:V_\Gamma\longrightarrow\PD_\Gamma$ in $\PD_\Gamma$ such that there is a surface-preserving bisection $\tilde\sigma$ of a finite path groupoid $\fPG$ and 
$\sigma_\Gamma(V)=\{\tilde\sigma(v_i):v_i\in V\}$ whenever $V$ is a subset of $V_\Gamma$.

The set surface-preserving bisections for a finite graph system forms a group equipped with $\ast_2$ and $^{-1}$ and is called the
\textbf{group $\mathfrak{B}_{\breve S,\diff}(\PD_{\Gamma})$ of surface-preserving bisections of a finite graph system and a surface set $\breve S$}.
\end{prop}

A right-translation $R_\sigma$ in the finite graph system $\PD_\Gamma$ is defined for a bisection $\sigma\in \mathfrak{B}_{\breve S,\diff}(\PD_{\Gamma})$ in definition \ref{defi bisecongraphgroupioid}.

For example for a path $\gamma$ that intersect a surface $S$ in $t(\gamma)$, then $\gamma_\sigma=\gamma\circ\sigma(t(\gamma))=\gamma$ is satisfied. For a path $\gp$ that intersects $S^\prime$ in $s(\gp)$ the transformed path is $\gamma_\sigma=\gamma\circ\sigma(t(\gamma))$. In both cases the surfaces satisfy $S=S_\sigma$ and $S^\prime=S_\sigma^\prime$.

\begin{center}
\includegraphics[width=0.6\textwidth]{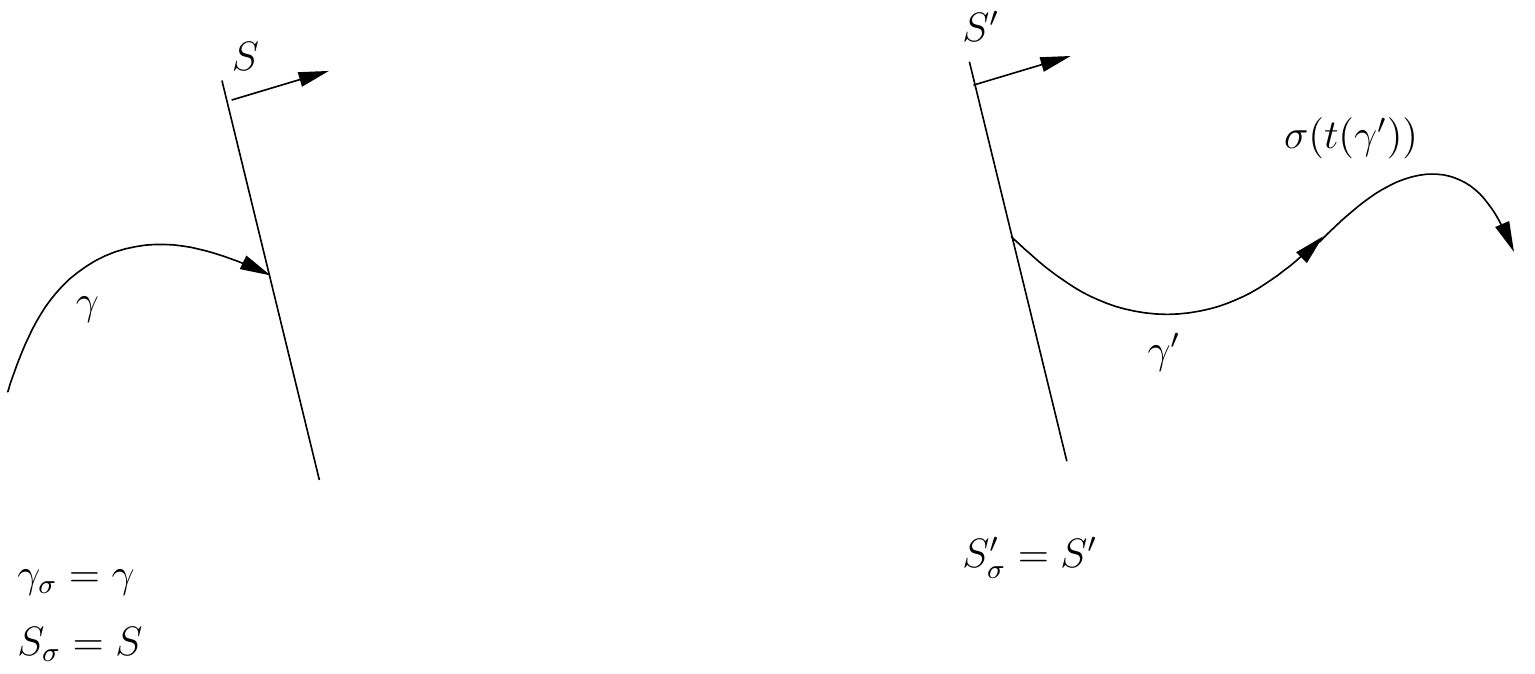}
\end{center}

Especially for a finite surface-preserving graph-diffeomorphism $(\varphi_{\Gamma},\Phi_{\Gamma})$ this implies that, the number of intersections between surfaces in $\breve S$ and a set of paths $\Gamma$ does not change. Certainly it follows that the surfaces are fixed, whereas the graph $\Gamma$ is changed to $\Gamma_\sigma$. 

Now the question arise whether another dissimilar situation is possible. The idea is to implement bisections of such a way that the orientation of the transformed surface with respect to path obtained by the bisection is preserved and paths that are ingoing w.r.t. the orientation of the unchanged surface are ingoing w.r.t. the transformed surface.  The same is true for paths that are outgoing w.r.t. the orientation of the unchanged surface.

Now consider a diffeomorphism or actions of bisections, which maps a surface $S$ into another $S_\sigma$. Then the modified path-diffeomorphisms are maps such that paths, which intersect a surface $S$ and are ingoing w.r.t. a surface $S$, to paths that are ingoing w.r.t the transformed surface $S_\sigma$.
It is possible that the path $\gamma$ lies above the surface $S$, whereas the transformed path $\Phi_\Gamma(\gamma)$ lies below the surface $\varphi(S)$. This case is needed to be excluded. The new action of bisections are required to preserve the orientation properties of the paths and the surfaces such that there is no interchange of the left and right unitary representation of $\bar G_{\breve S,\Gamma}$ on $\HS_\Gamma$. 
 
\begin{defi}
Let $\breve S$ be a set of surfaces and $\Gamma$ be a graph such that $\breve S\cap \Gamma\subset V_\Gamma$. Let $\varphi$ be a diffeomorphism in $\Sigma$, which maps each surface $S\in \breve S$ to a surface $S_\sigma\in\breve S$.

A map $\sigma$ is called \textbf{surface-orientation-preserving bisection for a finite path groupoid}, if $\sigma$ is a bisection in $\mathfrak{B}(\PD_\Gamma\Sigma)$ such that 
\begin{itemize}
\item the map $\varphi_\Gamma:V_\Gamma\longrightarrow V_\Gamma$, which is given by $\varphi_\Gamma=t\circ\sigma$, is bijective and $\varphi_\Gamma=\varphi\vert_{V_\Gamma}$, and
\item for each $\gamma\in\PD_\Gamma\Sigma$ that intersects a surface $S$ and, such that $\gamma\cap S=\{t(\gamma)\}$ holds, the non-trivial transformed path is given by $\gamma\circ\sigma(t(\gamma))=:\gamma_\sigma$ for a surface $S\in\breve S$ yields. Moreover if $\gamma$ lie above (or below) the surface $S$ and $\gamma_\sigma$ is non-trivial, then $\gamma_\sigma$ lie above (or below) the surface $S_\sigma$. Except of a vertex $s(\gamma)$ such that $s(\gamma)\cap S_\sigma=\{s(\gamma)\}$ , the vertex $t(\gamma_\sigma)$ is the only intersection vertex of $S_\sigma$ and $\gamma_\sigma$.
\item For each $\gamma\in\PD_\Gamma\Sigma$ that intersects a surface $S$, and such that $\gamma\cap S=\{s(\gamma)\}$, it is  true that $\sigma(s(\gamma))=\idf_{s(\gamma)}$, $(t\circ\sigma)(s(\gamma))=s(\gamma)$ and hence $\gamma\circ\sigma(s(\gamma))=\gamma\circ\idf_{s(\gamma)}=\gamma$ for a surface $S\in\breve S$. Furthermore if $\gamma$ is located above $S$, then $\gamma$ is located above the surface $S_\sigma$. 
\item The map $\sigma:V\rightarrow \PD_\Gamma\Sigma$, where $V=V_\Gamma\setminus V_{\breve S}$ and $V_{\breve S}=V_\Gamma\cap \{S_i:S_i\in\breve S\}$, is such that $(t\circ\sigma)(v)\in V$ for all $v\in V$ yields. Then for each $\gamma\in\PD_\Gamma\Sigma$ such that $\gamma\cap S=\{\varnothing\}$, the transformed path is given by $\gamma\circ\sigma(t(\gamma))$.
\end{itemize}
\end{defi}
Clearly this concept can be generalised to surface-orientation-preserving bisections of a finite graph system. Moreover similar to surface-orientation-preserving bisections of a finite graph system the corresponding surface-preserving-orientation graph-diffeomorphisms can be defined (\cite[Section 6.2]{KaminskiPHD}).

\begin{cor} 
The set $\mathfrak{B}_{\breve S,\ori}(\PD_\Gamma\Sigma)$ of all surface-orientation-preserving bisections of a finite path groupoid equipped with multiplication $\ast$ and inversion $^{-1}$ forms a group and it is called  the \textbf{group of surface-orientation-preserving bisections of a finite path groupoid associated to surfaces}.

The set $\mathfrak{B}_{\breve S,\ori}(\PD_\Gamma)$ of all surface-orientation-preserving bisections of a finite graph system equipped with multiplication $\ast_2$ and inversion $^{-1}$ forms a group and is called  the \textbf{group $\mathfrak{B}_{\breve S,\ori}(\PD_\Gamma)$ of surface-orientation-preserving bisections associated to graphs and surfaces}.
\end{cor}
For example, consider a graph $\Gamma=\{\gamma,\gp,\gpp\}$ and the surfaces $\breve S:=\{S,S^\prime,S^{\prime\prime}\}$, which are presented in the picure below. Then let $(\varphi_\Gamma,\Phi_\Gamma)$ be a surface-orientation-preserving path-diffeomorphism for $\breve S$ defined by a bisection $\sigma$ such that $S$ is mapped to $S_\sigma$ and so on. Let the path $\gamma$ intersects the surface $S$ in $t(\gamma)$ such that $\gamma$ lies below $S$ and the path $\gamma_\sigma$ intersects $S_\sigma$ in $t(\gamma_\sigma)$ such that $\gamma_\sigma$ lies below $S_\sigma$. Moreover let
$\gp$ intersect $S^\prime$ in $t(\gp)$ such that $\gp$ lies below $S^\prime$ and the path $\gp_\sigma:=\gp\circ\sigma(t(\gp))$ intersects $S_\sigma^\prime$ in $t(\gp_\sigma)$ and $\gp_\sigma$ lies below $S_\sigma^\prime$. Finally for a path $\gpp$ intersecting $S^{\prime\prime}$ in $s(\gpp)$ and the path $\gpp$ lies above, then the transformed path $\gamma_\sigma^{\prime\prime}$ is equivalent to $\gpp$ and $\gpp$ is outgoing and lies above the surface $S_\sigma^{\prime\prime}$. Summarising, there is a map such thats $\Phi_\Gamma(\Gamma)=\{\gamma_\sigma,\gp_\sigma,\gamma_\sigma^{\prime\prime}\}$.

\begin{center}
\includegraphics[width=0.75\textwidth]{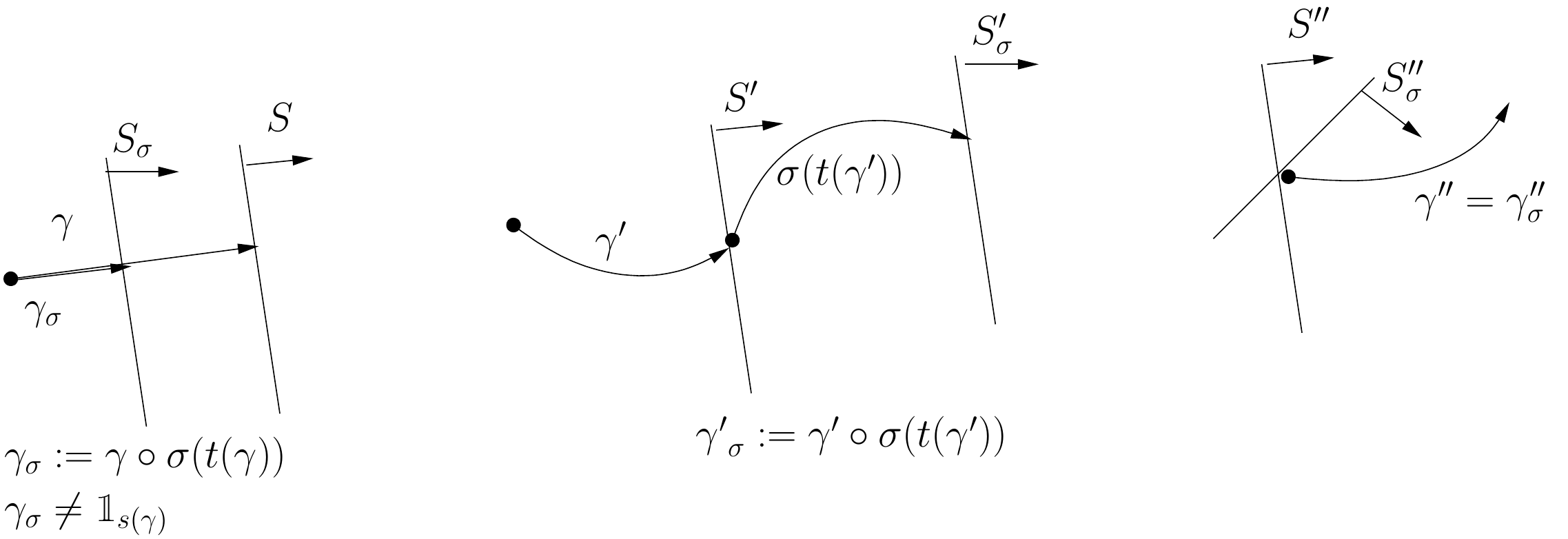}
\end{center}

There is a problem if the bisection $\sigma$ maps $s(\tilde\gamma)$ to a path $\tilde\gamma_\sigma$, which is not equivalent to $\tilde\gamma$. Since the resulting path $\tilde\gamma_\sigma$ is ingoing and lies above the surface $S_\sigma$ whereas $\tilde\gamma$ is outgoing and above $S$.

\begin{center}
\includegraphics[width=0.12\textwidth]{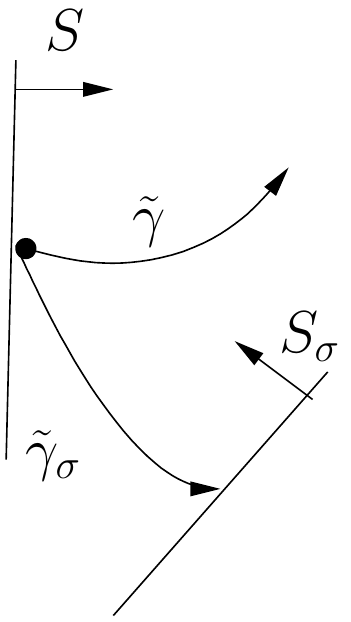}
\end{center}

The situation is restricted to the case that the surface set $\breve S$ is chosen such that each path $\gamma_i$ in $\PD_\Gamma\Sigma$, which intersects only the surfaces $S_j$, which is contained in $\breve S:=\{S_j\}_{1\leq j\leq K}$, and lies above and ingoing, (or below and ingoing, or above and outgoing, or below and outgoing) w.r.t. $S_j$ and there are no other intersection vertices of this path with any other surface $S_k$ for $k\neq j$. Then it is required that, the path $\gamma_i\circ\sigma(t(\gamma_i))$ (or $\sigma(s(\gamma_i))$) lies above and ingoing (or below and ingoing, or above and outgoing, or below and outgoing) w.r.t. each $\varphi(S_j)$. Hence all actions of these bisections preserve the quantum flux operators associated to different surface sets and graphs presented in subsection \ref{subsec fluxdef} can be treated.

\subsubsection*{Actions of the group of surface-preserving bisections of the $C^*$-algebra $\mathsf{W}(\bar G_{\breve S,\Gamma})$ of Weyl elements}

The next question is related to different actions of group of bisections of the algebra of Weyl elements. First consider the action of a surface-preserving group $\mathfrak{B}_{\breve S,\diff}(\PD_\Gamma)$ of bisections of a finite graph system on the $C^*$-algebra $\mathsf{W}(\bar G_{\breve S,\Gamma})$. 
\begin{lem}
The action is trivial, i.e. for $\Gp\in \PD_\Gamma$
\beqs (\zeta_{\sigma} U)(\rho_{S,\Gamma}(\Gp)) = U(\rho_{S,\Gamma}(\Gamma^\prime_\sigma))= U(\rho_{S,\Gamma}(\Gp))
\eqs yields for $\rho_{S,\Gamma}\in G_{\breve S,\Gamma}$. 
\end{lem}
\begin{proofs}
This is true, since $\rho_{S}(\gamma)=e_G$ holds if the path $\gamma\in\Gp$ does not intersect with a surface $S$ in $\breve S$ in a source or target vertex. For a subgraph $\Gp$ of $\Gamma=\{\gamma_1,...,\gamma_N\}$, where each path $\gamma_i$ in $\Gamma$ intersect a surface in $\breve S$ in a vertex, the action is given by $(\zeta_\sigma U)(\rho_{S,\Gamma}(\Gp))= U(\tilde R_\sigma(\rho_{S,\Gamma}(\Gp)))$ where $\tilde R_\sigma(\rho_{S,\Gamma}(\Gp))= \rho_{\varphi(S),\Gamma}(R_\sigma(\Gp))= \rho_{\varphi(S),\Gamma}(\Gamma^\prime_\sigma)$ for all $S\in\breve S$ and $\varphi(S)=S$, $(t\circ\sigma)(v)=v$ and hence $\sigma(v)=\idf_v$ for all $v\in V_\Gamma\cap \breve S$, then finally deduce $\Gamma^\prime_\sigma=\Gp$.
\end{proofs}

\subsubsection*{Action of the group of surface-orientation-preserving bisections of the $C^*$-algebra $\mathsf{W}(\bar G_{\breve S,\Gamma})$ of Weyl elements}

\begin{prop}Let $\breve S$ be a finite set of surfaces.

The action of a surface-orientation-preserving group $\mathfrak{B}_{\breve S,\ori}(\PD_\Gamma)$ of bisections of a finite graph system on the $C^*$-algebra $\mathsf{W}(\bar G_{\breve S,\Gamma})$ is presented by
\beqs (\zeta_{\sigma} U)(\rho_{S,\Gamma}(\Gp)) = U(\rho_{\varphi(S),\Gamma}(R_\sigma(\Gp)))\text{ for all }U\in\Rep(\bar G_{\breve S,\Gamma},\KD(\HS_\Gamma))
\eqs satisfies the following properties:
\begin{enumerate}
\item\label{prop bisec cond1} The action $\zeta$ of $\mathfrak{B}_{\breve S,\ori}(\PD_\Gamma)$ in $\mathsf{W}(\bar G_{\breve S,\Gamma})$, which is a $C^*$-subalgebra of $\LD(\HS_\Gamma)$ is automorphic,
\item\label{prop bisec cond2} The action $\zeta$ of $\mathfrak{B}_{\breve S,\ori}(\PD_\Gamma)$ in $\mathsf{W}(\bar G_{\breve S,\Gamma})$ is point-norm continuous. 

\item\label{prop bisec cond3} The automorphic action $\alpha$ the group of bisection $\mathfrak{B}_{\breve S,\ori}(\PD_\Gamma)$ on $\mathsf{W}(\bar G_{\breve S,\Gamma})$ is inner such that there exists an unitary representation $V$ of $\mathfrak{B}_{\breve S,\ori}(\PD_\Gamma)$ on the $C^*$-algebra $\mathsf{W}(\bar G_{\breve S,\Gamma})$, i.e. $V\in\Rep(\mathfrak{B}_{\breve S,\ori}(\PD_\Gamma),\mathsf{W}(\bar G_{\breve S,\Gamma}))$, which satisfy
\beq\label{cancomrel II} V_\sigma U(\rho_{S,\Gamma}(\Gp)) V^*_\sigma = \big(\zeta_{\sigma}(U)\big)(\rho_{S,\Gamma}(\Gp))\quad\forall U\in\mathsf{W}(\bar G_{\breve S,\Gamma}),\sigma\in\mathfrak{B}_{\breve S,\ori}(\PD_\Gamma)
\eq
\end{enumerate}

For each $\sigma\in\mathfrak{B}_{\breve S,\ori}(\PD_\Gamma)$ the unitary $V_\sigma$ are called the \textbf{unitary bisections of a finite graph system and surfaces in $\breve S$}.
\end{prop}
For example, for an appropriate surface set $\breve S$, a graph $\Gamma=\{\gamma,\gamma_1,...,\gamma_M\}$ and a subgraph $\Gp=\{\gamma\}$ it is true that
\beq (\zeta_{\sigma} U_{\overrightarrow{L}})(\rho_{S}(\gamma)) = U_{\overrightarrow{L}}(\rho_{\varphi(S)}(\gamma\circ\sigma(v)))
\eq whenever $v=t(\gamma)$ and $\Gamma^\prime_\sigma:=\{\gamma\circ\sigma(v)\}$. 

Certainly, there is also an action of the group of bisection $\mathfrak{B}_{\breve S,\ori}(\PD_\Gamma)$ on $\mathsf{W}(\bar G_{\breve S,\Gamma})$ given by a left translation or a inner automorphism of the path groupoid $\PD_\Gamma\Sigma$ over $V_\Gamma$. Each action is inner and hence for each action there is a unitary representation of $\mathfrak{B}_{\breve S,\ori}(\PD_\Gamma)$ on $\mathsf{W}(\bar G_{\breve S,\Gamma})$.

\begin{proofs}For the proof of the action being automorphic (property \ref{prop bisec cond1}) derive the following. Set $\varphi\vert_\Gamma=\varphi_\Gamma=t\circ\sigma$ and $\varphi^\prime\vert_\Gamma=\varphi^\prime_\Gamma=t\circ\sigma^\prime$ for two bisections $\sigma,\sigma^\prime\in\mathfrak{B}_{\breve S,\ori}(\PD_\Gamma)$.
For simplicity assume that, $\Gp=\{\gamma\}$ such that $\rho_{S,\Gamma}(\Gp)=\rho_S(\gamma)$ and $v=t(\gamma)$, $w=t(\sigma^\prime(v))$, then derive
\beqs  &((\zeta_{\sigma}\circ\zeta_{\sigma^\prime})U)(\rho_{S}(\gamma))\\
&=U\Big(\rho_{(\varphi\circ\varphi^\prime)(S)}(\big(\gamma\circ\sigma^\prime(v)\big)\circ\sigma(w)))\Big)
=U\Big(\rho_{(\varphi\circ\varphi^\prime)(S)}(\gamma\circ(\sigma^\prime(v)\circ\sigma(w)))\Big)\\
&=(\zeta_{\sigma\ast\sigma^\prime}U)(\rho_{S}(\gamma))
\eqs which proves condition \ref{def automorphic1} of definition \ref{def automorphic}.

Condition \ref{def automorphic3} of definition \ref{def automorphic} is shown by the observation that 
\beqs &(\zeta_{\sigma}U^*)(\rho_{S}(\gamma))= U(\rho^{-1}_{\varphi(S)}(\gamma\circ\sigma(v))=(\zeta_{\sigma}U)^*(\rho_{S}(\gamma))
\eqs where $v=t(\gamma)$ and $\varphi(S)^{-1}=\varphi(S^{-1})$ holds.

The action satisfies property \ref{prop bisec cond2} of the proposition, since it is true that,
\beqs &\lim_{\sigma\rightarrow \id}\|\zeta_{\sigma}(U(\rho_S(\Gp)))- U(\rho_S(\Gp))\|_\Gamma 
=\lim_{\sigma\rightarrow \id}\|U(\rho_{\varphi(S)}(\Gp_\sigma))- U(\rho_S(\Gp))\|_\Gamma\\
&=0
\eqs yields for a graph $\Gp=\{\gamma\}$ and $\Gamma^\prime_{\sigma}=\{\gamma \circ \sigma(v)\}$ being subgraphs of $\Gamma$, $v=t(\gamma)$, $\sigma\in\mathfrak{B}_{\breve S,\ori}(\PD_\Gamma)$ and
where $\id(v)=v$ for $v=t(\gamma)$ is the identity bisection.

The action is indeed inner (property \ref{prop bisec cond3} of the proposition), since there is a unitary representation $V:\mathfrak{B}_{\breve S,\ori}(\PD_\Gamma) \longrightarrow M(\mathsf{W}(\bar G_{\breve S,\Gamma}))$, where $M(\mathsf{W}(\bar G_{\breve S,\Gamma}))$ is the multiplier algebra of $\mathsf{W}(\bar G_{\breve S,\Gamma})$, such that
\beqs (\zeta_{\sigma}U)(\rho_{S}(\gamma)) = V_\sigma U(\rho_{S}(\gamma))V^*_\sigma\eqs is satisfied.

Observe that, $\rho_S(\idf_v)=e_G$ for any $v\in V_\Gamma$.
Set $(V_{\sigma} U)(\rho_{ S,\Gamma}(\Gp))= U(\rho_{S,\Gamma}(\Gp_\sigma))$ and $(V^*_\sigma U)(\rho_{S,\Gamma}(\Gp))= U(\rho_{S,\Gamma}(\Gp_{\sigma^{-1}}))$, in example
\beqs &(V_\sigma U)(\rho_{S}(\gamma)) = U(\rho_{S}(\gamma\circ\sigma(v))\\
&(V^*_\sigma U)(\rho_{S}(\idf_{w}))=(V_{\sigma^{-1}} U)(\rho_{S}(\idf_{w})) = U(\rho_{S,\Gamma}(\idf_{w}\circ(\sigma((t\circ\sigma)^{-1}(v)))^{-1}))\\
&V_{\sigma}V_{\sigma^{-1}} = \id =
V_{\sigma^{-1}}V_{\sigma} \\
&(V_{\sigma}V_{\sigma^{-1}} U)(\rho_{S}(\gamma))=U(\rho_{S}(\gamma\circ(\sigma\ast\sigma^{-1})(v))=U(\rho_{S}(\gamma))\\
&(V_{\sigma^{-1}}V_{\sigma} U)(\rho_{S}(\idf_{w}))=U(\rho_{S}(\idf_{w}\circ(\sigma^{-1}\ast\sigma)(v))=U(\rho_{S}(\idf_{w}))
\eqs yields whenever $w=(t\circ\sigma)^{-1}(v)$, $v=t(\gamma)$, $\Gp=\{\gamma,\idf_w,\gamma_1,...,\gamma_M\}$, $\Gamma^\prime_\sigma:=\{\gamma\circ\sigma(v),\idf_w,\gamma_1,...,\gamma_M\}$ and $\Gp_{\sigma^{-1}}:=\{\gamma,\idf_w\circ\sigma^{-1}(v),\gamma_1,...,\gamma_M\}$.

To show that, the $^*$-representation $V$ of $\mathfrak{B}_{\breve S,\ori}(\PD_\Gamma)$ on the Hilbert space $\HS_\Gamma$ is really unitary consider the following example.
Let $\gamma_1$ and $\gamma_2$ be two disjoint paths defining two disjoint graphs, then derive
\beq\label{eq partunit} (V_\sigma V^*_\sigma \psi_{\gamma_1})(\ho_{\gamma_1}(\gamma_1))&=\psi_{\gamma_1}(\ho_{\gamma_1}(\gamma_1))\\
(V^*_\sigma V_\sigma \psi_{\gamma_2})(\ho_{\gamma_2}(\gamma_2))&=\psi_{\gamma_2}(\ho_{\gamma_2}(\gamma_2))
\eq
and, consequently,
\beqs
&(V_\sigma V^*_\sigma)\Big(\psi_{\gamma_1}(\ho_{\gamma_1}(\gamma_1)),\psi_{\gamma_2}(\ho_{\gamma_2}(\gamma_2))\Big)
=\Big(\psi_{\gamma_1}(\ho_{\gamma_1}(\gamma_1)),\psi_{\gamma_2}(\ho_{\gamma_2}(\gamma_2))\Big)\\
&=(V^*_\sigma V_\sigma )\Big(\psi_{\gamma_1}(\ho_{\gamma_1}(\gamma_1)),\psi_{\gamma_2}(\ho_{\gamma_2}(\gamma_2))\Big)
\eqs
for  $v=t(\gamma_2)$, $t(\gamma_1)=(t\circ\sigma)^{-1}(v)$, $\psi_{\gamma_1}(\ho_{\gamma_1}(\gamma_1))\in\HS_{\gamma_1}$, $\psi_{\gamma_2}(\ho_{\gamma_2}(\gamma_2))\in\HS_{\gamma_2}$ where $\Gamma=\{\gamma_1,\gamma_2\}$ and $\HS_\Gamma=\HS_{\gamma_1}\otimes\HS_{\gamma_2 }$.

Then collect the following facts to conclude that $V$ is a unitary representation of $\mathfrak{B}_{\breve S,\ori}(\PD_\Gamma)$ in $\mathsf{W}(\bar G_{\breve S,\Gamma})$:
\begin{enumerate}
 \item\label{Vunitreppop1} $V_\sigma$ is unitary for any $\sigma\in \mathfrak{B}_{\breve S,\ori}(\PD_\Gamma)$
 \item $V_\sigma V_{\sigma^\prime}=V(\sigma\ast\sigma^\prime)$ for all $\sigma,\sigma^\prime\in \mathfrak{B}_{\breve S,\ori}(\PD_\Gamma)$
\item $V_\sigma$ is point-norm continuous, since the associated action $\zeta_\sigma$ is.
\end{enumerate} 
\end{proofs}

\subsubsection*{Dynamical systems of an action of the group of surface-preserving bisections of the $C^*$-algebra $\mathsf{W}(\bar G_{\breve S,\Gamma})$ of Weyl elements and states on $C(\Ab_\Gamma)$}

Then the last proposition implies the following.
\begin{prop}\label{lem cdynsysdiffeoweyl}Let $ \breve S$ be a set of surfaces.

The triple $(\mathfrak{B}_{\breve S,\ori}(\PD_\Gamma),\mathsf{W}(\bar G_{\breve S,\Gamma}),\zeta)$ of a surface-orientation-preserving group $\mathfrak{B}_{\breve S,\ori}(\PD_\Gamma)$ of bisections, a $C^*$-algebra $\mathsf{W}(\bar G_{\breve S,\Gamma})$ w.r.t. a set $\breve S$ of surfaces and a graph $\Gamma$ and the action $\zeta$ is a $C^*$-dynamical system.
\end{prop}

\begin{lem}Let $\Phi:\mathsf{W}(\bar G_{\breve S,\Gamma})\longrightarrow \LD(\HS_\Gamma)$ be the natural $^*$-homomorphism.

Then the set $\{\Phi(W)B:W\in \mathsf{W}(\bar G_{\breve S,\Gamma}), B\in \LD(\HS_\Gamma)\}$ is dense in $\LD(\HS_\Gamma)$. Consequently $\Phi$ is a morphism of $C^*$-algebras $\mathsf{W}(\bar G_{\breve S,\Gamma})$ and $\LD(\HS_\Gamma)$.
\end{lem}
Hence it is obvious to consider the following covariant pair.
\begin{prop}
The pair $(\Phi,V)$ consisting of a morphism $\Phi\in \Mor(\mathsf{W}(\bar G_{\breve S,\Gamma}),\LD(\HS_\Gamma))$ and a unitary representation $V$ of $\mathfrak{B}_{\breve S,\ori}(\PD_\Gamma)$ on the Hilbert space $\HS_\Gamma$, i.e. $V_\sigma\in\Rep(\mathfrak{B}_{\breve S,\ori}(\PD_\Gamma),\KD(\HS_\Gamma))$ such that
\beqs \Psi(\zeta_{\sigma}W)
=V_\sigma \Psi(W) V^*_\sigma
\eqs whenever $W\in\mathsf{W}(\bar G_{\breve S,\Gamma})$ and $\sigma\in\mathfrak{B}_{\breve S,\ori}(\PD_\Gamma)$, is a covariant representation of the $C^*$-dynamical system\\ $(\mathfrak{B}_{\breve S,\ori}(\PD_\Gamma),\mathsf{W}(\bar G_{\breve S,\Gamma}),\zeta_\sigma)$ in $\LD(\HS_\Gamma)$.
\end{prop}
\begin{proofs}Conclude for a subgraph $\Gp=\{\gamma\}$ and $\Gamma^\prime_\sigma=\{\gamma\circ\sigma(v)\}$ of $\Gamma$, $v=t(\gamma)$, a surface $S$ such that $\gamma$ is outgoing and lies below and $\psi_{\Gamma}\in\HS_{\Gamma}$ that
\beqs& (V_\sigma U(\rho_{S,\Gamma}(\Gp))V^*_\sigma\psi_{\Gamma})(\ho_{\Gamma}(\Gp)) 
=(V_\sigma U(\rho_{S,\Gamma}(\Gp))\psi_{\Gamma})(\ho_{\Gamma}(\Gp))=(V_\sigma \psi_{\Gamma})(\rho_{S}(\gamma)\ho_{\Gamma}(\gamma))\\
&= U(\rho_{S}(\gamma\circ\sigma(v)))\psi_{\Gamma})(\ho_{\Gamma}(\Gamma^\prime_\sigma))
=(\zeta_{\sigma}(U(\rho_{S,\Gamma}(\Gp)))\psi_{\Gamma})(\ho_{\Gamma}(\Gp))
\eqs holds for $U\in \mathsf{W}(\bar G_{\breve S,\Gamma})$. Hence the proposition is true.
\end{proofs}

Covariant pairs are constructed from the multiplication representation $\Phi_M$ of the holonomy algebra $C(\Ab_\Gamma)$ for a finite graph system $\PD_\Gamma$. The next proposition is valid for both identifications of the configuration space $\Ab_\Gamma$.

\begin{prop}\label{prop graph diffeo}Let $\Phi_M$ be the multiplication represention of $C(\Ab_\Gamma)$ on $\HS_\Gamma$.

Then there is an unitary representation $V_\sigma$ of $\mathfrak{B}(\PD_\Gamma)$ on $\HS_\Gamma$ such that
\beq\label{cancomrel III} 
&V_\sigma \Phi_M(f_\Gamma)V^*_\sigma=\Phi_M(\zeta_{\sigma}f_\Gamma)
\eq whenever $f_\Gamma\in C(\Ab_\Gamma)$, $\sigma\in\mathfrak{B}(\PD_\Gamma)$ and $(V,\Phi_M)$ is a covariant pair of the $C^*$-dynamical system $(\mathfrak{B}(\PD_\Gamma),C(\Ab_\Gamma),\zeta)$.

Then there is a $\mathfrak{B}(\PD_\Gamma)$-invariant state $\omega^\Gamma_{\mathfrak{B}}$ on $C(\Ab_\Gamma)$ such that
\beqs \omega^\Gamma_{\mathfrak{B}}(\zeta_{\sigma}(f_\Gamma))
&=\omega^\Gamma_{\mathfrak{B}}(f_{\Gamma})
:=\la \Omega^\Gamma_{\mathfrak{B}},\Phi_M(f_\Gamma) \Omega^\Gamma_{\mathfrak{B}}\ra
\eqs where $\Omega^\Gamma_{\mathfrak{B}}$ is a cyclic vector in $\HS_\Gamma$ for the GNS-triple $(\HS_\Gamma,\Phi_M,\Omega^\Gamma_{\mathfrak{B}})$.
\end{prop}

Notice that, the state $\omega^\Gamma_{\mathfrak{B}}$ is also a $\mathfrak{B}_{\breve S,\diff}(\PD_\Gamma)$-and $\mathfrak{B}_{\breve S,\ori}(\PD_\Gamma)$-invariant state on $C(\Ab_\Gamma)$. Hence 
\beqs V_{\sigma}\Omega^\Gamma_{\mathfrak{B}}
=\Omega^\Gamma_{\mathfrak{B}}\text{ for all }\sigma\in \mathfrak{B}_{\breve S,\ori}(\PD_\Gamma)
\eqs is true.

Recall the $\bar G_{\breve S,\Gamma}$-invariant state $\omega_M^\Gamma$ of $C(\Ab_\Gamma)$ defined in proposition \ref{prop Ginvstate}.

\begin{problem}\label{probl statewithfdiff}
For a path-diffeomorphism $(\varphi_\Gamma,\Phi_\Gamma)\in\Diff(\PD_\Gamma)$ such that for a graph\\ $\Gamma:=\{\gpe,\gppe,...,\gpm,\gppm,\gpppe,...,\gpppm\}$ where $\vert\Gamma\vert:=N=3M$ and $\Phi_\Gamma(\Gp)=(\gpe\circ\gppe,...,\gpm\circ\gppm)$ the natural identification is used, i.e. 
\beq\label{eq identificationfordiff}\Phi_\Gamma(\Gp)=(\gpe,\gppe,...,\gpm,\gppm)
\eq where 
$\Gp:=(\gpe,...,\gpm)$ and $\Gpp:=(\gppe,...,\gppm)$ are subgraphs of $\Gamma$, then 
\beqs &\omega_M^\Gamma(\zeta_{(\varphi_\Gamma,\Phi_\Gamma)}(f_\Gamma))=\int_{G^{2M}}\dif\mu_\Gamma(\ho_\Gamma(\Phi_\Gamma(\Gp)))f_\Gamma(\ho_\Gamma(\Phi_\Gamma(\Gp)))\\
&=\int_{G^{2M}}\dif\mu_\Gamma(\ho_\Gamma(\Phi_\Gamma(\Gp)))f_\Gamma(\ho_\Gamma(\gpe),\ho_\Gamma(\gppe),...,\ho_\Gamma(\gpm),\ho_\Gamma(\gppm))\\
&=\int_{G^{M}}\dif\mu_\Gamma(\ho_\Gamma(\Gp)) \int_{G^{M}}\dif\mu_\Gamma(\ho_\Gamma(\Gpp))f_\Gamma(\ho_\Gamma(\gpe),\ho_\Gamma(\gppe),...,\ho_\Gamma(\gpm),\ho_\Gamma(\gppm))\\
&= W \int_{G^M}\dif\mu_\Gamma(\ho_\Gamma(\Gp))
f_\Gamma(\ho_\Gamma(\gpe),...,\ho_\Gamma(\gpm))\\
&\neq \omega_M^\Gamma(f_\Gamma)
\eqs where $W$ is a suitable constant. Clearly the state $\omega_M^\Gamma$ is not graph-diffeomorphism invariant.

On the other hand, if instead of the natural identification \eqref{eq identificationfordiff} the non-standard identification is taken into account, then this state is graph-diffeomorphism invariant, since for a graph-diffeomorphism $(\varphi_\Gamma,\Phi_\Gamma)\in\Diff(\PD_\Gamma)$ such that for $\Gamma:=\{\gpe,\gppe,...,\gpm,\gppm,\gpppe,...,\gpppm\}$ and $N=3M$, 
\[\Phi_\Gamma(\Gp)=(\gpe\circ\gppe,...,\gpm\circ\gppm)\] Then derive
\beqs &\omega_M^\Gamma(\zeta_{(\varphi_\Gamma,\Phi_\Gamma)}(f_\Gamma))=\int_{G^{2M}}\dif\mu_\Gamma(\ho_\Gamma(\Phi_\Gamma(\Gp)))f_\Gamma(\ho_\Gamma(\Phi_\Gamma(\Gp)))\\
&=\int_{G^{M}}\dif\mu_\Gamma(\ho_\Gamma(\Phi_\Gamma(\Gp)))f_\Gamma(\ho_\Gamma(\gpe)\ho_\Gamma(\gppe),...,\ho_\Gamma(\gpm)\ho_\Gamma(\gppm))\\
&= \int_{G^M}\dif\mu_\Gamma(\ho_\Gamma(\Gp))
f_\Gamma(\ho_\Gamma(\gpe),...,\ho_\Gamma(\gpm))\\
&= \omega_M^\Gamma(f_\Gamma)
\eqs 
Furthermore recall that, the state $\omega_{M,f}^\Gamma$ defined in remark \ref{rem statewithL1}. Then this state is not path-diffeomorphism invariant, since for a path-diffeomorphism $(\varphi_\Gamma,\Phi_\Gamma)\in\Diff(\PD_\Gamma)$ such that for\\ $\Gamma:=\{\gpe,\gppe,...,\gpm,\gppm,\gpppe,...,\gpppm\}$ and $N=3M$, 
\[\Phi_\Gamma(\Gp)=(\gpe\circ\gppe,...,\gpm\circ\gppm)\] it is true that for $f\in\CD(\Ab_\Gamma)$ such that  $f(\ho_\Gamma(\Gp)\textbf{k}^{-1})=f(\ho_\Gamma(\Gp))$ for all $\textbf{k}\in G^N$ the state satisfies
\beqs &\omega_{M,f}^\Gamma(\zeta_{(\varphi_\Gamma,\Phi_\Gamma)}(f_\Gamma))\\
&=\int_{G^N}\dif\mu_\Gamma(\ho_\Gamma(\Gp))f_\Gamma(\ho_\Gamma(\Phi_\Gamma(\Gp)))f(\ho_\Gamma(\Gp))\\
&=\int_{G^{M}}\dif\mu_\Gamma(\ho_\Gamma(\Gp))f_\Gamma(\ho_\Gamma(\gpe)\ho_\Gamma(\gppe),...,\ho_\Gamma(\gpm)\ho_\Gamma(\gppm))
f(\ho_\Gamma(\gpe),...,\ho_\Gamma(\gpm))\\
&\neq \omega_{M,f}^\Gamma(f_\Gamma)
\eqs Consequently if the function $f$ satisfies addtionally
\beq\label{eq central for algraphs}  f(\ho_\Gamma(\Gp))=f(\ho_\Gamma(\gpe),...,\ho_\Gamma(\gpm))= f(\ho_\Gamma(\Phi_\Gamma(\Gp)))=f(\ho_\Gamma(\gpe)\ho_\Gamma(\gppe),...,\ho_\Gamma(\gpm)\ho_\Gamma(\gppm))
\eq for all $(\varphi_\Gamma,\Phi_\Gamma)\in\Diff(\PD_\Gamma)$ then the state $\omega_{M,f}^\Gamma$ is $\Diff(\PD_\Gamma)$-invariant. But the only function, which satisfies \eqref{eq central for algraphs}  for all graph-diffeomorphism for a finite graph groupoid, is the constant function.

Restrict the state $\omega_{H,f}^\Gamma$ presented in remark \ref{rem statewithL2}, which is $H^N$-invariant, to functions in $f\in \CD(\Ab_\Gamma)$ such that
\beqs &f(R(\textbf{k})(\ho_\Gamma(\Gp)))=f(\ho_\Gamma(\Gp))=f(L(\textbf{k})(\ho_\Gamma(\Gp)))\text{ for all }\textbf{k}\in H^N\text{ and }\\
&f(\ho_\Gamma(\Gp))=f(\ho_\Gamma(\Phi_\Gamma(\Gp)))\text{ for all }\Gp\in\PD_\Gamma\text{ and }(\varphi_\Gamma,\Phi_\Gamma)\in\Diff(\PD_\Gamma)
\eqs Then the state $\omega_{H,f}^\Gamma$ is $H^N$- and $\Diff(\PD_\Gamma)$-invariant. Observe that, this would give a new state for the holonomy-flux $C^*$-algebra, but the flux operators would be implemented by maps $H_{\breve S,\Gamma}$ instead of $G_{\breve S,\Gamma}$.
\end{problem}

Recognize that, the elements of $\Ab_\Gamma$ are of the form $\ho_\Gamma(\Gp)$ where $\Gp$ is a subgraph of $\Gamma$ and the natural or the non-standard identification is applied to identify $\Ab_\Gamma$ with $G^N$. Hence there exists, additionally to the actions in $\Act(\bar G_{S,\Gamma},C(\Ab_\Gamma))$, the actions of $\bar G^A_{S,\Gamma}$ or $\bar \ZD_{\breve S,\Gamma}$ on $C(\Ab_\Gamma)$, which are defined in remark \ref{rem fluxlikeoperators}.

Furthermore due to the fact that the number of subgraphs of $\Gamma$ generated by the edges of $\Gamma$ is finite, there exists a finite set $\mathfrak{B}^\Gamma_{\breve S,\ori}(\PD_\Gamma)$ of bisections such that each of bisection is a map from the set $V_\Gamma$ to a distinct subgraph of $\Gamma$ such that all elements of $\PD_\Gamma$ are construced from the finite set $\mathfrak{B}^\Gamma_{\breve S,\ori}(\PD_\Gamma)$. Call such a set of bisections a \textbf{generating system of bisections for a graph} $\Gamma$. 

\begin{prop}\label{prop invstate for holalg} Let $\mathfrak{B}^\Gamma_{\breve S,\ori}(\PD_\Gamma):=\{\sigma_l\in \mathfrak{B}_{\breve S,\ori}(\PD_\Gamma)\}_{1\leq l\leq k}$ be a subset of $\mathfrak{B}_{\breve S,\ori}(\PD_\Gamma)$ that forms a generating system of bisections for the graph $\Gamma$.

Then there is a state $\hat\omega^\Gamma_{\mathfrak{B}}$ on $C(\Ab_\Gamma)$ given by
\beqs \hat\omega^\Gamma_{\mathfrak{B}}(f_\Gamma)
&:=\frac{1}{k}\sum_{l=1}^k \omega^\Gamma_{M}(\zeta_{\sigma_l}(f_\Gamma))\text{ for }\sigma_l\in\mathfrak{B}^\Gamma_{\breve S,\ori}(\PD_\Gamma)
\eqs which is $\mathfrak{B}_{\breve S,\ori}(\PD_\Gamma)$-invariant and
where $k$ is the maximal number of subgraphs, which are generated by all edges and their compositions of the graph $\Gamma$. 

Consequently the state  $\hat\omega^\Gamma_{\mathfrak{B}}$ on $C(\Ab_\Gamma)$ is $\bar \ZD_{S,\Gamma}$-, $\mathfrak{B}_{\breve S,\ori}(\PD_\Gamma)$- and $\mathfrak{B}_{\breve S,\diff}(\PD_\Gamma)$-invariant, i.o.w. $\hat\omega^\Gamma_{\mathfrak{B}}$ is contained in the set $\Ss^{\ZD,\diff,\ori}(C(\Ab_\Gamma))$ of all $\bar \ZD_{S,\Gamma}$-, $\mathfrak{B}_{\breve S,\ori}(\PD_\Gamma)$- and $\mathfrak{B}_{\breve S,\diff}(\PD_\Gamma)$-invariant states on $C(\Ab_\Gamma)$. 

The actions $\alpha\in\Act(\bar \ZD_{\breve S,\Gamma},C(\Ab_\Gamma))$ and $\zeta\in\Act(\mathfrak{B}_{\breve S,\diff}(\PD_\Gamma),C(\Ab_\Gamma))$ commute, i.e.
\beq\label{eq commactions1}(\alpha(\rho_{S,\Gamma}(\Gp))\circ\zeta_\sigma)(f_\Gamma)
=(\zeta_\sigma\circ\alpha(\rho_{S,\Gamma}(\Gp)))(f_\Gamma)\quad\forall f_\Gamma\in\Alg_\Gamma
\eq 

The state $\hat\omega^\Gamma_{\mathfrak{B}}$ and the actions $\alpha\in\Act(\bar \ZD_{\breve S,\Gamma},C(\Ab_\Gamma))$ and $\zeta\in\Act(\mathfrak{B}_{\breve S,\ori}(\PD_\Gamma),C(\Ab_\Gamma))$ satisfy
\beq\label{eq commactions2}\hat\omega^\Gamma_{\mathfrak{B}}\circ\alpha(\rho_{S,\Gamma}(\Gp))\circ\zeta_\sigma
=\hat\omega^\Gamma_{\mathfrak{B}}\circ\zeta_\sigma\circ\alpha(\rho_{S,\Gamma}(\Gp))
\eq 
\end{prop}

\begin{problem}
 Let $\Gp:=\{ \gamma_1,...,\gamma_M\}$ be a subgraph of $\Gamma$.
Then the following computation for $\rho_{S,\Gamma}(\Gp)\in\bar G_{\breve S,\Gamma}$  
\beqs &(\zeta_{\sigma}\circ\alpha(\rho_{S,\Gamma}(\Gp))(f_\Gamma))(\ho_\Gamma(\Gp))
= (\zeta_{\sigma} f_\Gamma)(\ho_\Gamma(\gamma_1)\rho_S(\gamma_1),...,\ho_\Gamma(\gamma_M)\rho_S(\gamma_M))\\
&= f_\Gamma(\ho_\Gamma(\gamma_1)\ho_\Gamma(\sigma(t(\gamma_1)))\rho_{S_\sigma,\Gamma}(\gamma_1\circ\sigma(t(\gamma_1))),...,\ho_\Gamma(\gamma_M)\ho_\Gamma(\sigma(t(\gamma_M)))\rho_{S_\sigma,\Gamma}(\gamma_M\circ\sigma(t(\gamma_M))))\\
&\neq f_\Gamma(\ho_\Gamma(\gamma_1)\rho_{S_\sigma,\Gamma}(\gamma_1)\ho_\Gamma(\sigma(t(\gamma_1))),...,\ho_\Gamma(\gamma_M)\rho_{S_\sigma,\Gamma}(\gamma_M)\ho_\Gamma(\sigma(t(\gamma_M))))\\
&=((\alpha(\rho_{S_\sigma,\Gamma}(\Gp))\circ\zeta_{\sigma})(f_\Gamma))(\ho_\Gamma(\Gp))
\eqs where $S_\sigma=S$, $\rho_{S_\sigma,\Gamma}(\gamma_i\circ\sigma(t(\gamma_i)))=\rho_{S_\sigma,\Gamma}(\gamma_i)$ for $i=1,...,M$ yields for $\sigma\in\mathfrak{B}_{\breve S,\diff}(\PD_\Gamma)$, $\gamma_i\cap S=\{t(\gamma_i)\}$. Clearly equality holds for every $\rho_{S,\Gamma}(\Gp)\in\bar \ZD_{\breve S,\Gamma}$. Hence the action of $\bar G_{\breve S,\Gamma}$ and the action of $\mathfrak{B}_{\breve S,\diff}(\PD_\Gamma)$ on the analytic holonomy $C^*$-algebra do not commute.

In particular observe that, if the natural identification of $\Ab_\Gamma$ is assumed, then there exists no map $D$ such that
\beqs &(D(\zeta_{\sigma}\circ\alpha(\rho_{S,\Gamma}(\Gp))(f_\Gamma))(\ho_\Gamma(\Gp))\\
&= (D f_\Gamma)(\ho_\Gamma(\gamma_1)\ho_\Gamma(\sigma(t(\gamma_1)))\rho_{S_\sigma,\Gamma}(\gamma_1\circ\sigma(t(\gamma_1))),...,\ho_\Gamma(\gamma_M)\ho_\Gamma(\sigma(t(\gamma_M)))\rho_{S_\sigma,\Gamma}(\gamma_M\circ\sigma(t(\gamma_M))))\\
&= f_\Gamma(\ho_\Gamma(\gamma_1)\rho_{S_\sigma,\Gamma}(\gamma_1),\ho_\Gamma(\sigma(t(\gamma_1)))\rho_{S_\sigma,\Gamma}(\sigma(t(\gamma_1))),...,\ho_\Gamma(\gamma_M)\rho_{S_\sigma,\Gamma}(\gamma_M),\ho_\Gamma(\sigma(t(\gamma_M)))\rho_{S_\sigma,\Gamma}(\sigma(t(\gamma_M))))\\
&= f_\Gamma(\ho_\Gamma(\gamma_1)\rho_{S_\sigma,\Gamma}(\gamma_1),\ho_\Gamma(\sigma(t(\gamma_1))),...,\ho_\Gamma(\gamma_M)\rho_{S_\sigma,\Gamma}(\gamma_M),\ho_\Gamma(\sigma(t(\gamma_M))))
\eqs holds for every element $\rho_{S,\Gamma}\in G_{\breve S,\Gamma}$, $\rho_{S_\sigma,\Gamma}\in G_{\breve S,\Gamma}$ and whenever $\sigma(t(\gamma_i))$ and $S_\sigma$ do not intersect each other for every $i=1,...,M$. 

But if $\rho_{S,\Gamma}\in \ZD_{S,\Gamma}$, then there is a map $D$ such that
\beqs &(D(\zeta_{\sigma}\circ\alpha(\rho_{S,\Gamma}(\Gp)))(f_\Gamma))(\ho_\Gamma(\Gp))\\
&= (D f_\Gamma)(\ho_\Gamma(\gamma_1)\rho_{S_\sigma,\Gamma}(\gamma_1)\ho_\Gamma(\sigma(t(\gamma_1))),...,\ho_\Gamma(\gamma_M)\rho_{S_\sigma,\Gamma}(\gamma_M)\ho_\Gamma(\sigma(t(\gamma_M))))\\
&= f_\Gamma(\ho_\Gamma(\gamma_1)\rho_{S_\sigma,\Gamma}(\gamma_1),\ho_\Gamma(\sigma(t(\gamma_1))),...,\ho_\Gamma(\gamma_M)\rho_{S_\sigma,\Gamma}(\gamma_M),\ho_\Gamma(\sigma(t(\gamma_M))))
\eqs is fulfilled.
\end{problem}
Consequently the actions of the flux group and the group of bisections are treated simulaneously only in the case of the commutative flux group $\bar \ZD_{\breve S,\Gamma}$.

\begin{proofs}
First observe that, the state $\hat\omega^\Gamma_{\mathfrak{B}}$ defined in the proposition is well-defined in both cases of a natural or non-standard identication of $\Ab_\Gamma$ and $G^N$. To conclude that, \eqref{eq commactions1} yields for an action of $\bar \ZD_{\breve S,\Gamma}$ on $C(\Ab_\Gamma)$, investigate the computation
\beqs \hat\omega^\Gamma_{\mathfrak{B}}((\zeta_\sigma\circ\alpha(\rho_{S,\Gamma}(\Gp)))(f_\Gamma))
&=\frac{1}{k}\sum_{l=1}^k \omega^\Gamma_{M}(\zeta_{\sigma_{l+1}}(\alpha(\rho_{S,\Gamma})(f_\Gamma)))\\
&=\frac{1}{k}\sum_{l=1}^k \omega^\Gamma_{M}(\alpha(\rho_{S_\sigma,\Gamma}(\Gp))(\zeta_{\sigma_{l+1}}(f_\Gamma)))\\
&=\frac{1}{k}\sum_{l=1}^k \omega^\Gamma_{M}(\zeta_{\sigma_{l+1}}(f_\Gamma))\\
&=\hat\omega^\Gamma_{\mathfrak{B}}(f_\Gamma)=\hat\omega^\Gamma_{\mathfrak{B}}((\alpha(\rho_{S_\sigma,\Gamma}(\Gp))\circ\zeta_\sigma)(f_\Gamma))
\eqs
for $\rho_{S_\sigma,\Gamma}(\Gp),\rho_{S,\Gamma}(\Gp)\in\bar \ZD_{\breve S,\Gamma}$, $\sigma\in\mathfrak{B}_{\breve S,\diff}(\PD_\Gamma)$ or $\sigma\in\mathfrak{B}_{\breve S,\ori}(\PD_\Gamma)$.
\end{proofs}

It is possible to construct a state $\breve\omega^\Gamma_{\mathfrak{B}}$ on $C(\Ab_\Gamma)$, which is $\bar G^A_{S,\Gamma}$-invariant. But this need more technical details, which are not concerned in this article.

\subsubsection*{An action of the local flux group on the holonomy algebra for finite graph systems}

Let $\Gp:=\{\gamma_1,...,\gamma_M\}$ be a subgraph of $\Gamma$. Recall the maps $G^{\loc}_{\Gamma}$ presented in definition \ref{def Gloc}. For an element $\textbf{g}_\Gamma\in G^{\loc}_{\Gamma}$ it is true that, $\textbf{g}_\Gamma(\Gp)$ is identified with the element $(g_\Gamma(s(\gamma_1)),...,g_\Gamma(s(\gamma_M)))$ in $G^{\vert\Gamma\vert}$. Notice that, it is not necessary to focus on natural identified graphs in a finite graph system. Then there is an action $\alpha_{\loc}$ of $\bar G^{\loc}_{\Gamma}$ on $C(\Ab_\Gamma)$ given by
\beqs (\alpha_{\loc}(\textbf{g}_\Gamma(\Gp))f_\Gamma)(\ho_\Gamma(\Gp))&:=f_\Gamma( g_\Gamma(s(\gamma_1))\ho_\Gamma(\gamma_1)g_\Gamma(t(\gamma_1))^{-1},...,g_\Gamma(s(\gamma_M))\ho_\Gamma(\gamma_M)g_\Gamma(t(\gamma_M))^{-1})
\eqs 

Consider the example, which is given by a graph $\Gamma:=\{\gamma_1,\gamma_2\}$ and asubgraph $\Gp:=\{\gamma_1\circ\gamma_2\}$. Then calculate
\beqs (\alpha_{\loc}(\textbf{g}_\Gamma(\Gp))f_\Gamma)(\ho_\Gamma(\Gp))
&=(D_S\alpha_{\loc}(\textbf{g}_\Gamma(\Gamma))D_S^{-1}f_\Gamma)(\ho_\Gamma(\gamma_1\circ\gamma_2))\\
&=(D_S\alpha_{\loc}(\textbf{g}_\Gamma(\Gamma))f_\Gamma)(\ho_\Gamma(\gamma_1),\ho_\Gamma(\gamma_2))\\
&=(D_Sf_\Gamma)( g_\Gamma(s(\gamma_1))\ho_\Gamma(\gamma_1)g_\Gamma(t(\gamma_1))^{-1},g_\Gamma(s(\gamma_2))\ho_\Gamma(\gamma_2)g_\Gamma(t(\gamma_2))^{-1})\\
&=f_\Gamma( g_\Gamma(s(\gamma_1))\ho_\Gamma(\gamma_1)g_\Gamma(t(\gamma_1))^{-1}g_\Gamma(s(\gamma_2))\ho_\Gamma(\gamma_2)g_\Gamma(t(\gamma_2))^{-1})\\
&=f_\Gamma( g_\Gamma(s(\gamma_1))\ho_\Gamma(\gamma_1\circ\gamma_2)g_\Gamma(t(\gamma_2))^{-1})
\eqs
 
\begin{defi}
The $\bar G^{\loc}_{\Gamma}$-fixed point subalgebra of $C(\Ab_\Gamma)=:\Alg_\Gamma$ is given by
\beqs \Alg^{\loc}_\Gamma:=\{f_\Gamma\in\Alg_\Gamma:\alpha_{\loc}(\textbf{g}_\Gamma(\Gp))(f_\Gamma)=f_\Gamma\quad\forall \textbf{g}_\Gamma(\Gp)\in\bar G^{\loc}_{\Gamma}\}
\eqs and this algebra is called the \textbf{$C^*$-algebra of gauge invariant holonomies restricted to finite graph systems}.
\end{defi}

Notice that, there is a isomorphism between the $C^*$-algebras $\Alg^{\loc}_\Gamma$ and $C(\Ab_\Gamma/\bar\SimGroup_\Gamma)$.

\subsection{Weyl $C^*$-algebras associated to surfaces and inductive limits of finite graph systems} \label{subsec weylalg}
\subsubsection*{Weyl $C^*$-algebras associated to surfaces and finite graph systems}

\begin{defi}
Let $\Gamma$ be a graph and $\PD_\Gamma$ be a finite graph system associated to $\Gamma$ and $\breve S$ a surface set. Let $\surf$ be the set of all suitable surface sets for $\Gamma$.

The algebra generated by all elements of $C(\Ab_{\Gamma})$ and $\mathbf{W}(\bar G_{\breve S,\Gamma})$, which satisfy the canonical commutator relations \eqref{cancomrel I}, form an \textbf{abstract Weyl $^*$-algebra $\mathbb{W}(\breve S,\Gamma)$ for a surface set and a finite graph system} associated to a graph $\Gamma$. The algebra generated by all elements of $C(\Ab_{\Gamma})$ and $\mathbf{W}(\bar G_{\surf,\Gamma})$ forms an \textbf{abstract Weyl $^*$-algebra $\mathbb{W}(\surf,\Gamma)$ for surfaces and a finite graph system} associated to a graph $\Gamma$.
\end{defi}

Due to the fact that all unitaries $U$ (or $V$) define a homomorphism of $\bar G_{\breve S,\Gamma}$ (or $\mathfrak{B}_{\breve S,\ori}(\PD_{\Gamma})$) into a unitary group of $\LD(\HS_\Gamma)$, the abstract Weyl $^*$-algebra $\mathbb{W}(\breve S, \Gamma)$ is completed to a $C^*$-algebra.

Summarising the Weyl $C^*$-algebra of Loop Quantum Gravity is generated by continuous functions depending on holonomies along paths and the (strongly) continuous unitary flux operators.   

\begin{prop}
Let $\HS_\Gamma$ be the Hilbert space $L^2(\Ab_\Gamma,\mu_\Gamma)$ with norm $\|.\|^\Gamma_2$.

The $^*$-algebra\footnote{modulo the two-sided self-adjoint ideal of the $^*$-algebra defined by $I=\{W:\|W\|_2=0\}$} generated by all elements of $C(\Ab_{\Gamma})$ and $\mathbf{W}(\bar G_{\breve S,\Gamma})$ for every surface set $\breve S$ in $\surf$, which satisfy the canonical commutator relations \eqref{cancomrel I}, completed w.r.t. the $\|.\|^\Gamma_2$-norm is a $C^*$-algebra. This $C^*$-algebra is called the \textbf{Weyl $C^*$-algebra for surfaces and a finite graph system}. 
Denote this $C^*$-algebra by $\mathsf{Weyl}(\surf, \Gamma)$.

The $^*$-algebra\footnote{modulo the two-sided self-adjoint ideal of the $^*$-algebra defined by $I=\{W:\|W\|_2=0\}$} generated by all elements of $C(\Ab_{\Gamma})$ and $\mathbf{W}(\bar\ZD_{\breve S,\Gamma})$ for every surface set $\breve S$ in $\surf_\ZD$, which satisfy the canonical commutator relations \eqref{cancomrel I}, completed w.r.t. the $\|.\|^\Gamma_2$-norm is a $C^*$-algebra. This $C^*$-algebra is called the \textbf{commutative Weyl $C^*$-algebra for surfaces and a finite graph system}. 
Denote this $C^*$-algebra by $\mathsf{Weyl}_\ZD(\surf_\ZD, \Gamma)$.
\end{prop}

The set $\surf$ of surface sets, which is used to define the $C^*$-algebra $\mathsf{Weyl}(\surf, \Gamma)$, and the set $\surf_\ZD$, which defines $\mathsf{Weyl}_\ZD(\surf_\ZD, \Gamma)$, are distinguished from each other. The set $ \surf_\ZD$ contains the set $\surf$. 

\begin{prop}
The $^*$-algebra\footnote{modulo the two-sided self-adjoint ideal of $\mathbf{Weyl}(\breve S, \Gamma)$ defined by $I=\{W:\|W\|=0\}$ for all $\breve S\in \surf$} generated by all elements of $C(\Ab_{\Gamma})$ and $\mathbf{W}(G_{\breve S,\Gamma})$ for every surface set $\breve S$ in $\surf$, which satisfy the canonical commutator relations \eqref{cancomrel I}, completed w.r.t. the norm
\beqs \| W\|:=\sup\{\|\pi_r(W)\|_r: \pi_r\text{ a unital }^*\text{-representation of } \mathbb{W}(\breve S, \Gamma)\text{ on }\HS_r \text{ }\forall \breve S\in \surf\}
\eqs is a $C^*$-algebra. This $C^*$-algebra is called the \textbf{universal Weyl $C^*$-algebra for surfaces and a finite graph system} and will be denoted by $\mathcal{W}\text{\textit{eyl}}(\surf, \Gamma)$.
\end{prop}

\subsubsection*{The Weyl algebra for surfaces}

\begin{prop}
Define the action of $\mathfrak{B}_{\breve S,\ori}(\PD_\Gamma)$ (or $\mathfrak{B}_{\breve S,\diff}(\PD_\Gamma)$) on $\mathsf{Weyl}_\ZD(\surf_\ZD,\Gamma)$ by
\beqs \zeta_\sigma(U(\rho_{S,\Gamma}(\Gp)) f_\Gamma):= \big(\zeta_\sigma(U)\big)(\rho_{S,\Gamma}(\Gp))\zeta_\sigma(f_\Gamma)
\eqs whenever $\sigma\in\mathfrak{B}_{\breve S,\ori}(\PD_\Gamma)$ and $U(\rho_{S,\Gamma}(\Gp)),f_\Gamma\in \mathsf{Weyl}_\ZD(\surf_\ZD,\Gamma)$ for every surface $S$ in the set $\breve S$, which is contained in $\surf_\ZD$. This action is automorphic and point-norm continuous. 
\end{prop}

Let $\breve S$ and $\breve S^\prime$ be two disjoint surface sets in $\surf_\ZD$.
Then the action of $\mathfrak{B}_{\breve S,\ori}(\PD_\Gamma)$ on $\mathsf{Weyl}_\ZD(\surf_\ZD,\Gamma)$ satisfies
\beqs \big(\zeta_\sigma(U)\big(\rho_{S,\Gamma}(\Gp)):= \idf_\Gamma
\eqs for $\sigma\in\mathfrak{B}_{\breve S,\ori}(\PD_\Gamma)$ and $U(\rho_{S,\Gamma}(\Gp))\in \mathsf{Weyl}_\ZD(\surf_\ZD,\Gamma)$ and $S\in\breve S^\prime$.

\begin{prop}\label{prop invariant state of weyl}
\begin{enumerate}
 \item\label{prop item 1} The state $\omega_{M}^\Gamma$ on $C(\Ab_\Gamma)$, which is defined in \ref{prop Ginvstate} and which is $\bar G_{\breve S,\Gamma}$-invariant for every surface set $\breve S$ in $\surf$, extends to a state $\breve\omega_{M}^{\Gamma}$ on $\mathsf{Weyl}(\surf, \Gamma)$.  The state $\breve\omega_{M}^{\Gamma}$ is pure and unique.
\item\label{prop item 2} The state $\hat\omega_{\mathfrak{B}}^\Gamma$ on $C(\Ab_\Gamma)$, which is defined in \ref{prop invstate for holalg} and which is $\bar\ZD_{\breve S,\Gamma}$-, $\mathfrak{B}_{\breve S,\ori}(\PD_{\Gamma})$- and $\mathfrak{B}_{\breve S,\diff}(\PD_{\Gamma})$-invariant for every surface set $\breve S$ in $\surf_\ZD$, extends to a state on $\mathsf{Weyl}_\ZD(\surf_\ZD, \Gamma)$.  
The state $\omega_{M,\mathfrak{B}}^{\Gamma}$ on $\mathsf{Weyl}_\ZD(\surf_\ZD, \Gamma)$ is $\bar\ZD_{\breve S,\Gamma}$-, $\mathfrak{B}_{\breve S,\ori}(\PD_{\Gamma})$- and $\mathfrak{B}_{\breve S,\diff}(\PD_{\Gamma})$-invariant. 
\end{enumerate}
\end{prop}
\begin{defi}
 The set of all not necessarily pure states on $\mathsf{Weyl}_\ZD(\surf_\ZD, \Gamma)$ that are $\bar\ZD_{\breve S,\Gamma}$-, $\mathfrak{B}_{\breve S,\ori}(\PD_{\Gamma})$- and $\mathfrak{B}_{\breve S,\diff}(\PD_{\Gamma})$-invariant are denoted by $\Ss^{\diff,\ori}(\mathsf{Weyl}_\ZD(\surf_\ZD, \Gamma))$. 
\end{defi}

\begin{proofo}of proposition \ref{prop invariant state of weyl}\\
The part \ref{prop item 1} of the proposition follows from the corollary \ref{cor uniqueness} and the proposition \ref{prop Ginvstate}.
The part \ref{prop item 2} of the proposition follows from the following derivations. First fix a surface set $\breve S$ in $\surf_\ZD$.\\
\textbf{Step 1:}\\
For the covariant pair $(\Phi_M,V)$ of $(\mathfrak{B}_{\breve S,\ori}(\PD_{\Gamma}),C(\Ab_\Gamma),\zeta)$ or $(\mathfrak{B}_{\breve S,\diff}(\PD_{\Gamma}),C(\Ab_\Gamma),\zeta)$ on $\HS_\Gamma$ there exists a invariant state. In proposition \ref{prop invstate for holalg} this state is defined and satisfies
\beqs \hat\omega_{\mathfrak{B}}^\Gamma(\zeta_{\sigma}(f_\Gamma))
=\hat\omega_{\mathfrak{B}}^\Gamma(f_\Gamma)
\eqs for all $f_\Gamma\in C(\Ab_\Gamma)$ and for arbitrary $\sigma\in \mathfrak{B}_{\breve S,\diff}(\PD_{\Gamma})$ or $\sigma\in \mathfrak{B}_{\breve S,\ori}(\PD_{\Gamma})$.
Recall that, the state $\omega_{\mathfrak{B}}^\Gamma$ on $C(\Ab_\Gamma)$ is required to be $\bar\ZD_{\breve S,\Gamma}$-invariant, too. Hence the state satisfies
\beqs \hat\omega_{\mathfrak{B}}^\Gamma(\alpha(\rho_{S,\Gamma}(\Gamma))(f_\Gamma))=\hat\omega_{\mathfrak{B}}^\Gamma(f_\Gamma)\eqs for all $f_\Gamma\in C(\Ab_\Gamma)$ and $\rho_{S,\Gamma}(\Gamma)\in\bar G_{\breve S,\Gamma}$. While $\alpha(\rho_{S,\Gamma}(\Gamma))(f_\Gamma)\in C(\Ab_\Gamma)$ and the actions $\zeta$ and $\alpha$ commute, the state $\hat\omega_{\mathfrak{B}}^\Gamma$ fulfill
\beqs \hat\omega_{\mathfrak{B}}^\Gamma(\alpha(\rho_{S,\Gamma}(\Gamma))(f_\Gamma))
&=\hat\omega_{\mathfrak{B}}^\Gamma(\zeta_{\sigma}(\alpha(\rho_{S,\Gamma}(\Gamma))(f_\Gamma)))
=\hat\omega_{\mathfrak{B}}^\Gamma(\alpha(\rho_{S_\sigma,\Gamma}(\Gamma))(\zeta_{\sigma}(f_\Gamma)))\\
&=\hat\omega_{\mathfrak{B}}^\Gamma(f_\Gamma)
\eqs 
Clearly there is a morphism $\Phi\in\Mor(C(\Ab_\Gamma),\mathsf{Weyl}_{\ZD}(\breve S, \Gamma))$.

\textbf{Step 2:}\\
On the other hand, there are covariant representations $(\Psi,V)$ of the $C^*$-dynamical systems\\ $(\mathfrak{B}_{\breve S,\ori}(\PD_\Gamma),\mathsf{W}(\bar G_{\breve S,\Gamma}),\zeta)$ and $(\mathfrak{B}_{\breve S,\diff}(\PD_\Gamma),\mathsf{W}(\bar G_{\breve S,\Gamma}),\zeta)$ in $\LD(\HS_\Gamma)$. There is a $\bar G_{\breve S,\Gamma}$-invariant, $\mathfrak{B}_{\breve S,\diff}(\PD_\Gamma)$-invariant and $\mathfrak{B}_{\breve S,\ori}(\PD_\Gamma)$-invariant state $\tilde\omega_{M,\mathfrak{B}}^\Gamma$ on $\mathsf{W}(\bar G_{\breve S,\Gamma})$.

\textbf{Step 3:}\\
There are covariant representations $(\Phi_{\Gamma},V)$ of the $C^*$-dynamical systems\\ $(\mathfrak{B}_{\breve S,\ori}(\PD_\Gamma),\mathsf{Weyl}_{\ZD}(\breve S, \Gamma),\zeta)$ and $(\mathfrak{B}_{\breve S,\diff}(\PD_\Gamma),\mathsf{Weyl}_{\ZD}(\breve S, \Gamma),\zeta)$ in $\LD(\HS_\Gamma)$, where\\ $\Phi_{\Gamma}(W)=\Psi(W)$ for $W=U\in \mathbf{W}(\bar\ZD_{\breve S,\Gamma})$ or $\Phi_{\Gamma}(W)=\Phi_M(W)$ for $W=f_\Gamma\in C(\Ab_\Gamma)$.
Consequently there exists a $\bar \ZD_{\breve S,\Gamma}$-invariant, $\mathfrak{B}_{\breve S,\diff}(\PD_\Gamma)$-invariant and $\mathfrak{B}_{\breve S,\ori}(\PD_\Gamma)$-invariant state $\omega_{M,\mathfrak{B}}^\Gamma$ on $\mathsf{Weyl}_{\ZD}(\breve S, \Gamma)$. This state is an extension of the state $\omega_{\mathfrak{B}}^\Gamma$ on $C(\Ab_\Gamma)$ by Hahn-Banach theorem.

Then the state restricted to $C(\Ab_\Gamma)$ is given by
\beqs \omega_{M,\mathfrak{B}}^\Gamma(f_\Gamma)=\hat\omega_{\mathfrak{B}}^\Gamma(f_\Gamma)\quad\forall f_\Gamma\in C(\Ab_\Gamma)\eqs
and restricted to $\mathbf{W}(\bar\ZD_{\breve S,\Gamma})$ it is given by
\beqs \omega_{M,\mathfrak{B}}^\Gamma(W)
=\tilde\omega_{M,\mathfrak{B}}^\Gamma(\idf_\Gamma)\quad\forall W\in \mathbf{W}(\bar\ZD_{\breve S,\Gamma})
\eqs such that $\tilde\omega_{M,\mathfrak{B}}^\Gamma(W^*W)=1$ holds.
Then
\beqs \omega_{M,\mathfrak{B}}^\Gamma(Wf_\Gamma)=\omega_{M,\mathfrak{B}}^\Gamma(f_\Gamma)=\omega_{M,\mathfrak{B}}^\Gamma(f_\Gamma W)
\eqs yields for all $W\in \mathbf{W}(\bar\ZD_{\breve S,\Gamma})$ and $f_\Gamma\in C(\Ab_\Gamma)$.

Observe that, 
\beqs \omega_{M,\mathfrak{B}}^\Gamma(V^*_\sigma V_\sigma)=1 \quad\forall V\in\Rep(\mathfrak{B}_{\breve S,\ori}(\PD_{\Gamma}),\KD(\HS_{\Gamma}))\text{ or }V\in\Rep(\mathfrak{B}_{\breve S,\diff}(\PD_{\Gamma}),\KD(\HS_{\Gamma}))
\eqs holds.

Finally all steps are true for every surface set $\breve S$ and hence for all surface sets in $\surf_\ZD$.
\end{proofo}

Now, in the next investigations the focus lies on actions of fluxes and diffeomorphisms on the inductive limit algebra $C(\Ab)$.

\begin{defi}Let $\Gamma_\infty$ be the inductive limit of a inductive family $\{\Gamma_i\}$ of graphs. Then $\PD_{\Gamma_\infty}$ denotes the inductive limit of a inductive family of finite graph systems $\PD_{\Gamma_i}$. Moreover let $G^\infty$ be the projective limit of the family of groups $\{G^{N_i}\}$ if $G^{N_i}=G\times ...\times G$ and $G$ is a compact group.

Let $(\varphi,\Phi)\in\Diff(\PD)$ be a path-diffeomorphism of a path groupoid $\PG$ such that $\varphi:\Sigma\longrightarrow \Sigma$ and $\Phi:\PD\longrightarrow\PD$. Then a \textbf{graph-diffeomorphism for a limit graph system} $\PD_{\Gamma_\infty}$ is given by the pair $(\varphi_\Sigma,\Phi_\infty)$ of maps such that $\varphi_\Sigma:\Sigma\longrightarrow \Sigma$, $\Phi_{\infty}:\PD_{\Gamma_\infty}\longrightarrow\PD_{\Gamma_\infty}$ and 
\beqs \Phi_{\infty}(\Gamma)= (\Phi(\gamma_1),...,\Phi(\gamma_N))=\Gamma_{\Phi}
\eqs for $\Gamma:=\{\gamma_1,...,\gamma_N\}$ and $\Gamma_{\Phi}$ being two subgraphs of $\PD_{\Gamma_\infty}$. The set of such graph-diffeomorphism for a limit graph system $\PD_{\Gamma_\infty}$ is denoted by $\Diff(\PD_{\Gamma_\infty})$.

Then there is an \textbf{action of graph-diffeomorphisms for a limit graph system} $\PD_{\Gamma_\infty}$ on the analytic holonomy $C^*$-algebra $C(\Ab)$ defined by
\beq
(\theta_{(\varphi_\Sigma,\Phi_\infty)}f )(\ho (\Gamma))
&:= (f )(\ho (\Phi_\infty(\Gamma)))=(f )(\ho (\Gamma_{\Phi}))
\eq whenever $\Gamma,\Gamma_{\Phi}\in\PD_{\Gamma_\infty}$, $(\varphi_\Sigma,\Phi_\infty)\in\Diff(\PD_{\Gamma_\infty})$ and
for 
\beq
f (\ho (\Gamma))= (f_\Gamma\circ\pi_{\Gamma})(\ho (\Gamma))=(\beta_{\Gamma}f_\Gamma)(\ho_{\Gamma}(\Gamma)) 
\eq where $\pi_\Gamma:\Ab \longrightarrow\Ab_\Gamma$ is a surjective projection and $ \beta_\Gamma:C(\Ab_\Gamma)\longrightarrow C(\Ab)$ are injective unit-preserving $^*$-homomorphism satisfying consistency conditions.
\end{defi} 

The group of bisections $\mathfrak{B}(\PD)$ is defined to be the set of all smooth maps $\sigma$ from $\Sigma$ to the path groupoid $\PGm$ such that $s\circ \sigma=\id_\Sigma$ and $t\circ\sigma:\Sigma\longrightarrow\Sigma$ is a diffeomorphism. 
Therefore due to the group morphism $\mathfrak{B}(\PD)\ni\sigma\mapsto t\circ\sigma\in\Diff(\Sigma)$ there exists also an action of the bisections of the $C^*$-algebra $C(\Ab)$. Recognize that, it is possible to rewrite
\beqs \Phi_{\infty}(\Gamma)=: \Gamma_{\sigma_\Sigma}\text{ for a bisection }\sigma_\Sigma\in\mathfrak{B}(\PD_{\Gamma_\infty})\text{ on the limit graph system }\PD_{\Gamma_\infty}
\eqs

Recall the action $\zeta$ of the group $\mathfrak{B}(\PD_{\tilde\Gamma})$ of bisections for a finite graph system $\PD_{\tilde\Gamma}$ on $C(\Ab_{\tilde\Gamma})$, which is given in proposition \ref{prop groupbisecdynsys} by
\beqs (\zeta_\sigma f_{\tilde\Gamma})(\ho_{\tilde\Gamma}(\Gamma))=( f_{\tilde\Gamma}\circ R_\sigma)(\ho_{\tilde\Gamma}(\Gamma))= f_{\tilde\Gamma}(\ho_{\tilde\Gamma}(\Gamma_\sigma))
\eqs whenever $\sigma\in \mathfrak{B}(\PD_{\tilde\Gamma})$, $f_{\tilde\Gamma}\in C(\Ab_{\tilde\Gamma})$ and $\PD_{\Gamma_\sigma}\leq\PD_{\tilde\Gamma}$.

\begin{defi}There is an \textbf{action of the group of global bisections} $\mathfrak{B}(\PD_{\Gamma_\infty})$ on the algebra $C(\Ab)$ given by
\beqs (\zeta_{{\sigma_\Sigma}} f )(\ho (\Gamma))
&:= f (R_{\sigma_\Sigma}\ho (\Gamma))
= f (\ho (\Gamma_{\sigma_\Sigma}))\\
&=\left((\beta_{\tilde\Gamma}\circ\beta_{\Gamma_{\sigma_\Sigma},\tilde\Gamma})f_{\Gamma_{\sigma_\Sigma}}\right)(\ho_{\Gamma_{\sigma_\Sigma}}(\Gamma_{\sigma_\Sigma}))\\
&=\left(\beta_{\tilde\Gamma}f_{\tilde\Gamma}\right)(\ho_{\tilde\Gamma}(\Gamma_{\sigma}))=\left(\beta_{\tilde\Gamma}(\zeta_{\sigma} f_{\tilde\Gamma})\right)(\ho_{\tilde\Gamma}(\Gamma))\\
&= \left(\beta_{\tilde\Gamma}f_{\tilde\Gamma}\right)(R_{\sigma}(\ho_{\tilde\Gamma}(\Gamma)))\\ 
\eqs whenever $\PD_{\Gamma}\leq \PD_{\Gamma_{\sigma_\Sigma}}\leq \PD_{\tilde\Gamma}$, for a function $f \in C(\Ab)$,
where $\beta_{\tilde\Gamma}:C(\Ab_{\tilde\Gamma})\longrightarrow C(\Ab)$, $\beta_{\Gamma_{\sigma_\Sigma},\tilde\Gamma}:C(\Ab_{\Gamma_{\sigma_\Sigma}})\longrightarrow C(\Ab_{\tilde\Gamma})$ are unit-preserving injective $^*$-homomorphisms satisfying consistency conditions and
for a global bisections $\sigma_\Sigma\in\mathfrak{B}(\PD_{\Gamma_\infty})$ such that for a bisection $\sigma\in\mathfrak{B}(\PD_{\tilde\Gamma})$ on a finite graph system $\PD_{\tilde\Gamma}$ it is true that $\sigma_\Sigma(V_{\tilde\Gamma})=\sigma(V_{\tilde\Gamma})$.
\end{defi}
The limit Hilbert space $\HS_\infty$ with norm $\|.\|_\infty$ is constructed from the inductive family of Hilbert spaces $\HS_\Gamma$.
 
But these actions related to the limit graph system $\PD_{\Gamma_\infty}$ are not norm-point continuous. This is proved by the following argument. Since from $\ho_\Gamma$ is not a continuous groupoid morphism between $\PG$ to $G$ over $\{e_G\}$ it follows that,
\beqs&\lim_{\sigma_\Sigma(\Sigma)\rightarrow \id(\Sigma)}
\|\zeta_{\sigma_\Sigma}(f )- f \|_{\infty}
=\lim_{\sigma_\Sigma(\Sigma)\rightarrow \id(\Sigma)}
\|\beta_{\tilde\Gamma}(\zeta_{\sigma}f_{\tilde\Gamma})- f \|_{\infty}
\\&= \lim_{\sigma_\Sigma(\Sigma)\rightarrow \id(\Sigma)}
\|(\beta_{\tilde\Gamma}f_{\tilde\Gamma})(\ho_{\tilde\Gamma}(\Gamma)\ho_{\tilde\Gamma}(\sigma(V^t)),\ho_{\tilde\Gamma}(\sigma(V))) - (\beta_{\tilde\Gamma}f_{\tilde\Gamma})(\ho_{\tilde\Gamma}(\Gamma))\|_{\infty}
\neq 0
\eqs yields for a function $f \in C(\Ab)$, a subgraph $\Gamma:=\{\gamma_1,...,\gamma_N\}$ of $\tilde\Gamma$, a subset $V_\Gamma:=V^t\cup V$ of $V_{\tilde\Gamma}$ where $V^t:=\{t(\gamma_1),...,t(\gamma_N)\}$ and $N=\vert\Gamma\vert$, a global bisection $\sigma_\Sigma\in \mathfrak{B}(\PD_{\Gamma_\infty})$ such that $\sigma_\Sigma(V_\Gamma)=(\tilde\sigma_\Sigma(v_1),...,\tilde\sigma_\Sigma(v_{2N}))$ where $\tilde\sigma_\Sigma\in\mathfrak{B}(\PD)$ and there is a bisection $\sigma\in\mathfrak{B}(\PD_{\tilde\Gamma})$ such that $\sigma(V_{\tilde\Gamma})=\sigma_\Sigma(V_{\tilde\Gamma})$ yields.
Since there is an group morphism between $\mathfrak{B}(\PD)$ and the group of diffeomorphisms $\Diff(\Sigma)$ on the spatial manifold $\Sigma$, the diffeomorphism cannot be implemented as strongly or weakly continuous representations on the limit Hilbert space $\HS_{\infty}$. Nevertheless the action $\zeta$ of $\mathfrak{B}(\PD_{\Gamma_\infty})$ on $C(\Ab)$ is automorphic. Denote the set of automorphic actions of a group $\mathfrak{B}(\PD_{\Gamma_\infty})$ on the commutative $C^*$-algebra $C(\Ab)$ by $\Act_0(\mathfrak{B}(\PD_{\Gamma_\infty}),C(\Ab))$.

Despite the discontinuity of the action on the inductive limit of $C^*$-algebras $C(\Ab)$, there are injective $^*$-homomorphisms $\beta_{\Gamma,\Gamma_\sigma}$ such that
\beqs f (\ho (\Gp))=\left((\beta_{\tilde\Gamma}\circ\beta_{\Gamma_\sigma,\tilde\Gamma})f_{\Gamma_\sigma}\right)(\ho_{\Gamma_\sigma}(\Gp))
\eqs yields for any graphs such that $\PD_{\Gp}\leq\PD_{\Gamma_\sigma}\leq\PD_{\tilde\Gamma}$.

In LQG literature the set of surfaces is not restricted, the set $\breve S$ is an infinite set of surfaces and the inductive limit of a family of graph systems is constructed from a limit of a family of graphs. In this article the infinite set of surfaces is decomposed into several finite sets. To implement an action of $\bar G_{\breve S,\Gamma}$ the inductive limit structure of graphs has to preserve the particular sort of the action for a fixed suitable surface set $\breve S$. 

\begin{prop}Let $\Gamma_\infty$ be the inductive limit of a family of graphs $\{\Gamma_i\}$ such that the set $\breve S$ of surfaces has the same surface intersection property for each graph $\Gamma_i$ of the family. Let $\check S$ be a suitable surface set with same right surface intersection property for each graph $\Gamma_i$ of the family.
Then $\PD_{\Gamma_\infty}^{\op}$ is the inductive limit of a inductive family $\{\PD_{\Gamma_i}^{\op}\}$ of finite orientation preserved graph systems.

Then there is an action of $\bar G_{\breve S,\Gamma_\infty}$ on $C(\Ab)$ given by
\beqs (\alpha(\rho_{S,\Gamma_\infty}(\Gamma))f )(\ho (\Gamma))
&:= f (L(\rho_{S^\prime,\Gamma_\infty}(\Gamma))(\ho (\Gamma)))\\
&=(\beta_{\tilde\Gamma}\circ\alpha(\rho_{S,\tilde\Gamma}(\Gamma))f_{\tilde\Gamma})(\ho_{\tilde\Gamma}(\Gamma))=(\beta_{\tilde\Gamma}f_{\tilde\Gamma})(L(\rho_{S,\tilde\Gamma}(\Gamma))(\ho_{\tilde\Gamma}(\Gamma)))
\eqs for $\PD_{\Gamma}^{\op}\leq\PD_{\tilde\Gamma}^{\op}\leq\PD_{\Gamma_\infty}^{\op}$, injective unit-preserving $^*$-homomorphism $\beta_{\tilde\Gamma}:C(\Ab_{\tilde\Gamma})\longrightarrow C(\Ab)$ satisfying consistency conditions, elements $\rho_{S,\Gamma_\infty}(\Gamma)\in \bar G_{\breve S,\Gamma_\infty}$ and there are element $\rho_{S,\tilde\Gamma}(\Gamma)\in \bar G_{\breve S,\tilde\Gamma}$ such that $\rho_{S,\tilde\Gamma}(\Gamma)=\rho_{S,\Gamma_\infty}(\Gamma)$ for all $\Gamma \in\PD_{\tilde\Gamma}$ and every $S\in\breve S$.

There is another action of $\bar G_{\check S,\Gamma_\infty}$ on $C(\Ab)$ defined by
\beqs (\alpha(\rho_{S^\prime,\Gamma_\infty}(\Gamma))f )(\ho (\Gamma))
&:= f (R(\rho_{S^\prime,\Gamma_\infty}(\Gamma))(\ho (\Gamma)))\\
&=(\beta_{\tilde\Gamma}\circ\alpha(\rho_{S^\prime,\tilde\Gamma}(\Gamma))f_{\tilde\Gamma})(\ho_{\tilde\Gamma}(\Gamma))=(\beta_{\tilde\Gamma}f_{\tilde\Gamma})(R(\rho_{S^\prime,\tilde\Gamma}(\Gamma))(\ho_{\tilde\Gamma}(\Gamma)))
\eqs for $\PD_{\Gamma}^{\op}\leq\PD_{\tilde\Gamma}^{\op}\leq\PD_{\Gamma_\infty}^{\op}$, injective unit-preserving $^*$-homomorphism $\beta_{\tilde\Gamma}:C(\Ab_{\tilde\Gamma})\longrightarrow C(\Ab)$ satisfying consistency conditions, elements $\rho_{S^\prime,\Gamma_\infty}(\Gamma)\in \bar G_{\check S,\Gamma_\infty}$ and there are elements $\rho_{S^\prime,\tilde\Gamma}(\Gamma)\in \bar G_{\check S,\tilde\Gamma}$ such that $\rho_{S^\prime,\tilde\Gamma}(\Gamma)=\rho_{S^\prime,\Gamma_\infty}(\Gamma)$ for all $\Gamma \in\PD_{\tilde\Gamma}$ and every surface $S^\prime\in\check S$.

These actions of the flux group $\bar G_{\breve S,\Gamma_\infty}$ for the surface set $\breve S$ and $\bar G_{\check S,\Gamma_\infty}$ for surfaces in $\check S$ on $C(\Ab)$ are  automorphic and point-norm continuous.
\end{prop}
\begin{proofs}
The point-norm continuity follows from the observation that 
\beqs&\lim_{\rho_{S,\Gamma_\infty}(\Gamma)\rightarrow \id_{S,\Gamma_\infty}(\Gamma)}
\big\|\alpha(\rho_{S,\Gamma_\infty}(\Gamma) (f )- f \big\|_{\sup}\\
&=\lim_{\rho_{S,\Gamma_\infty}(\Gamma)\rightarrow \id_{S,\Gamma_\infty}(\Gamma)}
\Big\|\left(\beta_{\tilde\Gamma}(\alpha(\rho_{S,\tilde\Gamma}(\Gamma))f_{\tilde\Gamma})\right)- \beta_{\tilde\Gamma}(f_{\tilde\Gamma})\Big\|_{\sup}\\
&=0
\eqs holds whenever $\PD_{\Gamma}\leq\PD_{\tilde\Gamma}$, $\rho_{S,\Gamma_\infty}(\Gamma)\in\bar G_{\breve S,\Gamma_{\infty}}$ and $\id_{S,\Gamma_\infty}(\Gamma)\in\bar G_{\breve S,\Gamma_{\infty}}$, which is defined by $\id_{S,\Gamma_\infty}(\Gamma)=(\id_S(\gamma_1),...,\id_S(\gamma_N))=(e_G,..,e_G)$ for a graph $\Gamma:=\{\gamma_1,...,\gamma_N\}$. 
\end{proofs}

\begin{defi}
Let $\Gamma_\infty$ be the inductive limit of a family of graphs $\{\Gamma_i\}$ such that the set $\breve S$ of surfaces has the surface intersection property for each graph $\Gamma_i$ of the family. 
Then $\PD_{\Gamma_\infty}$ is the inductive limit of a inductive family $\{\PD_{\Gamma_i}\}$ of finite graph systems.

Let $(\varphi,\Phi)\in\Diff(\PD)$ be a path-diffeomorphism of a path groupoid $\PG$ such that
\begin{itemize}
 \item $\varphi:\Sigma\longrightarrow \Sigma$, which leave each surface in $\breve S$ and a suitable neighborhood of each surface in $\breve S$ invariant and $\Phi:\PD\longrightarrow\PD$;
\item if a path $\gamma$ in $\PD$ does not intersect all surfaces, then $\Phi(\gamma)$ does not intersect all surfaces and
\item the number of all generators $\{\gamma_j\}$ of $\Gamma_i$ and the number of all transformed paths $\{\Phi(\gamma_j)\}$ that intersect each surface in $\breve S$ in their target vertices are constant and equal
\end{itemize}
is called a \textbf{surface-preserving path-diffeomorphism for a path groupoid} $\PG$ and a surface set $\breve S$.

Then a \textbf{surface-preserving graph-diffeomorphism for a limit graph system} $\PD_{\Gamma_\infty}$ is given by the pair $(\varphi_\Sigma,\Phi_\infty)$ of maps such that 
 \begin{itemize}
  \item $\varphi_\Sigma:\Sigma\longrightarrow \Sigma$, $\Phi_{\infty}:\PD_{\Gamma_\infty}\longrightarrow\PD_{\Gamma_\infty}$ and 
\beqs \Phi_{\infty}(\Gamma)= (\Phi(\gamma_1),...,\Phi(\gamma_N))=\Gamma_{\Phi}
\eqs for $\Gamma:=\{\gamma_1,...,\gamma_N\}$ and $\Gamma_{\Phi}$ being two subgraphs of $\PD_{\Gamma_\infty}$ and
 \item $(\varphi_\Sigma,\Phi)$ is a surface-preserving path-diffeomorphism for a path groupoid $\PG$ and a surface set $\breve S$. 
 \end{itemize}

The set of surface-preserving graph-diffeomorphism for a limit graph system $\PD_{\Gamma_\infty}$ is denoted by $\Diff_{\diff}(\PD_{\Gamma_\infty})$.
\end{defi}

With no doubt the group $\mathfrak{B}_{\diff}(\PD_{\Gamma_\infty})$ of surface-preserving bisections of a limit graph system $\PD_{\Gamma_\infty}$ can be defined, too.

\begin{defi}
Let $\Gamma_\infty$ be the inductive limit of a family of graphs $\{\Gamma_i\}$ such that the set $\breve S$ of surfaces has the simple surface intersection property for each graph $\Gamma_i$ of the family. 
Then $\PD_{\Gamma_\infty}^{\op}$ is the inductive limit of a inductive family $\{\PD_{\Gamma_i}^{\op}\}$ of finite orientation preserved graph systems.

Let $(\varphi,\Phi)\in\Diff(\PD)$ be a path-diffeomorphism of a path groupoid $\PG$ such that
\begin{itemize}
 \item $\varphi:\Sigma\longrightarrow \Sigma$ such that each surface $S$ in $\breve S$ is mapped to another surface $S_\sigma$ in $\breve S$ and $\Phi:\PD\longrightarrow\PD$;
\item if a path $\gamma$ in $\PD$ does not intersect a surface in $\breve S$, then $\Phi(\gamma)$ does not intersect a surface in $\breve S$ and
\item if a path intersects a surface $S$, lies below and is outgoing (or above and outgoing, below and ingoing, above and ingoing)  and the transformed path $\Phi(\gamma)$ is non-trivial, then $\Phi(\gamma)$ intersects the transformed surface $S_\sigma$, lies below and is outgoing (or above and outgoing, below and ingoing, above and ingoing), too, 
\end{itemize}
is called a \textbf{surface-orientation-preserving path-diffeomorphism for a path groupoid} $\PG$ and a surface set $\breve S$.

Then a \textbf{surface-orientation-preserving graph-diffeomorphism for a limit orientation preserved graph system} $\PD_{\Gamma_\infty}^{\op}$ is given by the pair $(\varphi_\Sigma,\Phi_\infty)$ of maps such that 
 \begin{itemize}
  \item $\varphi_\Sigma:\Sigma\longrightarrow \Sigma$, $\Phi_{\infty}:\PD_{\Gamma_\infty}^{\op}\longrightarrow\PD_{\Gamma_\infty}^{\op}$ and 
\beqs \Phi_{\infty}(\Gamma)= (\Phi(\gamma_1),...,\Phi(\gamma_N))=\Gamma_{\Phi}
\eqs for $\Gamma:=\{\gamma_1,...,\gamma_N\}$ and $\Gamma_{\Phi}$ being two subgraphs of $\PD_{\Gamma_\infty}^{\op}$ and
 \item $(\varphi_\Sigma,\Phi)$ is a surface-orientation-preserving path-diffeomorphism for a path groupoid $\PG$ and a surface set $\breve S$. 
 \end{itemize}

The set of surface-preserving graph-diffeomorphism for a orientation preserved limit graph system $\PD_{\Gamma_\infty}^{\op}$ is denoted by $\Diff_{\ori}(\PD_{\Gamma_\infty}^{\op})$.
\end{defi}
With no doubt the group $\mathfrak{B}_{\ori}(\PD_{\Gamma_\infty}^{\op})$ of surface-orientation-preserving bisections of a limit orientation preserved graph system $\PD_{\Gamma_\infty}^{\op}$ can be defined.

Apart from the problems of defining a point-norm continuous action of the group of global bisections of the inductive limit holonomy $C^*$-algebra, the Weyl algebra can be realized as inductive limit $C^*$-algebra, too.

\begin{defi}
Let $\breve S$ be a finite set of surfaces in $\Sigma$.
Moreover let $\Gamma_\infty$ be the inductive limit of a family of graphs $\{\Gamma_i\}$ such that each graph $\Gamma_i$ of the family has the surface intersection property for the set $\breve S$ of surfaces. 
Then $\PD_{\Gamma_\infty}$ is the inductive limit of a inductive family $\{\PD_{\Gamma_i}\}$ of finite graph systems. 

The \textbf{Weyl $C^*$-algebra $\WF\text{eyl}(\breve S)$ for a surface set} is generated by the inductive limit $C(\Ab)$ of the family of $C^*$-algebras $\{(C(\Ab_{\Gamma_i}),\beta_{\Gamma_i,\Gamma_j}):\PD_{\Gamma_i}\leq \PD_{\Gamma_j}, i,j\in\N\}$ and the Weyl elements $\mathsf{W}(\bar G_{\breve S,\Gamma_\infty})$, which satisfy the canonical commutator relations \eqref{cancomrel I}, completed w.r.t. the $\|.\|_{\infty}$-norm defined by the Hilbert space $\HS_{\infty}$, which is given as the limit of the Hilbert spaces $\HS_\Gamma$. 
\end{defi}

\begin{defi}
Let $\Gamma_\infty$ be the inductive limit of a family of graphs $\{\Gamma_i\}$ and let $\PD_{\Gamma_\infty}$ be the inductive limit of a inductive family $\{\PD_{\Gamma_i}\}$ of finite graph systems associated to the family of graphs $\{\Gamma_i\}$. Let $\surf$ and $\surf_\ZD$ be two sets of suitable surfaces for each graph $\Gamma_i$ of the family of graphs $\{\Gamma_i\}$. 

Moreover $\PD_{\Gamma_\infty}^{\op}$ is the inductive limit of a inductive family $\{\PD_{\Gamma_i}^{\op}\}$ of finite orientation preserved graph systems.

The \textbf{Weyl $C^*$-algebra $\WF\text{eyl}(\surf)$ for surfaces} is generated by the inductive limit $C(\Ab)$ of the family of $C^*$-algebras $\{(C(\Ab_{\Gamma_i}),\beta_{\Gamma_i,\Gamma_j}):\PD_{\Gamma_i}\leq \PD_{\Gamma_j}, i,j\in\N\}$ and the Weyl elements $\mathsf{W}(\bar G_{\breve S,\Gamma_\infty})$ for each surface set $\breve S$ in $\surf$, which satisfies the canonical commutator relations \eqref{cancomrel I}, completed w.r.t. the $\|.\|_{\infty}$-norm defined by the Hilbert space $\HS_{\infty}$, which is given as the limit of the Hilbert spaces $\HS_\Gamma$. 

The \textbf{commutative Weyl $C^*$-algebra $\WF\text{eyl}_\ZD(\surf_\ZD)$ for surfaces} is generated by the inductive limit of the family of $C^*$-algebras $\{(C(\Ab_{\Gamma_i}),\beta_{\Gamma_i,\Gamma_j}):\PD_{\Gamma_i}\leq \PD_{\Gamma_j}, i,j\in\N\}$ and the Weyl elements $\mathsf{W}(\bar \ZD_{\breve S})$ for each surface set $\breve S$ in $\surf_\ZD$, which satisfies the canonical commutator relations \eqref{cancomrel I}, completed w.r.t. the $\|.\|_{\infty}$-norm defined by the Hilbert space $\HS_{\infty}$.
\end{defi}

\subsection{Flux and graph-diffeomorphism group-invariant states of the Weyl $C^*$-algebra for surfaces}\label{subsec invstate}

Consider the inductive limit algebra $C(\Ab)$ of the family of $C^*$-algebras\\ $\{(C(\Ab_{\Gamma_i}),\beta_{\Gamma_i,\Gamma_j}):\PD_{\Gamma_i}\leq \PD_{\Gamma_j}, i,j\in\N\}$, where $C(\Ab_\Gamma)$ is isomorphic to $C(G^N)$ by the natural or non-standard identification. 

\begin{prop}\label{prop limit invstate for holalg}Let $\breve S$ be a finite set of surfaces in $\Sigma$. Moreover let $\Gamma_\infty$ be the inductive limit of a family of graphs $\{\Gamma_i\}$ such that each graph $\Gamma_i$ of the family has the surface intersection property for the set $\breve S$ of surfaces.  
Then $\PD_{\Gamma_\infty}$ is the inductive limit of a inductive family $\{\PD_{\Gamma_i}\}$ of finite graph systems. Let $\Ab_{\Gamma_i}$ be identified naturally with $G^{\vert \Gamma_i\vert}$. 

The limit $\hat\omega_{\mathfrak{B}_\Sigma}$ on $C(\Ab)$ is defined by
\beqs \hat\omega_{\mathfrak{B}_\Sigma}(f)
&:=\lim_{\Gamma_i\rightarrow\Gamma_\infty}\frac{1}{k_{\Gamma_i}}\sum_{l=1}^{k_{\Gamma_i}} \omega^{\Gamma_i}_{M}(\zeta_{\sigma_l}(f_{\Gamma_i}))\text{ for }\sigma_l\in\mathfrak{B}^{\Gamma_i}_{\breve S,\diff}(\PD_{\Gamma_i})
\eqs for $f\in C(\Ab)$ and where $k_{\Gamma_i}$ is the maximal number of subgraphs in $\PD_{\Gamma_i}$, which are generated by all edges and their compositions of the graph $\Gamma_\infty$. The limit $\hat\omega_{\mathfrak{B}_\Sigma}$ is $\mathfrak{B}_{\breve S,\diff}(\PD_{\Gamma_\infty})$-invariant and does not converge in weak $^*$-topology.
\end{prop}
\begin{proofs}
There are two disjoint families of graph $\{\Gamma^\prime_i\}$ and $\{\Gamma_i\}$ such that the union converges to $\Gamma_\infty$ and such that for a suitable constants $k_{\Gamma^\prime_i}$ it is true that
\beqs  &\lim_{\Gamma_i\longrightarrow\Gamma_\infty}\Big\vert\hat\omega_{\mathfrak{B}_\Sigma}(f)-\frac{1}{k_{\Gamma_i}}\sum_{l=1}^{k_{\Gamma_i}} \omega^{\Gamma_i}_{M}(\zeta_{\sigma_l}(f_\Gamma)) \Big\vert
> \lim_{\Gamma^\prime_i\longrightarrow\Gamma_\infty}\Big\vert \frac{1}{k_{\Gamma^\prime_i}}\sum_{l=1}^{k_{\Gamma^\prime_i}} \omega^{\Gamma^\prime_i}_{M}(\zeta_{\sigma_l}(f_{\Gamma^\prime_i})) \Big\vert >0 
\eqs
\end{proofs}
If there would be a asymptotic condition for the state such that for the limit to the infinite graph the state does not depend on the action of $\zeta$ anymore, then the state defined above would be weakly converging.

Consequently if the natural identication of $\Ab_\Gamma$ with $G^N$ is used, then there is no state, which is $\bar\ZD_{S,\Gamma_\infty}$-, $\mathfrak{B}_{\breve S,\ori}(\PD_{\Gamma_\infty})$- and $\mathfrak{B}_{\breve S,\diff}(\PD_{\Gamma_\infty})$-invariant on $C(\Ab)$. 

Let each finite graph system $\PD_{\Gamma_i}$ be identified naturally or in the non-standard way, then the following observations are made. 
\begin{prop}\label{cor_mean}Let $\breve S$ be a finite set of surfaces in $\Sigma$. Moreover let $\Gamma_\infty$ be the inductive limit of a family of graphs $\{\Gamma_i\}$ such that each graph $\Gamma_i$ of the family has the surface intersection property for the set $\breve S$ of surfaces.  
Then $\PD_{\Gamma_\infty}$ is the inductive limit of a inductive family $\{\PD_{\Gamma_i}\}$ of finite graph systems. 

There is a $\bar G_{\breve S,\Gamma_\infty}$-invariant state $\omega_M$ on $C(\Ab)$ presented by
\beqs \omega_M(f)
&=\int_{G^{N_i}}f_{\Gamma_i}(\ho_{\Gamma_i}(\Gamma^\prime_i))
\dif\mu_{\Gamma_i}(\ho_{\Gamma_i}(\Gamma^\prime_i))\\
\eqs for all $f \in C(\Ab)$, $\PD_{\Gamma^\prime_i}\leq\PD_{\Gamma_i}$, which satisfies
\beqs \omega_M\circ\beta_{\Gamma_i}=\omega^{\Gamma_i}_M
\eqs where $\beta_{\Gamma_i}: C(\Ab_{\Gamma_i})\rightarrow C(\Ab)$ is an injective $^*$-homomorphism.

Moreover there is a state on $C(\Ab)$ given by
\beqs\omega_{\mathfrak{B}}(f)=(\omega_{\mathfrak{B}}\circ \beta_{\Gamma_i})(f)
&:=\frac{1}{k_{\Gamma_i}}\sum_{l=1}^{k_{\Gamma_i}} \omega_{M}^{\Gamma_i}(\zeta_{\sigma_l}(f_{\Gamma_i}))
\eqs  for $\sigma_l\in\mathfrak{B}^{\Gamma_i}_{\breve S,\ori}(\PD_{\Gamma_i})$ and $f_{\Gamma_i}\in C(\Ab_{\Gamma_i})$,
which is invariant under the automorphic actions of the groups $\Diff_{\breve S,\ori}(\PD_{\Gamma_i})$, $\mathfrak{B}_{\breve S,\ori}(\PD_{\Gamma_i})$ for a fixed graph $\Gamma_i$ and $\bar \ZD_{\breve S}$ for a suitable set $\breve S$ of surfaces in $\surf_\ZD$. 

In other words,
\beqs \omega_{\mathfrak{B}}(\theta_{(\varphi_{\Gamma_i},\Phi_{\Gamma_i})} f)
&=\omega_{\mathfrak{B}}(f),\quad
\omega_{\mathfrak{B}}(\zeta_{\sigma}f)
=\omega_{\mathfrak{B}}(f),\\
\omega_{\mathfrak{B}}(\alpha(\rho_{S,\Gamma_\infty}(\Gamma^\prime_i))(f))
&=\omega_{\mathfrak{B}}(f)
\eqs yields for all $f\in C(\Ab)$, $(\varphi_{\Gamma_i},\Phi_{\Gamma_i})\in \Diff_{\breve S,\ori}(\PD_{\Gamma_i})$, $\theta\in\Act_0(\Diff_{\breve S,\ori}(\PD_{\Gamma_i}),C(\Ab))$, $\sigma\in \mathfrak{B}_{\breve S,\ori}(\PD_{\Gamma_i})$,\\ $\zeta\in\Act_0(\mathfrak{B}_{\breve S,\ori}(\PD_{\Gamma_i}),C(\Ab))$, $\rho_{S,\Gamma_\infty}(\Gamma^\prime_i)\in\bar\ZD_{\breve S}$, $\alpha\in\Act(\bar\ZD_{\breve S},C(\Ab))$ for any surface set $\breve S$ in $\surf_\ZD$.

Furthermore the state $\omega_{\mathfrak{B}}$ and the actions $\alpha\in\Act(\bar \ZD_{\breve S},C(\Ab))$ and $\zeta\in\Act(\mathfrak{B}_{\breve S,\ori}(\PD_{\Gamma_i}),C(\Ab))$ satisfy
\beq\label{eq commactions2}\omega_{\mathfrak{B}}^{\Gamma_i}\circ\alpha(\rho_{S,\Gamma_i}(\Gp))\circ\zeta_\sigma
=\omega_{\mathfrak{B}}^{\Gamma_i}\circ\zeta_\sigma\circ\alpha(\rho_{S,\Gamma_i}(\Gp))
\eq 
\end{prop}
Note that, the state $\omega_{\mathfrak{B}}$ is even invariant under a graph-diffeomorphism $(\varphi_{\Gamma_i},\Phi_{\Gamma_i})$ such that there exists a diffeomorphism $\varphi:\Sigma\rightarrow\Sigma$ that maps surfaces into surfaces in $\breve S$ and $\varphi(v)=\varphi_{\Gamma_i}(v)$ for all $v\in V_{\Gamma_i}$. In the following only graph-diffeomophisms in $\Diff(\PD_\Gamma)$ and consequently also the induced bisections of $\mathfrak{B}(\PD_\Gamma)$, which satisfy this requirement are considered. Therefore the restricted sets are denoted by $\Diff_{\breve S}(\PD_\Gamma)$ and $\mathfrak{B}_{\breve S}(\PD_\Gamma)$.

\begin{proofs}
Recall the inductive limit $C(\Ab)$ of the $C^*$-algebras $\{(C(\Ab_{\Gamma_i}),\beta_{\Gamma_i,\Gamma_j}):\PD_{\Gamma_i}\leq \PD_{\Gamma_j}, i,j\in\N\}$ where $\beta_{\Gamma_i,\Gamma_j}$ is an injective $^*$-homomorphism satisfying $\beta_{\Gamma_i,\Gamma_j}=\beta_{\Gamma_i,\Gamma_k}\circ\beta_{\Gamma_k,\Gamma_j}$ whenever $\PD_{\Gamma_i}\leq\PD_{\Gamma_k}\leq\PD_{\Gamma_j}$ for all $i,k,j\in\N$.

There is a state $\omega_{M}^\Gamma$ on $C(\Ab_\Gamma)$, which is $\bar G_{\breve S,\Gamma}$-invariant, and a state $\omega^\Gamma_{\mathfrak{B}}$, which is $\bar \ZD_{\breve S,\Gamma}$- and $\mathfrak{B}_{\breve S,\ori}(\PD_{\Gamma})$-invariant due to proposition \ref{prop invstate} and \ref{prop invstate for holalg}. Every inductive limit of $C^*$-algebras corresponds to a projective limit of the state space of the $C^*$-algebras. Hence there are conjugate maps $\beta^*_{\Gamma,\Gp}:C(\Ab_\Gamma)\longrightarrow C(\Ab_{\Gp})$ such that $\omega_{M}^{\Gp}=\beta^*_{\Gamma,\Gp}\omega_{M}^{\Gamma}$ and $\omega_{M}^{\Gp}=\beta^*_{\Gamma,\Gp}\omega_{M}^{\Gamma}$. Denote the projective limit state of $\{(\omega_{M}^{\Gp},\beta^*_{\Gamma,\Gp}):\PD_{\Gamma}\leq \PD_{\Gp}\}$ by $\omega_{M}$ on $C(\Ab)$ respectively $\{(\omega_{\mathfrak{B}}^{\Gp},\beta^*_{\Gamma,\Gp}):\PD_{\Gamma}\leq \PD_{\Gp}\}$ by $\omega_{\mathfrak{B}}$ on $C(\Ab)$.  
 
Then the state on $C(\Ab)$ satisfies
\beqs 
\omega_{M}(f )
  & = \beta^*_{\Gamma}(\omega^\Gamma_{M}(f_\Gamma))
  =(\beta^*_{\Gp}\circ\beta^*_{\Gamma,\Gp})(\omega^\Gamma_{M}(f_\Gamma))
  =\beta^*_{\Gp}(\omega^{\Gp}_{M}(f_{\Gp}))\\
  &=\int_{G^{N}}f_{\Gamma}(\ho_{\Gamma}(\Gamma)) \dif\mu_{\Gamma}(\ho_{\Gamma}(\Gamma))\\
  &=\int_{G^{N^\prime}}f_{\Gp}(L(\rho_{S,\Gp}(\Gp))(\ho_{\Gp}(\Gp)))
  \dif\mu_{\Gp}(\ho_{\Gp}(\Gp))\\
  &=\int_{G^{N^{\prime\prime}}}f_{\Gpp}(R(\rho_{S^\prime,\Gpp}(\Gpp))(\ho_{\Gpp}(\Gpp)))
  \dif\mu_{\Gpp}(\ho_{\Gpp}(\Gpp))
\eqs for suitable surface $S$ and $S^\prime$, graphs $\Gamma,\Gp$ and $\Gpp$, maps $\rho_{S,\Gp}\in G_{\breve S,\Gp}$ and $\rho_{S^\prime,\Gpp}\in G_{\breve S^\prime,\Gpp}$.  
\end{proofs}

Notice that, the state $\omega_{\mathfrak{B}}$ is only invariant under the group $\bar\ZD_{\breve S}$. This follows from the fact that, the action $\zeta$ for $\mathfrak{B}_{\breve S,\ori}(\PD_\Gamma)$ and the action $\alpha$ for $\bar G_{\breve S,\Gamma_\infty}$ do not commute.

\begin{cor}\label{cor uniqueholalg}Let $\breve S$ be a finite set of surfaces in $\Sigma$. Moreover let $\Gamma_\infty$ be the inductive limit of a family of graphs $\{\Gamma_i\}$.  
Then $\PD_{\Gamma_\infty}$ is the inductive limit of a inductive family $\{\PD_{\Gamma_i}\}$ of finite graph systems. 
Let $\Ab_\Gamma$ be identified in the non-standard identification with $G^{\vert\Gamma\vert}$.

Then the state $\omega_M$ defined by
\beqs \omega_M(f)
&=\int_{G^{N_i}}f_{\Gamma_i}(\ho_{\Gamma_i}(\Gamma^\prime_i))
\dif\mu_{\Gamma_i}(\ho_{\Gamma_i}(\Gamma^\prime_i))\\
\eqs for all $f\in C(\Ab)$, $\PD_{\Gamma^\prime_i}\leq\PD_{\Gamma_i}$, which satisfies
\beqs \omega_M\circ\beta_{\Gamma_i}=\omega^{\Gamma_i}_M
\eqs
is the unique state on $C(\Ab)$, which is invariant under the automorphic actions of the groups $\Diff_{\breve S}(\PD_{\Gamma_i})$, $\mathfrak{B}_{\breve S}(\PD_{\Gamma_i})$ for each graph $\Gamma_i$ and $\bar \ZD_{\breve S}$.
\end{cor}

Let each finite graph system $\PD_{\Gamma_i}$ be identified naturally or in the non-standard way.
\begin{prop}\label{cor_mean2}Let $\breve S$ be a finite set of surfaces in $\Sigma$. Moreover let $\Gamma_\infty$ be the inductive limit of a family of graphs $\{\Gamma_i\}$ such that each graph $\Gamma_i$ of the family has the surface intersection property for the set $\breve S$ of surfaces\footnote{This condition is necessary, since otherwise $\bar G_{\breve S,\Gamma_i}$ doesn't form a group.}.  
Then $\PD_{\Gamma_\infty}$ is the inductive limit of a inductive family $\{\PD_{\Gamma_i}\}$ of finite graph systems. 

There is a GNS-representation $(\HS_{\infty},\Phi,\Omega_{M})$ of $\WF\text{eyl}(\breve S)$ on $\HS_{\infty}$ of the pure and unique state $\bar\omega_M$ on $\WF\text{eyl}(\breve S)$, which is given by
\beqs 
\bar\omega_{M}(f)&=(\omega^\Gamma_{M}\circ\beta_{\Gamma})(f )= \la\Omega_{M},\Phi_M(f )\Omega_{M}\ra\\
&=\int_{G^{N_i}}f_{\Gamma_i}(\ho_{\Gamma_i}(\Gamma_i))
\dif\mu_{\Gamma_i}(\ho_{\Gamma_i}(\Gamma_i))\\
\bar\omega_{M}(U^*(\rho_{S,\Gamma}(\Gamma))U(\rho_{S,\Gamma}(\Gamma)))
&=\la\Omega_{M},\Phi(U^*(\rho_{S,\Gamma}(\Gamma))U(\rho_{S,\Gamma}(\Gamma)))\Omega_{M}\ra
=\bar\omega_{M}(\idf)=1
\eqs  where $\idf$ is the identity on $\HS_{\infty}$, for $f \in C(\Ab)$ and $U(\rho_{S,\Gamma}(\Gamma))\in\mathsf{W}(\bar G_{\breve S,\Gamma_\infty})$ whenever $\Phi\in\Mor(\WF\text{eyl}(\breve S),\LD(\HS_{\infty}))$ and $\Phi_M:=\Phi\big\vert_{C(\Ab)}$.
The state $\bar\omega_{M}$ is invariant under the automorphic actions of the flux group $\bar G_{\breve S,\Gamma_\infty}$.
\end{prop}
\begin{proof} First use the proposition \ref{prop invariant state of weyl} for finite graph systems and the uniqueness of the construction of the limit of states on the inductive limit of $C^*$-algebras. 
Equivalently it can be shown that, the representation $\Phi_M\in\Rep(C(\Ab),\LD(\HS_{\infty}))$ extends uniquely to a representation $\Phi\in\Mor(\WF\text{eyl}(\breve S),\LD(\HS_{\infty}))$.
\end{proof}
Furthermore this proposition extends to sets of surface sets.
\begin{prop}
Let $\Gamma_\infty$ be the inductive limit of a family of graphs $\{\Gamma_i\}$ and let $\PD_{\Gamma_\infty}$ be the inductive limit of a inductive family $\{\PD_{\Gamma_i}\}$ of finite graph systems associated to the family of graphs $\{\Gamma_i\}$. Let $\surf$ be a set of suitable surfaces for each graph $\Gamma_i$ of the family of graphs $\{\Gamma_i\}$.

There is a GNS-representation $(\HS_{\infty},\Phi,\Omega_{M})$ of $\WF\text{eyl}(\surf)$ on $\HS_{\infty}$ of the pure and unique state $\bar\omega_M$ on $\WF\text{eyl}(\surf)$, which is given by
\beqs 
\bar\omega_{M}(f)&=(\omega^\Gamma_{M}\circ\beta_{\Gamma})(f )= \la\Omega_{M},\Phi_M(f )\Omega_{M}\ra\\
&=\int_{G^{N_i}}f_{\Gamma_i}(\ho_{\Gamma_i}(\Gamma_i))
\dif\mu_{\Gamma_i}(\ho_{\Gamma_i}(\Gamma_i))\\
\bar\omega_{M}(U^*(\rho_{S,\Gamma}(\Gamma))U(\rho_{S,\Gamma}(\Gamma)))
&=\la\Omega_{M},\Phi(U^*(\rho_{S,\Gamma}(\Gamma))U(\rho_{S,\Gamma}(\Gamma)))\Omega_{M}\ra
=\bar\omega_{M}(\idf)
\eqs  where $\idf$ is the identity on $\HS_{\infty}$, for $f \in C(\Ab)$ and $U(\rho_{S,\Gamma}(\Gamma))\in\mathsf{W}(\bar G_{\breve S,\Gamma_\infty})$ for every suitable surface set $\breve S$ in $\surf$ whenever $\Phi\in\Mor(\WF\text{eyl}(\surf),\LD(\HS_{\infty}))$ and $\Phi_M:=\Phi\big\vert_{C(\Ab)}$
and which is invariant under the automorphic actions of the flux group $\bar G_{\breve S,\Gamma_\infty}$ for each suitable surface set $\breve S$ in $\surf$. 
\end{prop}
Now the focus lies on graph-diffeomorphism invariant states on a Weyl $C^*$-algebra. Then the following theorem is stated.

\begin{theo}\label{theo uniquweylalg}Let $\breve S$ be a finite set of surfaces in $\Sigma$. Moreover let $\Gamma_\infty$ be the inductive limit of a family of graphs $\{\Gamma_i\}$ such that each graph $\Gamma_i$ of the family has the surface intersection property for the set $\breve S$ of surfaces.  
Then $\PD_{\Gamma_\infty}$ is the inductive limit of a inductive family $\{\PD_{\Gamma_i}\}$ of finite graph systems and each $\PD_{\Gamma_i}$ is identified in the non-standard way. 
 
The state $\bar\omega_M$ on $\WF\text{eyl}_\ZD(\breve S)$ given in proposition \ref{cor_mean2} is the unique state, which is invariant under the automorphic actions of the groups $\Diff_{\breve S}(\PD_{\Gamma_i})$, $\mathfrak{B}_{\breve S}(\PD_{\Gamma_i})$ for each graph $\Gamma_i$ and $\bar \ZD_{\breve S}$ and pure.

Furthermore the state $\bar\omega_M$ and the actions $\alpha\in\Act(\bar \ZD_{\breve S},\WF\text{eyl}_\ZD(\breve S))$ and $\zeta\in\Act(\mathfrak{B}_{\breve S}(\PD_{\Gamma_i}),\WF\text{eyl}_\ZD(\breve S))$ satisfy
\beq\label{eq commactions3}(\beta_{\Gamma_i}^*\bar\omega_{M}^{\Gamma_i})\circ\alpha(\rho_{S,\Gamma_i}(\Gp))\circ\zeta_\sigma
=(\beta_{\Gamma_i}^*\bar\omega_{M}^{\Gamma_i})\circ\zeta_\sigma\circ\alpha(\rho_{S,\Gamma_i}(\Gp))
\eq 
\end{theo}
Notice that, this theorem generalises to the non-standard identification and for a suitable set $\surf_\ZD$ of surface sets.

\begin{proofs}
This follows from corollary \ref{cor uniqueholalg} and proposition \ref{cor_mean2}. 
\end{proofs}
For the natural or non-standard and natural identification the theorem follows. 
\begin{theo}\label{prop invstateweylalg}Let $\breve S$ be a finite set of surfaces in $\Sigma$. Moreover let $\Gamma_\infty$ be the inductive limit of a family of graphs $\{\Gamma_i\}$.
Then $\PD_{\Gamma_\infty}$ is the inductive limit of a inductive family $\{\PD_{\Gamma_i}\}$ of finite graph systems. 

There is a GNS- representation $(\HS_{\infty},\Phi,\Omega_{M,\mathfrak{B}})$ of $\WF\text{eyl}_\ZD(\breve S)$ on $\HS_{\infty}$ of the state $\omega_{M,\mathfrak{B}}$ on $\WF\text{eyl}_\ZD(\breve S)$, which is given by
\beqs 
\omega_{M,\mathfrak{B}}(f_{\Gamma_i})
&:=\frac{1}{k_{\Gamma_i}}\sum_{l=1}^{k_{\Gamma_i}} \omega_{M}(\zeta_{\sigma_l}(f_{\Gamma_i}))\\
&= \la\Omega_{M,\mathfrak{B}},\Phi_M(f_{\Gamma_i})\Omega_{M,\mathfrak{B}}\ra\\
\omega_{M,\mathfrak{B}}(f )
&=\omega_M(f )\\
\omega_{M,\mathfrak{B}}(U^*(\rho_{S,\Gamma}(\Gp))U(\rho_{S,\Gamma}(\Gp)))
&=\la\Omega_{M,\mathfrak{B}},\Phi(U^*(\rho_{S,\Gamma}(\Gp))U(\rho_{S,\Gamma}(\Gp)))\Omega_{M,\mathfrak{B}}\ra
=\omega_{M,\mathfrak{B}}(\idf)
\eqs where $\idf$ is the identity on $\HS_{\infty}$, for $\sigma_l\in\mathfrak{B}^{\Gamma_i }_{\breve S,\diff}(\PD_{\Gamma_i})$, $f_{\Gamma_i}\in C(\Ab_{\Gamma_i})$, $f \in C(\Ab)$, $U(\rho_{S,\Gamma}(\Gp))\in\mathsf{W}(\bar\ZD_{S,\Gamma_\infty})$ whenever  $\Phi\in\Mor(\WF\text{eyl}_\ZD(\breve S),\LD(\HS_{\infty}))$ and $\Phi_M:=\Phi\big\vert_{C(\Ab)}$. 

The state $\omega_{M,\mathfrak{B}}$ on $\WF\text{eyl}_\ZD(\breve S)$ is $\Diff_{\breve S,\diff}(\PD_{\Gamma_i})$-, $\Diff_{\breve S,\ori}(\PD_{\Gamma_i})$-, $\mathfrak{B}_{\breve S,\diff}(\PD_{\Gamma_i})$-, $\mathfrak{B}_{\breve S,\ori}(\PD_{\Gamma_i})$-invariant for a fixed graph $\Gamma_i$ and $\bar \ZD_{\breve S}$-invariant.

In other words, it is true that
\beqs \omega_{M,\mathfrak{B}}(\theta_{(\varphi_{\Gamma_i},\Phi_{\Gamma_i})}(W))&=\omega_{M,\mathfrak{B}}(W),\quad
\omega_{M,\mathfrak{B}}(\zeta_{\sigma}(W))=\omega_{M,\mathfrak{B}}(W),\\
\omega_{M,\mathfrak{B}}(\alpha(\rho_{S,\Gamma_\infty}(\Gp))(W))&=\omega_{M,\mathfrak{B}}(W)
\eqs holds for all $W\in\WF\text{eyl}_\ZD(\breve S)$ and $\theta\in\Act_0(\Diff(\PD_{\Gamma_i}),\WF\text{eyl}_\ZD(\breve S))$,\\ $\zeta\in\Act_0(\mathfrak{B}_{\breve S,\diff}(\PD_{\Gamma_i}),\WF\text{eyl}_\ZD(\breve S))$ or $\zeta\in\Act_0(\mathfrak{B}_{\breve S,\ori}(\PD_{\Gamma_i}),\WF\text{eyl}_\ZD(\breve S))$ for each fixed graph $\Gamma_i$,\\ $\alpha\in\Act(\bar \ZD_{S,\Gamma_\infty},\WF\text{eyl}_\ZD(\breve S))$ for any surface $S\in\breve S$.

Finally the state $\omega_{M,\mathfrak{B}}$ and the actions $\alpha\in\Act(\bar \ZD_{\breve S},\WF\text{eyl}_\ZD(\breve S))$ and $\zeta\in\Act(\mathfrak{B}_{\breve S,\ori}(\PD_{\Gamma_i}),\WF\text{eyl}_\ZD(\breve S))$ satisfy
\beq\label{eq commactions3}(\beta_{\Gamma_i}^*\omega_{M,\mathfrak{B}}^{\Gamma_i})\circ\alpha(\rho_{S,\Gamma_i}(\Gp))\circ\zeta_\sigma
=(\beta_{\Gamma_i}^*\omega_{M,\mathfrak{B}}^{\Gamma_i})\circ\zeta_\sigma\circ\alpha(\rho_{S,\Gamma_i}(\Gp))
\eq 
\end{theo}

This theorem can be generalised to $\WF\text{eyl}_\ZD(\surf_\ZD)$, since this theorem is true for all surface sets in $\surf_\ZD$ .


\section*{Acknowledgements}
The work has been supported by the Emmy-Noether-Programm (grant FL 622/1-1) of the Deutsche Forschungsgemeinschaft.

\end{document}